\documentclass[a4paper,titlepage]{book}

\usepackage{hyperref}
\hypersetup{pdfauthor={Di Gioacchino},pdftitle={Euclidean correlations in combinatorial optimizatino problems: a statistical physics approach},%
	colorlinks=true,%
	urlcolor=magenta, linkcolor=cyan, citecolor=red,
}

\usepackage{amsmath}
\usepackage{amsfonts}
\usepackage{amssymb}
\usepackage{color}
\usepackage{physics}
\usepackage[pdftex]{graphicx}
\usepackage{epstopdf}
\usepackage{subcaption}
\usepackage{amsthm}

\usepackage{tikz,pgfplots,bbm,bm,mathrsfs,dsfont}
\usetikzlibrary{decorations,calc}

\newcommand{\plas}{\mathfrak{p}}
\newcommand{\Pad}{\mathop{\rm Pad}}

\newtheorem{lemma}{Lemma}

\newtheorem{pros}{Proposition}[section]

\usepackage{geometry}
\savegeometry{origin}
\begin{document}
\thispagestyle{empty}%
\noindent\makebox[\textwidth][c]{%
	\begin{minipage}{\linewidth}
		\centering
		\includegraphics[width=0.3\linewidth]{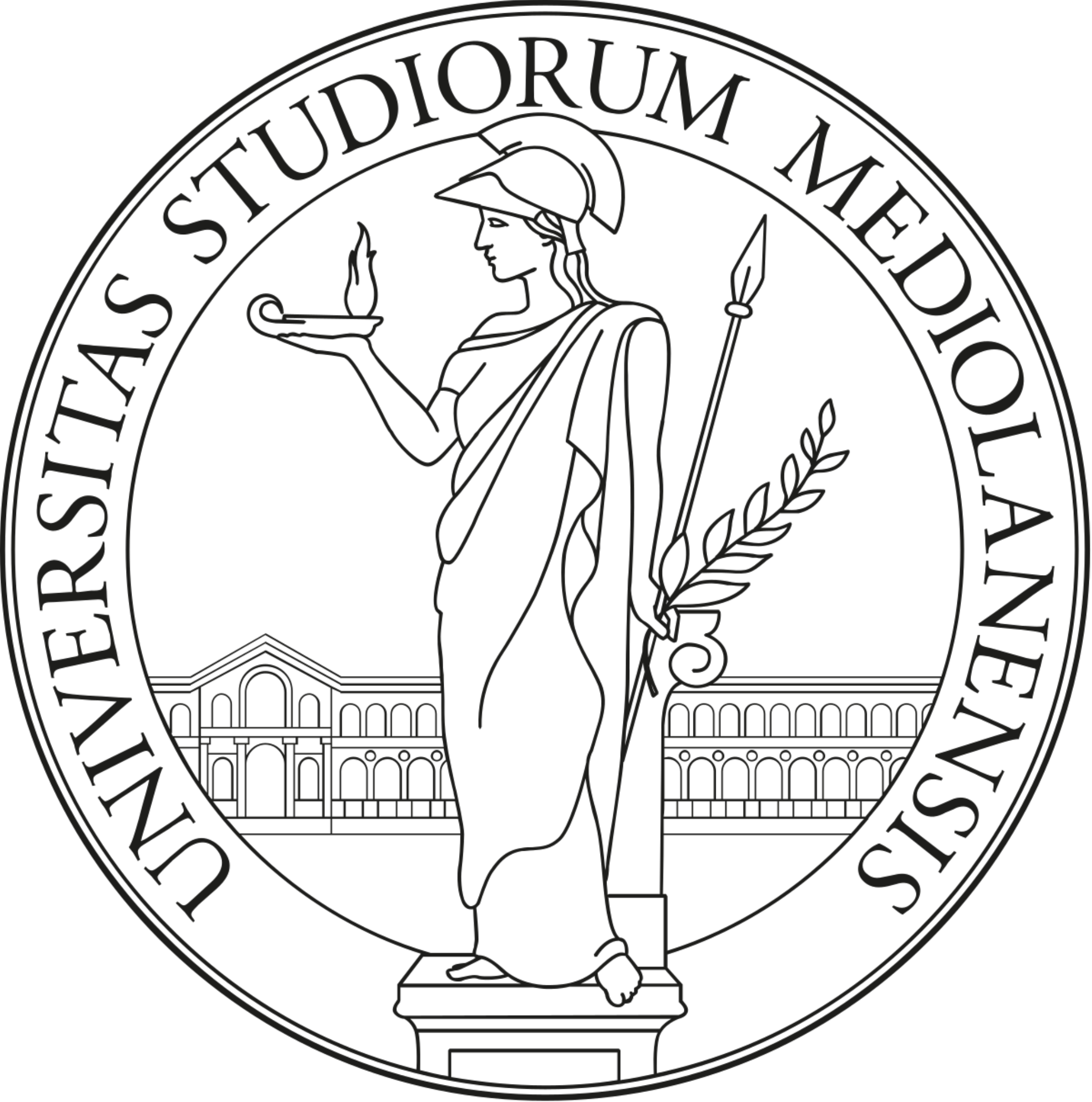}\\
		\medskip
		{{\LARGE \textbf{UNIVERSIT\`A DEGLI STUDI DI MILANO}}\\
		\Large Department of Physics\\
		\medskip
		\textbf{PhD School in Physics, Astrophysics, and Applied Physics}\\
		\bigskip
		Cycle XXXII}
		\par\vspace{2.5cm}
		{\uppercase{\Large Euclidean correlations in combinatorial optimization problems: a statistical physics approach \par}}
		\medspace
		{\large Disciplinary scientific sector: FIS/02 \par}
		\vspace{2.5cm}
		
		\begin{minipage}{0.6\linewidth}
			{\large \textbf{Director of the School:}} \\
			{\large Prof.~Matteo Paris}\par
			\medskip
			{\large \textbf{Supervisor of the Thesis:}} \\
			{\large Prof.~Sergio Caracciolo}\\
		\end{minipage}
		\hfill
		\begin{minipage}{0.35\linewidth}
			{\Large \textbf{PhD Thesis of:}} \par
			\smallskip
			{\Large Andrea Di Gioacchino}
		\end{minipage}
		\par
		\vspace{2.8cm}
		{\Large A.Y.~2019-2020}
	\end{minipage}}
\newpage

\loadgeometry{origin}

\section*{Acknowledgements}

My PhD has been a long journey, full of both beautiful and difficult moments. 

I would never have been able to arrive here without the help of many people, and I am writing these lines to express my deepest thanks to each of them.
First of all, I wish to thank my supervisor, Professor Sergio Caracciolo. I have had the luck of enjoying his precious advices, and a fruitful scientific collaboration. And I also am grateful to him for the freedom he always granted me, for all the time he spent helping me in many situations, and for his lectures on Conformal Field theory.

In addition to my supervisor, I also had other people guiding me in the stormy waters of the PhD: Professor Luca Guido Molinari, who has been an invaluable mentor in several scientific projects, as well as a wonderful teacher of Random Matrix theory, and Salvatore Mandr\`a, who supervised my work during my internship at NASA Ames in California and taught me many things about optimization and quantum computing.

In a sense, the people involved with my PhD can be compared with family members, and two of them have the role of elder brothers: Pietro Rotondo and Marco Gherardi. They have been always present for technical discussions about (the most diverse topics of) physics, but also for advices about schools, conferences, applications for post-docs, and many other important choices I made during my PhD.

At the end of a doctoral course, it is common to have several collaborators. I had the privilege to have many of them also as friends: Riccardo Capelli, Vittorio Erba, Riccardo Fabbricatore, Enrico Malatesta, Alessandro Montoli, Mauro Pastore, German Sborlini, Federica Simonetto. Thank you for all the scientific interactions, but also for all the time spent having fun together.

I also want to thank the QuAIL group at NASA Ames, for their kindness and scientific training during my internship, especially Eleanor Rieffel, Davide Venturelli, Gianni Mossi, Norm Tubman, Jeff Marshall, Eugeniu Plamadeala.

In addition to the scientific staff of the University of Milan and the ``Istituto Nazionale di Fisica Nucleare'', I also want to thank the administrative/management staff, and in particular Andrea Zanzani.

Few other very special people helped me during my PhD even though they are not experts in statistical physics: my family. They have provided me with constant support and love, and gave me the strength to get through the hardest moments. \\
Finally, thank you Greta. This (and many other things) would have been impossible without you.

\clearpage
\pagebreak

\tableofcontents

\pagebreak

\chapter{Why bothering with combinatorial optimization problems?} 
\chaptermark{Why bothering with COPs?}
Combinatorial optimization problems (COPs) arise each time we have a (finite) set of choices, and a well defined manner to assign a ``cost'' to each of them.
Given this very general (as well as rough) definition, it should not be surprising that we encounter many COPs in our everyday life: for example, it happens when we use Google Maps to find the fastest route to our workplace, or to a restaurant. 
But we deal with COPs also in much more specific situations, ranging from the creation of safer investment portfolios to the training of neural networks.
Despite their ubiquity, COPs are far from being completely understood. The most impressive example of our lack of knowledge is the so-called ``P vs NP'' problem, which puzzles theoretical computer scientists and mathematics since 1971, when Levin and Cook discovered that the Boolean satisfiability problem is NP-complete \cite{Cook1971}.

The study of COPs attracted soon the statistical physics community which, in those years, was beginning the study of spin glasses and thermodynamics of disordered systems. The connection between COPs and thermodynamics was clear since the work of Kirckpatrick, Gelatt and Vecchi \cite{Kirkpatrick1983}, and after that it became even stronger when physicists realized that ``random'' COPs (RCOPs) display phase-transition like behaviors (the so-called SAT-UNSAT transitions) \cite{Kirkpatrick1994}. The application of statistical mechanics techniques to COPs flourished after the seminal paper by Mezard and Parisi \cite{Mezard1985}, where they applied the so-called replica method to study typical properties of the random matching problem. Their results, together with those obtained after them, are astonishing and elegant, but they heavily rely on a sort of ``mean-field'' assumption: the cost of each possible solution of the COP studied is a sum of \emph{independent} random variables. Let us be more precise with an example: consider the problem of going from the left-bottom corner of a square city to the opposite one. The possible solutions (or configurations) are the sequences of streets that connect these two corners of the city, and the cost of a possible solution is the total length of the path. In the random version of the problem, one consider an \emph{ensemble} of cities, each of them with its pattern of streets, and a distribution of probability on them. In this case one is interested statistical properties of the ensemble such as the average cost of the solution, rather than the cost of the minimum-length path for a specific city (instance). 
In this example the mean-field approximation would consist in choosing the ensemble such that the length of each road is an independent random variable. On the opposite, in the original problem the \emph{Euclidean} structure of the problem introduces correlations between the street lengths, which are completely neglected in the mean-field version of the problem.
Euclidean correlations are not the only possible which are neglected in mean-field problems: using again the examples given before, assets that can be used in a portfolio are correlated (for example, shares of two companies in the same business area) and images used to train a neural network are typically ``structured'', in opposition with the hypothesis of mean-field problems.

Most of this manuscript will deal with the introduction of Euclidean correlations in RCOPs. We will see that several RCOPs can be analyzed with a well-understood formalism in one spatial dimension, and this can sometimes be extended, in very non-trivial ways, to two-dimensional problems. 

We will also discuss another route toward solutions of COPs that physicists (together with mathematics and computer scientists) are exploring in this years with intense interest: using quantum computers to solve hard combinatorial optimization problems. Even though the original idea has been discussed by Feynman in 1982 \cite{Feynman1982}, many questions are still without an answer. Here we will use Euclidean COPs as workhorse to analyze some of the open questions of the field.

This manuscript is organized as follows:
\begin{itemize}
\item In Chap.~\ref{chap::first} we introduce all the necessary formalism to deal with COPs and RCOPs from the statistical mechanics point of view. In particular we start by defining formally what a   COP is, and explaining why the statistical physics framework is a useful point of view to study COPs (Secs.~\ref{sec::cop}, \ref{sec::statphys}, \ref{sec::complexity}). We also briefly review the more relevant points (for our discussion) of spin glass theory, using the spherical p-spin problem as an example (Secs.~\ref{sec::spinglass}, \ref{sec::pspin}). Finally, we discuss large deviation theory (again using the spherical p-spin model) as a possible path to go beyond the study of the typical-case complexity for RCOPs.
\item In Chap.~\ref{chap::second} we address the problem of Euclidean correlations in RCOPs. We discuss firstly why we tackle problems starting from the 1-dimensional case (Sec.~\ref{sec::1d}), then we analyze in details several problems where our techniques can be used (Sec.~\ref{sec::matching}, \ref{sec::tsp}, \ref{sec::2factor}).
\item In Chap.~\ref{chap::third} we briefly discuss why quantum computation can be useful to solve COPs (Sec.~\ref{sec::quantcomp}), using the famous Grover problem as an example (Sec.~\ref{sec::groverorig}). Then we introduce two general algorithms of quantum computing which are used to solve COPs, the quantum adiabatic algorithm (Sec.~\ref{sec::qaa}) and the quantum approximate optimization algorithm (Sec.~\ref{sec::qaoa}). In the QAA case, we also analyze the performance of the DWave 2000Q quantum annealer to solve a specific COP problem, and we use the results obtained to address one of the current problems for the QAA, the so-called parameter setting problem (Sec.~\ref{sec::qaa_parset}).
\item Finally, in Chap.~\ref{chap::final} we summarize the main results of this work and explore the possibility of future works to further extend our understanding of COPs with correlations.
\end{itemize}
Throughout the manuscript, we make the effort to relegate the technical details of computations in the appendix, whenever possible, to lighten the text and to ease the reading. To do that, we have an appendix for each main chapter where we put the corresponding technical computations.

\chapter{Statistical physics for combinatorial optimization problems}\label{chap::first}
\chaptermark{Statistical physics for COPs}
Combinatorial optimization problems (COPs) have been addressed by using methods coming from statistical physics almost since their introduction. In this chapter we give a concise introduction of both COPs and the statistical physics of disordered systems (spin glass theory), with a focus on the links between these two fields. 
An important remark is due: both COPs and spin glass theory are deep and well-developed topics, and we do not want (neither would be able) to give a comprehensive review of them. In fact, we will limit ourselves to introduce the basic notions that we will need here and in the following chapters.
\section{Combinatorial optimization problems}\label{sec::cop}
Consider a finite set $\Omega$, that is $|\Omega|<\infty$, and a \emph{cost function} $C$ such that
\begin{equation}
	C: \Omega \to \mathbb{R}.
\end{equation}
The combinatorial optimization problem defined by $\Omega$ and $C$ consists in finding the element $\sigma^\star \in \Omega$ such that 
\begin{equation}
	\sigma^\star=\arg \min_{\sigma\in\Omega} C(\sigma).
\end{equation}
We will call the set $\Omega$ \emph{configuration space}, each element of the configuration space will be a configuration (of the system). We will call $C$ also \emph{Hamiltonian} of the system (sometimes we will also use the label $H$ for it, instead of $C$) and $C(\sigma)$ will be the cost or energy of the configuration $\sigma$. Notice that we are willingly using a terminology borrowed from the physics (and statistical mechanics) context, but up to this point this is pure appearance. However, as we will see in the following, this choice has deep root and can lead to extremely useful insights.

Let us now consider an example of COP. Suppose you and a friend of yours are invited to a bountiful feast. The two of you sit at the table, and then start discussing about who should eat what, since each dish is there in a single portion. Therefore you assign a ``value'' to each dish, and try to divide all of them in two equally-valued meal. This is the so-called ``integer partitioning'' problem: given a set $\{a_1, \dots, a_N \}$ of $N$ integer positive numbers, find whether there is a subset $A$ such that the sum of elements in $A$ is equal to the sum of those not in $A$, or their difference is 1, if $\sum_{i=1}^N a_i$ is odd. More precisely, this is the ``decision'' version of the problem, that is it admits a yes/no answer. We will see later the importance of decision problems, while we will focus here on restating the problem as an optimization one: given our set $\{a_1, \dots, a_N \}$, find the subset $A$ which minimizes the cost function
\begin{equation}\label{eq::cost_ipp}
	C(A) = \abs{\sum_{j\in A} a_j - \sum_{j\notin A} a_j}.
\end{equation}
Therefore in this case the configuration of a system is the subset $A$, and its cost is the ``un-balance'' between the elements of $A$ and those not belonging to $A$. Also notice that if one can solve the optimization problem, then the solution to the decision problem is readily obtained.


Each COP 
has some parameters which fully specify it, which most of the times are inside the cost function. These parameters are, basically, the input of our problem. When the full set of these parameters is given, we say that we have an \emph{instance} of our COP. For example, an instance of the integer partitioning problem is specified by the set $\{a_1, \dots, a_N \}$.\\
If we decide to deal with a COP in general, that is without specifying an instance, we have two choices: we can start searching for an algorithm to solve our problem for each possible value of the input, or try to say something more general about the solutions.
The first one is the direction (mostly!) taken by computer scientists (however, we will say something about it later), while physicists (mostly!) prefer to analyze the problem from the second point of view. We will follow this second road, but to do that we have to deal with the fact that the solution will depend drastically on the specific instance of the problem. \\
The way out this thorny situation consists in defining an \emph{ensemble} of instances and in giving to each of them a certain probability to be selected. Then many interesting quantities can be computed by averaging over this ensemble, so they do not depend anymore on any specific instance.
For example, let $\Omega$ and $C$ be respectively the configuration space and the cost function of a given COP. An instance is specified by the continuous parameters $\underline{x}$, so we will have $C=C_{\underline{x}}$ and the joint probability $p(\underline{x})$ over the parameters (and therefore over the instances). A quantity that we will be interested in is the average cost of the solution of our problem, which is given by
\begin{equation}\label{eq::gen_avg_on_dis}
	\overline{C^\star} = \overline{\min_{\sigma\in\Omega} C_{\underline{x}}(\sigma)} = \int \, dx \, p(\underline{x}) \min_{\sigma\in\Omega} C_{\underline{x}}(\sigma).
\end{equation}

How do we choose $p(\underline{x})$? In general, we would like to have an ensemble and a $p(x)$ such that the averages over the ensemble are representative of the typical case of our COP. That is, we hope that if we define an ensemble of integer partitioning problems, than our findings will be useful for our banquet problem.\\
This observation brings us to another important point: on one hand, we would like to have simple ensembles, where we can carry out as much analytic computations as possible; on the other hand, this is typically a oversimplified situation. For example, the standard ensemble defined for integer partitioning is composed of all the possible instances made of $N$ integers of the set \{0, 1, \dots, $2^{b-1}$\} (for a certain parameter $b$), and each of them has the same probability. In practice, this is done by choosing at random $N$ integers from our possibility set, each time we need an instance. \\
We will say something about what can be learned from this ensemble of integer partitioning, but one can immediately see that our banquet problem is considerably different: when choosing the value of each dish, you will probably have a lot of correlations. For example, you could decide to give to an ingredient, say sea bass, a high value and therefore all dishes containing sea bass will have a high, correlated value. And you could (and probably would) do the same with many other ingredients. This is an example of \emph{structure} in our instance, which is often difficult to capture with simpler ensembles where each parameter of the problem is uncorrelated with the others.

In the following, we will deal a lot with a specific kind of structure, that is the one induced by Euclidean correlations.

\section{Why statistical physics?}\label{sec::statphys}
The paradigms of statistical physics, and in particular those of spin glass theory, are particularly suited to deal with RCOPs. There are three main reasons for this fact, that we will now discuss.

\subsection{Partition functions to minimize costs}
A COP is defined by its configuration space $\Omega$ and its cost function $C$, and we are interested in finding its minimum. We introduce a fictitious temperature $T$ and its inverse $\beta$, and define the partition function of our problem as
\begin{equation}\label{eq::gen_part_func}
	Z = \sum_{\sigma \in \Omega} e^{-\beta C(\sigma)},
\end{equation}
where the name, partition function, arises from the fact that we are interpreting the cost of a configuration the energy of a (statistical) physics system. 
When the temperature is sent to zero, only the solutions of the COP, which minimizes $C(\sigma)$, are relevant in Eq.~\eqref{eq::gen_part_func}. Therefore, in this sense, a COP can be seen as the zero-temperature limit of a statistical physics problem.
We can compute many quantities starting from this point of view, but we will mainly be interested in the following:
\begin{equation}\label{eq::gen_free_en}
	F(\beta) = - \frac{1}{\beta} \log Z,
\end{equation}
since when we send $\beta\to\infty$ this quantity is the cost of the solution of our COP.
A useful consequence of the parallelism between low-temperature thermodynamics and COP that we just described is that we can use the well-developed techniques coming from the first field to address problems in the second. The first successful example of this program is the celebrate \emph{simulated annealing} algorithm, introduced by Kirkpatrick, Gelatt and Vecchi in~\cite{Kirkpatrick1983}.

Of course, there is no way we are able to compute $Z$ and $F$ for a given instance of a realistic (not over simplified) COP since both of these quantities depend on the parameters which define our instance. Here the idea of RCOPs comes in our help, and we can connect our formalism with that of disordered systems. We define an average, labeled by an overline, exactly as in Eq.~\eqref{eq::gen_avg_on_dis} and
\begin{equation}\label{eq::gen_free_en_av}
	F_{\text{av}}(\beta) = \overline{F(\beta)} = - \frac{1}{\beta} \overline{\log Z}.
\end{equation}
The computation of this quantity is at the heart of the so-called spin glass theory (see, for example, the books \cite{Mezard1987book, Dotsenko2005,Nishimori2001}) and several methods have been devised to deal with this kind of problems. Later we will review in detail one of these methods, the celebrated replica method.

Before moving to the next section, we want to add an important remark: the average done in Eq.~\eqref{eq::gen_free_en_av} is called ``with \emph{quenched} disorder'' and it very different from computing the average of the partition function first, $\overline{Z}$, and then taking its logarithm (which is called ``with \emph{annealed} disorder''). In general, the difference is that in the annealed case the disorder degrees of freedom are considered on the same footing of the ``configurational'' degrees of freedom of our systems, while in the quenched case the thermodynamic degrees of freedom are only the configurational ones, and the average over the disorder is done \emph{after} the computation of the partition function.\\
This distinction is very sharp when we take the COP/RCOP point of view: computing Eq.~\eqref{eq::gen_free_en_av} (quenched case) corresponds to take many instances of our COP, computing each time the cost of the solution, and then take the average of that. On the other hand, when we compute the annealed version of Eq.~\eqref{eq::gen_free_en_av} (that is, the one with $\log\overline{Z}$ instead of $\overline{\log Z}$) we are solving \emph{one single} instance of a COP, which in general will be different from the one we started with because of our average operation.

\subsection{Phase transition in RCOPs}\label{sec::phase_trans}
The connection between statistical physics and RCOPs  goes beyond the simple fact that we can use methods developed for the former to deal with the latter. This became clear after a first sequence of works~\cite{Goerdt1990, Chvatal1992, Kirkpatrick1994}, where it has been discovered that a certain COP, the $k$-SAT problem, when promoted to its random version, exhibits a behavior which is strongly reminiscent of a statistical-mechanics phase transition. 
In the $k$-SAT problem, the input is a sequence of $M$ clauses, in each of which $k$ variables are connected by the logical operation OR ($\vee$). There are $N$ different variables, which can appear inside the $M$ clauses also in negated form. For example $x\vee\overline{y}\vee\overline{z}$ is a possible clause of an instance of 3-SAT. The problem consists in finding an assignment to each variable such that all the clauses return \emph{TRUE}, or to say that such an assignment does not exist.\\
At the beginning of the 90s it has been discovered that, given the ratio $\alpha = M/N$ and giving the same probability to each instance of $k$-SAT with parameter $\alpha$, when $\alpha<\alpha_c$ the probability of finding an instance that can be solved goes to zero when $N\to\infty$, and if $\alpha>\alpha_c$ this probability goes to one when $N\to\infty$. This is the so-called SAT-UNSAT transition for the $k$-SAT problem, and $\alpha_c$ is a quantity which depends on $k$.\\ 
Actually, $k$-SAT problems with $k\geq3$ exhibit a sequence of phase transitions, discovered in following works (which are reviewed, for example, in Chapter 14 of~\cite{Moore2011} and treated in detail in~\cite{Zdeborova2009}), which the (random) system encounter if we change $\alpha$ from zero to $\alpha_c$. 

Other SAT-UNSAT transitions have been found in many other problems quite different from the $k$-SAT, for example the Traveling Salesman Problem (TSP)~\cite{Gent1996}, which we will discuss later, and the familiar integer partitioning problem (IPP)~\cite{Mertens1998}. Most of the times the transition is found by extensive numerical experiments, while for the IPP the critical point can be computed analytically. To have a feeling of why this transition happens, we will present here an intuitive argument which allows to obtain the correct transition point. Following Gent and Walsh~\cite{Gent1996phase}, we consider the IPP problem where $n$ values are taken from the set $\{0, \dots, B \}$. Giving a choice of a subset $A$, we compute $C$ as in Eq.~\eqref{eq::cost_ipp} and we notice that $C \leq n B$, so we can write it as a sequence of about $\log_2 n + \log_2 B$ bits. Remember that we want to take the limit $n\to\infty$, so there is no need to be very precise with the number of bits since we are in any case neglecting sub-dominant terms. Now, the problem has a solution if $C=0$ or $C=1$, therefore all the bits of $C$ but the last have to be $0$ for the problem to admit a solution. This corresponds to $\simeq \log_2 n + \log_2 B$ constraints on the choice of $A$. Let us suppose now that, for a random instance of the problem, each given choice of $A$ has probability $1/2$ of respecting each constraint. This is false, as one can easily argue, but it turns out to be a correct approximation to the leading and sub-leading order in $n$. \\
Given that there are $2^n$ different partitions of $n$ objects, the expected number of partitions which respect all constraints is
\begin{equation}
	\mathbb{E}[\mathcal{N}] = 2^{n- (\log_2 n + \log_2 B)}.
\end{equation}
The critical point is given by $ \mathbb{E}[\mathcal{N}]=1$, so
\begin{equation}
	\log_2 B = n - \log_2 n,
\end{equation}
that is $B \simeq 2^n$ to the first order. Actually, the approximation of independent constraints used is correct up to the second order, but the number of bits in $C$, that is the number of constraints, is overestimated by this simple argument. Indeed we have used the maximum $C$, which is a crude approximation of the typical one.

A more formal treatment giving the same result at the first order and the correct one at the second order can be found in the beautiful book of Moore and Mertens~\cite{Moore2011} (chapter 14) or, in a language more familiar to the statistical physics community, in~\cite{Mertens2001}. In Fig.~\ref{fig::ipp_phase_trans} we report the result of a numerical experiment showing the SAT-UNSAT phase transition of the IPP.

\begin{figure}
	\centering
	\includegraphics[width=0.95\columnwidth]{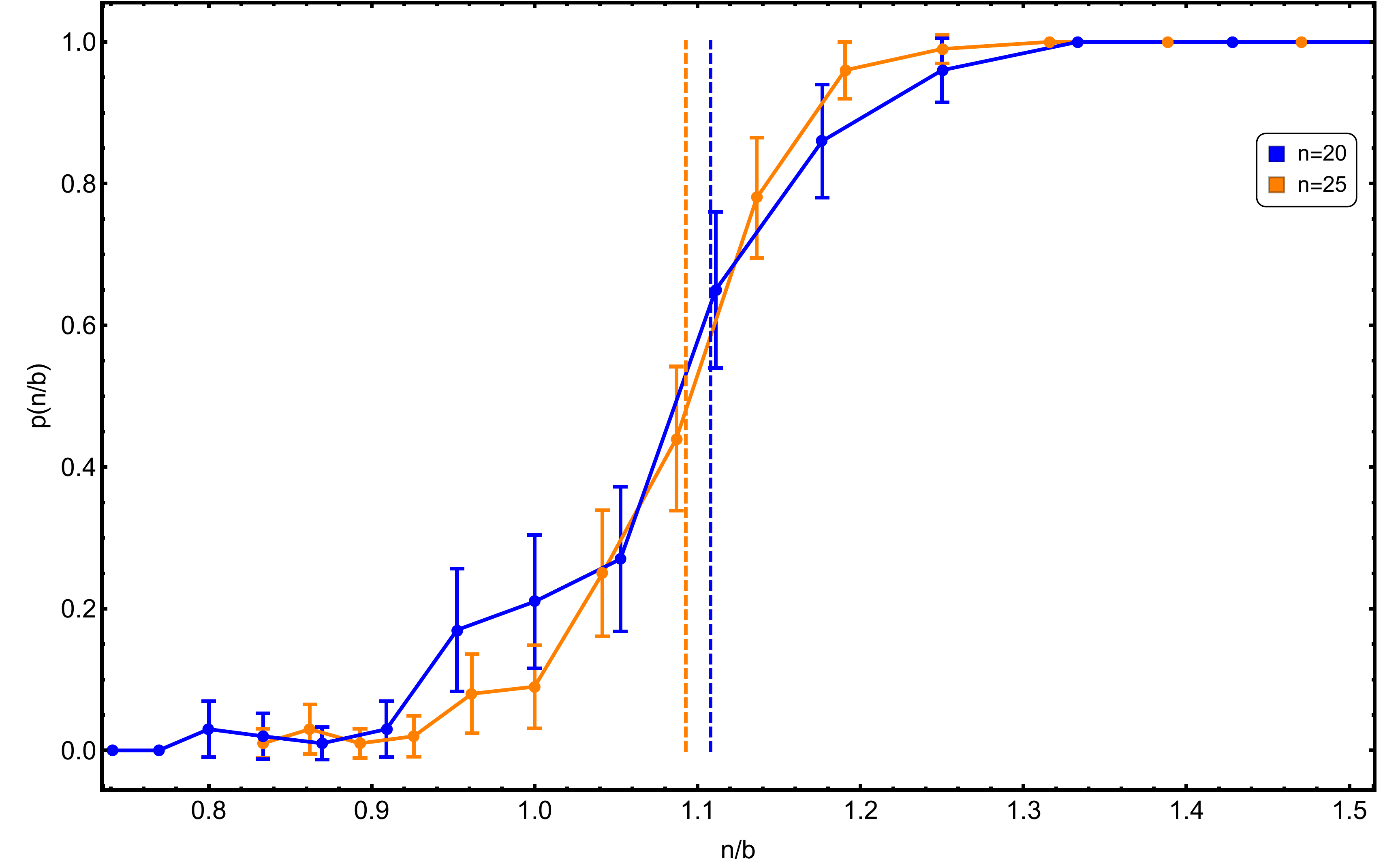}
	\caption{Numerical results for the probability of an instance of the integer partition problem with $n$ integers in the set $\{1,\dots, 2^b\}$ to have a solution, as a function of the ratio $n/b$. Each point is obtained by randomly extracting 100 instances of the problem, solving them and computing the number of instances with a solution. As we can see, at $n/b\simeq 1.1$ there is a transition from instance with a small probability of success to instances with a high probability of success. The vertical, dashed lines are the critical points given by $b = n - \log_2(n)/2$ (see~\cite{Moore2011}, chapter 14), computed up to the sub-leading term in $n$ to account for the finite size of the system. }\label{fig::ipp_phase_trans}
\end{figure}

\subsection{Back to spin models}\label{sec::cop_spin}
Another point in common of COPs and RCOPs with statistical physics, is that we can often write COPs cost functions as Hamiltonians in which the thermodynamical degrees of freedom are spin variables. For many COPs studied with statistical physics techniques this has been actually the first step. In this case a configuration of the system is given by specifying the state of all the spins. \\
Although the re-writing of the COP as a spin problem can be very useful, there is not a general procedure and in many problems with constraints (as we will discuss later) there is often a certain freedom in choosing the spin system. Indeed the minimum request the spin system has to satisfy is that given its ground state we can obtain the solution of the original COP. 

To be more concrete, let us discuss a spin system associated to the familiar IPP. Given a set $\{a_1, \dots, a_n \}$, we can specify a partition $A$ by assigning a spin variable $\sigma_i$ to each $a_i$, such that $a_i = 1$ (or $-1$) if $a_i$ belongs (or not) to $A$. As a function of these new variables, the cost function given in Eq.~\eqref{eq::cost_ipp} can be written as
\begin{equation}
	C = \abs{\sum_{i=1}^N a_i \sigma_i}.
\end{equation}
We can get rid of the absolute value by defining the Hamiltonian
\begin{equation}
	H = C^2 = \sum_{i,j=1}^{N} J_{ij} \sigma_i \sigma_j,
\end{equation}
where $J_{ij}=a_i a_j$, whose ground states correspond to the solutions of the original instance of IPP. Starting from this Hamiltonian, the problem has been analyzed in~\cite{Ferreira1998}, where the anti-ferromagnetic and random nature of the couplings $J_{ij}$ makes the thermodynamics non-trivial.\\
As a final remark, notice that there are many other Hamiltonians which are good spin models for our IPP, for example:
\begin{equation}
	H = C^4 = \sum_{i,j,k,\ell} J_{i,j,k,\ell} \sigma_i \sigma_j \sigma_k \sigma_\ell,
\end{equation}
with $J_{i,j,k,\ell}=a_i a_j a_k a_\ell$. The choice of one model rather than another is driven by the search for the simplest possible one which is well-suited to the techniques that we want to use.

\section{Complexity theory and typical-case complexity}\label{sec::complexity}
Consider a generic problem, not necessarily a combinatorial optimization one: we have an input and, according to certain rules which specify the problem, we want to get the output, that is the solution to the problem. Is a certain problem difficult or easy? Can we quantify this difficulty and say that some problems are harder than others? These are deep questions which are not completely understood, and are the holy grail (in their formalized version) of a branch of science which involves computer science, mathematics and physics and it is called \emph{complexity theory}. 

In this section we want to briefly introduce some concepts from complexity theory that will be relevant in the following and elaborate on the differences between the worst-case analysis of a COP and a typical-case analysis, which is the one usually carried out by means of disordered systems techniques.

\subsection{Worst-case point of view}
Let us focus on a specific kind of problems, those called \emph{decision} problems. In this case, we have a problem and an input and the output has to be a yes/no answer. The paradigmatic example is the k-SAT problem introduced in Sec.~\ref{sec::phase_trans}, and another example is the definition of the integer partition problem that we gave in Sec.~\ref{sec::cop}. A first attempt to measure the ``hardness'' or complexity of a problem could be done by measuring the number of operations needed to solve it. In the following, sometimes we will say ``time'' instead of ``number of operations'' for brevity, even if these two quantities are related but not equivalent. 
This definition of complexity, however, has several weaknesses:
\begin{itemize}
	\item the complexity of a problem depends on the algorithm used, while we would like to characterize the problem itself;
	\item the complexity depends on the specific instance.
\end{itemize}
Both problems are solved by introducing the concept of complexity classes. Before defining them, let us address a tricky point: in general the number of operations needed to solve a problem will increase with the input size. For example, if we have an algorithm to solve IPP, whatever algorithm it is, we expect that we will need to wait longer for the solution if we have a set of $N=100$ integers with respect to the case with $N=10$. \\
However, the exact determination of the size of an instance is a subtle point, since there is a certain freedom in deciding it. For our purposes, we will always deal with problems that admit a re-writing in terms of spin systems, so we can safely define the number of spins as the size of the instance.

Now we are ready to introduce complexity classes. A problem is said to be ``nondeterministic polynomial'' (NP) or in the NP complexity class if, given a configuration of an instance of size $N$ of the problem, the time needed to check whether this configuration is a solution of the problem scales as $\mathcal{O}(N^\alpha)$ for $N\to\infty$, where $\alpha\in \mathbb{R}$ does not depend on the configuration and on the instance. $\mathcal{O}$ is the ``Landau big-O'' notation, that is $f(x)=\mathcal{O}(g(x))$ if there are $c>0$ and $x_0>0$ such that $f(x)\leq c g(x)$ for all $x>x_0$.\\
For example, IPP is in NP: given a partition $A$ of $N$ objects, to check that this is or not the solution it is sufficient to compute the sum of $|A|$ objects, those of $N-|A|$ objects and a single operations to compare this two quantities, so $\mathcal{O}(N)$ operations.\\
Many COPs are not so easy to place in the class NP: consider the optimization variant of IPP, that is the problem of finding the minimum of the cost function Eq.~\eqref{eq::cost_ipp}, even if this is not 0 or 1. Now given a partition $A$ we can again compute its cost in $\mathcal{O}(N)$ time, but this is not a certificate that this partition is or is not the one with minimum cost. To check that, we would need to solve the whole problem, so the complexity of checking whether a given partition is the solution is the same of solving the original problem.

Another very important class is the ``polynomial'' (P) complexity class. A problem is said to be in P if there exists an algorithm which is guaranteed to solve each instance of size $N$ in a time which scales as $\mathcal{O}(N^\alpha)$ for $N\to\infty$, with $\alpha\in \mathbb{R}$. 
This class is defined such that the two issues in our definition of complexity are now solved: for a problem to be in P it is now sufficient that \emph{one} polynomial-time algorithm exists, and, for a given algorithm, the complexity is computed on the worst-case instance, that is the one where our algorithm needs more operations to reach the solution.

In 1971, Cook~\cite{Cook1971} discovered a special property of the SAT problem, a relaxation of k-SAT where the clauses are allowed to be of any size: each other problem in NP can be mapped to SAT in polynomial time, in such a way that if we know are able to solve SAT, we can obtain the solution to each other problem in NP with a polynomial overhead. In particular, if SAT turned out to be in P, each other NP problem would be in P as well. After the work of Cook, Karp~\cite{Karp1972} discovered several other problems with this property (among them, our friend the decision version of IPP) and many others have been found since then. This class of problems is called NP-complete, and these problems are sometimes referred to ``the hardest problem in NP''. \\
The question whether an algorithm which solves one of these in polynomial time exists or not is still open, it is called ``P vs NP'' problem and is one of the biggest theoretical challenge for todays computer scientists and mathematics.

Another question remains open: if the the decision version of IPP is NP-complete, to what class the optimization version of IPP belongs to? The answer is the so-called ``NP-hard'' complexity class, which contains all those problems which are almost as difficult as NP-complete problems. This means that if we were able to solve the optimization version of IPP in polynomial time, we would be able to solve also the decision IPP in polynomial time, and then all the NP problems in polynomial time.

Fig.~\ref{fig::comp_cl} is a cartoon representing the relations among the various complexity classes discussed here.

\begin{figure}
	\centering
	\includegraphics[width=0.95\columnwidth]{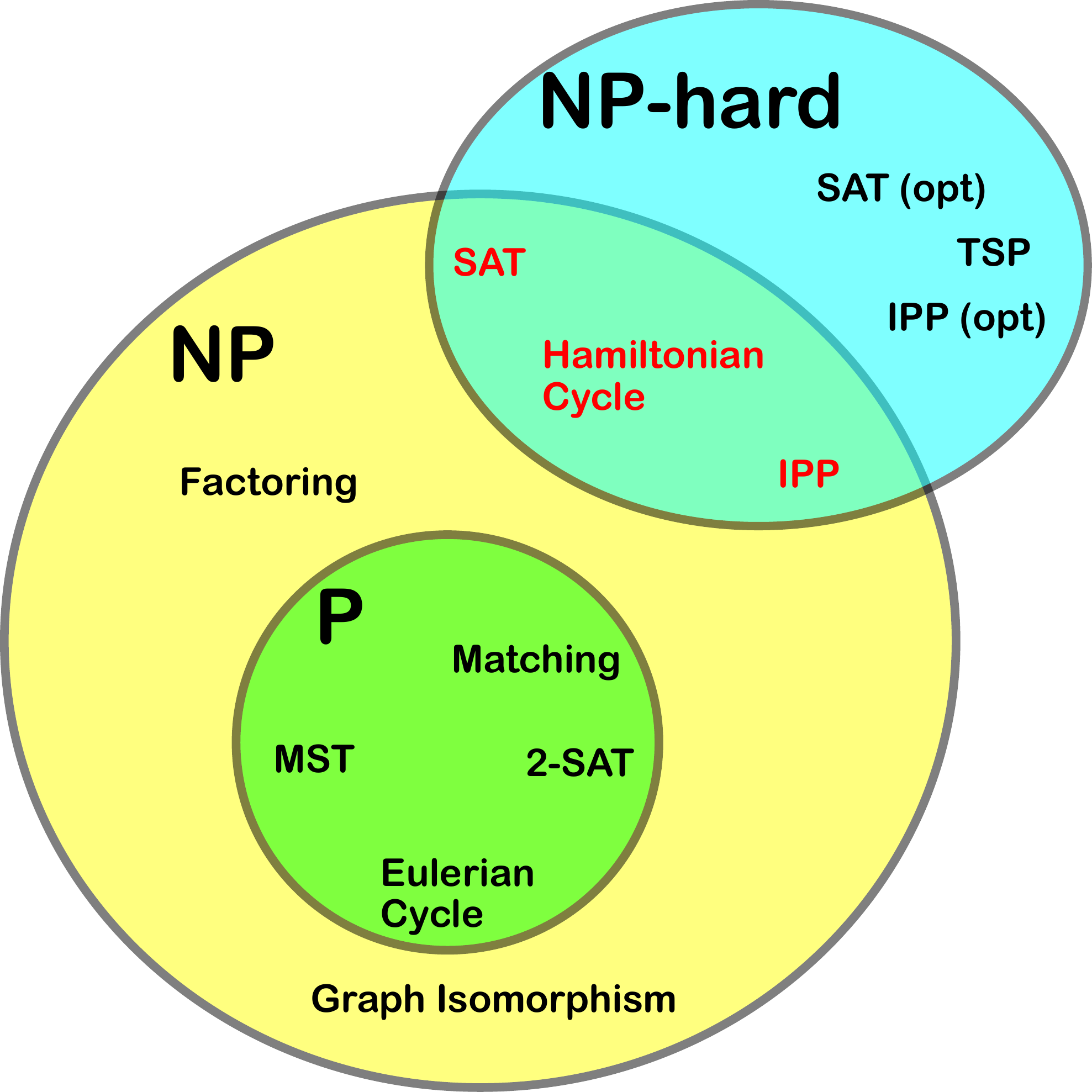}
	\caption{A cartoon highlighting the relations between the complexity classes P, NP, NP-complete and NP-hard, under the hypothesis that P$ \neq $NP. Inside each class, we wrote the names of some COPs belonging to that class. Problems in the NP-complete class are written in red and are positioned at the border of the NP class, to indicate that if one of the turned out to be in P, then all the NP class would be in P. Notice that the problems such that SAT or IPP are understood in their decision version, and their optimization version is present in the class NP-hard (we use the word ``opt'' to indicate optimization version). Finally, a list of the acronyms used: \mbox{MST $\to$ Minimum} Spanning Tree, \mbox{SAT $\to$ Satisfiability}, \mbox{IPP $\to$ Integer} Partition Problem, \mbox{TSP $\to$ Traveling} Salesman Problem.}\label{fig::comp_cl}
\end{figure}

There are many other complexity classes, and various refinement of the ideas presented here. For example, we can incorporate in our class definitions the scaling with the instance size of the usage of \emph{memory} (in addition to the time of computation). These topics are treated in detail in several very good textbooks, for example~\cite{Papadimitriou2003,Johnson1979}.

\subsection{Optimistic turn: typical case and self-averaging property}

As we have already discussed, a possible point of view consists in working out some \emph{typical} properties of a COP and this can be done through the definition of an ensemble of instances and a probability weight of each instance. This is the standard program carried out by physicists, since spin glass techniques are particularly suited for it.\\
But has this something to do with the complexity-theoretical perspective described before?
As we have seen, for a given algorithm its running time is computed on the worst-case scenario, that is the hardest instance for a that particular algorithm. However, it can be that this difficult instance has very low probability in an ensemble, so it gives a very little contribution to whatever typical quantity we are computing.\\
This reasoning could lead us to abandon the idea of describing COPs through their random formulations, but there are also some positive facts about adopting this perspective. For example, the problems that we usually observe in practical application are far from being those worst-case instances for our algorithm (and, even if they are, we can in principle use a different algorithm for which the hardest instances are different from the typical instances we encounter in practical applications).\\
This is actually related to a well-known phenomenon, called self-averaging property, which takes place in many physical systems. In practice, consider a random variable, for example the average cost of the solution of a given RCOP, $\overline{E}$. This quantity will depend on the instance size $N$ and it is said to be self-averaging if the limit for $N\to\infty$ of it is not a random variable anymore. In other words, we have
\begin{equation}
	E = \lim_{N\to\infty} \overline{E_N}
\end{equation}
and
\begin{equation}
	\lim_{N\to\infty} \overline{E_N^2} - E^2 = 0.
\end{equation}
As we will see in the following, this happens for some problems and does not happen for others, and it is an indicator which is telling us whether a random large instance of the problem is well characterized by the typical one. \\
Notice that even if each quantity of interest of the RCOP (for example, cost of the solution, number of solutions and others) is self-averaging, we can still build rare instances, which in our ensemble will have a probability which is going to zero with their size, in which this quantities are completely different from the typical case. Actually, there is a branch of physics which deals with this rare instances: it is called \emph{large deviation theory} and we will briefly discuss it at the end of this chapter, while in most of this work the focus will be on typical properties of RCOPs.

As a final remark, there is another point that we want to mention regarding the worst-case versus typical-case problems. Using the typical case approach we can locate phase transitions in RCOPs, as the SAT-UNSAT transition that we discussed in Sec.~\ref{sec::phase_trans}. It turns out that most of the times the presence of ``intrinsic hard instances'' is related to the presence of these phase transitions: for example, in the 3-SAT problems we know an algorithm (described in~\cite{Coja2009}) which is guaranteed to find the solution in polynomial time as long as the parameter $\alpha$ defined in Sec.~\ref{sec::phase_trans} is such that $\alpha<\alpha_r$, where $\alpha_r$ is a critical point where the so-called ``rigidity'' phase transition takes place (see~\cite{Achlioptas2008}).

\section{Spin Glass theory comes into play}\label{sec::spinglass}
As we said several times, the paradigm of statistical mechanics most suited to the application to RCOPs is spin glass theory. 
The two main ingredients which distinguish spin glasses from the ``standard'' statistical mechanics spin systems are \emph{quenched disorder} and \emph{frustration}, two things that we already met in our general discussion about RCOPs.
\begin{itemize}
	\item \emph{Quenched disorder}: the Hamiltonian of our system has some random variables in it, and the probability density of these variables is explicitly given. Moreover, the word quenched means that the thermodynamics of the system has to be considered \emph{after} that a specific realization of these random variables is chosen. We have already met that in the definition of RCOPs.
	\item \emph{Frustration}: consider a spin Hamiltonian (this discussion can be trivially extended to non binary-spin as well) and a random configuration of the degrees of freedom. Now randomly select a spin (or a set of $k$ spin with fixed $k$) and flip it (or all of them) only if this lowers the total energy of the system, and keep doing that until a minimum of energy is reached, choosing each time a random spin. 
	Repeat the whole experiment many times: start from a random configuration and flip spins at random until a local energy minimum is reached. If it happens that the final state is not always the same or the final states are not related by a symmetry of the initial Hamiltonian, the system is said to be frustrated.\footnote{there is a more simple but less generic definition of frustration, see for example~\cite{Dotsenko2005} (chapter 1), based on \emph{frustrated plaquettes}, that is closed chains of interactions among spins whose product is negative; however, for many RCOPs this definition is not immediately applicable, so we prefer to stick with that given here.} 
	For many systems it happens that most of these local minima have different energy from the ground state. In these cases, a frustrated system is such that we cannot find its ground state by local minimization of the energy. For example, is easy to see that (most of the instances of) optimization IPP is frustrated, and this situation applies.
\end{itemize}

The spin glass theory exploration begun with the so-called Edwards-Anderson model~\cite{Edwards1975}, whose Hamiltonian describes spins arranged in a 2-dimensional square lattice, and it is
\begin{equation}
	H_{EA} = \sum_{\left<i,j\right>} J_{ij} \sigma_i \sigma_j,
\end{equation}
where $\sigma_i$ is the spin variable number $i$ and the brackets mean that the sum has to be performed on first neighbors. The $J_{ij}$ are independent and identically distributed (IID) random variables, and two choices that are often used as probability density are
\begin{equation}
	p(J_{ij}) = \frac{1}{\sqrt{2 \pi J^2}} e^{-\frac{(J_{ij}-J_0)^2}{2 J^2}}
\end{equation}
(gaussian disorder)
\begin{equation}
	p(J_{ij}) = \frac{1}{2} \delta(J_{ij} - 1) + \frac{1}{2} \delta(J_{ij} + 1)
\end{equation}
(bimodal disorder).\\
As for the ferromagnetic 2-dimensional Ising model, the analytical solution, which basically coincide with the calculation of the partition function, could be difficult (or even impossible) to find, so a good starting point is to consider a mean-field approximation of the problem, which in this case takes the name ``Sherrington-Kirkpatrick'' model~\cite{Sherrington1975}:
\begin{equation}
	H_{SK} = \sum_{i < j} J_{ij} \sigma_i \sigma_j.
\end{equation}
Although their paper is called ``Solvable model of a spin-glass'', Sherrington and Kirkpatrick were not able to solve it in the low temperature phase, where they obtained an unphysical negative entropy. The model turned out to be actually solvable also in the low temperature phase, after a series of papers by Parisi~\cite{Parisi1979_1,Parisi1979_2,Parisi1980_1,Parisi1980_2}, where he introduced the remarkable \emph{replica-symmetry breaking} scheme.\\
The full solution of the SK model goes far beyond what we need to discuss in this work, so we suggest to the interested readers one of the several good books on the subject~\cite{Mezard1987book,Dotsenko2005,Nishimori2001} or the nice review presented in~\cite{Malatesta2019}. We will, however, perform a detailed spin glass calculation of the so-called p-spin spherical model in the next section. This will have a two-fold usefulness: it will allow us to review some important concepts of spin glass theory, such as the replica method, the concept of pure states and the replica-symmetry breaking; it will also constitute the basis for our analysis of large deviations in disordered system, which we will consider at the end of this chapter.

\subsection{An ideal playground: the spherical p-spin model}\label{sec::pspin}
The p-spin spherical model (without magnetic field) Hamiltonian is:
\begin{equation}
	H = - \sum_{i_1<i_2<\cdots<i_p} J_{i_1\cdots i_p} \sigma_{i_1} \cdots \sigma_{i_p},
\end{equation}
where $p\geq 3$ (one can also choose $p\ge2$, but the $p=2$ case is qualitatively different from the others, see for example \cite{Kosterlitz1976}), the spin variables are promoted to continuous variable defined on the real axis and subject to the ``spherical constraint''
\begin{equation}
	\sum_{i=1}^N \sigma_i^2 = N.
\end{equation}
The interaction strengths are IID random variables and their probability density is
\begin{equation}
	p(J) = \frac{N^{\frac{p-1}{2}}}{\sqrt{p! \, \pi}} e^{-J^2 \frac{N^{p-1}}{p!}}.
\end{equation}
Notice that the choice of the power of $N$ in $p(J)$ is fixed by the request of extensivity of the \emph{annealed} free energy, indeed
\begin{equation}\label{eq::pspin_annealed_comp}
\begin{split}
	\overline{Z} & = \int_{-\infty}^\infty \left( \prod_{i_1<\dots<i_p} dJ_{i_1\cdots i_p} \, p(J_{i_1\cdots i_p}) \right) \int_{-\infty}^\infty d\sigma_{i_1} \cdots d\sigma_{i_p} \, e^{-\beta H} \, \delta\left( N - \sum_i \sigma_i^2 \right)\\
	& = \int_{-\infty}^\infty D\sigma \left( \prod_{i_1<\dots<i_p} \int_{-\infty}^\infty dJ_{i_1\cdots i_p} \, p(J_{i_1\cdots i_p}) \, e^{\beta J_{i_1\cdots i_p} \sigma_{i_1} \cdots \sigma_{i_p}} \right) \\
	& = \int_{-\infty}^\infty D\sigma \prod_{i_1<\dots<i_p} e^{\frac{\beta^2}{4 N^{p-1}} p! \sigma_{i_1}^2 \cdots \sigma_{i_p}^2} \sim \int_{-\infty}^\infty D\sigma \, e^{\frac{\beta^2}{4 N^{p-1}} \left(\sum_i\sigma_{i}^2\right)^p} = \Omega_N \, e^{N \beta^2/4} \\
\end{split}
\end{equation}
where we used that
\begin{equation}\label{eq::pspin_sum_approx}
	p!\sum_{i_1<\cdots<i_p} f_{i_1,\dots,i_p} \sim \sum_{i_1, \dots, i_p} f_{i_1,\dots,i_p}
\end{equation} 
when $f_{i_1,\dots,i_p}$ is a symmetric under permutations of the indexes and in the thermodynamical ($N\to\infty$) limit (the error comes from the fact that in the right-hand side of the approximation we are also considering terms with equal indexes). We also defined
\begin{equation}
	D\sigma = d\sigma_{i_1} \cdots d\sigma_{i_p} \, \delta\left( N - \sum_i \sigma_i^2 \right).
\end{equation}
Finally, $\Omega_N$ is the surface of a $N$ dimensional sphere of radius $\sqrt{N}$. Therefore the annealed free energy is
\begin{equation}\label{eq::pspin_f_annealed}
	F_{\text{ann}} = - \frac{1}{\beta} \log \overline{Z} = - N (\beta/4 + T S_\infty),  
\end{equation}
where $S_\infty = \log \Omega_N / N \sim (1+\log (2\pi) )/2 $ (see Appendix \ref{app::sphere}). Therefore, thanks to the factor $N^{p-1}$, this quantity which has to be correct in the high temperature limit where the values of the couplings $J_{i_1 \dots i_p}$ become irrelevant, is extensive as it should be. For this reason the label $S_\infty$ is used: this quantity (as we will check) is also the infinite-temperature entropy of the quenched-disorder case.

This model has been introduced by Crisanti and Sommers in~\cite{Crisanti1992}, and used many times to probe the various aspects of spin glass theory (see, for example, the beautiful review article~\cite{Castellani2005}).\\
Our aim for this section is the computation of the quenched free energy:
\begin{equation}
	F = - \frac{1}{\beta} \overline{\log Z}.
\end{equation}

\subsubsection{The replica trick at work}
To start our computation, we will introduce the replica trick in its standard form, that is
\begin{equation}
	\overline{\log Z} = \lim_{n \to 0} \frac{\overline{Z^n} - 1}{n}.
\end{equation}
This exact identity is not a trick, so why the name? The real trick is in our way to use it: we will consider \emph{integer} values for $n$, so that we can compute $\overline{Z^n}$, which is simply the average over the disorder of the partition function of a system made by $n$ non-interacting replicas of the original system. Starting from the knowledge of $\overline{Z^n}$ for integer $n$ we will try to extend analytically our function to $n\in\mathbb{R}$ to obtain the $n\to 0$ limit. According to this interpretation, it should be clear why this procedure is called \emph{replica} trick, and we will call $n$ ``replica index'', or ``number of replicas''. \\ 
We will see that many of the manipulations done to recover a meaningful analytic extension to real values of $n$ will be impossible to justify formally, but the whole strategy, sometimes called ``replica method'', has been proved to be exact by many numerical simulations and, for some problems, also by analytical and rigorous arguments (see~\cite{Guerra2002,Guerra2003}).

The computation of $\overline{Z^n}$ is carried out in Appendix~\ref{app::pspin}, the result is
\begin{equation}\label{eq::pspin_zn_2}
	\overline{Z^n} = e^{n N \log(2\pi)/2} \int DQ \, D\lambda \, e^{-N S(Q,\lambda)},
\end{equation}
where
\begin{equation}\label{eq::pspin_zn_3}
	S (Q,\lambda) = -\frac{\beta^2}{4} \sum_{a,b} Q^p_{ab} - \frac{1}{2}\sum_{a,b} \lambda_{ab} Q_{ab} + \frac{1}{2} \log\det(\lambda),
\end{equation}
the $Q$ and $\lambda$ are $n \times n$ matrices, $Q_{aa}=1$ for all $a$, the integral with measure $DQ$ is done over the symmetric real matrices with 1 on the diagonal, the integral with measure $D\lambda$ is done over all symmetric real matrices.\\
We integrate over the $\lambda$ matrices by exploiting the saddle-point method\footnote{actually the integral on $\lambda_{ab}$ should be done over the imaginary line (with the methods discussed, for example, in Appendix B.1 of~\cite{Peliti2011}); however, at the end of the day this is perfectly equivalent to the usual saddle point, as pointed out in~\cite{Crisanti1992}.}, so at the saddle point $\lambda$ is such that
\begin{equation}
	\frac{\partial}{\partial \lambda_{ab}} S(Q,\lambda) = 0.
\end{equation}
Therefore, exploiting the formula valid for a generic matrix M
\begin{equation}\label{eq::matrix_id_deriv}
	\frac{\partial}{\partial M_{ab}} \log \det M = (M^{-1})_{ba},
\end{equation}
we obtain the equation for $\lambda$
\begin{equation}
	Q_{ab} = (\lambda^{-1})_{ab}.
\end{equation}
Putting that back into Eq.~\eqref{eq::pspin_zn_3}, we obtain
\begin{equation}\label{eq::pspin_last_action}
	S (Q) = -\frac{\beta^2}{4} \sum_{a,b} Q^p_{ab} - \frac{1}{2} \log\det Q + \frac{n}{2}.
\end{equation}
The term $n/2$ can be pulled out of the integral and together with the exponential outside the integral in Eq.~\eqref{eq::pspin_zn_2} will give a constant shift of the free energy of $-N T S(\infty)$, the same that we met already in the annealed calculation Eq.~\eqref{eq::pspin_f_annealed}.
The last step is again a saddle-point integration on the Q variables, and we obtain the free energy density, $f=F/N$,
\begin{equation}\label{eq::pspin_f}
	f = \lim_{n \to 0} - \frac{\beta}{4 n} \sum_{a,b} Q^p_{ab} - \frac{1}{2 n \beta} \log\det Q + T S(\infty),
\end{equation}
where the matrix $Q$ has $Q_{aa}=1$, is symmetric and the off-diagonal entries are given by the saddle point equations (we use again Eq.~\eqref{eq::matrix_id_deriv})
\begin{equation}\label{eq::pspin_sp}
	\frac{\beta^2 p}{2} Q_{ab}^{p-1} + (Q^{-1})_{ab} = 0.
\end{equation}

\subsubsection{Replica-symmetric ansatz and its failure (for this problem!)}
At this point of the discussion, the introduction of $n$ replicas of our system is simply a technical trick, a formal manipulation. Therefore it seems reasonable to impose \emph{symmetry} among replicas to deal with Eq.~\eqref{eq::pspin_sp}. In other words, we consider the following ansatz for the matrix $Q$:
\begin{equation}
	Q_{ab} = \delta_{ab} + q_0 (1-\delta_{ab}),
\end{equation}
that is $Q$ has 1 on the diagonal entries and $q_0$ on the off-diagonal entries. The inversion of a matrix with this form is
\begin{equation}
	(Q^{-1})_{ab}= \frac{1}{1-q_0} \delta_{ab} - \frac{q_0}{(1-q_0)(1+(n-1)q_0)}
\end{equation}
and from Eq.~\eqref{eq::pspin_sp} we obtain the equation for $q_0$ when $n\to 0$:
\begin{equation}\label{eq::pspin_q0_rs}
	\frac{\beta^2 p}{2} q_0^{p-1} - \frac{q_0}{(1-q_0)^2} = 0.
\end{equation}
The first observation is that $q_0=0$ is a solution. In this case one obtains for the free energy density
\begin{equation}
	f_{RS} = -\beta/4 + T S(\infty)
\end{equation}
where the subscript RS stands, here and in the following, for ``Replica Symmetric''. This is exactly the same result we obtained with the annealed computation, and indeed is the correct result for the high-temperature regime. This does not happen by chance: if the overlap matrix $Q$ has null off-diagonal entries, the whole replica-trick computation coincide with the annealed one, as can be checked confronting Eq.~\eqref{eq::pspin_zn_1} and Eq.~\eqref{eq::pspin_annealed_comp}.
However, the annealed calculation and the RS ansatz differ when the temperature is decreased: for $T<T_c$, Eq.~\eqref{eq::pspin_q0_rs} develop another solution, with $q_0\neq0$. Unfortunately, this solution is not \emph{stable}: one can compute the Hessian (with the second derivatives of Eq.~\eqref{eq::pspin_last_action} or directly of Eq.~\eqref{eq::pspin_f}) in the saddle point and check that the eigenvalues have different signs\footnote{for a saddle point to be stable, in general the eigenvalues have to be all positive, so that the matrix is positive-definite and we are actually in a minimum; however, in this case, for (actually quite nebulous) reasons connected to the limit of vanishing number of replicas, the saddle point would be stable if all the eigenvalues were negative.}.
This stability problem has been first noticed for the Sherrington-Kirkpatrick model in~\cite{deAlmeida1978}, and today it is well known that to go beyond this impasse, we need to give up our RS ansatz.

\subsubsection{Replicas and pure states}
The conceptual error that we made in the previous part of our calculation is to think about replicas as purely ``abstract'' mathematical objects that we exploited to ease our computation. This idea brought us to the RS ansatz, which turned out to be wrong, since it gives a unstable saddle point under a certain $T=T_c$.\\
Before trying to modify our ansatz, let us introduce some useful concepts for the description of the physics of spin glasses. The first one is the idea of \emph{pure state}. A pure state can be defined as a part of the configuration space such that the connected correlation functions decay to zero at large distances. A standard example of pure states are the two ferromagnetic phases of a Ising model in more than 2 dimension, below the critical temperature: one with positive magnetization, $\left< \sigma \right>_+ = m>0$, the other with negative magnetization $\left< \sigma \right>_- = - m<0$. In this case, the Gibbs measure splits in two components with the same statistical weight (due to the symmetry of the model), so that we have for the thermodynamical average $\left< \cdot \right>$
\begin{equation}
	\left< \cdot \right> = \frac{1}{2} \left< \cdot \right>_- + \frac{1}{2} \left< \cdot \right>_+.
\end{equation}
As one can easily see $\left< \sigma \right> = 0$, and for the connected two-point correlation function
\begin{equation}
	\left< \sigma_i \sigma_j \right> \sim \frac{1}{2} \left< \sigma \right>_-^2 + \frac{1}{2} \left< \sigma \right>_+^2  = m^2 \neq 0,
\end{equation}
where we used that the $\left< \cdot \right>_-$ and $\left< \cdot \right>_+$ are averages done inside the two pure states.

It can happen that there are more than 2 pure states, and in this case we have for the thermodynamical average
\begin{equation}
\begin{split}
	\left< A \right> & = \frac{1}{Z} \sum_{\sigma} e^{-\beta H(\sigma)} A = \frac{1}{Z} \sum_\alpha \sum_{\sigma \in \alpha} e^{-\beta H(\sigma)} A = \sum_\alpha w_\alpha \frac{1}{Z_\alpha} \sum_{\sigma \in \alpha} e^{-\beta H(\sigma)} A \\
	& = \sum_\alpha w_\alpha \left< A \right>_\alpha
\end{split}
\end{equation}
where the sum over $\alpha$ runs over all the pure states,
\begin{equation}
	Z_\alpha = \sum_{\sigma \in \alpha} e^{-\beta H(\sigma)}
\end{equation}
and
\begin{equation}\label{eq::pure_state_weight}
	w_\alpha = \frac{Z_\alpha}{Z}.
\end{equation}

Now, we go back to our p-spin spherical model and consider the quantity
\begin{equation}
	q_{\alpha \beta} = \frac{1}{N} \sum_i \left<\sigma_i\right>_\beta \left< \sigma_i \right>_\alpha
\end{equation}
where $\alpha$ and $\beta$ are indexes which label pure states, and the angle brackets mean thermodynamical average (possibly, as in this case, done inside a pure state). This quantity is the overlap between pure states, and depends on the specific instance.
Now, given the statistical weights of the pure states defined as in Eq.~\eqref{eq::pure_state_weight}, we introduce the probability $P_J(q)$ that two pure states, chosen according to their statistical weight, have overlap $q$ and is
\begin{equation}
	P_J(q) = \sum_{\alpha,\beta} w_\alpha w_\beta \, \delta(q - q_{\alpha \beta}).
\end{equation}
The index $J$ simply means that this quantity is instance dependent, so we average on the disorder to obtain $P(q) = \overline{P_J(q)}$.\\
Now, one can prove that~\cite{Mezard1987book}
\begin{equation}\label{eq::physics_of_replicas_-1}
	q^{(k)} = \frac{1}{N^k} \sum_{i_1,\dots,i_k} \overline{\left< \sigma_{i_1} \cdots \sigma_{i_k}\right>^2} = \int dq \, P(q) \, q^k.
\end{equation}
These quantities can be computed also exploiting the replica method. Consider $q^{(1)}$,
\begin{equation}\label{eq::physics_of_replicas_-11}
	q^{(1)} = \frac{1}{N} \sum_i \overline{\left< \sigma_i \right>^2}.
\end{equation}
We can insert the identity $1=\lim_{n \to 0}Z^n$ and write
\begin{equation}
\begin{split}
	q^{(1)} = \frac{1}{N} \sum_i \lim_{n \to 0} \overline{\left< \sigma_i \right>^2 Z_J^n},
\end{split}
\end{equation}
where $Z_J$ is the partition function at fixed disorder.
Now, in the spirit of the replica trick, we consider $n$ integer. We have
\begin{equation}\label{eq::physics_of_replicas_0}
\begin{split}
	q^{(1)} & = \lim_{n \to 0} \frac{1}{N} \sum_i \overline{Z_J^{n} \frac{1}{Z_J} \int D\sigma^a \, \sigma_i^a \, e^{-\beta H[\sigma_i^a]}  \frac{1}{Z_J} \int D\sigma^b \, \sigma_i^b \, e^{-\beta H[\sigma_i^b]}}\\
	& = \lim_{n \to 0} \frac{1}{N} \sum_i \overline{\int \left( \prod_{c=1}^n D\sigma^c \right) \sigma_i^a \, \sigma_{i}^b \, e^{-\beta \sum_{c=1}^n H[\sigma_i^c]}},
\end{split}
\end{equation}
where we considered $D\sigma^a = \prod_i \sigma^a_i \, \delta(N- \sum_i (\sigma_i^a) ^2)$.
Following the same step used to evaluate the free energy, we obtain
\begin{equation}\label{eq::physics_of_replicas_1}
	q^{(1)} = Q_{ab}^{(SP)},
\end{equation}
where we have necessarily $a\neq b$ because of the steps done in Eq.~\eqref{eq::physics_of_replicas_0} and $SP$ labels the value of the quantity $Q$ computed on the (correct) saddle point.
Notice that Eq.~\eqref{eq::physics_of_replicas_1} makes sense only if $Q_{ab}$ does not depend on the choice of the replica $a$ and $b$. This would have been true if the RS ansatz had been correct. Unfortunately, this is not the case. Therefore, as discussed in~\cite{Parisi1983,DeDominicis1983}, we need to average over the contribution of all the (different) pairs of replicas and we finally obtain
\begin{equation}\label{eq::physics_of_replicas_2}
	q^{(1)} = \lim_{n \to 0} \frac{2}{n(n-1)} \sum_{a<b} Q_{ab}.
\end{equation}
This can actually be generalized to
\begin{equation}
	q^{(k)} = \lim_{n \to 0} \frac{2}{n(n-1)} \sum_{a<b} Q_{ab}^k
\end{equation}
which, because of Eq.~\eqref{eq::physics_of_replicas_-1}, brings us to
\begin{equation}
	P(q) = \lim_{n \to 0} \frac{2}{n(n-1)} \sum_{a<b} \delta(q-Q_{ab}).
\end{equation}
This equation elucidates the physical meaning of the matrix $Q$: at the saddle point, the probability that two pure states have overlap $q$ is given by the \emph{fraction} of entries equal to $q$ in $Q$, or, equivalently, each entry $Q_{ab}=q$ implies the existence of to two pure states with overlap $q$.\\
There is only a little problem: the matrix $Q$ has $n(n-1)/2$ independent entries, and $n$ is going to zero! We can still define (and this is what is done) in some way the ``fraction'' of entries by considering $n$ integer (and large) and, only after all the formal manipulation, by sending $n$ to 0. But the physical intuition of $Q$ suffers this weird situation, and this is one of the reason why the replica method is rather ill-defined under a mathematical point of view. However, each single time this method has been carried out up to the end and later compared with exact methods or very precise simulations, the resulting free energy (or whatever observable one wants to compute) computed with replicas turned out to be correct. Therefore let us take as a guide the physical intuition build around the matrix $Q$ to propose a new ansatz to overcome the problems encountered with the RS one.

\subsubsection{The magic of replica-symmetry breaking}
Consider again our RS ansatz: because of the discussion on $Q$, we have seen that there is a correspondence between the values inside the matrix $Q$ and the (average) properties of the free energy landscape. In particular, the presence of a single variational parameter can be interpreted as an ansatz on the free energy landscape, that is the presence of a single pure states. Indeed consider two configurations: if we choose twice the same configuration we obtain an overlap of $Q_{aa}=1$, if we choose two different configurations in the pure state we have that their average (on the disorder and on the Gibbs measure inside the pure state) overlap is $q_0$. Since this picture turned out to be wrong (the corresponding saddle point is unstable), we need to assume the presence of more than one pure state.\\
However, the simplest possible ansatz in this direction if far from obvious, and it required a deep intuition pointed out for the first time by Parisi~\cite{Parisi1979_1}: we consider that there are many pure states of ``size'' $m$, and two possible values of the overlap between configurations taken from them, that is $q_1$ if the two configurations belong to the same pure state, $q_0$ if they belong to two different pure states. Notice that this interpretation implies that $q_1\geq q_0$. This is called \emph{one-step replica-symmetry breaking} (1RSB) ansatz and the corresponding matrix is
\begin{equation}\label{eq::pspin_1rsb_q}
	Q = (1-q_1) \, \mathbb{I} + (q_1 - q_0) \, \mathbb{E} + q_0 \,\mathbb{C},
\end{equation}
where $\mathbb{I}$ is the identity matrix, $\mathbb{E}$ is a block diagonal matrix, where each block is a $m \times m$ block with all entries equal to 1 and $\mathbb{C}$ is a matrix with constant entries equal to 1.\\
By using this form of $Q$ in Eq.~\eqref{eq::pspin_f}, we find
\begin{equation}\label{eq::pspin_f_1rsb}
\begin{split}
	- 2 \beta f_{\text{1RSB}} = & \frac{\beta^2}{2} \left( 1 + (m-1) q_1^p - m q_0^p \right) + \frac{m-1}{m} \log(1-q_1) + \\
	& + \frac{1}{m} \log\left( m (q_1-q_0) + 1- q_1 \right) + \frac{q_0}{m(q_1-q_0) + 1-q_1} -2 S(\infty) .
\end{split}
\end{equation}
The details of this computation are given in Appendix~\ref{app::pspin}. The parameters are such that $f_{\text{1RSB}}$ is minimum, so they can be found by extremizing Eq.~\eqref{eq::pspin_f_1rsb}. Notice that the 1RSB ansatz includes the RS one, since taking $m=1$ or $q_1=q_0$ gives back the RS free energy density. What we have done, in other words, is to enlarge our ansatz to search for new, stable, saddle points, in a way suggested by the underlying physical interpretation. \\

We have that the equation $\frac{\partial}{\partial q_0} f_{\text{1RSB}} = 0$ implies $q_0=0$ to have a solution which is different from the unstable RS under the critical temperature. The other two equations are:
\begin{equation}
	(m-1) \left( \frac{\beta^2}{2} p q_1^{p-1} - \frac{q_1}{(1-q_1) (1 + (m-1)q_1)} \right) = 0
\end{equation}
and
\begin{equation}
	\frac{\beta^2}{2} q_1^p + \frac{1}{m^2} \log(\frac{1-q_1}{1 + (m-1) q_1}) + \frac{q_1}{m (1 + (m-1)q_1)} =0 .
\end{equation}
The $m=1$ solution of the first equation makes the 1RSB ansatz to coincide with the RS one, and is the only solution for $T>T_c$. For $T<T_c$ (notice that this critical temperature is different from the one where the unstable replica-symmetric solution appears, see Fig.~\ref{fig::pspin_rsb}), another solution with $m\neq1$ appears, and actually is the one which gives the most relevant and \emph{stable} saddle point. A plot of the situation is given in Fig.~\ref{fig::pspin_rsb}.\\
Therefore the system at the critical temperature $T_c$ has a phase transition between the paramagnetic phase and the so called ``spin glass'' phase, where the order parameter $q^{(1)} = (1-m) q_1$ ($q^{(1)}$ is defined in Eq.~\eqref{eq::physics_of_replicas_-11}, and because of Eq.~\eqref{eq::physics_of_replicas_2}, $m$ and $q_1$ are the values of the variational parameters at the saddle point) starts to be different from 0. 

\begin{figure}
	\centering
	\includegraphics[width=0.95\columnwidth]{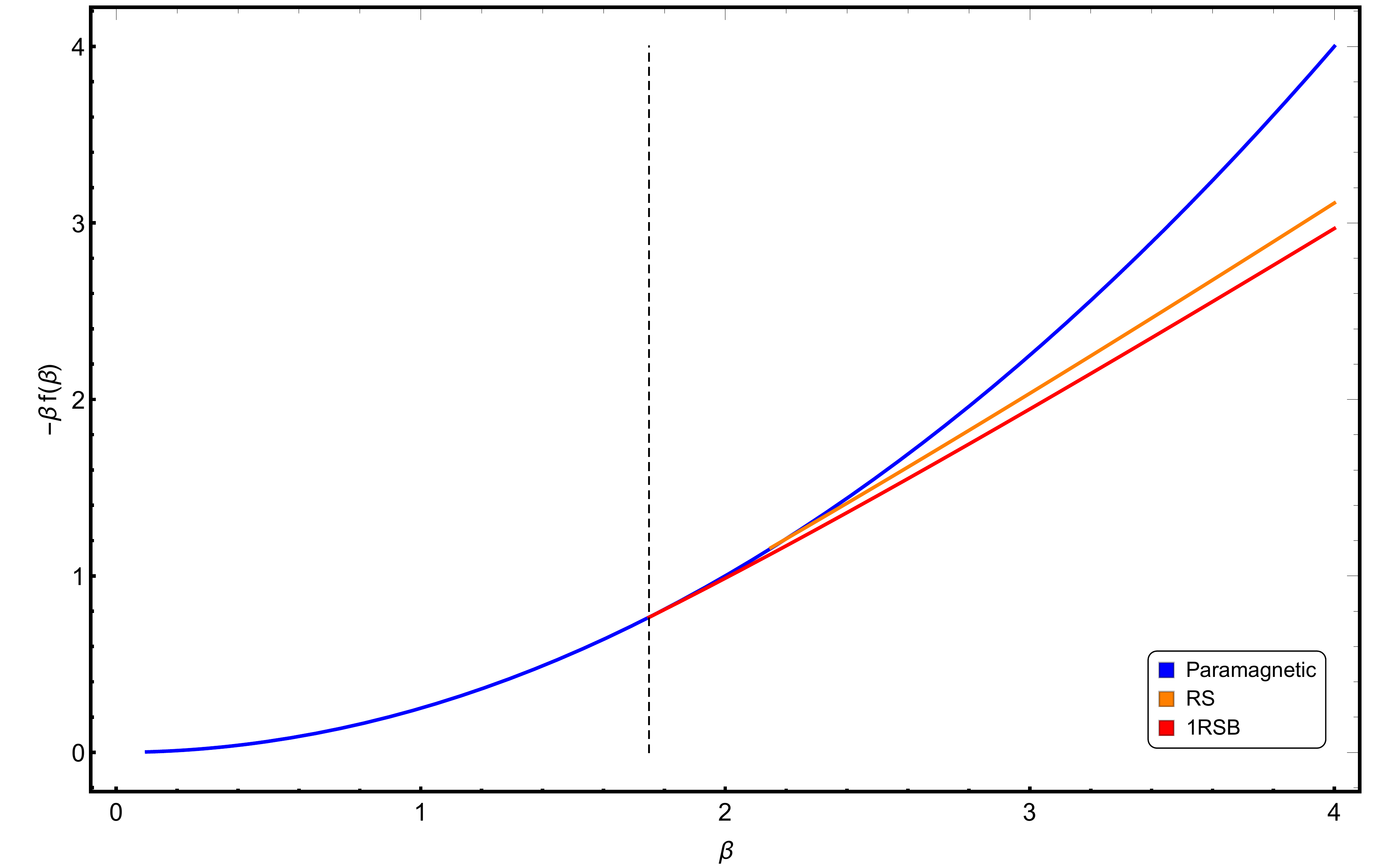}
	\caption{Numerical evaluation of the free energy density of the $p$-spin model with $p=3$, with the various ans\"atze . Notice that the function plotted is $-\beta f(\beta)$ (without the constant term $S(\infty)/\beta$) at the saddle points obtained by using the various ans\"atze described in Sec.~\ref{sec::pspin}: the blue curve is given by the annealed computation, the orange curve is given by the RS ansatz and the red curve is given by the 1RSB ans\"atze. Since this function appear in a saddle-point integration (Eq.~\eqref{eq::pspin_zn_2}), the correct one is always the smallest.
	As we can see, for $\beta<\beta_c \approx 1.7$ the paramagnetic solution (blue curve) is the only saddle point and it is the correct one. For $\beta>\beta_c$ the 1RSB solution becomes the most relevant stable saddle point, while at a smaller temperature ($\beta \approx 2.2$) the replica symmetric solution appear, but this is an unstable and not relevant saddle point.}\label{fig::pspin_rsb}
\end{figure}

\subsubsection{Spin glass and optimization problems}

What we learned with the $p$-spin spherical model is that when we deal with disorder and frustration, it can happen that the free energy landscape breaks into a plethora of pure states, which are taken into account via a RSB ansatz (we write RSB and not 1RSB because for some other models, for example the Sherrington-Kirkpatrick one, a more sophisticated ansatz, called \emph{full replica symmetry breaking}, is needed). \\
The presence of pure states is related to metastable states, that is groups of configurations separated by free energy barriers which become infinitely high in the thermodynamical limit. In turn, the presence of such metastable states results in the so-called ``ergodicity breaking''. Intuitively, this means that if the system is in a given configuration in a metastable state, even in the presence of thermal fluctuations (up to a certain temperature) it will stay in the same metastable state ``forever'', even if there are other regions of the configuration space with the same (or lower) free energy (for a more precise definition, see \cite{Vulpiani2014}, chapter 2).

Another interesting fact is that spin glass models, such as the Edwards-Anderson or the Sherrington Kirkpatrick one, can be seen as COPs whose cost function is the model Hamiltonian. Actually, it has been proved that the COP consisting in finding the ground state of Edwards-Anderson models in dimensions greater than 3 are NP-hard~\cite{Barahona1982,Bachas1984}.

Putting together the spin-glass and the optimization-problem perspectives, we can learn something about COPs (or, at least, get an interesting point of view): the difficulty in finding an algorithm to solve in polynomial time some problems seems to be related to the presence of ergodicity breaking, and therefore to RSB, in their thermodynamics. Actually, as far as we know there are no cases where there is RSB for a problem which is in the P complexity class\footnote{at first sight, the XORSAT problem (a SAT where the clauses use the logic operation XOR instead of OR) could seem a counterexample: it is in the P complexity class, but shows 1RSB when the thermodynamics is studied. However, the tractable problem consists actually in answering the question ``does this system admit solutions?'', while the optimization problems, ``what is the configuration that minimizes the number of FALSE clauses'' is NP-hard. Clearly, the thermodynamic can only say something about the optimization problem, or the generalized decision problem where we ask (for any given $n$): ``does this system admit a configuration which has up to $n$ FALSE clauses?''.}. On the opposite, it can happen that a NP-hard problem can be solved via the RS ansatz. This could be related to the fact that all the discussions about the energy landscape that we have done here are in fact about the \emph{typical} situation, while a problem is NP-hard even if only \emph{one} instance is hard (for each known algorithm).
In other words, consider a NP-hard COP. To study the thermodynamics of the corresponding disordered system, as already discussed, we need to introduce an ensemble of instances and a probability measure on it, obtaining a RCOP. Now, it can be that the ``hard'' instances belong to the ensemble but have zero weight in the thermodynamical limit for a certain choice of probability measure. Therefore also if the speculated connection between RSB and NP-hardness is correct, we would not see RSB in the thermodynamics of a problem unless we change in a suitable way our probability measure.


\section{Large deviations}\label{sec::lardev}
The standard approach of spin glass theory regards only the average over the disorder (or sometimes, also the variance) of some quantities, such as we have seen in Sec.~\ref{sec::pspin} with the p-spin spherical model free energy. On the opposite, the standard perspective of complexity theory is based on the idea of worst-case scenario. \\
A possible way to reduce the gap between these two fields is the large deviation theory. Basically, as we will see in a minute, large deviation theory (LDT) deals with the non-typical properties of random variables which depends on many other random variables. We will now introduce briefly the basic concepts of LDT, while for a more formal and comprehensive discussion we suggest to read one of the many good books~\cite{Ellis2007,Vulpiani2014} or the beautiful review~\cite{Touchette2009}.

\subsubsection{Large deviation principle}
We introduce now the large deviation principle (LDP). Consider a random variable $A_N$, which depends on an integer $N$. Let $p_{A_N}(a)$ be the probability density of $A_N$, such that $\int_B p_{A_N}(a) da = P(A_N \in B)$ is the probability that $A_N$ assumes a value in the set $B$. We say that for $A_N$ a LDP holds if the limit
\begin{equation}
	\lim_{N\to\infty} - \frac{1}{N} \log(p_{A_N}(a))
\end{equation}
exists, and in that case we introduce the \emph{rate function} of $A_N$, $I$, as
\begin{equation}\label{eq::ldp_general}
	\lim_{N\to\infty} - \frac{1}{N} \log(p_{A_N}(a)) = I(a).
\end{equation}
In other words, in a less precise but more transparent way we can write
\begin{equation}
	p_{A_N}(a) \simeq e^{-N I(a)},
\end{equation}
where the meaning of ``$\simeq$'' is given by Eq.~\eqref{eq::ldp_general}. Sometimes, as we will see, the situations where $I=\infty$ or $I=0$ in an interval are of particular interest. In these cases we say, respectively, that $p_{A_N}(a)$ decays faster than exponentially in $N$ (these are the so-called \emph{very} large deviations) or that it decays slower than exponentially.
LDT essentially consists in taking a random variable of interest and trying to understand whether a LDP holds for it, and what is its rate function.

\subsubsection{Recovering the law of large numbers and the central limit theorem}
A first comment on the LDP is that it encompasses both the law of large numbers and the central limit theorem. Indeed, consider a set of $N$ IID random variables\footnote{all this can be extend to non-IID random variables, provided that they are not too much correlated, but we will use IID random variables to keep things as simple as possible.} $x_1,\dots, x_N$ with finite mean $\left< x_i \right> = x$ and variance $\left< x \right>^2 - x^2 = \sigma^2$. Their empirical average is
\begin{equation}
	A_N = \frac{1}{N} \sum_{i=1}^{N} x_i
\end{equation}
and the law of large numbers guarantees that
\begin{equation}\label{eq::lln}
	\lim_{N\to\infty} P(|A_N - x|>\epsilon) = 0 
\end{equation}
for each $\epsilon>0$.

The central limit theorem for IID random variables extends this result by giving the details on the shape of the probability of obtaining $A_N$ inside an interval $[a,b]$:
\begin{equation}
	\lim_{N\to\infty} P(\sqrt{N}(A_N - x) \in [a,b]) = \frac{1}{\sqrt{ 2 \pi \sigma^2}}\int_a^b dz \, e^{- \frac{z^2}{2 \sigma^2}},
\end{equation}
The analogous for the probability density is
\begin{equation}
	\lim_{N\to\infty} - \frac{1}{N} \log(p_{A_N}(a)) =  \frac{(a-x)^2}{2 \sigma^2}.
\end{equation}

Now, consider that a LDP holds for our empiric average $A_N$. We then have
\begin{equation}
	p_{A_N}(a) \simeq e^{- N I(a)}.
\end{equation}
The only values of $a$ such that $\lim_{N\to\infty} p_{A_N}(a) \neq 0$ have to be all the values $a^\star$ such that $I(a^\star)=0$. Therefore we recover the law of large number by noticing that $a^\star= x$, where $x$ is the one in Eq.~\eqref{eq::lln}. Moreover, we can expand $I$ around its zero, $x$, and obtain
\begin{equation}\label{eq::ld_vs_clt}
	I(a) = \frac{1}{2} I''(x) (a-x)^2 + o((a-x)^3),
\end{equation}
where the ``small $o$'' notation means that we are neglecting terms of order $(a-x)^3$ or less relevant in the limit $a\to x$.
Therefore, we have
\begin{equation}
p_{A_N}(a) \simeq e^{- N \left(\frac{1}{2} I''(x) (a-x)^2 + o((a-x)^3) \right)},
\end{equation}
which is the central limit theorem, after identifying $I''(x) = \sigma^2$. Notice that this approximation is valid up to $|a-x| \sim N^{-1/2}$, while for larger distances from the average one needs to keep into account higher terms in the expansion Eq.~\eqref{eq::ld_vs_clt}. \\
If one has the full form of $I$, then the probability of each value of $A_N$ can be computed, also for values very far from the average $x$. This is the reason why this field is called \emph{large deviation} theory and in this sense we can consider LDT a generalization of the central limit theorem and of the law of large numbers.

\subsubsection{The G\"artner-Ellis theorem}\label{sec::scgf_props}
But how to compute rate functions? Unfortunately, there is not a general way. However, often the rate function can be computed by means of the \emph{G\"artner-Ellis} theorem, which in its simpler formulation states the following. \\
Consider the random variable $A_N$, where $N$ is an integer parameter. The \emph{scaled cumulant generating function} (SCGF) is defined as
\begin{equation}
	\psi(k) = \lim_{N\to\infty} \frac{1}{N} \log \left< e^{N k A_N} \right>,
\end{equation}
where $k\in \mathbb{R}$ and
\begin{equation}
	\left< e^{N k A_N} \right> = \int da \, p_{A_N}(a) \, e^{N k a}.
\end{equation}
If $\psi(k)$ exists and is differentiable for all $k\in\mathbb{R}$, then $A_N$ satisfies a large deviation principle, with rate function $I$ given by the Legendre transform of the SCGF, that is
\begin{equation}\label{eq::ld_ratefunction_gen}
	I(a) = \sup_{k\in\mathbb{R}} \left( k a - \psi(k) \right).
\end{equation}

We will not prove this theorem here, but the interested reader can find the proof, for example, on Ellis' book~\cite{Ellis2007} or on Touchette's review~\cite{Touchette2009}.

Here we will limit ourself to some consideration about the SCGF. First of all, its name is given by the fact that
\begin{equation}
	\left. 	\frac{\partial^n}{\partial k^n} \psi(k) \right|_{k=0} = \lim_{N\to\infty} N^{n-1} C_n,
\end{equation}
where $\partial_k^n$ denotes $n$ derivatives with respect to $k$ and $C_n$ is the $n$-th order cumulant of $A_N$. In particular,
\begin{equation}
	\left. \frac{\partial}{\partial k} \psi_k \right|_{k=0} = \lim_{N\to \infty} \left< A_N \right>
\end{equation}
and
\begin{equation}
	\left. \frac{\partial^2}{\partial k^2} \psi_k \right|_{k=0} = \lim_{N\to \infty} N \left( \left< A_N^2 \right> - \left< A_N \right>^2\right),
\end{equation}
that is the first and second derivatives of the SCGF $\psi(k)$ evaluated in $k=0$ are, respectively, the mean and the variance (times $N$) of $A_N$, in the limit of large $N$.\\
The SCGF $\psi(k)$ has some remarkable properties, that will be useful in the following:
\begin{enumerate}
	\item $\psi(0)=0$, because of normalization of the probability measure.
	
	\item The function $\psi(k)$ is convex, as can be proven by using the H\"older inequality:
	\begin{equation}
	\left< X Y \right> \le \left<X^{1/p}\right>^{p} \left<Y^{1/q}\right>^{q},\qquad 0\le p,\quad q\le 1,\quad p + q = 1.
	\label{eq:LDP_Holder}
	\end{equation}
	Indeed, if we choose $X=e^{p k_1 NA_N}$, $Y=e^{(1-p) k_2 NA_N}$, so that
	\begin{equation}
		\left< e^{[p k_1 + (1-p) k_2]NA_N} \right> \le \left<e^{ k_1 NA_N}\right>^{p} \left<e^{k_2 NA_N}\right>^{1-p},
	\end{equation}
	we now take the logarithm, divide by $N$ and, since this inequality is valid for all $N$, we can take the limit $N\to\infty$ to obtain 
	\begin{equation}
	\psi( p k_1 + (1-p) k_2) \le p \, \psi(k_1) + (1-p) \, \psi(k_2).
	\end{equation}
	
	\item The function $\psi(k)/k$ is a monotonic non-decreasing function, as can be proven from another usage of the H\"older inequality: this time we choose $X=e^{k p N A_N}$, $Y=1$. We have now
	\begin{equation}
	\left<e^{k p N A_N}\right> \le \left<e^{ k N A_N }\right>^{p}
	\end{equation}
	and we take the logarithm, divide by $N$ and get
	and taking the log	\begin{equation}
	\frac{1}{N} \log \left<e^{k p N A_N}\right> \le \frac{p}{N} \log \left<e^{ k N A_N }\right>,
	\end{equation}
	which implies
	\begin{equation}
	\psi (p k) \le p \, \psi(k).
	\end{equation}
	Since $p$ is an arbitrary number between 0 and 1, we have
	\begin{equation}
		\frac{\psi (p k)}{p} \le \, \psi(k)
	\end{equation}
	and, by dividing by $k$, we obtain that the function $\psi_N(k)/k$ must be non decreasing.

\end{enumerate}

\subsection{Large deviations of the p-spin model}\label{sec::lardev-pspin}
In this section, we follow Pastore, Di Gioacchino and Rotondo~\cite{Pastore2019} in their discussion about the large deviations of the p-spin spherical model introduced in Sec.~\ref{sec::pspin}. An interesting relation between large deviation and replica method (and replica symmetry breaking) is firstly elucidated, then exploited. We notice that similar techniques can be applied to RCOPs, once they are written as a spin glass problem.

\subsubsection{Replica trick and large deviation theory}
As we have seen, the theory of disordered systems has been mainly developed to describe the average behavior of physical observables, which one hopes to coincide with the typical one (this is true if the physical observable under discussion is self-averaging).

However, as it has been argued since the early days of the subject, one can employ spin glass techniques in a more general setting, to estimate probability distributions \cite{Toulouse1981} and fluctuations around the typical values \cite{Tanaka1989,Crisanti1990} of quantities of interest. More recently, Rivoire \cite{Rivoire2005}, Parisi and Rizzo \cite{Rizzo2008_1,Rizzo2008_2,Rizzo2010_1,Rizzo2010_2} and others \cite{Andreanov2004,Nakajima2008,Nakajima2009} followed this line of thought, providing a bridge between spin glasses (and disordered systems more in general, as in \cite{Malatesta2019_2}) 
and the theory of large deviations.
%
The key quantity providing the bridge is:
\begin{equation}
G(k) = \lim_{N\to\infty} -\frac{1}{\beta N} \log \overline{Z_N^k},
\label{eq::paperld_SCGF}
\end{equation}
where $Z_N$ is the partition function for a system of size $N$ and the bar above quantities denotes average over disorder. The argument of the logarithm is the averaged replicated partition function and $k$ is the so-called replica index. We have changed our notation for the replica number to emphasize that we will not deal here only with vanishing number of replicas. \\
From the viewpoint of large deviation theory, $S(k)$ is simply related to the scaled cumulant generating function (SCGF) of the free energy $f=\lim_{N\to\infty}f_N$ by
\begin{equation}
\psi(k) = \lim_{N\to\infty} \frac{\log \overline{e^{k N f_N}}}{N} = -\beta G(-k/\beta).
\label{eq::paperld_psifromG}
\end{equation}

We are interested in obtaining, by using the G\"artner-Ellis theorem, as much information as possible on the full form of the rate function $I (x)$. To do that, one needs to work out the SCGF for finite replica index $k$. This problem is clearly equivalent to determine the full analytical continuation of the averaged replicated partition function from integer to real number of replicas $k$ and it was extensively investigated in the early stage of the research in disordered systems in order to understand the manifestation of the (at that time surprising) mechanism of replica symmetry breaking we encountered in Sec.~\ref{sec::pspin}. Since these results are particularly interesting from the more modern large deviation viewpoint, we briefly mention the main ones in the following.\\
Van Hemmen and Palmer~\cite{vanHemmen1979} were the first ones to observe that the expression in Eq.~\eqref{eq::paperld_SCGF} must be a convex function of the replica index $k$, as we discussed Sec.~\ref{sec::scgf_props}. 
Shortly after, Rammal~\cite{Rammal1981} added that $\psi(k)/k$ must be monotonic. 
However, in some situations, the replica symmetric (RS) ansatz gives a trial SCGF which is not convex, or such that $\psi(k)/k$ is not monotonic. This problem has been analyzed for the first time in the context of the Sherrington-Kirkpatrick model. After Parisi introduced his remarkable hierarchical scheme for replica symmetry breaking, Kondor \cite{Kondor1983} argued that his full RSB solution was very likely to provide a good analytical continuation of Eq.~\eqref{eq::paperld_SCGF}, not only around $k=0$.

These results may be considered nowadays as the initial stage of a work that attempted to give mathematical soundness to the replica method. Although this vaste program is mostly unfinished, Parisi and Rizzo realized that the original analysis presented by Kondor is fundamental to investigate the large deviations of the free-energy in the SK model. Large deviations have been examined only for a few other spin glass models: Gardner and Derrida discussed the form of the SCGF in the random energy model (REM) in a seminal paper \cite{Gardner1975}, and many rigorous results have been established later on \cite{Fedrigo2007}; on the other side of the story Ogure and Kabashima \cite{Ogure2004,Ogure2009_1,Ogure2009_2} considered analyticity with respect to the replica number in more general REM-like models; Nakajima and Hukushima investigated the $p$-body SK model \cite{Nakajima2008} and dilute finite-connectivity spin glasses \cite{Nakajima2009} to specifically address the form of the SCGF for models where one-step replica symmetry breaking (1RSB) is exact.  

In this section we add one more concrete example to this list, considering the $p$-spin spherical model. In zero external magnetic field, we will show that the 1RSB calculation at finite $k$ produces a SCGF with a linear behavior below a certain value $k_c$ and a nice geometrical interpretation of this, dating back to Kondor's work on the SK model \cite{Kondor1983}, is discussed. Accordingly, the rate function is infinite for fluctuations of the free energy above its typical value, which are then more than exponentially suppressed in $N$, giving rise to a regime of \emph{very-large} deviations. This happens for several other spin glass problems, as discussed for example in \cite{Rizzo2010_1}, and many other systems showing extreme value statistics\cite{Dean2008}.

The situation changes dramatically when a small external magnetic field is turned on: the rate function becomes finite everywhere, although highly asymmetric around the typical value, and the very-large deviation feature disappears accordingly. We explain intuitively the reason of this change of regime in light of the geometrical interpretation discussed for the case without magnetic field, and argue that the introduction of a magnetic field could act as a regularization procedure for resolving the anomalous scaling of the large deviation principle for this kind of systems.

\subsubsection{Large deviations of the \texorpdfstring{$p$}{p}-spin spherical model free energy}
\label{sec::paperld_nofield}


\begin{figure}
	\centering
	\includegraphics[width=0.95\columnwidth]{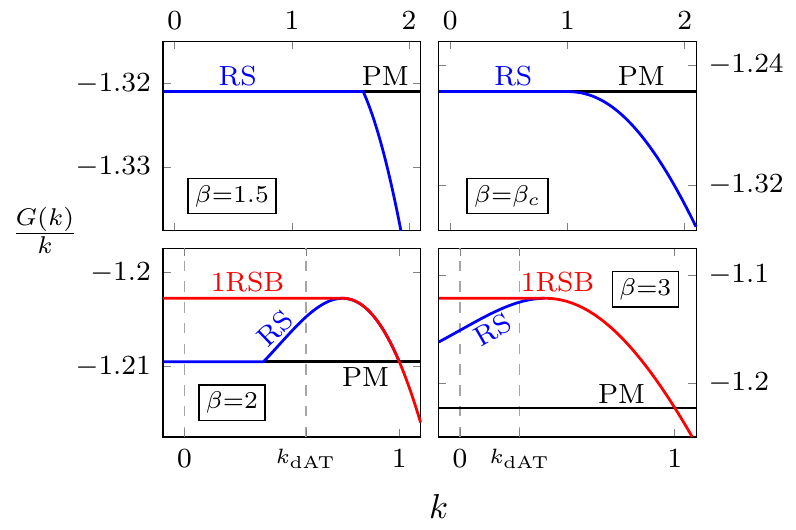}
	\caption{The function $G(k)/k$ for the ($p=3$)-spin in zero external magnetic field, for different values of $\beta$. Top-left: at high temperature ($\beta=1.5$) the 1-RSB anstatz coincides with the RS one (blue curve); the solution joins the paramagnetic line (in black) in a point $k_c>1$, where the function is not differentiable. Top-right: at $\beta=\beta_c\approx 1.706$, the junction is in $k_c=1$ and becomes smooth. Bottom line: for $\beta=2$ (left) and $\beta=3$ (right), the 1RSB solution (red curve) departs from the RS one and becomes a straight line for all the $k<k_c$, which is the point where the RS function loses its monotonicity. The critical value $k_c$ approaches zero for $\beta\to\infty$.}
	\label{fig::paperld_G(k)/k}
\end{figure}

We start our analysis from Eq.~\eqref{eq::pspin_zn_2}. After the integration on the $\lambda$ degrees of freedom, the partition function is (up to finite-size corrections in $N$):
\begin{equation}\label{eq::paperld_Zk_0}
\overline{Z_N^k} = \int DQ  \, e^{-N S(Q)}\,,
\end{equation}
where 
\begin{equation}\label{eq::paperld_action}
S(Q) = -\frac{\beta^2}{4} \sum_{\alpha, \beta = 1}^k Q_{\alpha \beta}^p - \frac{1}{2}\log \det \mathbf{Q} - k S(\infty).
\end{equation}
To evaluate the integrals on $Q$ we use again the saddle point method together with the 1RSB ansatz, which is formulated in terms of the three parameters $(q_1, q_0, m)$ in Eq.~\eqref{eq::pspin_1rsb_q}.\\
We compute $S(Q)$ in terms of the 1RSB parameters as discussed in Appendix \ref{app::pspin}, but now we do not take the limit $k\to 0$ and we obtain: 
\begin{multline}\label{eq::paperld_g_h_0}
S(k;q_0,q_1,m) = - \frac{(\beta J)^2}{4}\left[k + k(m-1)q_1^p + k(k-m)q_0^p \right]\\
-\frac{k(m-1)}{2m}\log\left({\eta_0}\right) - \frac{k}{2m}\log\left({\eta_1}\right)- \frac{1}{2}\log\left(1+ \frac{ k q_0}{\eta_1}\right)- k s(\infty),
\end{multline}
where $\eta_0 = 1- q_1$ and $\eta_1 = 1-(1-m)q_1$ are the two different eigenvalues of the 1RSB matrix $Q$ once we use that $q_0=0$ at the saddle point.
This functional is evaluated numerically at the saddle point $(\bar{q}_1, \bar{q}_0, \bar{m})$ for the 1RSB parameters for each value of $k$. 
The three parameters take values in the domains $q_1\in[0,1]$, $q_0\in[0,q_1]$, $m\in[1,k]$ (if $k>1$) or $m\in[k,1]$ (otherwise), and for $k<1$ the saddle point is obtained with a maximization of the functional instead of a minimization, as usual within the replica-method framework.
Using Eq.~\eqref{eq::paperld_psifromG}, we obtain a SCGF $\psi(k)$ which becomes linear above a certain value $k=k_c$.\\
To ease the visualization of this feature, in Fig.~\ref{fig::paperld_G(k)/k} we plot the function
$G(k)/k = S(k;\bar{q}_1, \bar{q}_0, \bar{m})/(k\beta)$ which, when $\psi(k)$ is linear, becomes an horizontal line intercepting the vertical axis in $f_{\text{typ}}$. 
The figure does not change qualitatively for $p\ge 3$. For the $p=2$ case, at low temperature the 1RSB ansatz reduces to the RS one (that is, $\bar{q}_1=\bar{q}_0$) as long as $k \geq 0$, therefore the typical values of all the thermodynamic quantities are obtained under the RS ansatz. On the opposite, for $k<0$ we need to introduce again the 1RSB ansatz which, as in the $p \geq 3$ case, gives the linear behavior of the SCGF. In other words, $k_c = 0$ for the 2-spin spherical model for $\beta>\beta_c$.

Before turning to the evaluation of the rate function, we discuss an interesting geometrical interpretation of the SCGF shape. To this aim, let us consider the RS ansatz (that is, Eq.~\eqref{eq::paperld_g_h_0} with $q_1=q_0=q$ and $m=1$).
As we can see in Fig.~\ref{fig::paperld_G(k)/k}, the RS solution (blue curve) is not monotonic for $\beta<\beta_c$. 
But as we have seen, $G(k)/k$ has to be a monotonic quantity and therefore the RS solution can be ruled out. 
We can check that the 1RSB solution gives a perfectly fine monotonic $G(k)/k$ (red curve in Fig.~\ref{fig::paperld_G(k)/k}), as one could expect due to the fact that this ansatz gives the correct typical free energy for this model. Interestingly, however, exactly the same monotonic curve can be obtained by using a much simpler geometric construction: just consider the RS solution, which is the right one for large $k$, and when $G(k)/k$  starts to be non-monotonic continue with a straight horizontal line (in the $G(k)/k$ vs $k$ plot). 
This construction actually dates back to Rammal \cite{Rammal1981} and is discussed in more detail in Appendix \ref{app::paperld_rammal}. 
Here we limit ourselves to notice that $G(k)/k$ obtained by using the 1RSB ansatz or the Rammal construction are the same because of the following facts: (i) for $k > k_c$ the 1RSB and RS ans\"atze coincide ($\bar{q}_1 = \bar{q}_0 = q \neq 0$) and $k_c$ is exactly the point where $G(k)/k$ is not monotonic anymore if one uses the RS ansatz; (ii) from the saddle point equations obtained by extremizing Eq.~\eqref{eq::paperld_g_h_0} when $k < k_c$, one obtains $\bar{q}_0=0$; (iii) the remaining saddle point equations fix $q_1$ and $m$, and one can see that these equations are identical to those needed to perform the Rammal construction, which fix the point $k_c$ and the parameter of the RS ansatz $q$.

\subsubsection{Rate function and very large deviations}

\begin{figure}
	\centering
	\includegraphics[width=0.95\columnwidth]{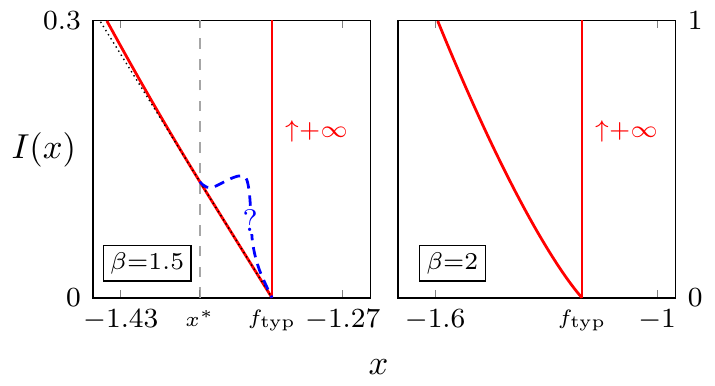}
	\caption{Rate function of the free energy for the ($p=3$)-spin in zero external magnetic field, for different values of $\beta$. The fluctuations above the typical value correspond to the linear part of the SCGF, so that the Legendre transformation gives an infinite rate function. The fluctuations below the typical value are described by the branch in red. For $\beta=1.5<\beta_c$ (left), as the SCGF is not differentiable, we obtain only the convex-hull of the true rate function; in the interval $[x^*,f_{\text{typ}}]$, where our result gives a straight segment (the part of the curve overlapping the dotted line), the true, unknown rate function is represented by the curve in blue. For $\beta=2>\beta_c$ (right) the SCGF is smooth and the G\"artner-Ellis theorem applies.}
	\label{fig::paperld_rate}
\end{figure}

Starting from the SCGF, we perform a numerical Legendre transformation to obtain the rate function according to Eq.~\eqref{eq::ld_ratefunction_gen}. The result is shown in Fig.~\ref{fig::paperld_rate} for different values of $\beta$. 
The rate function displays the following behavior: 
\begin{itemize}
	\item for $x=f_{\text{typ}}$, it is null as expected; 
	\item for $x<f_{\text{typ}}$, $I(x)$ is finite, indicating that a regular large deviation principle holds for fluctuations below the typical value. When $\beta>\beta_c$ the SCGF is smooth, so we obtain the rate function via the Gartner-Ellis theorem. On the other hand, when $\beta<\beta_c$ the SCGF is not differentiable in a point (see Fig.~\ref{fig::paperld_G(k)/k}), so we are only able to obtain the convex hull of the rate function (see Fig.~\ref{fig::paperld_rate});
	\item for $x>f_{\text{typ}}$, $I(x)=+\infty$. This is due to the linear behavior of the SCGF below $k_c$ discussed in the previous section and it is a signature of an anomalous scaling with $N$ of the rare fluctuations above the typical value.
\end{itemize}
An ambitious goal would be the identification of the correct behavior with $N$ of these very large deviations. Indeed, a more general way of stating a large deviation principle is
\begin{equation}\label{eq::paperld_ldp2}
P(f_N \in [x, x + d x]) \sim 
\begin{cases}
e^{- a_N I_-(x)}d x &	\text{if $x\le f_{\text{typ}}$}\,,\\
e^{- b_N I_+(x)}d x	&	\text{if $x>f_{\text{typ}}$}\,,\\
\end{cases}
\end{equation}
where $a_N, b_N \to \infty$ when $N\to\infty$. In other words, the fluctuations resulting in values of $x$ lower than $f_{\text{typ}}$ are given by the rate function $I_-(x)$, while those resulting in values larger than $f_{\text{typ}}$ have rate function $I_+(x)$, but with different scalings $a_N$, $b_N$. In our case, we have $a_N \sim N$, then the rate function defined in Eq.~\eqref{eq::ldp_general} can be written as
\begin{equation}
I(x) \sim \begin{cases}
I_-(x) &	\text{if $x\le f_{\text{typ}}$}\,,\\
\frac{b_N}{N} I_+(x)	&	\text{if $x>f_{\text{typ}}$}\,,\\
\end{cases}
\end{equation}
with  $b_N/N \to \infty$. For this reason, fluctuations above the typical value are referred to as ``very large deviations''. The physical explanation of the substantial difference in scaling of the deviations of thermodynamical quantities below and above their typical values resides in the different number of elementary degrees of freedom involved to obtain the corresponding fluctuation: while in the first case it is sufficient that only one of the elementary variables assumes an anomalous value below its typical, the others being fixed, in the second case all the variables have to fluctuate, a joint event with probability heavily suppressed with respect to the first one.

This argument shows the importance of the resolution of the anomalous scaling behavior leading to the very large deviations we explained above. In general, however, although the G\"artner-Ellis theorem can be extended to find rate functions for large deviation principles with arbitrary speed $a_N$, $b_N$, we lack techniques to compute the asymptotic scaling of $a_N$ and $b_N$ for large $N$, because of additional inputs needed to calculate the corresponding SCGF with a saddle-point approximation (for some other systems this problem has been solved with ad-hoc methods \cite{Andreanov2004,Dean2008}, while in \cite{Rizzo2010_1} a method is proposed in the context of the SK model).

In the next section we present the main result of our work, which could be useful to study this anomalous kind of fluctuations also in other problems: through an extension of the replica calculation to the case with an external magnetic field, we are able to numerically check that the very large deviation effect disappears. More in detail, we obtain that with a magnetic field, no matter how small, not only $a_N \sim N$ as before, but also $b_N \sim N$.

\subsection{A ``cure'' for very-large deviations: \texorpdfstring{$p$}{p}-spin model in a magnetic field}
\label{sec::paperld_field}

In this section we generalize the previous discussion to the case of non-zero magnetic field. The Hamiltonian for the model is
\begin{equation}
H = H_p  - h\sum_{i = 1}^N \sigma_i \, ,
\end{equation}
where $H_p$ is the p-spin Hamiltonian and $h$ represents an external magnetic field coupled with the spins. 

The computation of the SCGF at $h\neq 0$ goes beyond the approach of the work by Crisanti and Sommers, who only considered the typical case. In contrast to the problem with $h=0$, where the finite-$k$ calculation consists of a quite straightforward generalization of the standard one, here a more substantial effort is needed to extend the $k=0$ result. The derivation is quite technical, therefore to emphasize the discussion about the large deviation of the free energy we report here only the final expression we obtained for the SCGF, postponing the details in Appendix \ref{app::pspin_ld}.
The functional $g(\mathbf{q})$ in the 1RSB ansatz, for finite $k$ is
\begin{multline}\label{eq::paperld_g_h_non_0}
S(k;q_0,q_1,m) = -\frac{(\beta J)^2}{4}\left[k + k(m-1)q_1^p + k(k-m)q_0^p \right]
-  \frac{k\hat{q}_-}{2( \eta_2 -  k\hat{q}_{-})} \\
-\frac{k(m-1)}{2m}\log\left({\eta_0}\right) - \frac{k}{2m}\log\left({\eta_1}\right)- \frac{1}{2}\log\left(1 + \frac{k(q_0- \hat{q}_-)}{\eta_1}\right)\\
-\frac{(\beta h)^2}{2} k \left(\eta_2 - k\hat{q}_- \right) -k s(+\infty)\,,
\end{multline}
where $\hat{q}_-$ depends on the combination $\beta h$ and the parameters of the 1RSB ansatz (its full form is given in Eq.~\eqref{eq::paperld_hatq} of Appendix \ref{app::pspin_ld}) and now $\eta_0 = 1-q_1$, $\eta_1 = 1 - (1-m) q_1- m q_0$ and $\eta_2 = 1 - (1-m) q_1- (m-k) q_0 $ are the three eigenvalues of $Q$ (now we do not have anymore $q_0 = 0$). 

Again, we numerically compute and plot $G(k)/k = S(k;\bar{q}_1, \bar{q}_0, \bar{m})/(k\beta)$ in Fig.~\ref{fig::paperld_scgf_h_non_0}, where again $\bar{q}_1, \bar{q}_0, \bar{m}$ are the solutions of the saddle point equations, obtained by extremization of Eq.~\eqref{eq::paperld_g_h_non_0}. The most striking feature of these plots is the difference from those represented in Fig.~\ref{fig::paperld_G(k)/k}: all the horizontal lines disappear and their place is taken by curves (again given by the 1RSB ansatz) with non-null derivative.
Let us analyze more closely what is happening and why the external magnetic field is changing the behavior of the system.
As discussed in the last part of Sec.~\ref{sec::paperld_nofield}, one can apply the Rammal construction to correct the non-monotonic behavior of the RS version of $G(k)/k$ (plotted as a blue curve in Fig.~\ref{fig::paperld_scgf_h_non_0}). Exactly as in the $h=0$ case, the resulting function will be monotonic and will have an horizontal line, which is the smooth continuation of $G(k)/k$ from $k_m$, the point where it loses monotonicity. However, as one can see from Fig.~\ref{fig::paperld_scgf_h_non_0}), the result will not be the 1RSB solution.
This difference from the $h=0$ case can be seen as a consequence of the saddle point equations: now the equation for $q_0$ is non-trivial and so either $\bar{q}_0, \bar{q}_1$ and $\bar{m}$ depends on $k$ also in 1RSB phase, giving rise to the non-trivial behavior of $G(k)/k$ also for $k<k_c$. 
Notice that another interesting feature appears: when $h=0$ we have that $k_c$, the point where the 1RSB solution becomes different from the RS one, coincide with $k_m$, the point where $G(k)/k$ obtained by the RS ansatz loses monotonicity. With $h\neq 0$ we have that $k_c>k_m$ for $\beta>\beta_c$, that is the 1RSB branch departs from the RS one before (coming from large $k$) the point where $G(k)/k$ starts to be not monotonic.
Finally, we numerically checked that the shape of $G(k)/k$ below $k_c$ depends on $p$.

This change in the SCGF has an important effect, in turn, on the rate function: taking the numerical Legendre transformation of the SCGF we now obtain a continuous curve, meaning that very rare fluctuations are disappeared, see Fig.~\ref{fig::paperld_rate_h_non_0}. In other words, now the two quantities $a_N$ and $b_N$ introduced in Eq.~\eqref{eq::paperld_ldp2} are such that $a_N \sim N$ and $b_N \sim N$. This effect is present also for very small magnetic field, even though $I(x)$ is more and more asymmetrical around $x=f_{\text{typ}}$ as we decrease $h$.
This observation brings to a natural question, which for now remains open: can this effect be exploited to obtain insights on the very large fluctuations - that is how are they suppressed with the system size? And what is the corresponding (finite) rate function?

\begin{figure}
	\centering
	\includegraphics[width=0.95\columnwidth]{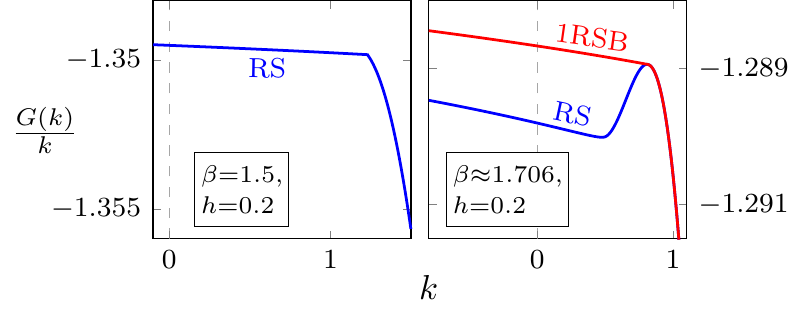}
	\caption{The function $G(k)/k$ for the ($p=3$)-spin in a magnetic field $h=0.2$, for different values of $\beta$: $\beta = 1.5<\beta_c(h)$ (left), $\beta=\beta_c(h=0)>\beta_c(h)$ (right). The application of a magnetic field washes out the linear behavior at small $k$ observed in zero magnetic field.
	}
	\label{fig::paperld_scgf_h_non_0}
\end{figure}

\begin{figure}
	\centering
	\includegraphics[width=0.95\columnwidth]{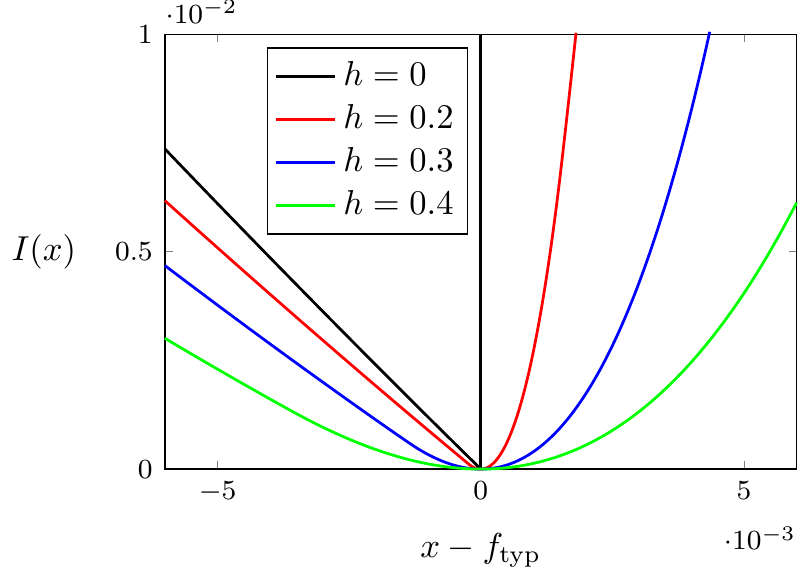}
	\caption{Rate function of the free energy for the ($p=3$)-spin at $\beta = 3$, for different values of the external magnetic field. The infinite branch of the rate functions in Fig.~\ref{fig::paperld_rate} is replaced by a curve gradually less steep as the magnetic field is increased.}\label{fig::paperld_rate_h_non_0}
\end{figure}

\chapter{In practice: from mean-field to Euclidean problems}\label{chap::second}
\chaptermark{From mean field to Euclidean problems}
From now on we will deal with a very specific class of COPs, the so-called \emph{Euclidean} problems. The main characteristics of these problems are:
\begin{itemize}
	\item an instance is specified by giving the positions of a certain number of points in a subspace (often compact) of $\mathbb{R}^d$;
	\item the cost function depends on the distances between pairs of these points;
	\item each of these problems allows for a natural definition in terms of a problem on a graph.
\end{itemize}
We will deal with certain specific problems, that is the matching and assignment problem, the traveling salesman problem and the 2-factor problem. \\ 
In all these cases, the RCOP version of these problems will be defined by considering an hypercube of side 1 and a certain factorized probability density for the point positions, that will then be IID random variables. Therefore the quantity of interest, which is in our case the cost of the solution, will be averaged over the point positions.\\
All these problems can be also studied in the so-called mean-field approximation, where instead of throwing the points according to a probability density and computing the distances, one directly chooses a probability density for the distances. If this probability density is factorized, each distance is a IID random variable and in this way we are neglecting correlations among distances. Notice that, on the opposite, these correlations emerge from the Euclidean structure of the space when we compute distances after having thrown the points, even if they are chosen in a IID way.

Often we will refer to these mean-field results to make a comparison with our finite-dimensional results, and also because under the mean-field approximation the replica method can be (most of the times) carried out to obtain the quantity of interest. 
On the other hand, in genuine Euclidean problems the emergence of the aforementioned correlations prevents us to successfully apply replica methods. To overcome this technical issue, we will deal with problems in low number of dimensions $d$ ($d=1$ and, when possible, $d=2$) since they are simpler, and we will focus on the search for a way to obtain the average cost of the solution without using the replica method. 

\section{Euclidean problems in low dimension}
\subsection{A very quick introduction to graph terminology}\label{sec::graph}
Here we introduce the key concepts about graphs that we will use profusely in the following.\\
Let us start with the definition of a graph: given a set $A$ of labels (typically $A=\mathbb{N}$), a graph $G$ is specified by two sets, the vertex set $V \subset A$ and the edge set $E \subset A\times A \times \dots \times A = A^\ell $ and we say that $G=(V, E)$, the element of $V$ are the \emph{vertices} or \emph{nodes} of $G$ and the elements of $E$ are the \emph{edges} or \emph{links} of $G$. \\
\emph{Multilinks} (or \emph{multiple edges}) are two or more edges connecting the same points. A \emph{self-loop} is an edge in which the same point appears more than once.
A graph without multilinks and self-loops is called \emph{simple} graph. From now on, we will consider always simple graphs with $E \subset A\times A$ (that is, $\ell=2$).\\
A graph is said to be \emph{undirected} if the following holds: given an edge $(i,j)\in E$, then $(j,i)\in E$ (or, alternatively, the edges are unorderd pairs of vertices). As a further restriction, we will deal only with undirected graphs.\\
It is customary to represent $G$ as a collection of points, which correspond to the vertices, and lines, which correspond to the edges, such that between vertices $v_i$ and $v_j$ there is a line if and only if $(v_i, v_j) \in E$.

We will say that a graph is \emph{weighted} if there is a weight $w_{ij} \in \mathbb{R}$ associated to each link $\epsilon_{ij} = (i,j)$.\\
Two vertices are said to be \emph{adjacent} if there is a link connecting them. Given a vertex $i$ we say that the \emph{neighborhood} of $i$, which we will denote as $\partial i$, is the set of all vertices adjacent to $i$. Given a vertex $i$, the number of vertices adjacent to him $|\partial i |$ is said to be its \emph{degree}.\\
We introduce the \emph{adjacency matrix} $A$ of a graph:
\begin{equation}
	A_{ij} = \left\{
	\begin{aligned}
	& 1 & & \text{if $(i,j) \in E$;}\\
	& 0 & & \text{otherwise.}
	\end{aligned}	
	\right.
\end{equation}
Notice that for undirected graphs, $A$ is symmetric.
We define the \emph{Laplacian matrix} $L$ of a graph:
\begin{equation}
	L_{ij} = \left\{
	\begin{aligned}
	 &- 1 &  &\text{if $i\neq j$, $(i,j) \in E$;}\\
	 &\sum_{j,j\neq i} 1 &  & \text{if $i=j$;}\\
	 & 0 & & \text{if $i\neq j$, $(i,j) \notin E$.}
	\end{aligned}	
	\right.
\end{equation}
When the graphs are weighted, we define the weighted adjacency matrix as
\begin{equation}
A_{ij} = \left\{
\begin{aligned}
& w_{ij} & & \text{if $(i,j) \in E$;}\\
& 0 & & \text{otherwise}
\end{aligned}	
\right.
\end{equation}
and the weighted Laplacian matrix as
\begin{equation}
L_{ij} = \left\{
\begin{aligned}
&- w_{ij} &  &\text{if $i\neq j$, $(i,j) \in E$;}\\
&\sum_{j,j\neq i} w_{ij} &  & \text{if $i=j$;}\\
& 0 & & \text{if $i\neq j$, $(i,j) \notin E$.}
\end{aligned}	
\right.
\end{equation}
where $w_{ij}$ is the weight associated to the edge $(i,j)$.\\

Another useful concept is the \emph{walk}, that is an alternating series of vertices and edges such that two consecutive vertices are linked by the interleaving edge. Pictorially, this is actually a ``walk'' on the graphical representation of a graph. When the vertices and edges are all different, the walk is called \emph{path}. A graph is \emph{connected} if there is a path connecting each pair of vertices, and it is said to be \emph{disconnected} otherwise. The length of a path is its number of vertices, and the \emph{distance} between two vertices is the length of the shortest path connecting them. If such a path does not exist, we say that the distance is infinite. A walk or path is closed if the starting vertex is the same of the ending one. A closed path is called \emph{cycle} or \emph{loop}. There are two special kinds of cycles: the \emph{Eulerian} cycle is such that it passes through each edge of the graph, the \emph{Hamiltonian} cycle is such that it passes through each vertex of the graph. If a graph does not contain any cycle, it is called \emph{forest}. If it is also connected, it is called \emph{tree}.\\

There are some classes of graphs which one encounters particularly often, because of their regularity properties (see Fig.~\ref{fig::graphs}):
\begin{itemize}
	\item the \emph{$k$-regular} graphs is such that for each $v\in V$, we have $|\partial V| = k$, that is each vertex has degree $k$; 
	\item a \emph{complete} graph is such that for each pair of vertices $i, j \in V$, $(i,j) \in E$, that is each pair of vertices are connected by an edge; in particular, a complete graph with $N$ vertices is $N$-regular and we will use for it the symbol $\mathcal{K}_N$;
	\item a graph $G=(V,E)$ is $p$-partite if we can partition $V$ in $p$ non-empty subsets such that there are no edges of $G$ connecting vertices which belong to the same subset; we will consider in the following only the case $p=2$: in this case we say that the graph is \emph{bipartite};
	\item a graph $G=(V,E)$ which is bipartite in such a way that each subset of vertices has the same number of vertices ($|V_1| = |V_2| = |V|/2$) and such that each vertex of a subset is connected with all the vertices of all the other subsets is called \emph{complete bipartite}; in particular, a complete bipartite graph with $2N$ vertices is $N$-regular and and we will use for it the symbol $\mathcal{K}_{N,N}$.
\end{itemize}

\begin{figure}
	\centering
  \begin{subfigure}[t]{.4\linewidth}
	\centering
		\begin{tikzpicture}[scale=1.9]
		\node[draw,circle,inner sep=2pt,fill=red] (1) at (0,0) {};
		\node[draw,circle,inner sep=2pt,fill=red] (2) at (2,0) {};
		\node[draw,circle,inner sep=2pt,fill=red] (3) at (2,2) {};
		\node[draw,circle,inner sep=2pt,fill=red] (4) at (0,2) {};
		\node[draw,circle,inner sep=2pt,fill=red] (5) at (0.5,0.5) {};
		\node[draw,circle,inner sep=2pt,fill=red] (6) at (1.5,0.5) {};
		\node[draw,circle,inner sep=2pt,fill=red] (7) at (1.5,1.5) {};
		\node[draw,circle,inner sep=2pt,fill=red] (8) at (0.5,1.5) {};
		\draw[line width=1pt,black]  (1) to (2);
		\draw[line width=1pt,black]  (2) to (3);
		\draw[line width=1pt,black]  (3) to (4);
		\draw[line width=1pt,black]  (4) to (1);
		\draw[line width=1pt,black]  (5) to (6);
		\draw[line width=1pt,black]  (6) to (7);
		\draw[line width=1pt,black]  (7) to (8);
		\draw[line width=1pt,black]  (8) to (5);
		\draw[line width=1pt,black]  (1) to (5);
		\draw[line width=1pt,black]  (2) to (6);
		\draw[line width=1pt,black]  (3) to (7);
		\draw[line width=1pt,black]  (4) to (8);
	\end{tikzpicture}
	\caption{Example of a 3-regular graph.}
\end{subfigure}
\begin{subfigure}[t]{.4\linewidth}
	\centering
	\begin{tikzpicture}[scale=0.5]
		\def\n{12}
		\foreach\x in{1,...,\n}{
			\draw (\x*360/\n: 4cm) node[draw,circle,inner sep=2pt,fill=red] (\x) {};
		}
		\foreach\x in{1,...,\n}{
			\foreach\y in{1,...,\n}{
				\ifnum\x=\y\relax\else
				\draw[line width=0.5pt,black]  (\x) to (\y);
				\fi
			}
		}
	\end{tikzpicture}
	\caption{Graphic representation of $\mathcal{K}_N$ with N=12.}
\end{subfigure}

\medskip

\begin{subfigure}[t]{.4\linewidth}
	\centering
	\begin{tikzpicture}[scale=1]
		\def\nr{6} 
		\def\nb{6} 
		\foreach\x in{1,...,6}{
			\node[draw,circle,inner sep=2pt,fill=red] (\x) at (\x,0) {};
		}
		\foreach\x in{7,...,12}{
			\node[draw,circle,inner sep=2pt,fill=blue] (\x) at (\x-6,-2) {};
		}
		\foreach\x in{1,...,6}{
			\foreach\y in{7,...,12}{
				\draw[line width=0.5pt,black]  (\x) to (\y);
			}
		}
	\end{tikzpicture}
	\caption{Graphic representation of $\mathcal{K}_{N,N}$ with N=6.}
\end{subfigure}
\caption{Examples of several classes of graphs.}\label{fig::graphs}
\end{figure}
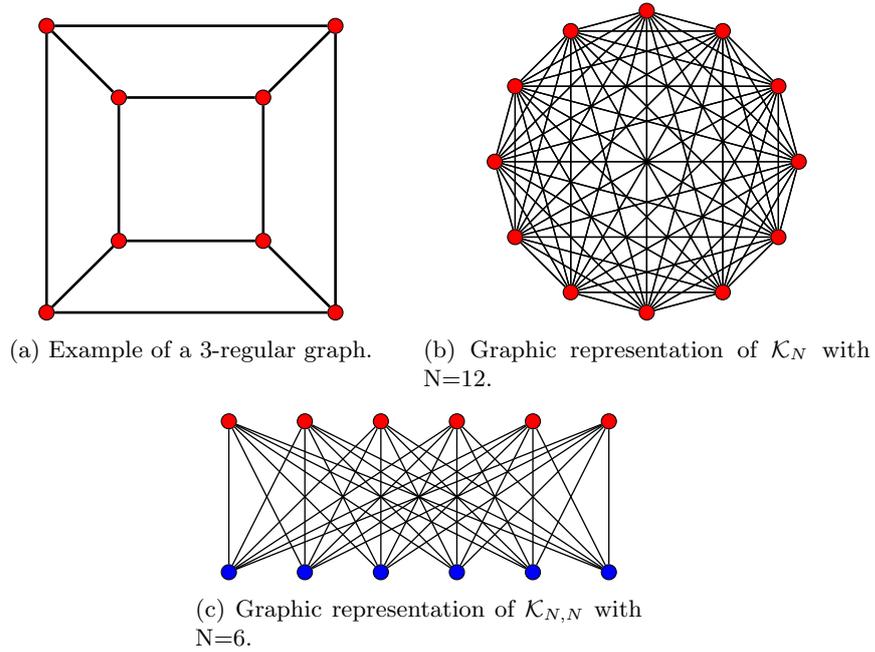

A \emph{subgraph} $G' = (V' , E')$ of the graph $G=(V,E)$ is such that $V'\subset V$ and $E'\subset E$. 
A \emph{spanning subgraph} or \emph{factor} is a subgraph such that $V'=V$.
A k-factor is a factor that is k-regular. In particular, 1-factors are also called \emph{(perfect) matchings}\footnote{1-regular subgraph which are not spanning are sometimes called matchings, and the word ``perfect'' is used if the subgraph is spanning; however, since we will only use this second case, we will from now on drop the adjective ``perfect''.}, or \emph{assignments} when the graph is bipartite. 2-factors are called \emph{loop coverings}. Finally, whenever a spanning subgraph is a tree, it is called \emph{spanning tree}.


Finally, for completeness, we add that sometimes the disorder in COPs defined over graphs is introduced directly at the graph level with the concept of \emph{random graphs}, that is a probability distribution over a set of graphs with certain properties. The most used random graphs are:
\begin{itemize}
	\item \emph{k-regular random graphs} - a graph is randomly chosen among all those with $N$ vertices which are $k$ regulars. Therefore, each graph has the same probability of being generated.
	\item \emph{Erd\"os-R\'enyi graphs} - given a set of $N$ vertices, each possible link is realized with fixed probability $p$.
	\item \emph{Barab\'asi-Albert graphs} (also known as \emph{preferential attachment graphs}) - one vertex at a time is added to the graph; if there are other vertices in the graph, the probability that the new one has a link with the already present vertex $v$ is $k_v/ \mathcal{N}$ where $k_v$ is the degree of $v$ and $\mathcal{N} = \sum_v k_v$ is the normalization constant.
\end{itemize}

\subsection{Why one dimension?}\label{sec::1d}
A Euclidean problem can be seen as a problem on a weighted graph, which typically is complete or complete bipartite. Indeed, an instance of such a problem is specified when the positions of all the involved points are given, and the cost function depends on the distances between pairs of these points. 
Therefore we can restate the problem on a graph as follows: each point chosen in the Euclidean space corresponds to a vertex of the graph and the weight of the link connecting two points is its distance computed in the Euclidean $d$-dimensional space. For this reason we say that the graph is \emph{embedded} in the Euclidean $d$-dimensional space. In the next sections, we will see how the cost function of Euclidean problems has often a very simple interpretation when the problem is casted in graph language.\\

From the next Section, we will start our analysis by considering problems on a graph embedded in one dimension. Now, the problems we will encounter are known to be in the P complexity class as long as they are in one dimension, so why do we are so interested in them?\\
The first reason is that interesting phenomena, such as non-self-averaging solution costs, can appear also in one dimension (as we will see later). The second is that the one-dimensional case is much simpler than the higher-dimensional one, and can lead to insight useful for the latter.

Let us now go through our general strategy to address one-dimensional RCOPs. Our aim is to compute the cost of the solution, averaged over the disorder (that is, over the ensemble of instances defined by probability density for the point positions).\\
The first step consists in finding the \emph{structure} of the solution of our problem. Indeed in one dimension (and, unfortunately, only in one!) we can sort the points according to their position, so that if we have the positions $x_1 < \dots < x_N$ for $N$ points, we say that that the \emph{first} point is the one in $x_1$, the \emph{second} is the one in $x_2$ and so on. Therefore, the point are \emph{ranked} according to their position.
We will see that often the solution is given in terms of this point rank, rather them the point specific positions. 
This will allow us to find the configuration which minimizes the cost, and then to reach our goal it will be enough to average the cost of this configuration over the point positions.\\
However, as we will see, for some problems the optimal configuration does depend on the specific point positions (and not only on their rank) even in one dimension.
In these cases we will still be able to work out bounds for the cost, by carefully analyzing the full set of possible solutions.\\

We will also see how this one-dimensional approach to Euclidean RCOPs will help us to make exact predictions for the limiting (in the large problem size) behavior of the average cost, even for some problems which are NP-hard (in two or more dimensions).

\section{Matching problem}\label{sec::matching}

\subsection{An easy problem?}
We start by the most general definition of the problem, which is the following:
consider a weighted graph $G=(V,E)$, $w_{e}$ being the weight of the edge $e \in E$. Let us denote by $\mathcal{M}$ of matchings of this graph. To each matching $M=(V,E_M)\in \mathcal{M}$ we associate a cost
\begin{equation}
	C_M = \sum_{e\in E_M} w_e.
\end{equation}
The matching problem consists in deciding whether $\mathcal{M} = \varnothing$ or not, and if $\mathcal{M} \neq \varnothing$ then the \emph{weighted} matching problem consists in finding the matching with the minimum cost. Notice that we can easily recast the matching problem as a weighted matching problem as follows: given $G=(V,E)$ with $N$ vertices, build a weighted complete graph $\mathcal{K}_N$, where the weight of a link $e$ is 0 if $e\in E$, 1 if $e\notin E$. Now solve the weighted matching problem on this weighted complete graph, and if the solution cost is 0 then $G$ has at least one matching, if the cost is greater than 0 then $G$ does not have matchings.\\
Therefore from now on we will only consider the weighted matching problem, which we will simply call matching problem. This problem, even in this very general graph setting, is in the P complexity class thanks to the work of Kuhn~\cite{Kuhn1955}, who discovered a polynomial algorithm called \emph{Hungarian algorithm}, and several other works~\cite{Edmonds1965,Micali1980, Edmonds2003,Lovasz2009} in which that algorithm is extended and made faster.\\
From this point on, we will only consider matchings on complete graphs $\mathcal{K}_{2N}$ and on complete bipartite graphs $\mathcal{K}_{N,N}$. Notice that usually, when the graph is bipartite, the matching problem is called \emph{assignment problem}.\\
Let us come back to the question of the problem complexity, and consider more closely the matching problem on complete graphs. The graph $\mathcal{K}_{2N}$ has always
\begin{equation}
	(2 N - 1)!! = \prod_{k=0}^{N-1} (2N- 1 -2k) = \frac{\prod_{k=0}^{2N} (2N-k)}{\prod_{k=0}^{N-1} (2N-2k)} = \frac{(2N)!}{2^k N!} \sim \sqrt{2} e^{N (\log (2N)+1)}
\end{equation}
perfect matchings, where in the last step we used the Stirling approximation for the factorial for large $N$,
\begin{equation}
	N! \sim \sqrt{2 \pi N} \left(\frac{N}{e}\right)^N.
\end{equation}
This is of course an enormous number which makes brute force approaches immediately unusable.
Similarly, a bipartite graph $\mathcal{K}_{N,N}$ has
\begin{equation}
	N! \sim \sqrt{2 \pi N} \left(\frac{N}{e}\right)^N = \sqrt{2 \pi} e^{N (\log N-1) + \frac{1}{2} \log N}
\end{equation}
assignments.\\

According to our discussion in Sec.~\ref{sec::statphys}, we can write down a spin Hamiltonian for this problem as follows: given a graph $G=(V, E)$ (which we restrict to be complete or complete bipartite), we associate a binary variable to each edge of the graph, and $x_{ij} = 1$ (or 0) if the edge $(i,j)$ is present (or not) in the configuration (set of edges) $x$.
The cost function is
\begin{equation}\label{eq::matching_binary}
C(x) = \frac{1}{2} \sum_{(i, j) \in E} w_{ij} x_{ij}.
\end{equation}
where the factor 1/2 is present because if $(i,j)\in E$, also $(j,i) \in E$.
We also need to require that $x$ is a matching, that is
\begin{equation}\label{eq::matching_constraint}
\sum_{j \in \partial i} x_{ij} = 1
\end{equation}
for each $i$. \\
%
%
We can proceed in two ways: 
\begin{itemize}
	\item we can add these constraints in a hard manner, that is by restricting the configuration space to those states for which Eq.~\eqref{eq::matching_constraint} is satisfied (constraints of this kind are ofter referred to as \emph{hard constraints});
	\item we can modify the cost function so that at least the minimum-energy configuration satisfies Eq.~\eqref{eq::matching_constraint}, for example by using
	\begin{equation}
		C_{\text{soft}}(x) = \frac{1}{2} \sum_{(i, j) \in E} w_{ij} x_{ij} + \lambda \left( \sum_{j \in \partial i} x_{ij}-1\right)^2,
	\end{equation}
	where $\lambda$ is a free parameter to be chosen sufficiently large (constraints of this kind are ofter referred to as \emph{soft constraints}).
\end{itemize}
In this Chapter we will always impose hard constraints, but in the next Chapter we will see that, to overcome some technical problems, sometimes it is necessary to use the soft variant.\\
At this point, to obtain a genuine spin Hamiltonian, we should do the change of variables
\begin{equation}\label{eq::matching_spin}
	x_{ij} = \frac{\sigma_{ij} + 1}{2},
\end{equation}
so that to the binary variable $x_{ij}=\{1,0\}$ we associate a spin variable $\sigma_{ij} = \{1,-1\}$. The resulting Hamiltonian is
\begin{equation}
	H(\sigma) = \frac{1}{4}\sum_{(i, j) \in E} w_{ij} \sigma_{ij} + C,
\end{equation}
where $C=\sum_w w /2$. As we can see, the Hamiltonian of this problem is trivial, and all the non-trivial part comes from the constraint term, which in terms of the new spin variables is 
\begin{equation}
	\sum_{j \in \partial i} \sigma_{ij} = 2 + C',
\end{equation}
where $C' = N-1$ for $G=\mathcal{K}_N$ and $C' = N$ for $G=\mathcal{K}_{N,N}$ and is the number of vertices adjacent to each vertex. By using either hard or soft constraints, one can check that the problem Hamiltonian is frustrated (in the sense discussed in Sec.~\ref{sec::spinglass}). Therefore, even if we know that an algorithm which solves the problem in polynomial time does exist, the energy landscape is far from being trivial for a generic choice of the weights $w_{ij}$.

\subsection{Mean field version}
We introduce the disorder in the matching/assignment problem to treat it as a RCOP. In the mean field case, we do it by choosing a probability density function for the weights $w_{ij}$ so that they are IID random variables. \\
The focus of this work is the Euclidean version of several problems, where the weights are correlated, but before than considering that more complicated case, we will very quickly review the mean field case following the original paper by M\'ezard and Parisi~\cite{Mezard1985} (the interested readers can find the details of the computations in that paper, but also in one of these PhD theses~\cite{Sicuro2016, Malatesta2019}).\\
For the weights, we consider the probability density
\begin{equation}\label{eq::matching_meanf_ddp}
	p(w) = \theta(w) e^{-w}.
\end{equation}
We can write the partition function for the complete graph $\mathcal{K}_{2N}$ using Eqs.~\eqref{eq::matching_binary} and \eqref{eq::matching_constraint} as
\begin{equation}
	Z = \left( \prod_{\substack{i,j=1 \\ i < j}}^{2N} \sum_{x_{ij}=0,1} \right) e^{- \beta  \sum_{i < j} w_{ij} x_{ij}} \left[ \prod_{i=1}^{2N} \delta\left( \sum_{j, j \neq i} x_{ij}, 1 \right) \right],
\end{equation}
where we have written the Kronecker delta as $\delta(a,b)$ instead of $\delta_{a,b}$ for notational convenience.
By using the integral representation of the $\delta$ we can sum over the binary variables and obtain (remember that $x_{ij} = x_{ji}$)
\begin{equation}
	Z = \left[ \prod_{i=1}^{2N} \frac{d \lambda_i}{2\pi} e^{i \lambda_i} \right] \prod_{i<j} \left(1 + e^{-\beta w_{ij} -i (\lambda_i + \lambda_j)}\right),
\end{equation}
which is the starting point for a replica computation of the free energy at zero temperature with quenched disorder, which coincide with the average cost of the solution. The computation is far from trivial, but it is similar conceptually (even though there are some technical differences) to the one we performed for the $p$-spin spherical model in Sec.~\ref{sec::pspin}. A remarkable difference is that this time a RS ansatz is enough to solve the problem. The result is
\begin{equation}
	\lim_{\beta\to\infty} \lim_{N\to\infty} - \frac{1}{\beta} \overline{\log Z} = \frac{\pi^2}{12},
\end{equation}
that is, for the average cost of the solution we have
\begin{equation}
	\overline{E_N} \sim \frac{\pi^2}{12}
\end{equation}
for $N\gg1$.
This same approach can be used for the assignment on the complete bipartite graph $\mathcal{K}_{N,N}$ and the result has a factor 2 of difference:
\begin{equation}
\overline{E^{(\text{bip})}_N} \sim \frac{\pi^2}{6}.
\end{equation}

At this point, one can wonder if there is a simple way to guess the fact that the cost of the solution is (on average) of order 1 for $N\to \infty$, and if there is a simple way to explain this factor 2 of difference. To answer that, notice that, even though $\overline{w_i}=1$, the minimum among $n$ IID random variables can be computed by obtaining the cumulative:
\begin{equation}
	P(\min w_i \geq x) = \prod_{i=1}^n P(w_i \geq x) = e^{- n x}.
\end{equation}
From this, we can obtain
\begin{equation}
	p(\min w_i) = - \frac{\partial}{\partial x} P(\min w_i \geq x) = n e^{-n x}
\end{equation}
and so, on average,
\begin{equation}
	\overline{\min w_i} = N \int_0^\infty dx \, x \,  e^{-n x} = \frac{1}{n}.
\end{equation}
Therefore, one can reasonably expect that, if we need to find the matching of minimum cost of $\mathcal{K}_{2N}$, each of the $N$ edges chosen in that matching will be have a cost close to the minimum of a set of $2N-1\sim 2N$ IID random variables drawn from the probability given in Eq.~\eqref{eq::matching_meanf_ddp}, so $1/(2N)$. Since we have $N$ edges in the cost function, if each edge of the matching were independent from the others, the total cost would have been $\sim 1/2$. However, there are the constraints, that prevent a matching from being composed only of minimum cost links, and the extra cost due to this fact raises the total cost to $\pi^2/12 \simeq 0.82$. 
A similar argument can be used for the bipartite matching, and in this case one has that again the matching is composed by $N$ edges, but each edge now has to be chosen among $N$ IID random variables with probability Eq.~\eqref{eq::matching_meanf_ddp}, and this gives the factor 2 of difference.

\subsection{Going in one dimension}\label{sec::mat_1d}
\subsubsection{Assignment problem on complete bipartite graphs}
In this section we will  focus on the Euclidean matching problem, beyond the mean field approximation. Therefore, let us state the problem in the Euclidean setting, starting from the assignment: consider two sets of points in $\mathbb{R}^d$ labeled by their coordinates, $\mathcal{R} = \{r_1, \dots, r_N\}$ (red points) and $\mathcal{B}= \{b_1, \dots, b_N \}$ (blue points). We want to match each red point to one and only one blue point such that a certain function of the distances between matched points is minimized. \\
Since we can connect the blue points only to the red points and vice versa, the problem can be seen as a matching on a bipartite complete graph $\mathcal{K}_{N,N}$. Now, each choice of a matching corresponds to a permutation of $N$ objects, $\pi\in S(N)$, and vice versa. 
The cost function assigns a cost to each permutation as follows:
\begin{equation}
	E_N^{(p)}[\pi] = \sum_i \left| r_i - b_{\pi(i)} \right|^p, 
\end{equation}
where $p\in\mathbb{R}$ is a parameter. We will focus here on the $p>1$ case. The points $p=1$ and $p=0$ are special points where there can be many solutions~\cite{Boniolo2014}, but apart from that they share the properties about the typical cost with, respectively, the $p>1$ and $p<0$ case. For a study on the properties of this problem with $p<0$, see~\cite{Caracciolo2017}, while some properties regarding the region $0<p<1$ are given in \cite{Caracciolo2019}.

In this problem, the disorder is introduced at the level of point positions: here we will consider the case when the coordinates of each point is an IID random variable, distributed with a flat probability density on the interval [0,1].\\
For this problem, there is no way to proceed directly with the plain replica method, due to the fact that now, even though the point positions are uncorrelated, their distances (which are the relevant variables in the cost function) are. One can try to take into account the Euclidean correlations as corrections to the mean field case~\cite{Mezard1988, Lucibello2017}, but here we will follow another path (which has been presented in~\cite{Caracciolo2019_1}).\\
We will use the fact that for $p>1$ the optimal solution is the identity permutation once both sets of points have been ordered~\cite{McCannRobert1999,Boniolo2014}. The proof of this is obtained by noticing that once one knows what is the optimal solution of the case $N=2$ (so two blue and two red points), then the solution is found by simply repeating this argument for each possible choice of two blue and two red points. 
It follows that, in these cases, the optimal cost is
\begin{equation}
E_N^{(p)} =  \sum_{i=1}^N | r_i - b_i |^p.
\end{equation}
Now we need to average over the disorder. To do that, we will use the Selberg integrals~\cite{Selberg1944}
\begin{equation}\label{eq::selberg}
\begin{split}
S_n(\alpha, \beta, \gamma) & :=  \left(\prod_{i=1}^n \int_0^1 \ dx_i \, x_i^{\alpha-1} (1-x_i)^{\beta -1}\right) \left |\Delta(x)\right |^{2 \gamma} \\
& = \prod_{j=1}^n \frac{\Gamma( \alpha + (j-1) \gamma) \Gamma(\beta + (j-1) \gamma)\Gamma(1 + j \gamma)}{\Gamma(\alpha+ \beta + (n+j-2) \gamma)\Gamma(1 +  \gamma)} 
\end{split}
\end{equation}
where
\begin{equation}
\Delta(x) := \prod_{1\leq i < j \leq n} (x_i-x_j)  \, 
\end{equation}
with $\alpha, \beta, \gamma \in \mathbb{C}$ and $\Re(\alpha)>0$, $\Re(\beta)>0$, $\Re(\gamma) > \min(1/n, \Re(\alpha)/(n-1), \Re(\beta)/(n-1) )$,
see~\cite[Chap. 8]{Andrews1999}. Selberg integrals are a generalization of Euler Beta integrals, which are recovered by setting $n=1$.

In Appendix \ref{app::orderstat} we compute the probability that, once we have ordered our points, the $k$-th point is in the interval $[x, x+dx]$. By using that result, given in Eq.~\ref{eq::ordstat_pk}, and the Selberg integral from Eq.~\eqref{eq::selberg}
\begin{equation}
\begin{split}
S_2\left(k, N-k+1, \frac{ p}{2}\right)  
& = \left(\prod_{i=1}^2 \int_0^1 \ dx_i \, x_i^{k-1} (1-x_i)^{N-k}\right) \left |x_2-x_1\right |^{p} \\
& = \frac{\Gamma(k) \Gamma(N-k+1)\Gamma\left(k+\frac{p}{2}\right)\Gamma\left(N-k+1+\frac{p}{2}\right)\Gamma(1+ p)}{\Gamma\left(N+1+\frac{p}{2}\right) \Gamma(N+1+p)\Gamma\left(1+\frac{p}{2}\right)}\, ,
\end{split}
\end{equation}
we get that the average of the $k$-th contribution is given by
\begin{equation}\label{ass}
\begin{split}
\overline{| r_k - b_k |^p} =  & \int_0^1 \ dx\, \int_0^1 \ dy\, P_k(x) \, P_k(y)\, |y-x|^p \\
= &  \left(  \frac{\Gamma(N+1)}{\Gamma(k) \, \Gamma(N-k+1)} \right)^2 \, \left(\prod_{i=1}^2 \int_0^1 \ dx_i \, x_i^{k-1} (1-x_i)^{N-k}\right) \left |x_2-x_1\right |^{p} \\
= & \left(  \frac{\Gamma(N+1)}{\Gamma(k) \, \Gamma(N-k+1)} \right)^2 \, S_2\left(k, N-k+1, \frac{p}{2}\right)\\
= & \frac{\Gamma^2(N+1)\Gamma\left(k+\frac{p}{2}\right)\Gamma\left(N-k+1+\frac{p}{2}\right)\Gamma(1+p)}{\Gamma(k) \Gamma(N-k+1) \Gamma\left(N+1+\frac{p}{2}\right) \Gamma(N+1+p)\Gamma\left(1+\frac{p}{2}\right) } \, 
\end{split}
\end{equation}
and therefore we get the exact result
\begin{equation}\label{eq::exact_cost_ass1d}
\begin{split}
\overline{E_N^{(p)}} & =  \frac{\Gamma^2(N+1)\Gamma(1+p)}{\Gamma\left(N+1+\frac{p}{2}\right) \Gamma(N+1+p)\Gamma\left(1+\frac{p}{2}\right) }   \sum_{k=1}^N\frac{\Gamma\left(k+\frac{p}{2}\right)\Gamma\left(N-k+1+\frac{p}{2}\right)}{\Gamma(k) \Gamma(N-k+1)} \\
& =  \frac{\Gamma\left(1 + \frac{p}{2}\right)}{p+1} \, N\, \, \frac{  \Gamma(N+1)}{\Gamma\left(N+1 + \frac{p}{2}\right)} ,
\end{split}
\end{equation}
where we made repeated use of the duplication and Euler's inversion formula for $\Gamma$-functions
\begin{subequations}
	\begin{align}
	\Gamma(z) \Gamma\left( z + \frac{1}{2} \right) & = 2^{1 - 2 z} \sqrt{\pi} \, \Gamma(2z) \\
	\Gamma(1-z) \Gamma(z) & = \frac{\pi}{\sin( \pi z)}\, .
	\end{align}
\end{subequations}
For large $N$ we obtain, at the first order,
\begin{equation}\label{eq::cost_assgn_1d_largeN}
\overline{E_N^{(p)}} \sim \frac{\Gamma\left(1 + \frac{p}{2}\right)}{p+1} \, N^{1-\frac{p}{2}}.
\end{equation}

\subsubsection{Matching problem on the complete graph}
A similar technique can be carried out to compute the cost of the matching problem on the complete graph $\mathcal{K}_{2N}$. Indeed, again by studying the case $N=4$, it can be shown that the optimal solution for $p>1$ consists always in, once we have sorted the point such that $x_1\leq \dots \leq x_{2N}$, matching $x_1$ with $x_2$, $x_3$ with $x_4$ and so on. Therefore the optimal cost is
\begin{equation}
	E_N^{(p)} =  \sum_{i=1}^N ( x_{2i} - x_{2i-1} )^p.
\end{equation}
To compute the average cost, it is convenient to define the variables $\phi_i =x_{i+1} - x_i$. The cost of the solution in this new variables reads
\begin{equation}
	E_N^{(p)} =  \sum_{i=1}^N \phi_{2i-1}^p 
\end{equation}
Since the $2N$ points are uniformly chosen in the unit interval and then ordered, their joint distribution is
\begin{equation}
	p(x_1, \dots, x_{2N}) = (2N)! \prod_{i=0}^{2N} \theta(x_{i+1}-x_i),
\end{equation}
with $x_0=0$ and $x_{2N+1}= 1$. Therefore, the probability distribution function of the $\phi_i$ variables is
\begin{equation}
	p(\phi_1, \dots, \phi_{2N}) = (2N)! \, \delta\left(\sum_{i=1}^{2N} \phi_i, 1\right) \, \prod_{i=0}^{2N} \theta(\phi_i).
\end{equation}
From this we can compute the marginal probability of the $k$-th spacing $\phi_k$, which is
\begin{equation}
\begin{split}
	p^{(1)}(\phi_k) & = (2N)! \, \left[ \prod_{\substack{a=0 \\ a \neq k}}^{2N} \int_0^\infty d \phi_a \right] \delta\left(\sum_{i=1}^{2N} \phi_i, 1\right)\\
	& = (2N)! \, i^{2N} \lim_{\epsilon\to 0^+} \int_{-\infty}^\infty \frac{d \lambda}{2\pi} \frac{e^{-i \lambda (1-\phi_k)}}{(\lambda + i \epsilon)^{2N}}\\
	& = \left\{
	\begin{aligned}
	& 2N \, (1-\phi_k)^{2N-1} & & \text{if $0 < \phi_k < 1$;}\\
	& 0 & & \text{otherwise}
	\end{aligned}	
	\right.
\end{split}
\end{equation}
where we inserted a small imaginary part at the denominator in the second step to be able to use the residue method to perform the integral. Notice that the result does not depend on $k$ and, by exploiting the Euler Beta integral (Eq.~\eqref{eq::selberg} with $n=1$), we get
\begin{equation}
	\overline{\phi_k^p} = \int d \phi_k \, p^{(1)}(\phi_k) \, \phi_{k}^p = \frac{\Gamma(2N+1)\Gamma(1+p)}{\Gamma(2N+1+p)}.
\end{equation}
Therefore we finally obtain
\begin{equation}\label{eq::matching_cost_exact_1d}
	E_N^{(p)} = N \frac{\Gamma(2N+1)\Gamma(1+p)}{\Gamma(2N+1+p)}.
\end{equation}
For large $N$ we obtain, at the first order,
\begin{equation}\label{eq::matching_cost_largeN}
\overline{E_N^{(p)}} \sim \frac{\Gamma(p+1)}{2^p} \, N^{1-p}.
\end{equation}

\subsection{Assignment in two dimensions and beyond}


We have seen how to exploit properties of the solution structure to compute the average cost of the matching problem solution in one spatial dimension, for both the complete and complete bipartite version of the problem. However, we just scratched the surface: there are many known results, and many open questions about this fascinating COP. We will mention some of them here.

First of all, the \emph{scaling} in $N$ for $N\to\infty$ of the average solution cost is known for all number of dimension $d$. In particular, for the Euclidean matching problem on the complete graph $\mathcal{K}_N$ embedded in $d$ dimensions, where the cost of a link is the distance between points to the $p\geq 1$, we have
\begin{equation}\label{eq::matc_scaling_mono}
	\overline{E^{(p)}_N} \sim A^{(p)}_d N^{1-p/d}
\end{equation}
These scalings can be obtained by a qualitative reasoning, such as the fact that if there are $N$ points in a volume $V=1$, then the distance between two first neighbors is $\sim N^{-1/d}$ and therefore the cost of that link is $\sim N^{p/d}$ and we have $N$ of these links. A formal proof of Eq.~\eqref{eq::matc_scaling_mono} is given in~\cite{Steele1997,Yukich2006}.
As we have seen, when $d=1$ we have
\begin{equation}
	A^{(p)}_1 = \frac{\Gamma(p+1)}{2^p},
\end{equation}
and we are actually able to compute the average cost for each finite $N$. It is known that the first correction, in every $d$, scales as $\mathcal{O}(N^{-p/d})$~\cite{Houdayer1998}, as we can again check in one dimension by starting from Eq.~\eqref{eq::matching_cost_exact_1d}. The exact value of the constant $A^{(p)}_d$ is not known for $d>1$.

For the assignment problem on the complete bipartite graph $\mathcal{K}_{N,N}$, in $d$ dimensions with the parameter $p\geq 1$ we know that~\cite{Talagrand1992,Ajtai1984}:
\begin{equation}\label{eq::matc_scaling_bi}
	E^{(p)}_N \sim \left\{
	\begin{aligned}
	& B_{1}^{(p)} N^{1-p/2}  												& & \text{$d=1$;}\\
	& B_{2}^{(p)} N	\left( \frac{\log(N)}{N} \right)^{p/2}					& & \text{$d=2$;}\\
	& B_{d}^{(p)} N^{1-p/d}													& & \text{$d>2$.}
	\end{aligned}	
	\right.
\end{equation}
The scaling differences between the bipartite problem and the one with a single kind of points are due, intuitively, to the fact that in the former case if we consider the problem restricted to a small region of space we can have fluctuations of the \emph{relative} density of points of one kind with respect to those of the other kind. Clearly, this is not possible when there is a single kind of points. This fact, in turn, implies the presence of longer links even at a ``microscopical'' level, giving rise to the different behavior between this two versions of the matching problem. However, this difference is less and less important as we go in higher number of dimensions.\\
When $p=2$, additional results are known~\cite{Caracciolo2014} for the case with periodic boundary conditions (so that the points are chosen on a $d$-dimensional torus):
\begin{equation}\label{eq::largeN_mat_bip_p2}
	E^{(p)}_N \sim \left\{
	\begin{aligned}
	& \frac{1}{3} + \mathcal{O}(1/N)										& & \text{$d=1$;}\\
	& \frac{\log(N)}{2 \pi} + \mathcal{O}(1)								& & \text{$d=2$;}\\
	& N^{1-p/d}\left(B_{d}^{(p)}+\frac{\zeta(1)}{2\pi^2}N^{2/d-1}\right)	& & \text{$d>2$,}
	\end{aligned}	
	\right.
\end{equation}
where $\zeta(x)$ is the Epstein $\zeta$ function.
Notice that in $d>2$ only the coefficient of the first correction is analytically known, while the leading-term coefficient is not.
The result in $d=2$, $p=2$ has been extended to the case of open boundary condition~\cite{Caracciolo2015, Ambrosio2019}, and the asymptotic result given in Eq.~\eqref{eq::largeN_mat_bip_p2} is proven to be correct also in this case.\\
As for the problem on the complete graph, in $d=1$ Eq.~\eqref{eq::exact_cost_ass1d} gives the cost and all the corrections for each value of $p$.
For the case with $p\neq2$ and $d>1$ neither the coefficient $B_{2}^{(p)}$ nor the scaling of the corrections in $N$ is known.

Several other results are known about the self-averaging property of the solution cost: 
\begin{itemize}
	\item for the matching problem (on complete graph), it has been proven~\cite{Steele1997,Yukich2006} that the cost is self averaging in any number of dimension $d$;
	\item for the assignment problem (on complete bipartite graph), in $d=1$ one can check that the cost is not self-averaging with methods similar to those used in Sec.~\ref{sec::mat_1d} (see~\cite{Caracciolo2017}), while for $d>2$ it is known that the cost is self-averaging~\cite{Houdayer1998}; in $d=2$, this question is still open. 
\end{itemize}


\section{Traveling salesman problem}\label{sec::tsp}
In this section we will analyze an archetypal combinatorial optimization problems, which has been fueling a considerable amount of research, from its formalization to the present day.
Given $N$ cities and $N (N-1)/2$ values that represent the cost paid for traveling between all pairs of them, the traveling salesman problem (TSP) consists in finding the tour that visits all the cities and finally comes back to the starting point with the least total cost to be paid for the journey.
The first formalization of the TSP can be probably traced back to the Austrian mathematician Karl Menger, in the 1930s~\cite{Menger1932}. As it belongs to the class of NP-complete problems, see Karp and Steele in~\cite{Lawler1985}, one of the reason for studying the TSP is that it could shed light on the famous P vs NP problem discussed in Sec.~\ref{sec::complexity}.
Many problems in various fields of science (computer science, operational research, genetics, engineering, electronics and so on) and in everyday life (lacing shoes, Google maps queries, food deliveries and so on) can be mapped on a TSP or a variation of it, see for example Ref.~\cite[Chap.~3]{Reinelt1994}  for a non-exhaustive list.
Interestingly, the complexity of the TSP seems to remain high even if we try to modify the problem. For example, the Euclidean TSP, where the costs to travel from cities are the Euclidean distances between them, remains NP-complete~\cite{Papadimitriou1977}. The bipartite TSP, where the cities are divided in two sub-sets and the tour has to alternate between them, is NP-complete too, as its Euclidean counterpart.\\
The traveling salesman problem is one of the most studied combinatorial optimization problems, because of the simplicity in its statement and the difficulty of its solution. In this section, after defining the problem explicitly, we review the most recent works regarding the average cost of the solutions in one and two dimensions~\cite{Caracciolo2018_1, Caracciolo2018_2, Caracciolo2019_2}.

\subsection{Traveling on graphs}
In a generic graph, the determination of the existence of an Hamiltonian cycle is an NP-complete problem (see Johnson and Papadimitriou in~\cite{Lawler1985}).
However, here we will deal with complete graphs $\mathcal{K}_N$, where at least one Hamiltonian cycle exists for $N>2$, and bipartite complete graphs $\mathcal{K}_{N,N}$, where at least an Hamiltonian cycle exits for $N>1$.\\
Let us denote by $\mathcal H$ the set of Hamiltonian cycles of the graph $\mathcal{G}$. Let us suppose now that a weight $w_e > 0$ is assigned to each edge $e \in \mathcal{E}$ of the graph $\mathcal{G}$. We can associate to each Hamiltonian cycle $h\in \mathcal{H}$ a total cost
\begin{equation}\label{eq::tsp_cost_general}
E(h) :=  \sum_{e\in h} w_e.
\end{equation}
In the (weighted) Hamiltonian cycle problem we search for the Hamiltonian cycle $h\in \mathcal{H}$ such that the total cost in Eq.~\eqref{eq::tsp_cost_general} is minimized, i.e., the optimal Hamiltonian cycle $h^*\in \mathcal{H}$ is such that
\begin{equation}
E(h^*) = \min_{ h\in \mathcal{H}} E(h)\, . \label{h^*}
\end{equation}
When the $N$ vertices of $\mathcal{K}_N$ are seen as cities and the weight for each edge is the cost paid to cover the route distance between the cities, the search for $h^*$ is called the \emph{traveling salesman problem} (TSP). For example, consider when the graph $\mathcal{K}_N$ is embedded in $\mathbb{R}^d$, that is for each $i\in [N] =\{1,2,\dots,N\}$ we associate a point $x_i\in \mathbb{R}^d$,  and for $e=(i,j)$ with $i,j \in [N]$ we introduce a weight which is a function of the Euclidean distance $w_e = |x_i-x_j|^p$ with $p\in \mathbb{R}$, as we did previously for the matching problem. When $p=1$, we obtain the usual Euclidean TSP.
Analogously for the bipartite graph $\mathcal{K}_{N,N}$ we will have two sets of points in $\mathbb{R}^d$, that is the red $\{r_i\}_{i\in [N]}$ and the blue $\{b_i\}_{i\in [N]}$ points and the edges connect red with blue points with a cost
\begin{equation}\label{eq::tsp_bip_pb}
w_e = |r_i-b_j|^p \, . 
\end{equation}
When $p=1$, we obtain the usual bipartite Euclidean TSP.

Also this COP can be promoted to be a RCOP in many ways, and the simplest correspond to the mean-field case: the randomness is introduced by considering the weights $w_e$ independent and identically distributed random variables, thus neglecting any correlation due to the Euclidean structure of the space. In this case the problem is called random TSP and has been extensively studied by disordered system techniques such as replica and cavity methods ~\cite{Vannimenus1984,Orland1985,Sourlas1986,Mezard1986,Mezard1986a,Krauth1989,Ravanbakhsh2014} and by a rigorous approach \cite{Wastlund2010}. In the random Euclidean TSP~\cite{Beardwood1959,Steele1981,Karp1985,Percus1996,Cerf1997}, instead, the point positions are generated at random as IID random variables, and as a consequence the weights will be correlated.
Also in this case we are interested in finding the average optimal cost
\begin{equation}
\overline{E}  = \overline{E(h^*)}\,,
\end{equation}
and its statistical properties.

\subsection{TSP on bipartite complete graphs}
\subsubsection{Hamiltonian cycles and permutations}
We shall now restrict to the complete bipartite graph $\mathcal{K}_{N,N}$. 
Before turning to the computation of the average cost of the TSP solution in one dimension, let us discuss some general properties, valid in every dimension number, and the relationship between the TSP on bipartite graphs and the assignment problem discussed before.\\
Let $\mathcal{S}_N$ be the group of permutation of $N$ elements. For each $\sigma, \pi\in  \mathcal{S}_N$, the sequence of edges for $i\in [N]$
\begin{equation}\label{eq::tsp_bip_corris}
\begin{aligned}
	e_{2i-1} = & \, (r_{\sigma(i)}, b_{\pi(i)} )  \\
	e_{2i} = & \, (b_{\pi(i)}, r_{\sigma(i+1)}) 
\end{aligned}
\end{equation}
where $\sigma(N+1)$ must be identified with $\sigma(1)$, defines a Hamiltonian cycle.
More properly, it defines a Hamiltonian cycle with starting vertex $r_1=r_{\sigma(1)}$ with a particular orientation, that is 
\begin{equation}
h[(\sigma, \pi)] := (r_1 b_{\pi(1)} r_{\sigma(2)} b_{\pi(2)} \cdots r_{\sigma(N)} b_{\pi(N)}) 
= (r_{1} C ) \,,
\end{equation}
where $C$ is an open walk which visits once all the blue points and all the red points with the exception of $r_1$. Let $C^{-1}$ be the open walk in opposite direction. This defines a new, dual, couple of permutations which generate the same Hamiltonian cycle
\begin{equation}
h[(\sigma, \pi)^\star] := (C^{-1} r_1) = (r_1 C^{-1} ) = h[(\sigma, \pi)]\, ,
\end{equation}
since the cycle $(r_1 C^{-1} )$ is the same as $(r_1 C)$ (traveled in the opposite direction).
By definition
\begin{align}
\begin{split}
h[(\sigma, \pi)^\star] = (r_1 b_{\pi(N)} r_{\sigma(N)} b_{\pi(N-1)} r_{\sigma(N-1)} \cdots b_{\pi(2)}  r_{\sigma(2)} b_{\pi(1)} ) \, .
\end{split}
\end{align}
Let us introduce the cyclic permutation $\tau \in \mathcal{S}_N$, which performs a left rotation, and the inversion $I \in \mathcal{S}_N$. That is $\tau(i) = i+1$ for $i\in [N-1]$ with  $\tau(N) = 1$ and $I(i) = N+1 -i$. In the following we will denote a permutation by using the second row in the usual two-row notation, that is, for example $\tau = (2, 3, \cdots ,N, 1)$ and $I= (N, N-1, \dots, 1)$. Then
\begin{equation}\label{eq::tsp_bip_d1}
h[(\sigma, \pi)^\star] =  h[( \sigma \circ \tau \circ I, \pi \circ I)] \, . 
\end{equation}
There are $N! \, (N-1)!/2$ Hamiltonian cycles for $\mathcal{K}_{N,N}$.
Indeed the couples of permutations are $(N!)^2$ but we have to divide them by $2N$ because of the $N$ different starting points and the two directions in which the cycle can be traveled. 

From Eq.~\eqref{eq::tsp_bip_corris} and weights of the form given in Eq.~\eqref{eq::tsp_bip_pb}, we get an expression for the total cost
\begin{align}\label{eq::tsp_bip_costgen}
\begin{split}
E[h[(\sigma, \pi)]] =  \sum_{i \in [N]}\left[  |r_{\sigma(i)} - b_{\pi(i)}|^p + |r_{\sigma \circ \tau(i)} - b_{\pi(i)}|^p \right] \, .
\end{split}
\end{align} 
Now we can re-shuffle the sums and we get
\begin{align}\label{eq::tsp_bip_dec}
\begin{split}
E[h[(\sigma, \pi)]] &= \sum_{i \in [N]}  |r_{i} - b_{\pi \circ \sigma^{-1}(i)}|^p + 
\sum_{i \in [N]}  |r_{i} - b_{\pi \circ \tau^{-1} \circ \sigma^{-1}(i)}|^p \\
& = E[m(\pi \circ \sigma^{-1})] + E[m(\pi \circ \tau^{-1} \circ \sigma^{-1})] 
\end{split}
\end{align}
where $E[m(\lambda)]$ is the total cost of the assignment $m$ in $\mathcal{K}_{N,N}$ associated to the permutation $\lambda\in \mathcal{S}_N$.
The duality transformation given in Eq.~\eqref{eq::tsp_bip_d1}, interchanges the two matchings because
\begin{subequations}
	\begin{align}
	\begin{split}
	\mu_1 := \pi \circ \sigma^{-1} \, \to & \;  \pi \circ I \circ I \circ \tau^{-1} \circ \sigma^{-1} = \pi \circ \tau^{-1} \circ \sigma^{-1}
	\end{split}
	\\
	\begin{split}
	\mu_2 := \pi \circ \tau^{-1} \circ \sigma^{-1}  \, \to & \; \pi \circ I \circ \tau^{-1} \circ I \circ \tau^{-1} \circ \sigma^{-1} = \pi \circ \sigma^{-1} 
	\end{split}
	\end{align}
\end{subequations}
where we used 
\begin{equation}\label{eq::tsp_bip_Itau}
I \circ \tau^{-1} \circ I = \tau.
\end{equation}
The two matchings corresponding to the two permutations $\mu_1$ and $\mu_2$ have no edges in common and therefore each vertex will appear twice in the union of their edges. Remark also that
\begin{equation}
\mu_2 = \mu_1 \circ \sigma \circ \tau^{-1} \circ \sigma^{-1}
\end{equation}
which means that $\mu_1$ and $\mu_2$ are related by a permutation which has to be, as it is $\tau^{-1}$, a unique cycle of length $N$. It follows that, if $h^*$ is the optimal Hamiltonian cycle and $m^*$ is the optimal assignment,
\begin{equation}\label{eq::ineq_mat_tsp}
E[h^*] \; \geq \; 2 \, E[m^*] \, .
\end{equation}

\subsubsection{Traveling on a line... and tying shoelaces!}
Here we shall focus on the one-dimensional case, where both red and blue points are chosen uniformly in the unit interval $[0,1]$.
Remember that, as seen in Sec.~\ref{sec::mat_1d}, given two sets of sorted points in increasing order, the optimal assignment is defined by the identity permutation $\mathcal{I}=(1,2,\dots, N)$. We will compute the average cost of the solution of the TSP on bipartite complete graphs, similarly to what we have done with the matching problem: as first step we will obtain the general structure of the solution, and as second step we will use this information to perform the average over the disorder.\\
From now on, we will assume $p>1$ and that both red and blue points are ordered, i.e. $r_1 \le \dots \le r_N$ and $b_1 \le \dots \le b_N$.
Let 
\begin{equation}\label{eq::tsp_bip_sigmatilde}
\tilde{\sigma}(i) = 
\begin{cases}
	2i-1 & i \leq  (N+1)/2 \\
	2N -2i +2 & i > (N+1)/2 
\end{cases}
\end{equation}
and 
\begin{equation}\label{eq::tsp_bip_pitilde}
\tilde{\pi}(i) = \tilde{\sigma}\circ I(i) = \tilde{\sigma}(N+1-i)  =
\begin{cases}
2i  & i < (N +1)/2 \\
2N - 2i +1 & i \geq  (N +1)/2 
\end{cases}
\end{equation}
the couple $(\tilde{\sigma}, \tilde{\pi})$  will define a Hamiltonian cycle $\tilde{h}\in \mathcal{H}$. More precisely, according to the correspondence given in Eq.~\eqref{eq::tsp_bip_corris}, it contains the edges
for even $N$, 
\begin{subequations}\label{eq::tsp_bip_1d_even}
	\begin{align}
	\tilde{e}_{2i-1} = & \,  
	\begin{cases}
	(r_{2i-1}, b_{2i})   & i \leq  N/2 \\
	(r_{2N-2i+2}, b_{2N-2i +1})    & i > N/2
	\end{cases} \\
	\tilde{e}_{2i} = & \,  \begin{cases}
	(b_{2i}, r_{2i+1})   & i <  N/2 \\
	(b_{N}, r_{N}) & i = N/2 \\
	(b_{2N-2i+1}, r_{2N-2i})    &  N/2< i < N \\
	(b_{1}, r_{1})    &  i = N
	\end{cases}
	\end{align}
\end{subequations}
while for  $N$ odd
\begin{subequations}\label{eq::tsp_bip_1d_odd}
	\begin{align}
	\tilde{e}_{2i-1} = & \,  
	\begin{cases}
	(r_{2i-1}, b_{2i})   & i <  (N-1)/2 \\
	(r_{N}, b_{N}) & i = (N-1)/2 \\
	(r_{2N-2i+2}, b_{2N-2i +1})    & i > (N-1)/2
	\end{cases} \\
	\tilde{e}_{2i} = & \,  \begin{cases}
	(b_{2i}, r_{2i+1})   & i <  (N-1)/2 \\
	(b_{2N-2i+1}, r_{2N-2i})    &  (N-1)/2< i < N \\
	(b_{1}, r_{1})    &  i = N \, .
	\end{cases}
	\end{align}
\end{subequations}

\begin{figure}
	\begin{center}
		\begin{tikzpicture}[scale=0.6]
		\node[draw,circle,inner sep=1.5pt,fill=black,text=white,label=above:{\footnotesize $r_1$}] (r1) at (1,3) {};
		\node[draw,circle,inner sep=1.5pt,fill=black,text=white,label=above:{\footnotesize $r_2$}] (r2) at (3,3) {};
		\node[draw,circle,inner sep=1.5pt,fill=black,text=white,label=above:{\footnotesize $r_3$}] (r3) at (5,3) {};
		\node[draw,circle,inner sep=1.5pt,fill=black,text=white,label=above:{\footnotesize $r_4$}] (r4) at (7,3) {};
		\node[draw,circle,inner sep=1.5pt,fill=white,text=white,label=below:{\footnotesize $b_1$}] (b1) at (1,0) {};
		\node[draw,circle,inner sep=1.5pt,fill=white,text=white,label=below:{\footnotesize $b_2$}] (b2) at (3,0) {};
		\node[draw,circle,inner sep=1.5pt,fill=white,text=white,label=below:{\footnotesize $b_3$}] (b3) at (5,0) {};
		\node[draw,circle,inner sep=1.5pt,fill=white,text=white,label=below:{\footnotesize $b_4$}] (b4) at (7,0) {};
		\draw[line width=1pt,gray]  (r1) to (b2);
		\draw[line width=1pt,gray]  (r2) to (b3);
		\draw[line width=1pt,gray]  (r3) to (b4);
		\draw[line width=1pt,gray]  (b1) to (r2);
		\draw[line width=1pt,gray]  (b2) to (r3);
		\draw[line width=1pt,gray]  (b3) to (r4);
		\draw[line width=1pt,gray] (r1) to[in=110, out=-110] (b1);
		\draw[line width=1pt,gray] (r4) to[in=70, out=-70] (b4);
		\end{tikzpicture}
	\end{center}
	\caption{The optimal Hamiltonian cycle $\tilde{h}$ for $N=4$ blue and red points chosen in the unit interval and sorted in increasing order.} \label{fig::tsp_bip_exh}
\end{figure}
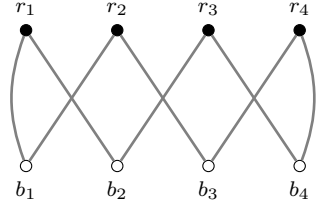

The cycle $\tilde{h}$ is the analogous of the \emph{criss-cross} solution introduced by Halton~\cite{Halton1995} (see Fig.~\ref{fig::tsp_bip_exh}). 
In his work, Halton studied the optimal way to tie a shoe. This problem can be seen as a peculiar instance of a 2-dimensional bipartite Euclidean TSP with the parameter which tunes the cost $p=1$.
One year later, Misiurewicz~\cite{Misiurewicz1996} generalized Halton's result giving the least restrictive requests on the 2-dimensional TSP instance to have the criss-cross cycle as solution.
Other generalizations of these works have been investigated in more recent papers~\cite{Polster2002,Garcia2017}.
In Appendix \ref{app::tsp_bip} we prove that for a convex and increasing cost function the optimal Hamiltonian cycle is provided by $\tilde{h}$.

\subsubsection{Statistical properties of the solution cost}
Similarly to the assignment problem, we can exploit a generalization of the Selberg integral given in Eq.~\eqref{eq::selberg} (see~\cite[Sec. 8.3]{Andrews1999}),
\begin{equation}\label{eq::selberg_gen}
\begin{split}
B_n(j, k; \alpha, \beta, \gamma) \! & :=  \! \left(\prod_{i=1}^n \int_0^1 \!\!\! dx_i \, x_i^{\alpha-1} (1-x_i)^{\beta -1}\!\right) \!\!\left( \prod_{s=1}^j \! x_s \! \right) \!\!\! \left( \prod_{s=j+1}^{j+k} \!\! (1-x_s) \! \right) \!\! \left|\Delta(x)\right |^{2 \gamma}\\
& \, = 
\, S_n(\alpha, \beta, \gamma)\, \frac{\prod_{i=1}^j [\alpha +(n-i)\gamma] \prod_{i=1}^k [\beta +(n-i)\gamma] }{\prod_{i=1}^{j+k} [\alpha+\beta +(2n-1-i)\gamma]},
\end{split}
\end{equation}
to compute the average solution cost for each $N$. By using Eq.~\eqref{eq::selberg_gen} and the probability that given $N$ ordered points on a line the $k$-th is in $[x,x+dx]$, Eq.~\eqref{eq::ordstat_pk}, we obtain:
\begin{equation}
\begin{split}
& B_2\left(1,1; k, N-k,\frac{p}{2} \right) = \\
& =  \int_{0}^{1} dx_1 \, \int_{0}^{1} dx_2 \, x_1^{k-1} \, x_2^k\, (1-x_1)^{N-k} (1-x_2)^{N-k-1} \left|x_1-x_2 \right|^p \\
& = \frac{\left(k+\frac{p}{2}\right)\left(N-k+\frac{p}{2}\right)}{(N+p)\left(N+\frac{p}{2}\right)} \,S_2\left(k,N-k,\frac{p}{2} \right) \\
& = \frac{\Gamma(k)\Gamma(N-k) \, \Gamma(p+1) \, \Gamma\left(k+\frac{p}{2}+1\right) \, \Gamma\left(N-k+\frac{p}{2}+1\right)}{\Gamma(N+p+1) \, \Gamma\left(N+\frac{p}{2}+1\right) \, \Gamma\left( 1+ \frac{p}{2}\right)} \, .
\end{split}
\end{equation}
Therefor we get
\begin{equation}
\begin{split}
\overline{\left| b_{k+1}-r_k \right|^p} & = \overline{\left| r_{k+1}-b_k \right|^p} = \int_0^1 \, dx\, \int_0^1 \, dy\, P_k(x) \, P_{k+1}(y) \, \left| x-y \right|^p  \\
& = \frac{\Gamma^2(N+1)}{\Gamma(k) \, \Gamma(N-k) \, \Gamma(k+1)\, \Gamma(N-k+1)} \times\\
& \hphantom{=} \times \int_{0}^{1} dx \, dy \, x^{k-1} \, y^k (1-x)^{N-k} (1-y)^{N-k-1} \left|x-y \right|^p \\
& =  \frac{\Gamma^2(N+1)}{\Gamma(k) \, \Gamma(N-k) \, \Gamma(k+1)\, \Gamma(N-k+1)}\,  B_2\left(1,1; k, N-k,\frac{p}{2} \right)\\
& = \frac{\Gamma^2(N+1) \, \Gamma(p+1) \, \Gamma\left(k+\frac{p}{2}+1\right) \, \Gamma\left(N-k+\frac{p}{2}+1\right)}{\Gamma(k+1) \, \Gamma(N-k+1) \, \Gamma(N+p+1) \, \Gamma\left(N+\frac{p}{2}+1\right) \, \Gamma\left( 1+ \frac{p}{2}\right)} \,.
\label{ass2}
\end{split}
\end{equation}
from which we obtain
\begin{equation}
\sum_{k=1}^{N-1} \overline{\left| b_{k+1}-r_k \right|^p} = 2 \, \Gamma(N+1) \Gamma(1+p) \left( \frac{(N+p+1) \, \Gamma(\frac{p}{2})}{4 (p+1) \, \Gamma(p) \, \Gamma(N+1+\frac{p}{2})}-\frac{1}{\Gamma(N+1+p)} \right).
\end{equation}
In addition
\begin{equation}
\begin{split}
\overline{\left| r_1 -b_1 \right|^p} = \overline{\left| r_N -b_N \right|^p} & = N^2 \int_{0}^{1} dx \, dy \, (xy)^{N-1} \left| x-y \right|^p \\
& = N^2 S_2\left(N,1,\frac{p}{2}\right) = \frac{N \, \Gamma(N+1) \, \Gamma(p+1)}{\left(N+\frac{p}{2}\right) \, \Gamma(N+p+1)} \,.
\end{split}
\end{equation}
Finally, the average optimal cost for every $N$ and every $p>1$ is
\begin{equation}\label{eq::tsp_bip_costtsp}
\begin{split}
\overline{E_N^{(p)}} &  = 2 \left(\overline{\left| r_1 -b_1 \right|^p} +\sum_{k=1}^{N-1} \overline{\left| b_{k+1}-r_k \right|^p} \right)\\
& = 2 \, \Gamma(N+1) \left[ \frac{(N+p+1) \, \Gamma \left( 1+ \frac{p}{2} \right)}{(p+1) \, \Gamma\left( N+1+\frac{p}{2} \right)} - \frac{2 \, \Gamma (p+1)}{(2N+p) \, \Gamma(N+p)} \right] \,.
\end{split}
\end{equation}
Notice that, for large $N$,
\begin{equation}
\lim\limits_{N \to \infty} N^{p/2-1} \overline{E_N^{(p)}} = 2\, \frac{\Gamma\left(\frac{p}{2}+1\right)}{p+1} \,,
\end{equation}
which is twice the cost of the assignment problem in the limit of large $N$, Eq.~\eqref{eq::cost_assgn_1d_largeN}.
The case $p=2$ of Eq.~\eqref{eq::tsp_bip_costtsp} is confronted with numerical simulation in Fig.~\ref{fig::tsp_bip_plot}.

\begin{figure}
	\centering
	\includegraphics[width=0.95\columnwidth]{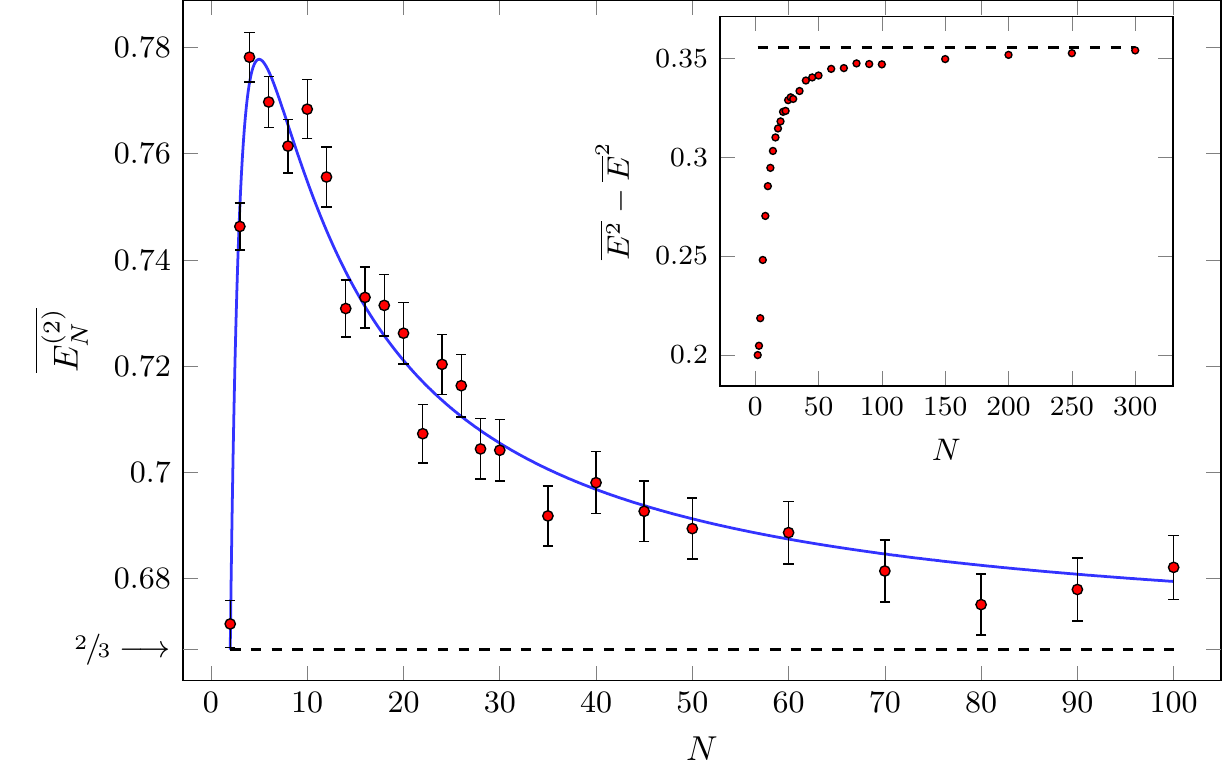}
	\caption{Numerical results for $\overline{E_N^{(2)}}$ for several values of $N$. The continuous line represents the exact prediction given in Eq.~\eqref{eq::tsp_bip_costtsp} and the dashed line gives the value for infinitely large $N$. For every $N$ we have used $10^4$ instances. In the inset we show the numerical results for the variance of the cost $E_N^{(2)}$ obtained using the exact solution provided by Eq.~\eqref{eq::tsp_bip_sigmatilde} and Eq.~\eqref{eq::tsp_bip_pitilde}. The dashed line represents the theoretical large $N$ asymptotic value. Error bars are also plotted but they are smaller than the mark size.}\label{fig::tsp_bip_plot}
\end{figure}

Finally, we can compute in the thermodynamical limit the variance of the solution cost, to check if this quantity is self-averaging or not. 
Given two sequences of $N$ points randomly chosen on the segment $[0,1]$, the probability for the difference $\phi_k$ in the position  between the $(k+1)$-th and the $k$-th points is
\begin{align}\label{eq::tsp_bip_phi_k}
\begin{split}
\Pr\left[\phi_k\in d \phi\right] & = k(k+1)\, \binom{N}{k} \, \binom{N}{k+1} \,d\phi_k \\
& \int dx \, dy  \, \delta(\phi_k - y +x) \, x^{k-1} \, y^{k} \, (1-x)^{N-k} (1-y)^{N-k-1} 
\end{split}
\end{align}
Proceeding as in the case of the assignment discussed in ~\cite{Boniolo2014, Caracciolo2014_2}, one can show that these random variables $\phi_k$ converge (in a weak sense specified by Donsker's theorem) to $\phi(s)$, which is a difference of two Brownian bridge processes~\cite{Caracciolo2017}. \\
One can write the re-scaled average optimal cost as
\begin{equation}
\overline{E_p} \equiv \lim_{N\to \infty} N^{\frac{p}{2}-1} \, \overline{ E_N^{(p)} }. 
\label{Ep}
\end{equation}
By starting at finite $N$ with the representation given in Eq.~\eqref{eq::tsp_bip_phi_k},
the large $N$ limit  can be obtained by setting $k=Ns+\frac{1}{2}$ and introducing the variables $\xi$, $\eta$ and $\varphi$ such that
\begin{equation}
x=s+\frac{\xi}{\sqrt N},\quad y=s+\frac{\eta}{\sqrt N},\quad \phi_k=\frac{\varphi(s)}{\sqrt N},
\end{equation}
in such a way that $s$ is kept fixed when $N\to +\infty$.
We obtain, at the leading order,
\begin{equation}
\begin{split}
\Pr \left[\varphi(s)\in  d\varphi\right] & = \, d\varphi\iint\delta\left(\varphi-(\eta-\xi)\right)\frac{\exp\left(-\frac{\xi^2+\eta^2}{2 s(1-s)}\right)}{2\pi s(1-s)}d\xi \,d\eta\\
=&\,\frac{1}{\sqrt{4 \pi s(1-s)}}\exp\left\{-\frac{1}{4 s(1-s)} \varphi ^2\right\}d\varphi. \label{varphi-d}
\end{split}
\end{equation}
Similarly, see for example~\cite[Appendix A]{Caracciolo2014_2}, it can be derived that the joint probability distribution $p_{t,s}(x,y)$ for $\varphi(s)$ is (for $t<s$) a bivariate Gaussian distribution 
\begin{align}
\begin{split}
p_{t,s}(x,y) = \overline{\delta(\varphi(t)-x) \, \delta(\varphi(s)-y)} = \frac{e^{-\frac{x^2}{4t}-\frac{(x-y)^2}{4(s-t)} - \frac{y^2}{4(1-s)} } }{4 \pi \sqrt{t(s-t)(1-s)}}.
\end{split}
\end{align}
This allows to compute, for a generic $p>1$, the average of the square of the re-scaled optimal cost 
\begin{equation}
\label{eq:costvariance}
\overline{E^2_p}=
4 \int_0^1 \!dt \int_0^1 ds \, \overline{\left|\varphi(s)\right|^p \left|\varphi(t)\right|^p},
\end{equation}
which is 4 times the corresponding one of a bipartite matching problem. 
In the case $p=2$, the average in Eq.~(\ref{eq:costvariance}) can be evaluated by  using the Wick theorem for expectation values in a Gaussian distribution
\begin{equation}
\overline{E^2_2}=
4 \int_0^1 \!ds \int_0^s \!dt 
\int_{-\infty}^{\infty}\!dx\,dy\,p_{t,s}(x,y) \, x^2 y^2=
\frac{4}{5} \,,
\end{equation}
and therefore
\begin{equation}
\overline{E^2_2}-\overline{E_2}^2 = \frac{16}{45} = 0.3\bar{5}.
\end{equation}
This result is in agreement with the numerical simulations (see inset of Fig.~\ref{fig::tsp_bip_plot}) and proves that the re-scaled optimal cost is not a self-averaging quantity.

\subsection{TSP on complete graphs}\label{sec::tsp_mono_1d}
Analogously to our previous analysis of the TSP on complete bipartite graphs, we can address the complete graph case. We will be able to study the problem not only in the $p>1$ case, but also in the $0<p<1$ and $p<0$ cases. However, as we will see, this last case is particularly tricky and we will not be able to univocally determine the structure of the solution. Nonetheless, we will overcome the difficulty and obtain an upper and lower bound for the average cost, which become strict in the large-$N$ limit. 
\subsubsection{Optimal cycles on the complete graph}
We shall consider the complete graph $\mathcal{K}_{N}$ with $N$ vertices, that is with vertex set $V=[N]:=\{1,\dots, N\}$. This graph has $(N -1)!/2$ Hamiltonian cycles. 
Indeed, each permutation $\pi$ in the symmetric group of $ N $ elements, $\pi \in \mathcal{S}_{N}$, defines an Hamiltonian cycle on $\mathcal{K}_{N}$. 
The sequence of points $(\pi(1), \pi(2),\dots,\pi(N),\pi(1))$ defines a closed walk with starting point $\pi(1)$, but the same walk is achieved by choosing any other vertex as starting point and also by following the walk in the opposite order, that is, $(\pi(1), \pi(N),\dots,\pi(2),\pi(1))$. As the cardinality of $\mathcal{S}_{N}$ is $N!$ we get that the number of Hamiltonian cycles in $\mathcal{K}_{N}$ is $N! / (2 N)$.

In this section, we characterize the optimal Hamiltonian cycles for different values of the parameter $p$ used in the cost function. Notice that $p=0$ and $p=1$ are degenerate cases, in which the optimal tour can be found easily by looking, for example, at the $0<p<1$ case.

\paragraph{The $p>1$ case}
We start by proving the shape of the optimal cycle when $p>1$, for every realization of the disorder.
Let us suppose, now, to have $N$ points $\mathcal{R} = \{r_i\}_{i=1,\dots,N}$ in the interval $[0,1]$. As usual we will assume that the points are ordered, i.e. $r_1 \le \dots \le r_N$. Let us define the following Hamiltonian cycle
\begin{equation}\label{eq::tsp_mono_p>1}
h^* = h[\tilde{\sigma}] = (r_{\tilde{\sigma}(1)}, r_{\tilde{\sigma}(2)}, \dots, r_{\tilde{\sigma}(N)}, r_{\tilde{\sigma}(1)})
\end{equation}
with $\tilde{\sigma}$ defined as in Eq.~\eqref{eq::tsp_bip_sigmatilde}. In Appendix~\ref{app::tsp_mono} we prove that the Hamiltonian cycle which provides the optimal cost is $h^*$.\\
The main ideas behind the proof is that we can introduce a complete bipartite graph in such a way that a solution of the bipartite matching problem on it is a solution of our original problem, with the same cost. Therefore, using the results known for the bipartite problem, we can prove the optimality of $h^*$. \\
A graphical representation of the optimal cycle for $p>1$ and $N=6$ is given in Fig.~\ref{fig::tsp_mono_fig1}, left panel.
\begin{figure*}[t]
	\begin{subfigure}{0.5\linewidth}
		\centering
		\includegraphics[width=0.9\columnwidth]{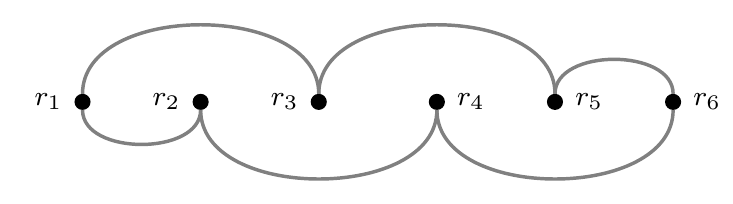}
	\end{subfigure} \hfill
	\begin{subfigure}{0.5\linewidth}
		\centering
		\includegraphics[width=0.9\columnwidth]{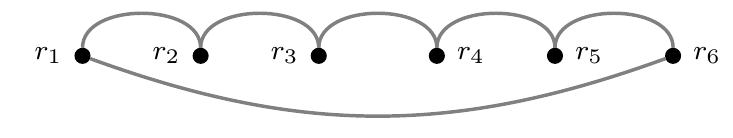}
	\end{subfigure}
	\caption{Optimal solutions for $N=6$, for the $p>1$ (left) and $0<p<1$ (right) cases. Notice how, in the $0<p<1$ case, when all the arcs are drawn in the upper half-plane above the points than there is no crossing between the arcs.}\label{fig::tsp_mono_fig1}
\end{figure*}

\paragraph{The $0<p<1$ case}
Given an ordered sequence $\mathcal{R} = \{r_i\}_{i=1,\dots,N}$ of $N$ points in the interval $[0,1]$, with $r_1 \le \dots \le r_N$, if $0 < p <1$ and if
\begin{equation}\label{eq::tsp_mono_0<p<1}
h^* = h[{\mathds{1}}] = (r_{{\mathds{1}}(1)}, r_{{\mathds{1}}(2)}, \dots, r_{{\mathds{1}}(N)}, r_{{\mathds{1}}(1)})
\end{equation}
where $\mathds{1}$ is the identity permutation, i.e.:
\begin{equation}
\mathds{1}(j) = j 
\end{equation}
then the Hamiltonian cycle which provides the optimal cost is $h^*$.\\
The idea behind this result is that we can define a crossing in the cycle as follows: let $\{r_i\}_{i=1,\dots,N}$ be the set of points, labeled in ordered fashion; consider two links $(r_i,r_j)$ and $(r_k,r_\ell)$ with $i<j$ and $k<\ell$; a crossing between them occurs if $i<k<j<\ell$ or $k<i<\ell<j$.
This corresponds graphically to a crossing of lines if we draw all the links as, for example, semicircles in the upper half-plane. In the following, however, we will not use semicircles in our figures to improve clarity (we still draw them in such a way that we do not introduce extra-crossings between links other than those defined above). An example of crossing is in the following figure

\begin{figure}[h]
	\centering
	\includegraphics[width=0.4\columnwidth]{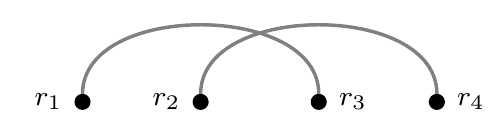}
\end{figure}

\noindent where we have not drawn the arcs which close the cycle to emphasize the crossing. Now, as shown in~\cite{Boniolo2014}, if we are able to swap two crossing arcs with two non-crossing ones, the difference between the cost of the original cycle and the new one simply consists in the difference between a crossing matching and a non-crossing one,  that is positive when $0<p<1$. Therefore the proof of the optimality of the cycle in Eq.~\eqref{eq::tsp_mono_0<p<1}, which is given in Appendix \ref{app::tsp_mono}, consists in showing how to remove a crossing (without breaking the cycle into multiple ones) and in proving that $h^*$ is the only Hamiltonian cycle without crossings (see Fig.~\ref{fig::tsp_mono_fig1}, right panel).

\paragraph{The $p<0$ case}
Here we study the properties of the solution for $p<0$. Our analysis is based, again, on the properties of the $p<0$ optimal matching solution. In~\cite{Caracciolo2017} it is shown that the optimal matching solution maximizes the total number of crossings, since the cost difference of a non-crossing and a crossing matching is always positive for $p<0$. This means that the optimal matching solution of $2N$ points on an interval is given by connecting the $i$-th point to the $(i+N)$-th one with $i = 1, \dots, N$; in this way every edge crosses the remaining $N-1$. 
Similarly to the $0<p<1$ case, suppose now to have a generic oriented Hamiltonian cycle and draw the connections between the vertices in the upper half plain (as before, eliminating all the crossings which depend on the way we draw the arcs). Suppose it is possible to identify a matching that is non-crossing, then the possible situations are the following two (we draw only the points and arcs involved in the non-crossing matching):
\begin{figure}[h]
	\centering
	\includegraphics[width=0.4\columnwidth]{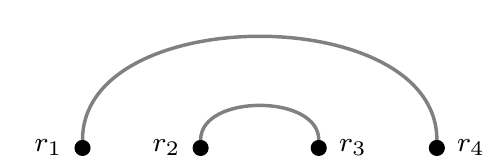} 
\end{figure}
\begin{figure}[h]
	\centering
	\includegraphics[width=0.4\columnwidth]{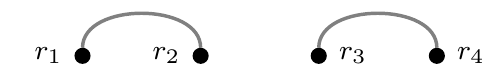}
\end{figure}
In Appendix~\ref{app::tsp_mono}, we discuss the move that allows to replace non-crossing matchings by crossing ones, in such a way that the cycle that contains the matching remains an Hamiltonian cycle.
This move is such that the cost of the new configuration is lower than the cost of the old one, since the cost gain is the difference between the costs of a non-crossing and a crossing matching, which is always positive for $p<0$.

In this manner the proof for $p<0$ goes on the same line of $0<p<1$, but instead of finding the cycle with no crossings, now we look for the one or ones that maximize them. However, as we will see in the following, one must distinguish between the $N$ odd and even case. In fact, in the $N$ odd case, only one cycle maximizes the total number of crossings, i.e. we have only one possible solution. In the $N$ even case, on the contrary, the number of Hamiltonian cycles that maximize the total number of crossings are $\frac{N}{2}$. 

\paragraph{The $p<0$ case: N odd}

Given an ordered sequence $\mathcal{R} = \{r_i\}_{i=1,\dots,N}$ of $N$ points, with $N$ odd, in the interval $[0,1]$, with $r_1 \le \dots \le r_N$, consider the permutation $\sigma$ defined as:
\begin{equation}\label{sigma}
\sigma(i) = 
\begin{cases}
	1 & \hbox{for } i = 1 \\
	\frac{N-i+3}{2} & \hbox{for } $even $ i>$1$ \\
	\frac{2N-i+3}{2} & \hbox{for } $odd $ i>$1$
\end{cases}
\end{equation}
This permutation defines the following Hamiltonian cycle:
\begin{equation}\label{eq::tsp_mono_p<0_odd}
h^*:=h[\sigma]=(r_{\sigma(1)}, r_{\sigma(2)},\dots,r_{\sigma(N)}) .
\end{equation}
The Hamiltonian cycle which provides the optimal cost is $h^*$.\\
The proof consist in showing that the only Hamiltonian cycle with the maximum number of crossings is $h^*$. As we discuss in Appendix~\ref{app::tsp_mono}, the maximum possible number of crossings an edge can have is $N-3$. The Hamiltonian cycle under exam has $N(N-3)/2$ crossings, i.e. every edge in $h^*$ has the maximum possible number of crossings. Indeed, the vertex $a$ is connected with the vertices $a+\frac{N-1}{2} \pmod N$ and $a+\frac{N+1}{2}  \pmod N$. The edge $(a,a+\frac{N-1}{2} \pmod N)$ has $2\frac{N-3}{2}=N-3$ crossings due to the $\frac{N-3}{2}$ vertices $a+1\pmod N, a+2\pmod N,\dots, a+\frac{N-1}{2}-1\pmod N$ that contribute with 2 edges each. This holds also for the edge $(a,a+\frac{N+1}{2} \pmod N)$ and for each $a\in[N]$. As shown in Appendix~\ref{app::tsp_mono} there is only one cycle with this number of crossings.

An example of an Hamiltonian cycle discussed here is given in Fig.~\ref{fig::tsp_mono_fig2}.
\begin{figure}
	\centering
	\includegraphics[width=0.5\columnwidth]{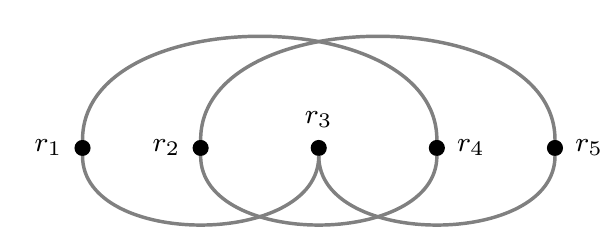}
	\caption{This is the optimal TSP and 2-factor problem solution for $N=5$, in the $p<0$ case. Notice that there are no couples of edges which do not cross and which can be changed in a crossing couple.}\label{fig::tsp_mono_fig2}
\end{figure}

\paragraph{The $p<0$ case: N even}

\begin{figure*}[ht]
	\begin{subfigure}[b]{0.5\linewidth}
		\centering
		\includegraphics[width=0.7\columnwidth]{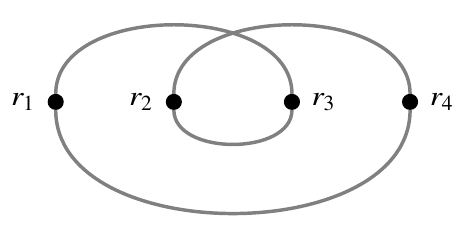}
	\end{subfigure}
	\begin{subfigure}[b]{0.5\linewidth}
		\centering
		\includegraphics[width=0.7\columnwidth]{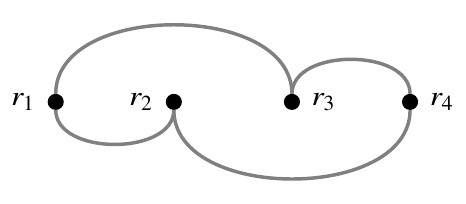}
	\end{subfigure}
	\caption{The two possible optimal Hamiltonian cycles for $p<0$, $N=4$. For each specific instance one of them has a lower cost than the other, but differently from all the other cases ($p>0$ or $N$ odd) the optimal cycle is not the same for each disorder instance. }\label{fig::tsp_mono_fig3}
\end{figure*}	

In this situation, differently from the above case, the solution is not the same irrespectively of the disorder instance. More specifically, there is a set of possible solutions, and at a given instance the optimal is the one among that set with the lowest cost. We will show how these solutions can be found and how they are related.\\
Given the usual sequence of points $\mathcal{R} = \{r_i\}_{i=1,\dots,N}$ of $N$ points, with $N$ even, in the interval $[0,1]$, with $r_1 \le \dots \le r_N$, if $p <0$, consider the permutation $\sigma$ such that:
\begin{equation}
\sigma(i)=\begin{cases}
	1  &\mbox{for } i =1\\
	\frac{N}{2}-i+3 &\mbox{for even } i \leq \frac{N}{2}+1  \\
	N-i+3 &\mbox{for odd } i \leq \frac{N}{2}+1\\
	i-\frac{N}{2} &\mbox{for even } i>\frac{N}{2}+1\\
	i &\mbox{for odd } i>\frac{N}{2}+1\\
\end{cases}
\end{equation}

Given $\tau \in \mathcal{S}_N$ defined by $\tau(i)=i+1$ for $i \in [N-1]$ and $\tau(N)=1$, we call $\Sigma$ the set of permutations  $\sigma_k, k=1,...,N$ defined as:
\begin{equation} 
\sigma_k(i)=\tau^k(\sigma(i))
\end{equation} 
where $\tau^k=\tau \circ \tau^{k-1}$.
The optimal Hamiltonian cycle is one of the cycles defined as
\begin{equation}\label{eq::tsp_mono_p<0_even}
	h_k^*:=h[\sigma_k]=(r_{\sigma_k(1)}, r_{\sigma_k(2)},\dots,r_{\sigma_k(N)}).
\end{equation}
An example with $N=4$ points is shown in Fig.~\ref{fig::tsp_mono_fig3}.
In Appendix~\ref{app::tsp_2f_mono} the 2-factor (or loop covering) with minimum cost is obtained. The idea for the proof of the TSP is to show how to join the loops in the optimal way in order to obtain the optimal TSP. The complete proof of the optimality of one among the cycles in Eq.~\ref{eq::tsp_mono_p<0_even} is given in Appendix~\ref{app::tsp_2f_mono}.

\subsubsection{Statistic properties of the solution cost}
Now we can use the insight on the solutions just obtained to compute typical properties of the optimal cost for various values of $p$.

\begin{figure}
	\centering
	\includegraphics[width=0.8\columnwidth]{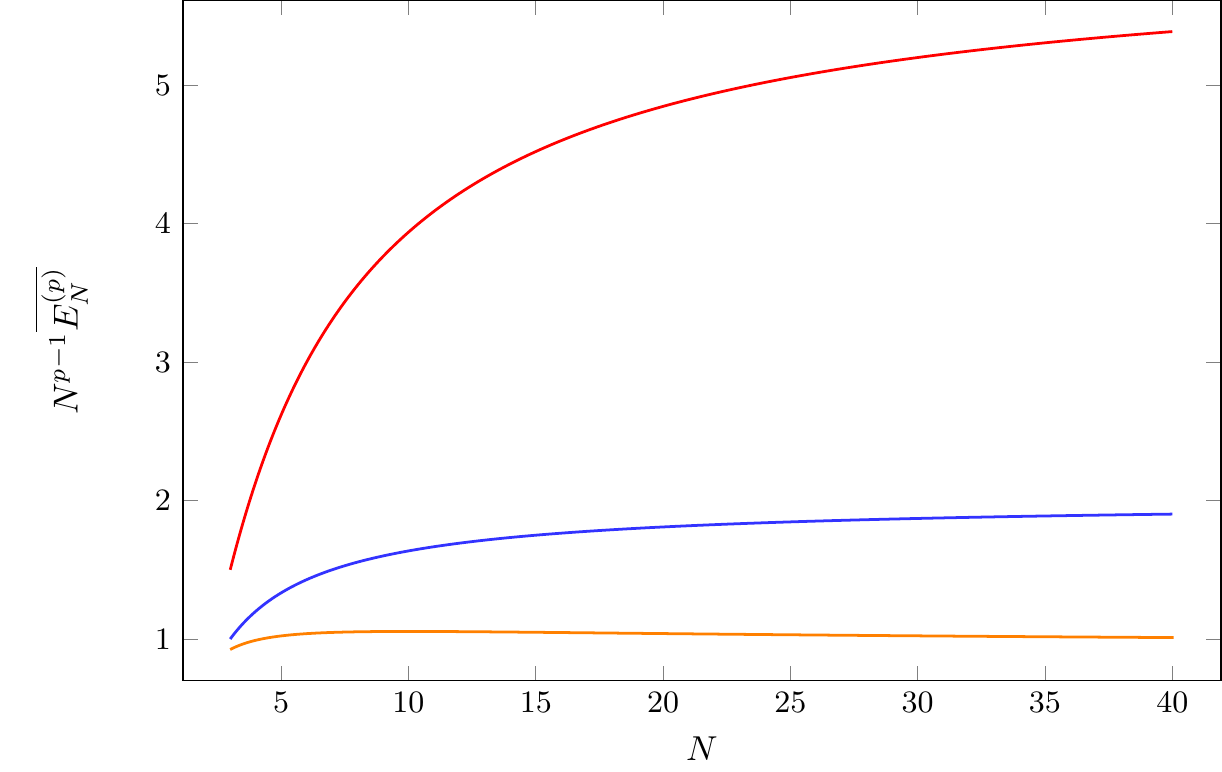} 
	\caption{Rescaled average optimal cost for $p=2$, 1, 0.5 (from top to bottom).} \label{fig::tsp_mono_TSP1}
\end{figure}
In Appendix \ref{app::orderstat} we computed the probability of finding the $l$-th point in $[x,x+dx]$, Eq.~\eqref{eq::ordstat_pk}, and the probability $p_{l, l+k}(x, y) \, dx \, dy$ of finding the $l$-th point in $[x,x+dx]$ and the $s$-th point in $[y,y+dy]$, Eq.~\eqref{eq::ordstat_pk_2}. From these equations, it follows that
\begin{equation}\label{eq::tsp_mono_k}
\int \,dx\, dy\, (y-x)^\alpha\, p_{l, \, l+k}(x, y) = \frac{\Gamma(N+1)\, \Gamma(k+\alpha)}{\Gamma(N+\alpha+1)\, \Gamma(k)   } 
\end{equation}
independently from $l$, and, therefore, in the case $p>1$ we obtain
\begin{equation}\label{eq::tsp_mono_AOC}
\overline{E_N[h^*]} = \left[ (N-2) (p+1) +2 \right] \, \frac{\Gamma(N+1)\, \Gamma(p+1)}{\Gamma(N+p+1)   } 
\end{equation}
and in particular for $p=2$ 
\begin{equation}
\overline{E_N[h^*]} = \frac{ 2\, (3 N - 4)}{(N+1) (N+2)}\, ,
\end{equation}
and for $p=1$ we get
\begin{equation}\label{eq::tsp_mono_p1}
\overline{E_N[h^*]} = \frac{ 2\, (N - 1)}{N+1}\,.
\end{equation}
In the same way one can evaluate the average optimal cost when $0<p<1$, obtaining
\begin{equation}
\begin{split}
	\overline{E_N[h^*]} = \frac{\Gamma(N+1)}{\Gamma(N+p+1)} \left[(N-1) \, \Gamma(p+1) + \frac{\Gamma(N + p - 1)}{\Gamma(N-1)} \right]
\end{split}
\end{equation}
which coincides at $p=1$ with Eq.~\eqref{eq::tsp_mono_p1} and, at $p=0$, provides  $\overline{E_N[h^*]}  = N$. For large $N$, we get
\begin{equation}
\lim_{N\to \infty} N^{p-1} \overline{E_N[h^*]}  = 
\begin{cases}
	\Gamma(p+2) & \hbox{for \,}  p\ge 1 \\
	\Gamma(p+1) & \hbox{for \,}  0< p<1\, .
\end{cases}
\end{equation}
The asymptotic cost for large $N$ and $p>1$ is $2(p+1)$ times the average optimal cost of the matching problem on the complete graph $\mathcal{K}_N$ given in Eq.~\eqref{eq::matching_cost_largeN} (notice that in Eq.~\eqref{eq::matching_cost_largeN} the cost is normalized with $N$ and the number of points is $2N$, differently from what we do here). This factor $2(p+1)$ is another difference with respect to the bipartite case, where we have seen that the cost of the TSP is twice the cost of the assignment problem for large $N$, independently of $p$.

For $p<0$ and $N$ odd we have only one possible solution, so that the average optimal cost is
\begin{equation}
\begin{split}
\overline{E_N[h^*]} = \frac{\Gamma(N+1)}{2\Gamma(N+p+1)} \left[(N-1) \frac{\Gamma\left(\frac{N+1}{2}+p\right)}{\Gamma \left( \frac{N+1}{2} \right)} + (N+1) \frac{\Gamma\left(\frac{N-1}{2}+p\right)}{\Gamma \left( \frac{N-1}{2} \right)}\right] \,.
\end{split}
\end{equation}
For large $N$ it behaves as
\begin{equation}
\lim_{N\to \infty}  \frac{\overline{E_N[h^*]}}{N} = \frac{1}{2^p} \,,
\end{equation}
which coincides with the scaling derived before for $p=0$. Note that for large $N$ the average optimal cost of the TSP problem is two times the one of the corresponding matching problem for $p<0$~\cite{Caracciolo2017}. 

\begin{figure}
	\centering
	\includegraphics[width=0.8\columnwidth]{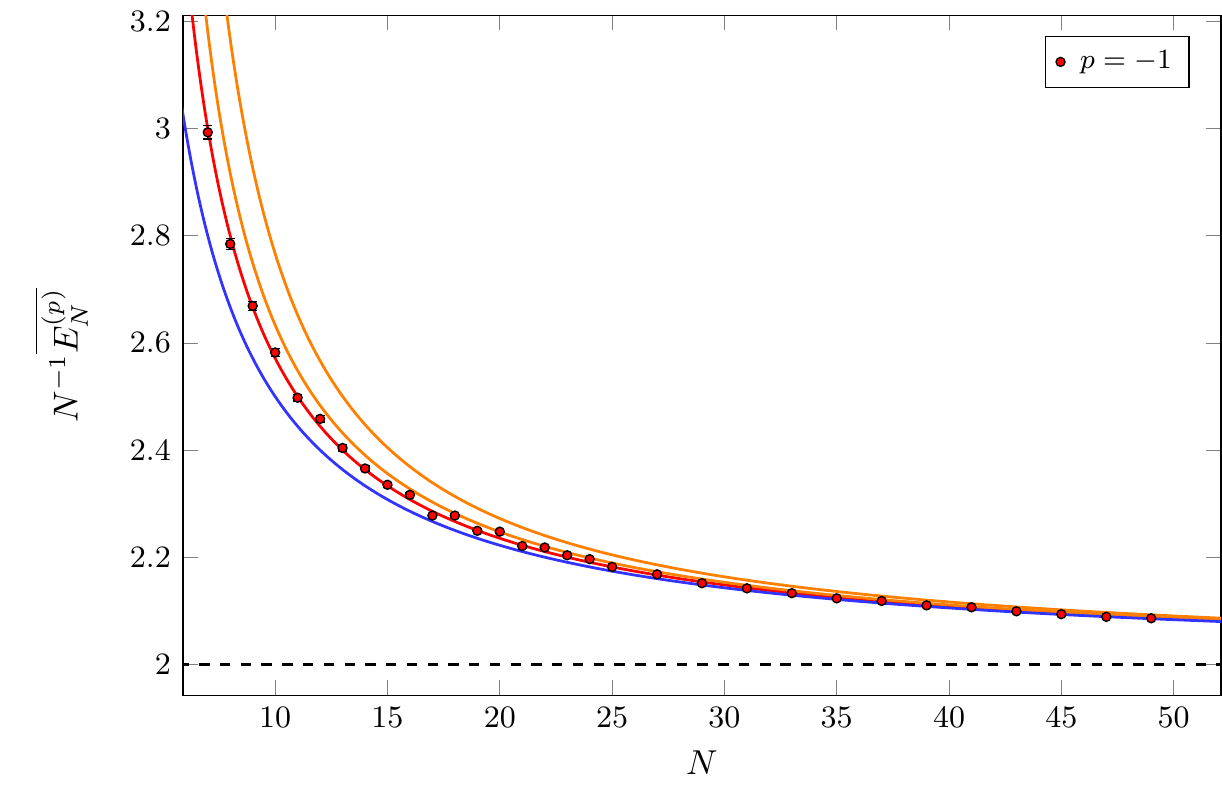} 
	\caption{Rescaled average optimal cost in the $p=-1$ case. The red points and line are respectively the result of a numerical simulation and the theoretical prediction in the odd $N$ case. The blue line is the 2 times the theoretic value of the optimal matching. The orange lines (from top to bottom) are the average costs $\overline{E_N[h_1]}$ and $\overline{E_N[h_2]}$ defined in Eqs.~\eqref{eq::tsp_mono_h_1} and \eqref{eq::tsp_mono_h_2} respectively. The dashed black line is the large $N$ limit of all the curves.} \label{fig::tsp_mono_TSP2}
\end{figure}

For $N$ even, instead, there are $N/2$ possible solutions. One can see $N/2-1$ of these share the same average energy, since they have the same number of links with the same $k$ of Eq.~\eqref{eq::tsp_mono_k}. These solutions have 2 links with $k=N/2$, $N/2$ links with $k=N/2+1$ and $N/2-2$ links with $k=N/2+1$. We denote this set of configurations with $h_1$ (although they are many different configurations, we use only the label $h_1$ to stress that all of them share the same average optimal cost) and its average cost is
\begin{equation}\label{eq::tsp_mono_h_1}
\begin{split}
\overline{E_N[h_1]} =& \frac{\Gamma(N+1)}{\Gamma(N+p+1)} \left[ \frac{N}{2} \frac{\Gamma\left( \frac{N}{2} + p-1 \right)}{\Gamma \left( \frac{N}{2}-1 \right)}  \right.\\
& \left. + \left( \frac{N}{2}-2 \right) \frac{\Gamma \left( \frac{N}{2} + p + 1 \right)}{\Gamma \left( \frac{N}{2} +1 \right)} + 2 \frac{\Gamma\left( \frac{N}{2}+p \right)}{\Gamma \left( \frac{N}{2} \right)}\right].
\end{split}
\end{equation}
The other possible solution, that we denote with $h_2$ has 2 links with $k=N/2-1$, $N/2$ links with $k=N/2+1$ and $N/2-1$ links with $k=N/2+1$ and its average cost is
\begin{equation}\label{eq::tsp_mono_h_2}
\begin{split}
\overline{E_N[h_2]} = & \frac{\Gamma(N+1)}{\Gamma(N+p+1)} \left[ \left(\frac{N}{2}-1\right) \frac{\Gamma\left( \frac{N}{2} + p-1 \right)}{\Gamma \left( \frac{N}{2}-1 \right)} \right. \\
& \left. + \left( \frac{N}{2}-1 \right) \frac{\Gamma \left( \frac{N}{2} + p + 1 \right)}{\Gamma \left( \frac{N}{2} +1 \right)} + 2 \frac{\Gamma\left( \frac{N}{2}+p \right)}{\Gamma \left( \frac{N}{2} \right)}\right].
\end{split}
\end{equation}
In Fig.~\ref{fig::tsp_mono_TSP1} we plot the analytical results for $p=0.5$, $1$, $2$ and in Fig.~\ref{fig::tsp_mono_TSP2} we compare analytical and numerical results for $p=-1$. In particular, since $\overline{E_N[h_1]} > \overline{E_N[h_2]}$, $\overline{E_N[h_2]}$ provides our best upper bound for the average optimal cost of the $p=-1$, $N$ even case. The numerical results have been obtained by solving $10^4$ TSP instances using its linear programming representation.

Now we investigate whether the optimal cost is a self-averaging quantity.
We collect in Appendix~\ref{app::tsp_mono_2Moment} all the technical details concerning the evaluation of the second moment of the optimal cost distribution $\overline{E_N^2}$, which has been computed for all number of points $N$ and, for simplicity, in the case $p>1$ and it is given in Eq.~\eqref{eq::app_tsp_mono_SecondMoment}.
In the large $N$ limit it goes like
\begin{equation}
\lim_{N\to \infty} N^{2(p-1)} \overline{E_N^2[h^*]} = \Gamma^2(p+2)
\end{equation}
i.e. tends to the square of the rescaled average optimal cost. This proves that the cost is a self-averaging quantity. 
Using Eq.~\eqref{eq::app_tsp_mono_SecondMoment} together with Eq.~\eqref{eq::tsp_mono_AOC} one gets the variance of the optimal cost. In particular for $p=2$ we get
\begin{equation}
\sigma_{E_N}^2=\frac{4 (N (5 N (N+13)+66)-288)}{(N+1)^2 (N+2)^2 (N+3) (N+4)} \,,
\label{Variance}
\end{equation}
which goes to zero as $\sigma_{E_N}^2 \simeq 20/N^3$.

\subsection{The bipartite traveling salesman problem in two dimensions}
We have seen that in one dimension the cost of the solution of the bipartite TSP is twice that the cost of the assignment problem. This actually holds also in two dimensions, where the bipartite TSP is a genuine NP-hard problem. I.e.~for any given choice of the positions of the points, in the asymptotic limit of large $N$,  the cost of the bipartite TSP converges to twice the cost of the assignment. However, this claim is non-trivial and it requires several results introduced previously, together with a scaling argument which we present in this section. This is another noticeable example where information about average properties of the solution of a hard COP can be obtained even in more than one dimension and in the presence of Euclidean correlations.

\subsubsection{Scaling argument}

Given an instance of $N$ blue and $N$ red point positions, let us consider the optimal assignment $\mu^*$ on them. Let us now consider $N$ points which are taken between the red an blue point of each edge in $\mu^*$ and call $\mathcal{T}^*$ the optimal ``monopartite'' TSP solution on these points. For simplicity, as these $N$ points we take the blue points. 

We shall use $\mathcal{T}^*$ to provide an ordering among the red and blue points. Given two consecutive points in $\mathcal{T}^*$, for example $b_1$ and $b_2$, let us denote by $(r_1,b_1)$ and $(r_2,b_2)$ the two edges in $\mu^*$ involving the blue points $b_1$ and $b_2$ and let us consider also the new edge $(r_1, b_2)$.
We have seen that, in the asymptotic limit of large $N$, the typical distance between two matched points in $\mu^*$ scales as $ (\log N/N)^{1/2}$ (see Sec.~\ref{sec::matching}) while the typical distance between two points matched in the monopartite case scales only as $1/N^{1/2}$~\cite{Beardwood1959}, that is (for all points but a fraction which goes to zero with $N$)
\begin{equation}
\begin{split}
& w_{(b_1,r_1)} = \left( \alpha_{11} \frac{\log N}{N} \right)^{\frac{p}{2}}, \\ 
& w_{(b_2,r_1)} =\left[\beta_{22} \frac{1}{N} + \alpha_{11} \frac{\log N}{N} - \gamma \frac{\sqrt{\log N}}{N} \right]^{\frac{p}{2}}.
\end{split}
\end{equation}
where $(\alpha_{11} \log N / N)^{1/2}$ is the length of the edge $(r_1,b_1)$ of $\mu^*$, $(\beta_{22} / N)^{1/2}$ is the length of the edge $(b_1, b_2)$ of $\mathcal{T}^*$ and $\gamma = 2 \sqrt{\alpha_{11} \beta_{22}} \cos\theta$, where $\theta$ is the angle between the edges $(r_1,b_1)$ of $\mu^*$ and $(b_1,b_2)$ of $\mathcal{T}^*$.

This means that, typically,  the difference in cost 
\begin{equation}
\Delta E = w_{(b_2,r_1)} - w_{(b_1,r_1)} \sim \frac{(\log N)^{\frac{p-1}{2}}}{N^{\frac{p}{2}}}
\end{equation}
is small as compared to the typical cost $(\log N/N)^\frac{p}{2}$ of one edge in the bipartite case.
To obtain a valid TSP solution, which we call $h^A$, we add to the edges $\mu^* = \{(r_1,b_1), \dots,(r_N, b_N)\}$ the edges $\{(r_1,b_2), \dots, (r_{N-1}, b_N), (r_N, b_1)\}$, see Fig.~\ref{Fig::tsp_2d_figscale}.

\begin{figure}
	\centering
	\includegraphics[]{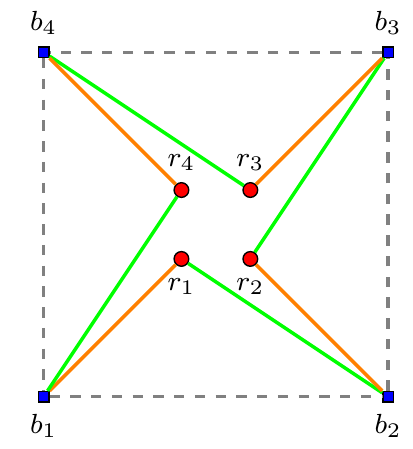}
	\caption{The optimal assignment $\mu^*$ is given by the orange edges $\{ (r_1, b_1), (r_2 ,b_2), (r_3, b_3), (r_4 ,b_4) \}$. 
		The monopartite TSP (gray dashed edges) among blue points provides the necessary ordering. In order to obtain the TSP ${b_1, r_1, b_2, r_2, b_3, r_3, b_4, r_4, b_1}$ in the bipartite graph we have to add the green edges $\{((r_1, b_2), (r_2, b_3), (r_3, b_4), (r_4, b_1)\}$.
	}\label{Fig::tsp_2d_figscale}
\end{figure}

Of course $h^A$ is not, in general, the optimal solution of the TSP. However, because of Eq.~\eqref{eq::ineq_mat_tsp}, we have that
\begin{equation}
E[h^A] \geq E[h^*]  \geq 2\, E[\mu^\star]
\end{equation}
and we have shown that, for large $N$, $E[h^A]$ goes to $2\, E[\mu^\star]$ and therefore also $E[h^*]$ must behave in the same way. 
Notice also that our argument is purely local and therefore it does not depend in any way on the type of boundary conditions adopted, therefore it holds for both open and periodic boundary conditions. 

An analogous construction can be used in any number of dimensions. However, the success of the procedure lies in the fact that the typical distance between two points in $\mu^*$ goes to zero slower than the typical distance between two consecutive points in the monopartite TSP. This is true only in one and two dimensions, and as we have already said, it is related to the importance of fluctuations in the number of points of different kinds in a small volume.

This approach allowed us to find also an approximated solution of the TSP which improves as $N\to\infty$. However, this approximation requires the solution of a \emph{monopartite} TSP on $N/2$ points, corroborating the fact that the bipartite TSP is a hard problem (from the point of view of complexity theory). 

\begin{figure*}[ht]
	\begin{subfigure}[t]{0.49\linewidth}
		\centering
		\includegraphics[width=0.95\columnwidth]{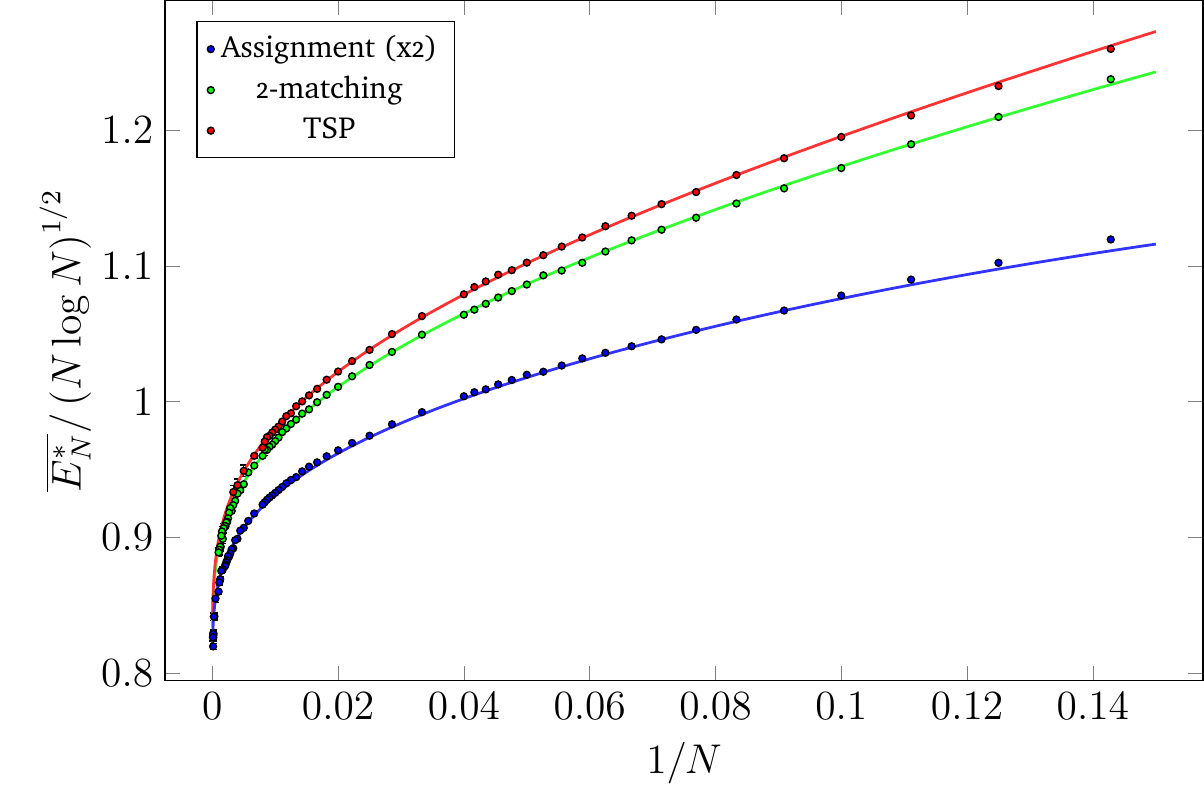}
		\label{Fig::p=1}
	\end{subfigure} \hfill
	\begin{subfigure}[t]{0.49\linewidth}
		\centering
		\includegraphics[width=0.95\columnwidth]{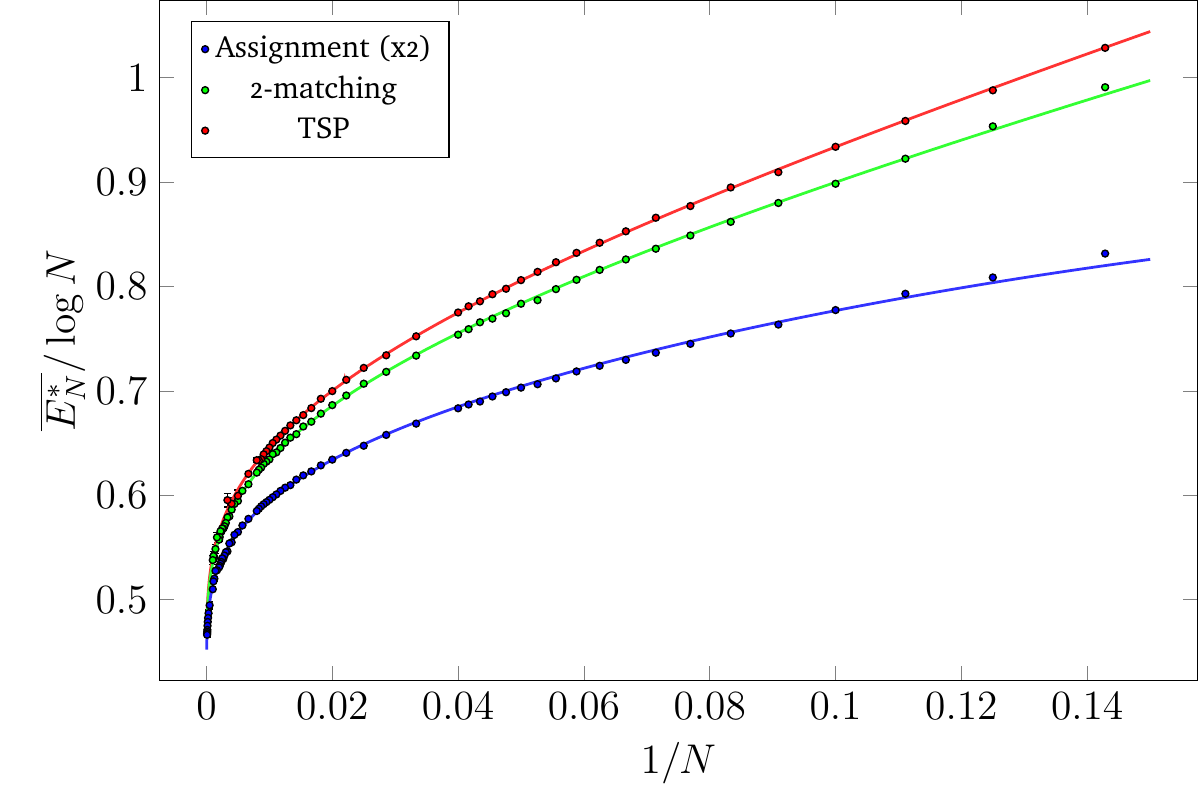}
		\label{Fig::p=2}
	\end{subfigure}
	\caption{Numerical results for $p=1$ (left panel) and $p=2$ (right panel) for the TSP (red points, top), the 2-factor, which is defined in Sec.~\ref{sec::2factor} (green points, middle), and 2 times the assignment problem (blue points, bottom) in the open boundary condition case. Continuous lines are numerical fit to the data.}
	\label{Fig::tsp_2d_AverageCost}
\end{figure*}
\begin{table*}[ht]
	\centering
	\begin{tabular}{llll} 
		{$p=1$} & $a_{1}$ & $a_{2}$ & $a_{3}$ \\
		\hline
		TSP & 0.717(2) & 1.32(1) & $-0.513(1)$ \\ 
		Assignment & 0.714(2) & 1.17(2) & $-0.77(2)$ \\  
	\end{tabular}
	\quad \quad \quad \quad \quad \quad \quad \quad
	\begin{tabular}{llll} 
		{$p=2$} & $a_{1}$ & $a_{2}$ & $a_{3}$ \\
		\hline
		TSP & 0.321(5) & 1.603(2) & $-0.428(6)$ \\ 
		Assignment & 0.31831 & 1.502(2) & $-1.05(1)$  \\ 
	\end{tabular}
	\caption{Comparison between fit factors in assignment and TSP, for $p=1$, $2$.  We have doubled the factors for the assignment to verify our hypothesis. For $p=2$, we have reported the theoretical value of $a_1$ which is $1/\pi$. 
	}
	\label{tab::tsp_2d_p1}
\end{table*}

\subsubsection{Numerical results}

We confirm our theoretical predictions performing numerical simulations on both assignment and bipartite TSP. We have considered the case of open boundary conditions.

For what concerns the assignment problem, we have implemented an in-house solver based on the LEMON optimization library~\cite{Deszo2011}, which is based on the Edmonds' blossom algorithm~\cite{Edmonds1965}. In the case of the TSP, the most efficient way to tackle the problem numerically is to exploit its \textit{linear} or \textit{integer programming} formulation. 

To validate our argument, we solved for the assignment problem (with $p=1,2$) $10^{5}$ independent instances for $2 \leq N \leq 125$, $10^{4}$ independent instances for $150 \leq N \leq 500$, and $10^{3}$ independent instances for $600 \leq N \leq 1000$. In the TSP case, the computational cost is dramatically larger; for this reason the maximum number of points we were able to achieve with a good numerical precision using integer programming was $N=300$, also reducing the total number of instances. 

An estimate of the asymptotic average optimal cost and finite size corrections has been obtained using the fitting function for $p=1$
\begin{equation}
\label{eq:p1_cost}
f^{(p=1)}(N) = \sqrt{N\log{N}} \left( a_{1}  +  \frac{a_{2}}{\log N} + \frac{a_{3}}{\log^2 N} \right)
\end{equation}
while, for $p=2$
\begin{equation}
\label{eq:p2_cost}
f^{(p=2)}(N) =  \log N \, \left( a_{1}  +  \frac{a_{2}}{\log N} + \frac{a_{3}}{\log^2 N} \right) \,.
\end{equation}
These are the first 3 terms of the asymptotic behavior of the cost of the assignment problem~\cite{Ajtai1984,Caracciolo2014}. 
Parameters $a_2$ and $a_3$ for $p=2$ were obtained fixing $a_1$ to $1/\pi$. In Fig.~\ref{Fig::tsp_2d_AverageCost} we plot the data and fit in the case of open boundary conditions. Results are reported in Table~\ref{tab::tsp_2d_p1}.

To better confirm the behavior of the average optimal cost of the TSP, we also performed some numerical simulations using a much more efficient solver, that is the Concorde TSP solver~\cite{Applegate2006}, which is based on an implementation of the Branch-and-cut algorithm proposed by Padberg and Rinaldi~\cite{Padberg1991}. The results for the leading term of the asymptotic average optimal cost are confirmed while a small systematic error due to the integer implementation of the solver is observed in the finite size corrections.

These numerical checks, together with our scaling argument, demonstrate that, as already obtained in one dimension,
\begin{equation}
\lim_{N\to\infty} \frac{\overline{E[h^*]}}{\overline{E[\mu^*]}} =  2  \, .
\end{equation}
This implies, for the special case $p=2$, by using the second line of Eq.~\eqref{eq::largeN_mat_bip_p2}, an exact, analytical result: $\lim_{N\to\infty} (\overline{E^[h^*]}/ \log N) =   1/\pi$.
In general, the evaluation of the large $N$ value of the cost of solutions of the bipartite TSP is reduced to the solution of the matching problem with the same number of points, which requires only polynomial time.
This seems to be a peculiar feature of the bipartite problem: the ``monopartite'' TSP cannot be approached in a similar way.

\subsection{Other known results}

There are many other very interesting research papers about the TSP. Here we limit ourselves to report some results regarding the average of the solution cost in higher dimension. 

As for the matching case, the mean field version of the problem can be studied with replica methods~\cite{Mezard1986} and the so-called cavity method~\cite{Krauth1989}. When the graph is complete and the links weights are IID random variables distributed according to the law
\begin{equation}
	\rho(\ell) \underset{\ell \to 0}{\sim} \frac{\ell^r}{r!}
\end{equation}
where $r$ is a parameter (notice that the behavior of the distribution far from $\ell=0$ is irrelevant), we have that, for large $N$,
\begin{equation}
	\overline{ E_N^{(r)} } \sim N^{1-1/(r+1)} L_r,
\end{equation}
where $L_r$ can be computed numerically up to the desired precision.
Notice that there result are, as in the matching case, obtained by using a RS ansatz (or the analogous hypothesis for the cavity method) and are confirmed by extensive numerical simulations (thus certifying the exactness of the RS ansatz for this problem).

Another famous result (which is actually one of the first about the RCOP version of the TSP) due to Beardwood, Halton and Hammersley~\cite{Beardwood1959} is about euclidean TSP in $d>1$ when the points are chosen with flat distribution in a volume $V$, and the cost function is the total length of the tour (that is $p=1$). In that case, we have, for large $N$
\begin{equation}
	\overline{ E_N }\sim C_d N^{1-1/d} V^{1/d}.
\end{equation}
where the constant $C_d$ is unknown analytically and has been estimated, up to a certain precision, for several number of dimensions by solving numerically the TSP and averaging the cost of the solution.

\section{Between matching and TSP: 2-factor problem}\label{sec::2factor}
\subsection{Meet the 2-factor problem}
In this Section we will deal with the 2-factor problem which consists, given an undirected graph, in finding a spanning subgraph that contains only disjoint cycles (that is, a 2-factor). For this reason this problem is also called \emph{loop covering} of a graph.\\
The 2-factor problem can be seen as a relaxation of the TSP, in which one has the additional constraint that there must be a unique cycle. We mention that also this problem can be studied using replica (and cavity) methods in the mean field case: one finds that, for large number of points, its average optimal cost is the same of that of the TSP.

In the following we will study the 2-factor problem in one dimension and in two dimensions, both on the complete bipartite graph and, only in one dimension, on the complete graph. The disorder in this problem will be introduced by drawing the points independently from the uniform distribution over the compact interval $[0,1]$ or over the square $[0,1]\times [0,1]$, as we did for the other random Euclidean COPs studied. 
As in the previous investigations, the weights on the edges are chosen as the Euclidean distance between the corresponding points on the interval or square, to the $p$. 

This problem can be seen as an intermediate problem between the assignment (or matching) and the TSP: indeed, in the former case we search for the minimum-cost 1-factor of a graph, while in the latter we are interested in the minimum-cost $N$-factor if the graph has $N$ vertices.\\
Nonetheless, when tackled in one dimension for $p>1$, we will see that there is an important difference between the 2-factor and the other studied problems: while almost for every instance of the problem there is only one solution, by looking at the whole ensemble of instances it appears an exponential number of possible solutions scaling as $\plas^N$, where $\plas$ is the plastic constant (see Appendix \ref{app::plastic_constant}). 
This is in contrast with the matching and TSP cases, where we have seen that, for $p>1$, for every realization of the disorder the configuration that solves the problem is always the same. Moreover, also for $p<0$, when for the TSP in the complete graph there are more than one possible optimal tour, they are $N$ different possibilities (when the graph has $2N$ vertices), while for the 2-factor we have an exponential number (in $N$) of them.

Let us start by defining formally the problem. Consider a graph $\mathcal{G}$ and the set of 2-factors of this graph, $\mathcal{M}_2$.
Suppose now that a weight $w_e > 0$ is assigned to each edge $e \in \mathcal{E}$ of the graph $\mathcal{G}$. We can associate to each 2-factor $\nu\in \mathcal{M}_2$ a total cost
\begin{equation}\label{eq::2f_E}
	E(\nu) :=  \sum_{e\in \nu} w_e \, .
\end{equation}
In the (weighted) 2-factor problem we search for the 2-factor $\nu^*\in \mathcal{M}_2$ such that the total cost in Eq.~\eqref{eq::2f_E} is minimized, that is
\begin{equation}\label{eq::2f-nu^*}
	E(\nu^*) = \min_{ \nu\in \mathcal{M}_2} E(\nu)\, . 
\end{equation}
If $\mathcal{H}$ is the set of Hamiltonian cycles for the graph $\mathcal{G}$, of course $\mathcal{H} \subset \mathcal{M}_2$ and therefore if $h^*$ is the optimal Hamiltonian cycle, we have
\begin{equation}\label{eq::ineq_2f_tsp}
	E[h^*] \geq E[\nu^*] \,,
\end{equation}
which is a relation between the cost of the solution of the 2-factor problem and the TSP on the same graph.

From now on we specialize to the Euclidean version of the problem, and so when the graph is complete or complete bipartite and is embedded in $[0,1]^d \subset \mathbb{R}^d$. 
For the complete case $\mathcal{G} = \mathcal{K}_N$ at each vertex $i\in [N] =\{1,2,\dots,N\}$ we associate a point $x_i\in [0,1]^d$, and for each $e=(i,j)$ with $i,j \in [N]$ we introduce a cost which is a function of their Euclidean distance
\begin{equation}
w_e = |x_i-x_j|^p \,  \label{eq::2f_p}
\end{equation}
with $p\in \mathbb{R}$. 
Analogously for the complete bipartite graph $\mathcal{K}_{N,N}$, we have two sets of points in $[0,1]^d$, that is, say, the red $\{r_i\}_{i\in [N]}$ and the blue $\{b_i\}_{i\in [N]}$ points, and the edges connect red points with blue points with a cost
\begin{equation}
w_e = |r_i-b_j|^p \, . \label{eq::2f_pb}
\end{equation}
For a discussion on this problem on an arbitrary graph $\mathcal{G}$, see~\cite{Zecchina2011} and references therein.

Let us now focus on the case of complete bipartite graph $\mathcal{K}_{N,N}$, where each cycle in a 2-factor must have an even length. 
Let $\mathcal{S}_N$ be the symmetric group of order $N$ and consider two permutations $\sigma, \pi \in \mathcal{S}_N$. If for every $i\in [N]$ we have that $\sigma(i) \neq \pi(i)$, then the two permutations define the 2-factor $\nu(\sigma,\pi)$ with edges
\begin{align}
e_{2i-1} \; := \; & (r_i, b_{\sigma(i)})\\
e_{2i} \; := \; & (r_i, b_{\pi(i)})
\end{align}
for $i\in[N]$. And, vice versa, for any 2-factor $\nu$ there is a couple of permutations  $\sigma, \pi \in \mathcal{S}_N$,  such that for every $i\in [N]$ we have that $\sigma(i) \neq \pi(i)$.

It will have total cost
\begin{equation}
E[\nu(\sigma, \pi)] = \sum_{i\in [N]} \left[ |r_i - b_{\sigma(i)}|^p +  |r_i - b_{\pi(i)}|^p \right] \, .
\end{equation}
By construction, if we denote by $\mu[\sigma]$ the matching associated to the permutation $\sigma$ and by
\begin{equation}
E[\mu(\sigma)] := \sum_{i\in [N]}  |r_i - b_{\sigma(i)}|^p
\end{equation}
its cost, we soon have that
\begin{equation}
E[\nu(\sigma, \pi)] = E[\mu(\sigma)] + E[\mu(\pi)]
\end{equation}
and we recover that
\begin{equation}
\label{eq::2f_bip::inequality}
E[\nu^*] \geq 2\, E[\mu^*],
\end{equation}
i.e.~the cost of the optimal 2-factor is necessarily greater or equal to twice the optimal 1-factor. Together with inequality~\eqref{eq::ineq_2f_tsp}, which is valid for any graph, we obtain that
\begin{equation}
\label{eq::2f_Inequalities} 
E[h^*] \geq E[\nu^*] \geq 2\, E[\mu^*] \, .
\end{equation}

Previously in this Chapter we have seen that in the limit of infinitely large $N$, in one dimension and with $p>1$, the average cost of the optimal Hamiltonian cycle is equal to twice the average cost of the optimal matching (1-factor). We conclude that the average cost of the 2-factor must be the same. Moreover, since inequality \eqref{eq::2f_Inequalities} holds also in 2 dimensions, also in that case the cost of the 2-factor problem has the same limit, for large $N$ of that obtained for assignment and bipartite TSP (see Fig.~\ref{Fig::tsp_2d_AverageCost}).\\
In the following we will denote with $\overline{E_{N,N}^{(p)}[\nu^*]}$ the average optimal cost of the 2-factor problem on the complete bipartite graph. 
Its scaling for large $N$ will be the same of the TSP and the matching problem, that is the limit
\begin{equation}
\lim\limits_{N\to \infty} \frac{\overline{E_{N,N}^{(p)}[\nu^*]}}{N^{1-p/2}} = E^{(p)}_B \,,
\end{equation}
is finite. 

On the complete graph $\mathcal{K}_N$ inequality (\ref{eq::2f_bip::inequality}) does not hold, since a general 2-factor configuration cannot always be written as a sum of two disjoint matchings, due to the presence of odd-length loops. Every 2-factor configuration on the complete graph can be determined by only one permutation $\pi$, satisfying $\pi(i) \ne i$ and $\pi(\pi(i)) \ne i$ for every $i\in [N]$. The cost can be written as
\begin{equation}
E[\nu(\pi)] = \sum_{i\in [N]}  |x_i - x_{\pi(i)}|^p \,.
\end{equation}
The two constraints on $\pi$ assure that the permutation does not contain fixed points and cycles of length 2. 
In the following we will denote with $\overline{E_{N}^{(p)}[\nu^*]}$ the average optimal cost of the 2-factor problem on the complete graph. Even though inequality~(\ref{eq::2f_bip::inequality}) does not hold, we expect that for large $N$, the average optimal cost scales in the same way as the TSP and the matching problem, i.e. as
\begin{equation}
\lim\limits_{N \to \infty} \frac{\overline{E_{N}^{(p)}[\nu^*]}}{N^{1-p}} = E^{(p)}_M \,.
\end{equation}
Later we will give numerical evidence for this scaling.

\subsection{2-factor in one dimension on complete bipartite graphs}
Here we will consider the case $p>1$, that is the weight associated to an edge is a convex and increasing function of the Euclidean distance between its two vertices. This section is taken from \cite{Caracciolo2018_3}.
Let us now look for the optimal solutions for the 2-factor.

The possible solutions for $N=6$ and are represented schematically in Fig.~\ref{fig::2f_N=6}. For $N=7$ there are three solutions and so on.

The first observation that we can do is that in any optimal 2-factor $\nu^*$ all the loops must be in the shoelace configuration, that is the one that we found for the TSP.\\
Indeed in each loop there is the same number of red and blue points and the result we proved for the one dimensional bipartite TSP shows indeed that the shoelace loop is always optimal (when the number of loops used has to be one).

Moreover, in any optimal 2-factor $\nu^*$ there are no loops with more than 3 red points.
Indeed, as soon as the number of red points (and therefore blue points) in a loop is larger than 3, a more convenient 2-factor is obtained by considering a 2-factor with two loops. In fact, as can be seen in Fig.~\ref{Fig::2f_bip}, the cost gain is exactly equal to the difference between an ordered and an unordered matching which we know is always negative for $p>1$.

From these two considerations, it follows that in any optimal bipartite 2-factor $\nu^*$ there are only shoelaces loops with 2 or 3 red points. \\
The reason why there is not a solution which is always the optimal independently on the point positions is that two different 2-factors in this class are not comparable, that is all of them can be optimal in particular instances. For example, the possible solutions for $N=6$ and are represented schematically in Fig.~\ref{fig::2f_N=6}. For $N=7$ there are three solutions and so on.

But how many of these possible solutions there are? 
At given number $N$ of both red and blue points there are at most $\Pad(N-2)$  optimal 2-factor $\nu^*$. $\Pad(N)$ is the $N$-th {\em Padovan} number, see Appendix~\ref{app::2factor}, where it is also shown that for large $N$
\begin{equation}
\Pad(N) \sim \plas^N
\end{equation}
with $\plas$ the {\em plastic} number (see Appendix~\ref{app::plastic_constant} for a discussion on this constant). 

Actually, for values of $N$ which we could explore numerically, we saw that all $\Pad(N-2)$ possible solutions appear as optimal solutions in the ensemble of instances.

\begin{figure}[ht]
	\begin{subfigure}[t]{0.48\linewidth}
		\centering
		\includegraphics[width=1\columnwidth]{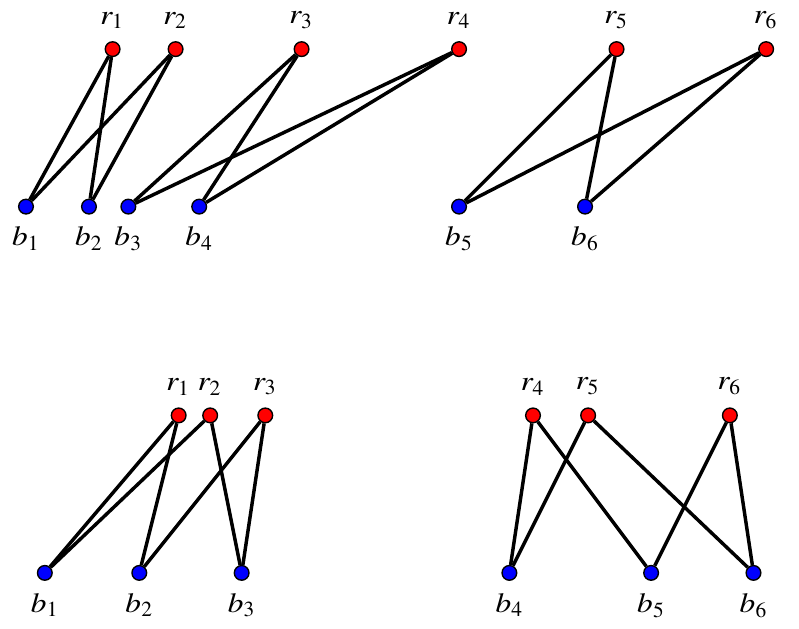}
		\caption{\footnotesize Two instances whose optimal solutions are the two possible $\nu^*$ for $N=6$ on the complete bipartite graph $\mathcal{K}_{N,N}$. For each instance the blue and red points are chosen in the unit interval and sorted in increasing order, then plotted on parallel lines to improve visualization.} \label{fig::2f_N=6}
	\end{subfigure} \hfill
	\begin{subfigure}[t]{0.48\linewidth}
		\centering
		\includegraphics[width=1\columnwidth]{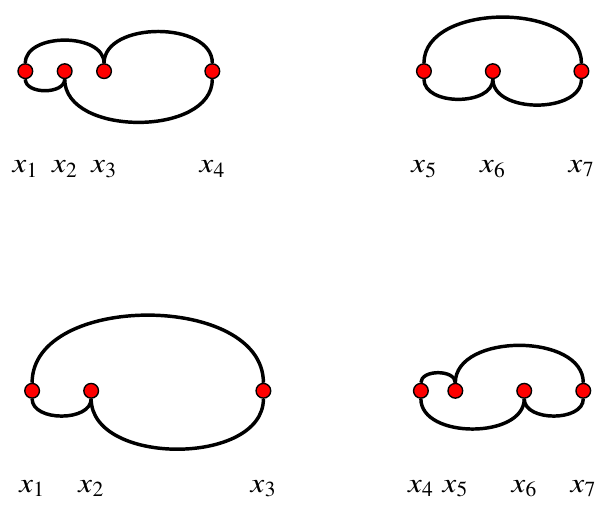}
		\caption{\footnotesize Two instances whose optimal solutions are the two possible  $\nu^*$ for $N=7$ on the complete graph $\mathcal{K}_N$. For each instance the points are chosen in the unit interval and sorted in increasing order. }  \label{fig::2f_N=7_mono}
	\end{subfigure}
	\caption{Optimal solutions for small $N$ cases.}
\end{figure}

\begin{figure*}[ht]
	\begin{subfigure}[t]{0.49\linewidth}
		\centering
		\includegraphics[width=0.9\columnwidth]{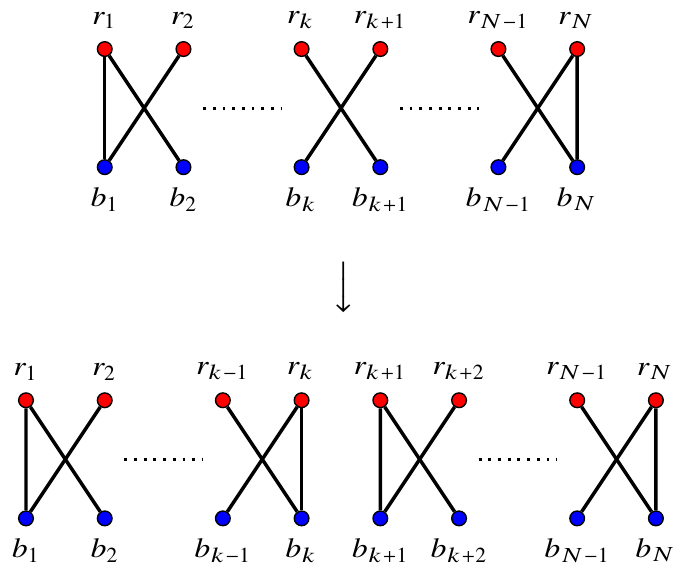}
		\caption{\footnotesize$\mathcal{K}_{N,N}$ case}
		\label{Fig::2f_bip}
	\end{subfigure} \hfill
	\begin{subfigure}[t]{0.49\linewidth}
		\centering
		\includegraphics[width=1\columnwidth]{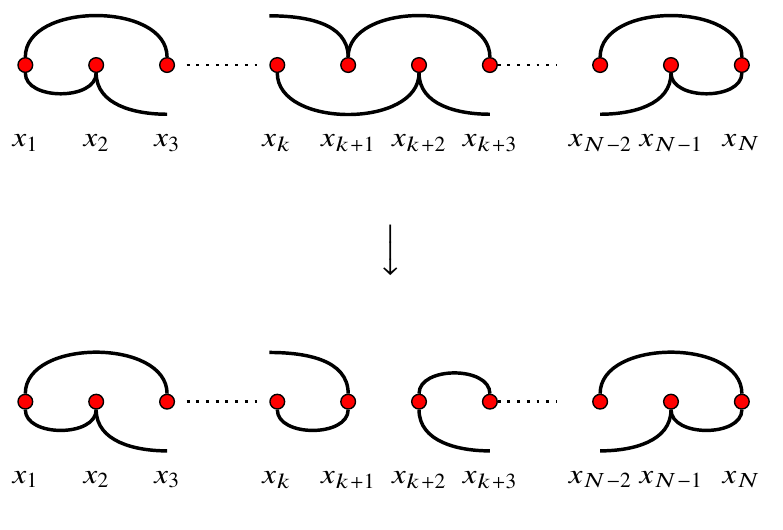}
		\caption{\footnotesize $\mathcal{K}_N$ case}
		\label{Fig::2f_mono}
	\end{subfigure}
	\caption{Result of one cut of the shoelace in two smaller ones for both the complete bipartite and complete graph cases. The cost gained is exactly the difference between an unordered matching and an ordered one.}
\end{figure*}

\subsubsection{Cost for finite $N$}
We have already seen that Eq.~\eqref{eq::2f_Inequalities} guarantees that in the large $N$ limit the average solution cost of the 2-factor problem is the same of the bipartite TSP (with the same $N$).

\begin{figure}[ht]
	\centering
	\includegraphics[scale=1, keepaspectratio]{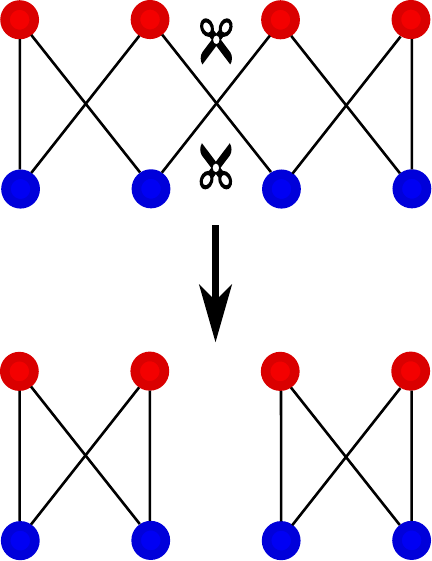} 
	\caption{Graphical representation of the cutting operation which brings from the optimal TSP cycle (top) to a possible optimal solution of the 2-factor problem (bottom). Here we have represented the $N=4$ case, where the cutting operation is unique. Notice that blue and red points are chosen on a interval, but here they are represented equispaced on two parallel lines to improve visualization.} \label{fig::2f_fig_cut2}
\end{figure}

We have proved that, for every value of $N$, the optimal 2-factor solution is always composed by an union of shoelaces loops with only two or three points of each color. As a consequence of this fact, differently from the assignment and the TSP cases, different instances of the disorder can have different spanning subgraphs that minimize the cost function. In particular these spanning subgraphs can always be obtained by ``cutting'' the optimal TSP cycle (see Fig.~\ref{fig::2f_fig_cut2}) in a way which depends on the specific instance. 
This ``instance dependence'' makes the computation of the average optimal cost particularly difficult. However, Eq.~\eqref{eq::2f_Inequalities} guarantees that the average optimal cost of the 2-factor problem is bounded from above by the TSP average optimal cost and from below by twice the assignment average optimal cost. Since in the large $N$ limit these two quantities coincide, one obtains immediately the large $N$ limit of the average optimal cost of the 2-factor problem.
Unfortunately, this approach is not useful for a finite-size system. But we can use Selberg integrals to obtain an upper bound: indeed we can compute the average cost obtained by ``cutting'' the TSP optimal cycle in specific ways.
When we cut at the $k$-position the optimal TSP into two different cycles we gain an average cost
\begin{equation}
E_k^{(p)} =   \overline{\left| b_{k+1}-r_k \right|^p} + \overline{\left| r_{k+1}-b_k \right|^p} - \overline{\left| b_{k+1}-r_{k+1} \right|^p} - \overline{\left| b_{k}-r_k \right|^p} \,.
\end{equation}
By using Eq.~\eqref{eq::ordstat_pk} and the generalized Selberg integral given in Eq.~\eqref{eq::selberg_gen}, we obtain
\begin{equation}
\begin{split}
\lefteqn{ \overline{\left| b_{k}-r_k \right|^p} - \overline{\left| b_{k+1}-r_k \right|^p} = }\\
& = \frac{\Gamma^2(N+1) \, \Gamma(p+1) \, \Gamma \left( k+\frac{p}{2} \right) \, \Gamma \left( N-k+ \frac{p}{2} +1 \right)}{\Gamma(k) \, \Gamma(N-k+1) \, \Gamma(N+p+1) \, \Gamma\left( N + \frac{p}{2} +1 \right) \, \Gamma \left( \frac{p}{2} +1 \right)}  \left[ 1- \frac{k+ \frac{p}{2}}{k} \right] \\
& = - \frac{p}{2}\frac{\Gamma^2(N+1) \, \Gamma(p+1) \, \Gamma \left( k+\frac{p}{2} \right) \, \Gamma \left( N-k+ \frac{p}{2} +1 \right)}{\Gamma(k+1) \, \Gamma(N-k+1) \, \Gamma(N+p+1) \, \Gamma\left( N + \frac{p}{2} +1 \right) \, \Gamma \left( \frac{p}{2} +1 \right)} \,,
\end{split}
\end{equation}
and similarly
\begin{equation}
\begin{split}
\lefteqn{\overline{\left| b_{k+1}-r_{k+1} \right|^p} - \overline{\left| r_{k+1}-b_k \right|^p} = }\\ 
& = \frac{\Gamma^2(N+1) \, \Gamma(p+1) \, \Gamma \left( k+\frac{p}{2}+1 \right) \, \Gamma \left( N-k+ \frac{p}{2} \right)}{\Gamma(k+1) \, \Gamma(N-k) \, \Gamma(N+p+1) \, \Gamma\left( N + \frac{p}{2} +1 \right) \, \Gamma \left( \frac{p}{2} +1 \right)} \left[ 1- \frac{N-k+ \frac{p}{2}}{N-k} \right] \\
& = - \frac{p}{2}\frac{\Gamma^2(N+1) \, \Gamma(p+1) \, \Gamma \left( k+\frac{p}{2}+1 \right) \, \Gamma \left( N-k+ \frac{p}{2} \right)}{\Gamma(k+1) \, \Gamma(N-k+1) \, \Gamma(N+p+1) \, \Gamma\left( N + \frac{p}{2} +1 \right) \, \Gamma \left( \frac{p}{2} +1 \right)} \,.
\end{split}
\end{equation}
Their sum is
\begin{equation}\label{costcut}
E_k^{(p)}  = \frac{p}{2} \frac{\Gamma^2(N+1) \, \Gamma(p+1) \, \Gamma \left( k+\frac{p}{2} \right) \, \Gamma \left( N-k+ \frac{p}{2} \right)}{\Gamma(k+1) \, \Gamma(N-k+1) \, \Gamma(N+p) \, \Gamma\left( N + \frac{p}{2} +1 \right) \, \Gamma \left( \frac{p}{2} +1 \right)} \,,
\end{equation}
For $p=2$ this quantity is 
\begin{equation}
E_k^{(2)} =  \frac{2}{(N+1)^2}.
\end{equation}
Since this quantity does not depends on $k$, for $p=2$ the best upper bound for the average optimal cost is given by summing the maximum number of cuts that can be done on the optimal TSP cycle. Therefore for $N$ even the 2-factor with lowest average energy is $\nu_{(2,2,\dots,2)}$ and then
\begin{equation}
\overline{E_{N,N}^{(2)}[\nu_{(2,2,\dots,2)}]} = \frac{2}{3} \frac{N^2 + 4 N -3}{(N+1)^2} - \frac{N-2}{(N+1)^2} = \frac{1}{3} \frac{N (2 N + 5)}{(N+1)^2}\,, \label{eq::2f_bound_even}
\end{equation}
is an upper bound for the optimal average cost since, even though this configuration has the minimum average cost, for every fixed instance of disorder there can be another one which is optimal. For $N$ odd, one of the 2-factors with lowest average energy is $\nu_{(2,2,\dots,2,3)}$ and
\begin{equation}
\overline{E_{N,N}^{(2)}[\nu_{(2,2,\dots,2,3)}]} = \frac{2}{3} \frac{N^2 + 4 N -3}{(N+1)^2} - \frac{N-3}{(N+1)^2} = \frac{1}{3} \frac{2 N^2 + 5 N + 3}{(N+1)^2}. \label{eq::2f_bound_odd}
\end{equation}
Therefore that essentially the upper bound for the optimal average cost for even and odd large $N$ is the same. For $p=2$, these bounds are compared with the results of numerical simulations in Fig.~\ref{fig::2f_bip_plot}.

For $p\neq2$, $E_k$ depends on $k$. In particular, for $1 < p < 2$ the cut near to 0 and 1 are (on average) more convenient than those near the center. For $p>2$ the reverse is true (see Fig.~\ref{fig::2f_fig_cut}).
 For $p\neq2$, however, this sum does not give a simple formula.

\begin{figure}[t]
	\centering
	\includegraphics[width=0.9\columnwidth]{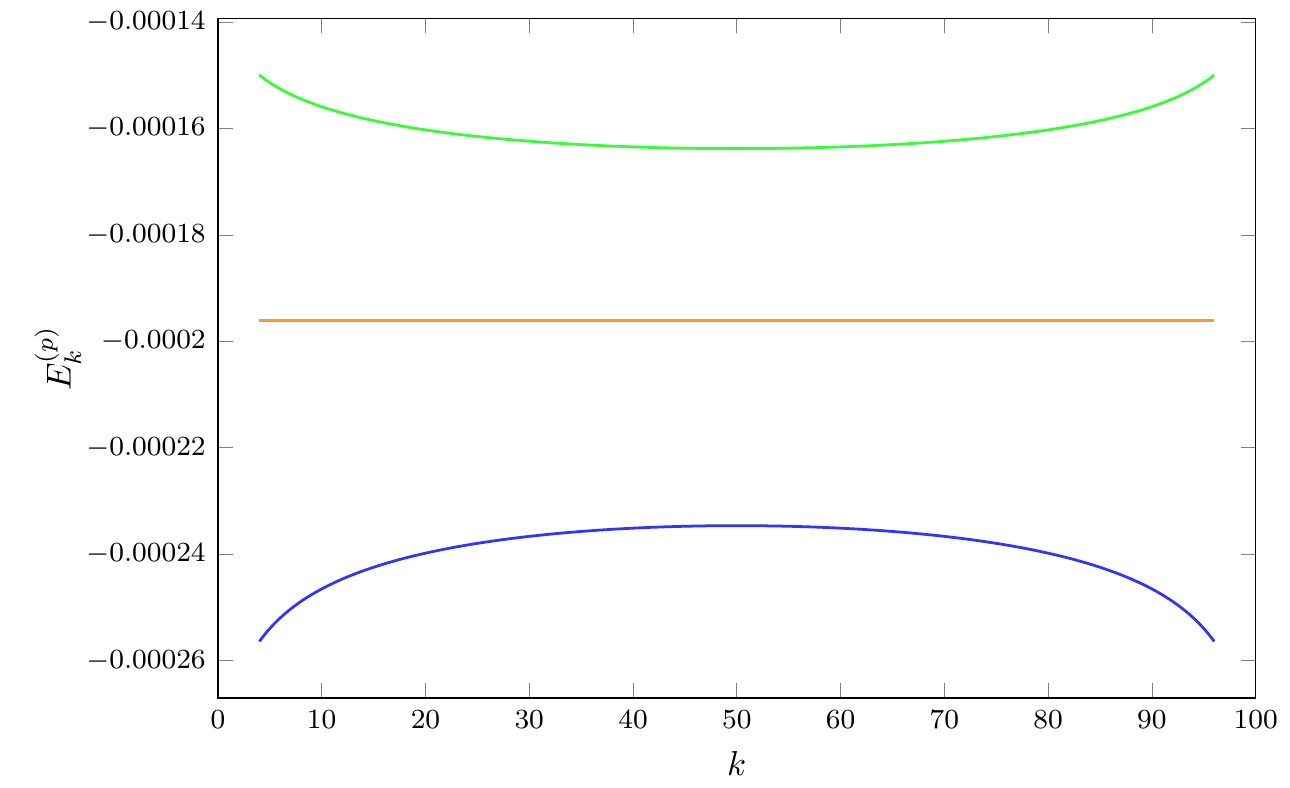} 
	\caption{Plot of $E_k^{(p)}$ given in Eq.~\eqref{costcut} for various values of $p$: the green line is calculated with $p=2.1$, the orange with $p=2$ and the blue one with $p=1.9$; in all cases we take $N=100$.} \label{fig::2f_fig_cut}
\end{figure}

\begin{figure*}
	\vspace{1cm}
	\begin{subfigure}[t]{0.48\linewidth}
		\centering\includegraphics[scale=0.5]{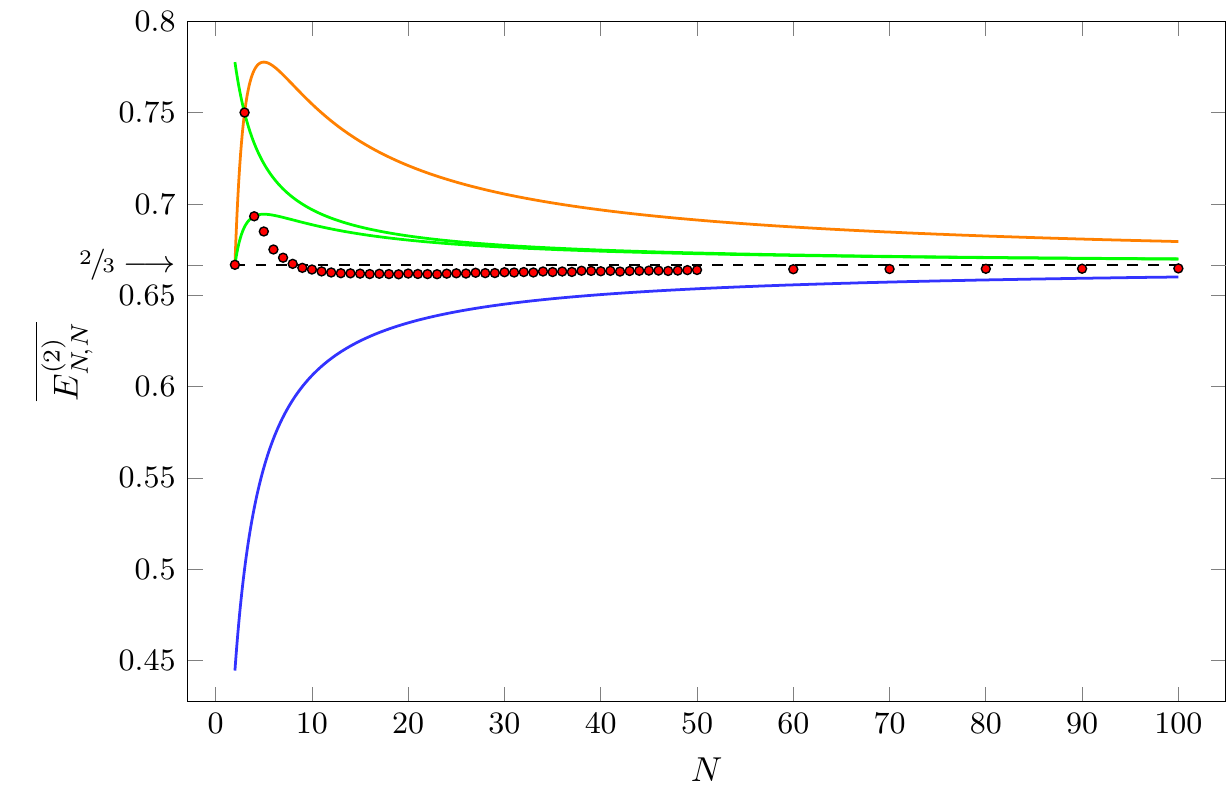} 
		\caption{\footnotesize $\mathcal{K}_{N,N}$ case with $p=2$. The orange line is the cost of the TSP given in Eq.~\eqref{eq::tsp_bip_costtsp} for $p=2$; the green lines are, from above, the cost of the optimal fixed 2-factor $\nu_{(2,2,\dots,2,3)}$ given in Eq.~\eqref{eq::2f_bound_odd} and $\nu_{(2,2,\dots,2)}$ given in Eq.~\eqref{eq::2f_bound_even}. The dashed black line is the asymptotic value $\frac{2}{3}$ and the blue continuous one is twice the cost of the optimal 1-matching $\frac{2}{3}\frac{N}{N+1}$. Red points are the results of a 2-factor numerical simulation, in which we have averaged over $10^7$ instances.} \label{fig::2f_bip_plot}
	\end{subfigure} \hfill
	\begin{subfigure}[t]{0.48\linewidth}
		\centering\includegraphics[scale=0.5]{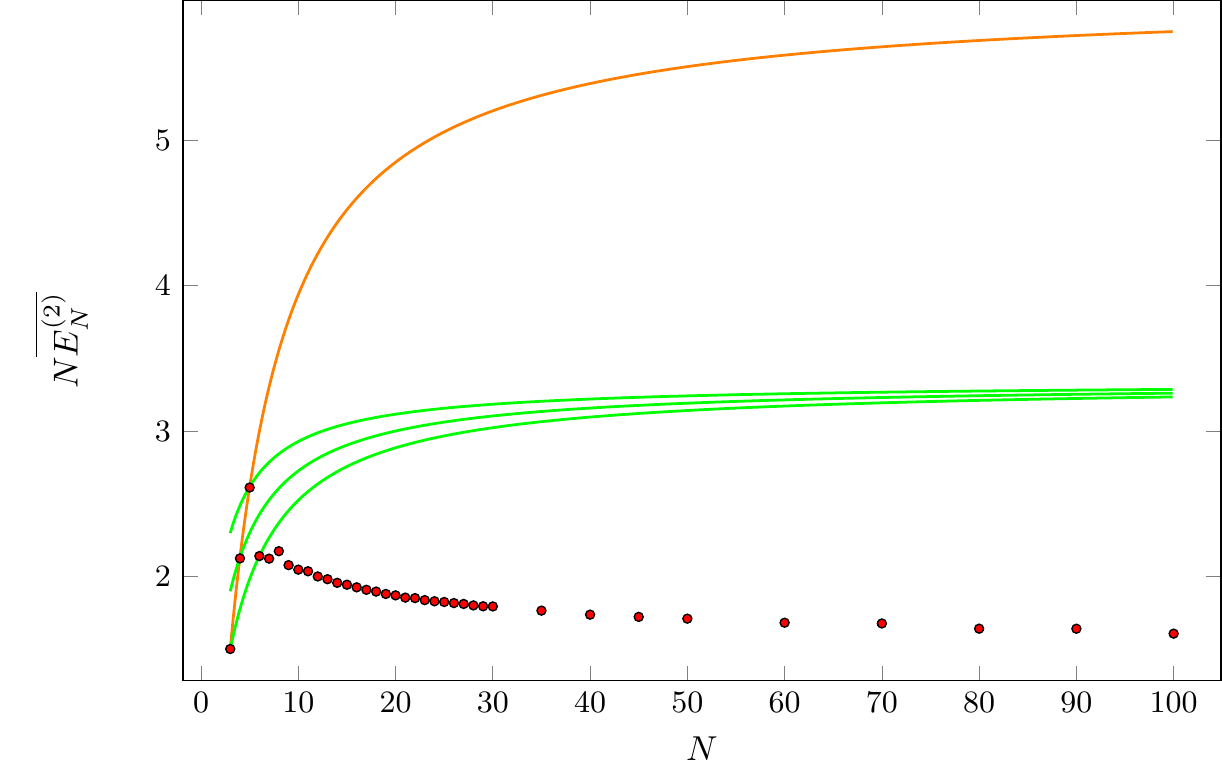} 
		\caption{\footnotesize $\mathcal{K}_N$ case with $p=2$.  Here the average cost is rescaled with $N$. The orange line is the cost of the TSP given in Eq.~\eqref{eq::tsp_mono_AOC} for $p=2$. The green lines are from above the cost of the fixed 2-factor $\nu_{(3,3,\dots,3,5)}$ given in Eq.~\eqref{multipleOfThree_plus2}, $\nu_{(3,3,\dots,4)}$ given in Eq.~\eqref{multipleOfThree_plus1} and $\nu_{(3,3,\dots,3)}$ given in Eq.~\eqref{multipleOfThree}.  
		Red points are the results of a numerical simulation for the 2-factor, in which we have averaged over $10^5$ instances for $N \le 30$, $10^4$ for $30 < N \le 50$ and $10^3$ for $N>50$.} \label{fig::2f_mono_plot}
	\end{subfigure}
	\caption{Average optimal costs for various $N$ and for $p=2$.}
\end{figure*}

\subsection{2-factor in one dimension on complete graphs}
Finally, we analyze here the 2-factor problem in one dimension on complete graphs, in the $p>1$ case.
The possible solutions for the 2-factor on complete graph can be constructed by cutting in a similar way the corresponding TSP solution into smaller loops as can be seen pictorially in Fig.~\ref{Fig::2f_mono}. Note that one cannot have a loop with two points. Analogously to the bipartite case we have analyzed before, each loop that form the 2-factor configuration must be a shoelace. However the length of allowed loops will be different, since one cannot cut, on a complete graph, a TSP of 4 and 5 points in two smaller sub-tours. 
Therefore, on the complete graph, in the optimal 2-factor $\nu^*$ there are only loops with 3, 4 or 5 points.

In Fig.~\ref{Fig::2f_mono} we represent the two solutions when $N=7$.
In Appendix~\ref{app::2factor} we prove that, similarly to the bipartite case, the number of  2-factor solutions is at most $g_N$ on the complete graph, which for large $N$ grows according to
\begin{equation}
g_N \sim \plas^N \,.
\end{equation}
Also in this case we verified numerically, for accessible $N$, that the set of possible solutions that we have identified is actually realized by some instance of the problem.

Using these informations on the shape of the solution, we turn to the evaluation of bounds on its cost. Let us first evaluate the cost gain when we cut a TSP solution cycle in two ``shoelaces'' (we keep using here the word shoelace to indicate the cycle which is the solution to the TSP on the complete graph) sub-cycles. For $p>1$ the cost gain doing one cut can be written as
\begin{equation}
\begin{split}
\overline{\left(x_{k+1}-x_k\right)^p} + \overline{\left(x_{k+3}-x_{k+2}\right)^p} - \overline{\left(x_{k+3}-x_{k+1}\right)^p} - \overline{ \left(x_{k+2}-x_{k}\right)^p} & \\
= - \frac{2 \,p \, \Gamma(N+1) \, \Gamma(p+1)}{\Gamma(N+p+1)} &  \,.
\end{split}
\end{equation}
For example for $N=6$ (in which the solution is unique since 6 can be written as a sum of 3, 4 and 5 in an unique way as 3+3) and $p=2$ we have
\begin{equation}
\overline{E_6^{(2)}} = \frac{1}{2} - \frac{1}{7} = \frac{5}{14} \,.
\end{equation}
If $N$ is multiple of 3, the lowest 2-factor is, on average, the one with the largest number of cuts i.e. $\nu_{(3,3,\dots,3)}$. The number of cuts is $(N-3)/3$ so that the average cost of this configuration is
\begin{equation}
\begin{aligned}
\overline{E_N^{(p)}[\nu_{(3,3,\dots,3)}]} & 
= N \left( \frac{p}{3} + 1\right) \frac{\Gamma(N+1)\, \Gamma(p+1)}{\Gamma(N+p+1)} \,.
\label{multipleOfThree}
\end{aligned}
\end{equation}
Instead if $N$ can be written as a multiple of 3 plus 1, the minimum average energy configuration is $\nu_{(3,3,\dots ,3,4)}$, which has $(N-4)/3$ cuts and
\begin{equation}
\begin{aligned}
\overline{E_N^{(p)}[\nu_{(3,3,\dots,4)}]} & 
= \left[ N \left( \frac{p}{3} + 1\right) + \frac{2}{3}p \right]\frac{ \Gamma(N+1) \, \Gamma(p+1)}{\Gamma(N+p+1)} \,.
\label{multipleOfThree_plus1}
\end{aligned}
\end{equation}
The last possibility is when $N$ is a multiple of 3 plus 2, so the minimum average energy configuration is $\nu_{(3,3,\dots ,3,5)}$, with $(N-4)/3$ cuts and
\begin{equation}
\begin{aligned}
\overline{E_N^{(p)}[\nu_{(3,3,\dots,5)}]} & 
= \left[ N \left( \frac{p}{3} + 1\right) + \frac{4}{3}p \right]\frac{ \Gamma(N+1) \, \Gamma(p+1)}{\Gamma(N+p+1)} \,.
\label{multipleOfThree_plus2}
\end{aligned}
\end{equation}
In the limit of large $N$ all those three upper bounds behave in the same way. For example
\begin{equation}
\lim\limits_{N \to \infty} \overline{E_N^{(p)}[\nu_{(3,3,\dots,3)}]} = N^{1-p} \left( 1 + \frac{p}{3} \right) \Gamma(p+1) \,.
\end{equation} 
Note that the scaling of those upper bounds for large $N$ is the same of those of matching and TSP.
For $p=2$, these bounds are compared with the results of numerical simulations in Fig.~\ref{fig::2f_mono_plot}.

\chapter{Quantum point of view}\label{chap::third}
\chaptermark{Quantum point of view}
In this Chapter we will deal with another field which lies between physics and computer science: quantum computing. Quantum computers have been considered for the first time by Feynmann to simulate quantum systems (or, better, physical systems in which quantum effects are relevant). 
We, on the other hand, will focus on the possibility of using quantum computers to solve hard combinatorial optimization problems. After the important works by Shor and Grover, many concepts about quantum algorithms to solve COPs have been understood, and we will discuss some of them.
We will then specialize in the so called \emph{quantum adiabatic algorithm}, in the form of the \emph{simulated annealing}, which is usable today in the largest chip that performs computations using quantum effects, i.e.~the D-Wave machine.
Finally, we will briefly comment on a recent and promising approach to approximate (and sometimes also solve) COPs in gate models, the famous \emph{quantum approximate optimization algorithm}. 
\section{Quantum computation for combinatorial optimization problems}\label{sec::quantcomp}
\sectionmark{Quantum computation for COPs}
The study of quantum computation is flourishing in these recent years for two main reasons: the discovery of powerful quantum algorithms (Shor~\cite{Shor1999} and Grover~\cite{Grover1997}) in the late 90s, and the advent of real computers able to exploit quantum effects during the computation.\\
As a consequence, there are many good books (\cite{Nielsen2000, Rieffel2011, Kaye2007, Mermin2007}) and reviews (for example, \cite{Aharonov1999}) where a complete introduction to the subject can be found. Here we will focus on quantum algorithms for COPs, disregarding completely other fundamental topics as, for example, quantum error correction and fault-tolerant quantum computation.

\subsection{Gate model of quantum computing}
The basic building block of classical computation is the bit, which can be in state 0 or 1. The quantum version of that is the \emph{qubit}, which is a two level system. Therefore its general state is
\begin{equation}
	\ket{q} = a \ket{0} + b \ket{1},
\end{equation}
where $a$ and $b$ are complex numbers we require $\abs{a}^2+\abs{b}^2 = 1$, so that the state is normalized. We represent
\begin{equation}\label{eq::qc_repr}
	\ket{0} = \begin{pmatrix}
	1 \\ 0
	\end{pmatrix} \quad \text{and} \quad 
	\ket{1} = \begin{pmatrix}
	0 \\ 1
	\end{pmatrix}
\end{equation}
and we will refer to this as the \emph{computational basis}.
When we have $N$ qubits, the computational basis is the set of states
\begin{equation}\label{eq::quant_comp_state}
	\ket{q_1} \otimes \ket{q_2} \otimes \cdots \otimes \ket{q_N},
\end{equation}
for each choice of $q_i \in \{0, 1\}$. Therefore the Hilbert space describing the state of a $N$ qubit system is $2^N$ dimensional.

Let us now come back for a moment to the classical world: if we have a system of $N$ bits, we have $2^N$ possible states of our system. Let us see the state of our system (computer) as a basis vector of the $2^N$-dimensional space $\mathbb{C}^{2^N}$,
\begin{equation}\label{eq::classical_comp_state}
	\ket{s} = \begin{pmatrix}
	b_1 \\ \vdots \\ b_{2^N}
	\end{pmatrix},
\end{equation}
where only one bit $b_i$ is 1, and all the other are 0 (we use the braket formalism also for this representation of classical states). 
For example, we have for a two-bit system
\begin{equation}
	\ket{00} = \begin{pmatrix}
	0 \\ 0 \\ 0 \\ 1
	\end{pmatrix}, \quad
	\ket{01} = \begin{pmatrix}
	0 \\ 0 \\ 1 \\ 0
	\end{pmatrix}, \quad
	\ket{10} = \begin{pmatrix}
	0 \\ 1 \\ 0 \\ 0
	\end{pmatrix}, \quad
	\ket{11} = \begin{pmatrix}
	1 \\ 0 \\ 0 \\ 0
	\end{pmatrix}.
\end{equation}
Therefore, it seems that quantum computers could be more powerful of classical computers simply because we can store much more information in $N$ qubits than in $N$ bits, since in the former case the system can be in any of the linear combinations (with unit $\ell_2$ norm) of the $2^N$ basis vectors, while in the latter it lives \emph{inside the basis}. \\
However, this is not the end of the story: a \emph{deterministic} program for a classical computer, in this formalism, can be seen as a matrix which is applied to $\ket{s}$ and modify the state of the system. For example, if we have a two-bit system and we want to have assign 1 to the second bit, we apply the matrix
\begin{equation}
	M = \begin{pmatrix}
	1 & 1 & 0 & 0 \\ 0 & 0 & 0 & 0 \\ 0 & 0 & 1 & 1 \\ 0 & 0 & 0 & 0 
	\end{pmatrix},
\end{equation}
so that
\begin{equation}
	M \ket{00} = \ket{01}, \quad M \ket{01} = \ket{01}, \quad M \ket{10} = \ket{11}, \quad M \ket{11} = \ket{11}.
\end{equation}
In general, a computation will be a matrix with elements $M_{ij} \in \{0,1\}$ such that $\sum_i M_{ij} = 1$ for each $j$, since this condition correspond to the fact that we want our matrix to map one basis state in another basis state.

Nonetheless, we can do something closer to quantum computing. For example, we could have in our code instructions like ``with probability 1/2, assign 1 to the second bit''. This kind of instructions, which are not deterministic, are captured by using \emph{stochastic} matrices, that is with elements $\sum_j M_{i,j}=1$ but now with the only restriction that $M_{ij}\geq0$. 
For our case:
\begin{equation}
	M_s = \begin{pmatrix}
	1 & 1/2 & 0 & 0 \\ 0 & 1/2 & 0 & 0 \\ 0 & 0 & 1 & 1/2 \\ 0 & 0 & 0 & 1/2 
	\end{pmatrix},
\end{equation}
and now we have, for example,
\begin{equation}
	M_s \ket{00} = \begin{pmatrix}
	0 \\ 0 \\ 1/2 \\ 1/2
	\end{pmatrix}
	= \frac{1}{2} \ket{00} + \frac{1}{2} \ket{01}.
\end{equation}
This result has to be interpreted as follows: ``if we use the computer program $M_s$ with input state $\ket{00}$, with 1/2 or probability the output state will be $\ket{00}$ and with 1/2 it will be $\ket{01}$''.\\
And this is very close to the meaning of a quantum state for a qubit: if the state is
\begin{equation}
	\ket{q} = \frac{1}{\sqrt{2}} \ket{0} + \frac{1}{\sqrt{2}} \ket{1},
\end{equation}
and we measure the qubit in the computational basis, we have 1/2 of probability of obtaining 0 and 1/2 of obtaining 1.

Therefore, if we allow for ``stochastic'' instructions in our code, we can really have ``superpositions'' of basis states of the form given in Eq.~\eqref{eq::classical_comp_state}, provided that their coefficients are positive and sum to 1:
\begin{equation}
	\ket{s} = \sum_{i=1}^{2^N} a_i \ket{s_i},
\end{equation}
where the $\ket{s_i}$ are the basis states given in Eq.~\eqref{eq::classical_comp_state},  $\sum_i a_i = 1$ and $a_i\geq 0$.

Let us now turn to the standard gate model of quantum computation. Similarly to the stochastic classical computation case, we have a state of $N$ qubits
\begin{equation}
	\ket{q} = \sum_{i=1}^{2^N} a_i \ket{q_i},
\end{equation}
where the states are as in the classical case, but now $a_i$ are complex numbers such that $\sum_i \abs{a_i}^2 = 1$.
Given a state, the computation is done by multiplying the state for a \emph{unitary} matrix $U$ and then measuring the state in the computational basis.\\
Notice that, physically, this means that the initial state $\ket{q_0}$ of the system is evolved with the Hamiltonian $H$ such that
\begin{equation}
	\text{Texp}\left(- \frac{i}{\hbar} \int_{t_0}^{t_1} dt \, H(t) \right) \ket{q_0} = U \ket{q_0}.
\end{equation}



\subsection{Quantum versus Classical: the story of interference and entanglement}
As we have seen, there are two main differences between stochastic classical and quantum computation: in the first case the ``amplitudes'' of each basis state are positive quantities which sum to 1 (so they are probabilities). In the second case, the amplitudes are complex numbers and their modulus squared sum to 1. In fact, it turns out that the power of quantum computing is not due to the fact that amplitudes are complex numbers, but rather to the (less stringent) fact that they can assume \emph{negative} values \cite{Bernstein1997}. The reason is that with negative amplitudes we can create interference phenomena to decrease the probability of unwanted output states and increase that of the solution to our problem. \\
Let us deepen this intuition with a practical example: consider a system of 2 qubits. We need to define two (actually very important) gates: the \emph{Hadamard} gate $H$, that is defined by
\begin{equation}
	H \ket{0} = \frac{\ket{0} + \ket{1}}{\sqrt{2}} = \ket{+}, \quad 	H \ket{1} = \frac{\ket{0} - \ket{1}}{\sqrt{2}} = \ket{-},
\end{equation}
and therefore in the representation used in Eq.~\eqref{eq::qc_repr}, we have:
\begin{equation}
	H = \frac{1}{\sqrt{2}}\begin{pmatrix}
	1 & 1 \\
	1 & -1
	\end{pmatrix}.
\end{equation}
The other gate we need is a two-qubit one, the \emph{CNOT} gate defined by
\begin{equation}
	C_{not} = \ket{0} \bra{0} \otimes \mathbb{I} + \ket{1} \bra{1} \otimes X , 
\end{equation}
where $\mathbb{I}$ is the identity $2\times2$ matrix and $X$ is the Pauli matrix
\begin{equation}
	X = \begin{pmatrix}
	0 & 1 \\
	1 & 0
	\end{pmatrix}.
\end{equation}
The Hadamard gate $H$ is such that a qubit in the state $\ket{0}$ or $\ket{1}$ has equal probability to be measured in $\ket{0}$ or $\ket{1}$ after $H$ is applied. In this sense, the application of $H$ has a similarity with the classic operation of randomly flipping a bit (to some extent!). \\
Also the $C_{not}$ gate has a simple actions on the computational basis states: if the first qubit is $\ket{0}$, it does nothing; if the first qubit is $\ket{1}$, the second qubit is flipped.\\
Now consider the state $\ket{01}$ and apply firstly $H$ to both qubits
\begin{equation}
	\ket{01} \rightarrow H \otimes H \ket{01} = \ket{+ -} = \frac{\ket{00} - \ket{01} + \ket{10} - \ket{11}}{\sqrt{2}} 
\end{equation}
and then the $C_{not}$ gate
\begin{equation}
	\ket{- +} \rightarrow  C_{not} \ket{- +} = \frac{\ket{00} - \ket{01} + \ket{11} - \ket{10}}{\sqrt{2}} = \ket{--}
\end{equation}
and, finally, again the $H$ gate to the second qubit
\begin{equation}\label{eq::qc_ex_int}
	\ket{- -} \rightarrow \mathbb{I} \otimes H \ket{--} = \ket{- 1} = \frac{\ket{01} - \ket{11}}{\sqrt{2}}.
\end{equation}
Therefore, after the process, we have equal probability to be in the states $\ket{11}$ and $\ket{01}$. 
If we try to replicate classically this short algorithm, we can try to do the following: take the two bits in $\ket{01}$, and randomly flips them. Notice that if we make our measurements here, the results of the classical and quantum systems are indistinguishable. Then, we ask a friend of ours to look the first bit and change the second if it is 1, otherwise do nothing. Again, at this point there has been no interference of probabilities and we could not distinguish the qubits and the bits systems: each possible outcome is equally probable. Finally, randomly flip the second bit again. Clearly, after the first step, each outcome has the same probability classically, in sharp contrast with \eqref{eq::qc_ex_int}: in the quantum system, the probability of the outcomes $\ket{10}$ and $\ket{00}$ is zero!
In Fig.~\ref{fig::qc_interference_paths} there is a graphical representation of the situation.
\begin{figure}
		\centering
	\begin{tikzpicture}[scale=1]
	\node (1) at (0,0) {$\ket{01}$};
	\node (2) at (2,1) {$\ket{10}$};
	\node (3) at (2,-1) {$\ket{01}$};
	\node (4) at (2,3) {$\ket{11}$};
	\node (5) at (2,-3) {$\ket{00}$};
	\node (6) at (5,1) {$\ket{10}$};
	\node (7) at (5,-1) {$\ket{01}$};
	\node (8) at (5,3) {$\ket{11}$};
	\node (9) at (5,-3) {$\ket{00}$};
	\node (10) at (8,1) {$\ket{10}$};
	\node (11) at (8,-1) {$\ket{01}$};
	\node (12) at (8,3) {$\ket{11}$};
	\node (13) at (8,-3) {$\ket{00}$};
	\draw[->,line width=1pt,black]  (1) to (2);
	\draw[->,line width=1pt,red]  (1) to (3);
	\draw[->,line width=1pt,red]  (1) to (4);
	\draw[->,line width=1pt,black]  (1) to (5);
	\draw[->,line width=1pt,black]  (2) to (8);
	\draw[->,line width=1pt,red]  (3) to (7);
	\draw[->,line width=1pt,red]  (4) to (6);
	\draw[->,line width=1pt,black]  (5) to (9);
	\draw[->,line width=1pt,red]  (6) to (10);	
	\draw[->,line width=1pt,black]  (6) to (12);
	\draw[->,line width=1pt,black]  (7) to (11);	
	\draw[->,line width=1pt,black]  (7) to (13);
	\draw[->,line width=1pt,black]  (8) to (10);
	\draw[->,line width=1pt,red]  (8) to (12);
	\draw[->,line width=1pt,black]  (9) to (11);	
	\draw[->,line width=1pt,black]  (9) to (13);
	\end{tikzpicture}
	\caption{Paths of amplitude if we apply $H\otimes H$, then $C_{not}$, then $\mathbb{I} \otimes H$ to the initial state $\ket{01}$ In this case, when more than one line originates from the same state, the probability is equally divide; if more than one line ends on the same state, the probability is summed; if a line is red, the amplitude is negative.}\label{fig::qc_interference_paths}
\end{figure}
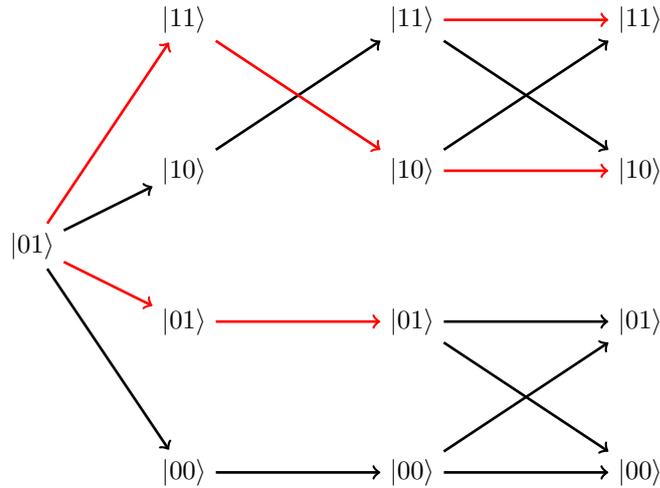

Another important difference between the classical and quantum case, is the \emph{entanglement}. A two-qubit state is said to be entangled if it cannot be written as tensor product of two single-qubit states\footnote{multiple-qubit states can be entangled or not depending on the tensor decomposition under consideration: for example, the state $(\ket{0000}+\ket{1111}+\ket{0101}+\ket{1010})/2$ is entangled in the sense that cannot be written as single-qubit tensor product, but it is un-entangled in the sense that it can be written as tensor product of two two-qubit states.}. 
A famous entangled two-qubit state is
\begin{equation}\label{eq::qc_bell}
	\ket{\psi} = \frac{\ket{00} + \ket{11}}{\sqrt{2}}.
\end{equation}
The question is: are there any differences between this state and a classical state of two strongly correlated qubits? Consider that two qubits are in the state $\ket{\psi}$. These two qubits are bring far away, and then one of them is measured (in the computational basis) and suppose that the outcome is 0: instantaneously we know that, whenever the other will be measured, the outcome will be again 0. This is not necessarily a quantum effect: suppose that your cousin randomly writes a 0 or a 1 on a paper and prepares two identical copies of that. Then she sends one copy to you and one to your brother into an envelope. When you will look at your paper, you will immediately know the content of the other envelope. So what is the point in quantum entanglement? The best explanation requires another ingredient: non-commuting observables.
Actually, we need two sets of two non-commuting single-qubit observables: let us consider those associated to the Pauli matrix $X$ and the Pauli matrix $Z$,
\begin{equation}
	Z = \begin{pmatrix}
	1 & 0 \\
	0 & -1
	\end{pmatrix},
\end{equation}
and those associated with Hadamard gate $H$ and with $H'$ defined as
\begin{equation}
	H' = X H X = \frac{1}{\sqrt{2}}\begin{pmatrix}
	-1 & 1 \\
	1 & 1
	\end{pmatrix}.
\end{equation}
Equipped with the ability to measure these observables, let us take two qubits in the state $\ket{\psi}$ and give one of them to Mario and one of them to Luigi\footnote{the usual names for these two guys are Alice and Bob.}. Now, Mario can measure his qubit, let us say the first one, with the observables associated to $X$ and $Z$. Therefore, if he measures on his qubit the observable associated to $X$, his expected outcome is 
\begin{equation}
	\left< X \otimes \mathbb{I} \right> = \bra{\psi} X \otimes \mathbb{I} \ket{\psi},
\end{equation} 
and similarly for $Z$. An analogous situation holds for Luigi, with $H$ and $H'$ instead of $X$ and $Z$. Now, let us suppose that Mario and Luigi randomly choose which measurement they do. 

There are 4 possible different situations, and the expected values of the measurements are
\begin{equation}
	\left< Z \otimes H \right> = - 	\left< Z \otimes H' \right> = 	\left< X \otimes H \right> = 	\left< X \otimes H' \right> = \frac{1}{\sqrt{2}}.
\end{equation}
Therefore, if we take the quantity
\begin{equation}
	W = Z \otimes H - Z \otimes H' + X \otimes H + X \otimes H',
\end{equation}
we expect that
\begin{equation}\label{eq::bell_eq}
	\left< W \right> = 2 \sqrt{2}.
\end{equation}

But here is the best part: we ask Mario and Luigi to go \emph{very} far away, like a light-year or so. Then we hypothesize that whatever Mario does with his qubit, it will not change in any way Luigi's qubit (this is called \emph{locality} hypothesis). Moreover, we assume that Mario's qubit \emph{does} have a value for the measurement of both $X$ and $Z$, and analogously for Luigi's qubit (this is the \emph{realistic} hypothesis.) Let us elaborate a little about that: classically we could have non-commutating observables, in the sense that some measurements can interfere with others, and so the order in which we perform measurements matters. But if we have two devices which do those measurements, we (in the classical point of view) do not doubt that the system does have at any moment certain values that we could measure. We are doing exactly that hypothesis here: Mario's qubit has a certain value for the measure related to $X$, say $x_M$, and another value for $Z$, say $z_M$. The problem is that we do not know that values, because both $x_M$ and $z_M$ can be -1 or 1 with probability 1/2 because the initial state is $\ket{\psi}$ (a completely analogous situation holds for Luigi's qubit).
But under these hypothesis, we can write
\begin{equation}
	W = z_M h_L - z_M h_L' + x_M h_L + x_M h_L' = z_M (h_L - h_L') + x_M (h_L + h_L').
\end{equation}
Now, remember that all these quantities can be only -1 or 1. Therefore if $h_L=h_L'$ we have $W = 2 x_M$, otherwise $W = 2 z_M$. As a consequence, we obtain a form of the so-called \emph{Bell inequalities}:
\begin{equation}\label{eq::bell_ineq}
	\abs{W} \leq 2,
\end{equation}
which contradicts Eq.~\eqref{eq::bell_eq}. This ``paradox'' has been noticed for the first time by Einstein, Podolsky and Rosen~\cite{Einstein1935}, but today many experiments have confirmed that the inequality in \eqref{eq::bell_ineq} is violated: our reality is not local and realistic. 
Moreover, this explains the difference between entanglement and classical correlation: for entangled qubits Eq.~\eqref{eq::bell_eq} holds, while the reasoning which brought us to the inequality in \eqref{eq::bell_ineq} is correct if we have classical correlated variables.

It has been showed that for an algorithm working with pure states, entanglement among a number of qubit which scale as $\mathcal{O}(N)$ ($N$ being the input size) is necessary for that algorithm not to be efficiently simulated by classical computers~\cite{Jozsa2003}. However, even though entanglement and interference are two important resources which are not available to classical computers, the power of quantum computation has more subtle origins which are not completely understood today \cite[Section 13.9]{Rieffel2011}.

\subsection{An example: Grover algorithm}\label{sec::groverorig}
Grover algorithm is an excellent example of the power of quantum computing at work. We can state the problem as follows. We are given a oracle $f$ such that $f(i) \in \{0,1\}$ for each $i\in \{1,\dots,N\}=[N]$. We do not know anything about the internal structure of the oracle, that is we have no idea of what the oracle is actually computing. The only thing we know is that \emph{one} among the set of possible inputs, $k\in[N]$, is such that $f(k)=1$ and $f(i)=0$ for $i\neq k$. Our aim is to find $k$.\\
Now, classically the only way to proceed is to try all the possible inputs: on average, we will need $N/2$ queries to the oracle (i.e.~applications of $f$), and $N-1$ queries in the worst-case.\\
Let us now start with the quantum algorithm. We have two registers, one can be in any state $\ket{i}$ where $i\in [N]$ (therefore, it can be represented by $\log N$ qubits) and the other is an additional qubit in state $\ket{q}$. Therefore the state of the whole system is $\ket{i} \otimes \ket{q}$. Let us suppose that the oracle works as follows: it is implemented by a unitary $U_f$ such that, for $\ket{q}=\ket{0}$ or $\ket{1}$, $U_f \ket{i} \otimes \ket{q} = \ket{i} \otimes\ket{q \oplus f(i)} $, where $\oplus$ denotes addition modulo 2. Equivalently, we can see that
\begin{equation}
	U_f \ket{i} \otimes \ket{q} = \begin{cases}
	\ket{i} \otimes X \ket{q} & \text{if $i=k$};\\
	\ket{i} \otimes \ket{q} & \text{if $i\neq k$}.
	\end{cases}
\end{equation}
In other words, $U_f \ket{i} \otimes \ket{q} = \ket{i} \otimes X^{f(i)} \ket{q} $. Therefore the oracle flips the qubit in the second register if in the first one there is the value $k$ such that $f(k)=1$, otherwise it left the qubit untouched (this reasoning is valid for the qubit in computational basis states).
Now, we prepare the first register in the superposed state 
\begin{equation}\label{eq::grov_1_unif}
	\ket{\psi} = \frac{1}{\sqrt{N}} \sum_{j=1}^N \ket{j}
\end{equation}
and the qubit in the second register in the state $\ket{-}$. When we apply the oracle to the system, we obtain
\begin{equation}
	U_f \ket{\psi -} = \frac{1}{\sqrt{N}} \sum_{j\neq k}\ket{j -} - \frac{1}{\sqrt{N}}  \ket{k -} = U \otimes \mathbb{I} \ket{\psi -},
\end{equation}
where the operator $U$ is defined by
\begin{equation}\label{eq::grover_1_phase_op}
	U = \mathbb{I} - 2 \ket{k} \bra{k}.
\end{equation}
Since the second qubit is left untouched by the application of $U_f$, we will stop writing him down (the remaining part of the algorithm works on the first register). However, keep in mind that each time we apply $U$, we are querying the oracle once.\\
We also need another operator, the \emph{diffusion} operator $D$, defined as
\begin{equation}\label{eq::grover_1_diff_op}
	D = 2 \ket{\psi} \bra{\psi} - \mathbb{I},
\end{equation}
where $\ket{\psi}$ is given in Eq.~\eqref{eq::grov_1_unif}.
This operator is unitary (it can be written as $-e^{i \pi \ket{\psi} \bra{\psi}}$) and can be efficiently implemented in $\sim \log N$ elementary gates (see, for example, \cite[Section 9.1.3]{Rieffel2011}).\\
A great simplification for the analysis of Grover algorithm comes from the fact that the only operators involved are those in Eq.~\eqref{eq::grover_1_phase_op} and in Eq.~\eqref{eq::grover_1_diff_op}. Since these operators can be written in terms of the projector on $\ket{\psi}$ and on $\ket{k}$ (and identities), we can restrict our analysis to the two-dimensional space spanned by these two vectors.
In this space, a basis is composed by the two vectors $\{\ket{k}, \ket{\nu}\}$, where
\begin{equation}
\ket{\nu} = \frac{1}{\sqrt{N-1}} \sum_{j\neq k} \ket{j}.
\end{equation}
We represent
\begin{equation}
	\ket{k} = \begin{pmatrix}
	0 \\
	1
	\end{pmatrix} \quad
	\ket{\nu} = \begin{pmatrix}
	1 \\
	0
	\end{pmatrix}
\end{equation}
and we have
\begin{equation}
	DU = \begin{pmatrix}
	\cos \theta & - \sin \theta \\
	\sin \theta & \cos \theta
	\end{pmatrix},
\end{equation}
where $\cos \theta = 1-2/N$ and since $\cos \theta \sim 1 - \theta^2/2$ for small $\theta$, we obtain $\theta \sim 2/\sqrt{N}$. Therefore the application of the operator $DU$ corresponds to a rotation of an angle $\theta$.\\
Now, we start from the state $\ket{\psi}$, which is close to $\ket{\nu}$ for large $N$. The state $\ket{i}$ is orthogonal to $\ket{\nu}$, so their relative angle is $\pi/2$. Therefore we need to apply the operator $DU$ for 
\begin{equation}
	t = \frac{\pi/2}{\theta} \sim \frac{\pi}{4} \sqrt{N}
\end{equation}
times to rotate the initial state to the target state.\\
After this operation, the probability of obtaining $k$ with a measure is
\begin{equation}
\abs{\bra{k} (DU)^t \ket{\psi}} \geq \cos^2\theta \sim 1.
\end{equation}
Since each usage of the operator U corresponds to a query to the oracle, we are doing $\pi/4 \sqrt{N}$ queries for large $N$, which is much less than in the classical case. Finally, we note that this algorithm is \emph{optimal}, in the sense that it has been proved that $\pi/4 \sqrt{N}$ is the minimum number of queries to the oracle to solve the problem, independently of the algorithm~\cite{Bennett1997,Boyer1998,Zalka1999}.

\section{Quantum Adiabatic Algorithm}\label{sec::qaa}
The gate model of quantum computing is not the only model possible. Actually, there are many others and in this section we will focus on the \emph{quantum adiabatic computation} model. Its introduction dates back to the works of Apolloni~\cite{Apolloni1989}, and at the beginning it was called \emph{quantum annealing}. The original idea was to design an algorithm similar to the simulated annealing one, but able to exploit \emph{quantum}, rather then thermal, fluctuations to escape local minima. \\
Only later, when experiments with quantum systems able to physically implement quantum annealing~\cite{Brooke1999} started to appear, the quantum annealing (or adiabatic) algorithm (QAA) becomes something which required a dedicated (quantum) device~\cite{Farhi2000}. Up to that point, the Hamiltonians used to evolve the quantum systems were composed of non-positive off-diagonal entries in the computational basis (\emph{stoquastic} Hamiltonians), but it turned out that if we allow the system to evolve with non-stoquastic Hamiltonians, then the QAA is as general as the gate model (that is, each gate-model algorithm can be re-casted as a QAA with a polynomial overhead)~\cite{Aharonov2007}.\\
Today the interest in QAA is still very high, mainly because the hope that this model of computation can provide speedup to solve NP-hard combinatorial optimization problems. To this end, some devices are available to test QAA, the most famous of them being the D-Wave system: their latest architecture, called \emph{Pegasus}, has more than 5000 qubit arranged in a topology which allows each of them to be connected with 15 others.

In this section we will review the quantum adiabatic theorems which are the theoretical backbone of QAA, and we will see QAA at work with the Grover problem. After that we will discuss one of the many unsolved problems regarding QAA (and, more specifically, implementation on real-world devices) which is called \emph{parameter setting} problem~\cite{DiGioacchino2019}.


\subsection{Why it could work... and why not}\label{sec::qaa_pandc}
The QAA consists in the following: consider a starting Hamiltonian, $H_0$, which is easy to implement and with a known ground state easy to prepare as well. Now encode the solution of a COP in the ground state of another Hamiltonian $H_1$ and define the Hamiltonian
\begin{equation}\label{eq::qaa_ham}
	H(s) = A(s) H_0 + B(s) H_1
\end{equation}
so that $A(0)=1, B(0)=0, A(1)=0$ and $B(1) = 1$. Now prepare a system in the ground state of $H_0$ and let it evolve with $H(s)$, changing $s$ from 0 to 1. The functions $A$ and $B$ are called \emph{schedule}, and we are guaranteed that the system will always remain in the instantaneous ground state of $H(s)$ provided that the change of $H$ is ``slow enough''. Therefore at the end of the evolution the system will be in the ground state of $H_1$ and a measurement will give as outcome the result of our original problem. But how slow is ``slow enough''? The answer is in the \emph{adiabatic theorem}, that we now state (in its simpler form, a nice review of the various versions can be found in~\cite{Albash2018}).

Consider a Hamiltonian $H_{t_f}(t)$ which depends on time and on the parameter $t_f$, such that $H_{t_f}(s t_f) = H(s)$ with $s \in[0,1]$. Basically this is equivalent to the requirement that once the Hamiltonian $H_{t_f}(t)$ depends on time only through the form $s=t/t_f$, which is the case for the QAA. Now, consider the set of eigenstates $\ket{\epsilon_j(s)}$ with $j\in\{0,1,\dots\}$ such that
\begin{equation}
	H(s) \ket{\epsilon_j(s)} = \epsilon_j(s) \ket{\epsilon_j(s)},
\end{equation}
and the values $\epsilon_j(s)$ are ordered in increasing order. Therefore $\ket{\epsilon_0(s)}$ is the instantaneous ground state. Now, the adiabatic theorem~\cite{Amin2009} states that, if the system is prepared in the state $\ket{\epsilon_j(0)}$ at $s=0$, it will remain in the same instantaneous eigenstate provided that
\begin{equation}\label{eq::adiab_theo}
	\frac{1}{t_f}\max_{s\in[0,1]} \frac{\abs{\bra{\epsilon_i(s)} \partial_s H(s) \ket{\epsilon_j(s)}}}{\abs{\epsilon_i(s)- \epsilon_j(s)}^2}\ll 1
\end{equation}
for each $i\neq j$. Since one is typically interested in the ground state we can set $i=0$. Moreover, notice that we can always bound the numerator from above with 1, therefore we are guaranteed to stay in the ground state if $t_f \Delta^2 \gg 1$, where
\begin{equation}
	\Delta = \min_{s\in[0,1]} \left( \epsilon_1(s) - \epsilon_0(s) \right)
\end{equation}
is usually called \emph{spectral gap} (or simply gap). 
In conclusion, the adiabatic theorem suggests us to choose $t_f = \eta \Delta^{-2}$, with $\eta \gg 1$. Notice that the typical situation is that $\Delta$ depends on the problem size $N$, as we will see in the following. Since $\eta$ has to be large but we can fix it such that it will not depend on $N$, the complexity of the QAA is entirely given by the dependence on $N$ of $\Delta$.

\subsection{A solvable case: Grover again}\label{sec::qaa_grover}
The adiabatic version of Grover's algorithm has a nice story: it has been proposed as one of the first example of application of QAA~\cite{Farhi2000}, but the result was disappointing. Indeed, no speedup with respect to the classical case was found. Only later, Roland and Cerf~\cite{Roland2002} understood how to recover the Grover speedup in the adiabatic setting. Here we review their results.

As in the standard Grover case, we have $N$ states $\ket{i}$ and a marked state $\ket{m}$ which we do not know a priori, and we want to find. We use as initial state the uniform superposition
\begin{equation}
	\ket{\psi} = \frac{1}{N} \sum_i \ket{i}.
\end{equation}
The Hamiltonian that we use to evolve the system is
\begin{equation}
	H(s) = (1-a(s)) H_0 + a(s) H_1,
\end{equation}
with
\begin{equation}
	H_0 = \mathbb{I} - \ket{\psi} \bra{\psi}
\end{equation}
and
\begin{equation}
	H_1 = \mathbb{I} - \ket{k} \bra{k}.
\end{equation}
Notice that $\ket{\psi}$ is the ground state of $H_0$ with eigenvalue 0, and $\ket{k}$ is the ground state of $H_1$, again with eigenvalue 0. For this problem the schedule is completely determined by the choice of $a$.\\
Now we need to evaluate the eigensystem of $H(s)$ in order to choose a proper schedule $s=s(t)$. Notice that we start from the state $\ket{\psi}$ and therefore, since the Hamiltonian only depends on projectors on $\ket{\psi}$, $\ket{k}$ and identities, the evolution remains in the subspace spanned by $\ket{\psi}$ and $\ket{k}$. A basis of this space is \{$\ket{k}$,$\ket{\nu}$\}, where
\begin{equation}
	\ket{\nu} = \frac{1}{\sqrt{N-1}} \sum_{j\neq k} \ket{j}.
\end{equation}
In this subspace, where the non-trivial evolution of the initial state happens, we use
\begin{equation}
	\braket{\nu}{\psi} = \sqrt{1- \frac{1}{N}}, \quad \braket{k}{\psi} = \frac{1}{\sqrt{N}}, \quad \braket{k}{\nu} = 0
\end{equation}
and we obtain
\begin{equation}
\begin{split}
	& \bra{k} H_0 \ket{k} = 1- \frac{1}{N}, \ \bra{k} H_0 \ket{\nu} = \bra{\nu} H_0 \ket{k} = -\frac{1}{\sqrt{N}} \sqrt{1- \frac{1}{N}}, \ \bra{\nu} H_0 \ket{\nu} = \frac{1}{N},\\
	& \bra{k} H_1 \ket{k} = 0, \ \bra{k} H_1 \ket{\nu} = \bra{\nu} H_1 \ket{k} = 0, \ \bra{\nu} H_1 \ket{\nu} = 1.
\end{split}
\end{equation}
At this point, we compute the eigenvalues of the matrix $H(s)$ restricted to this $2$-dimensional space (the other eigenvalue is 1, with degeneracy $N-2$) and we obtain
\begin{equation}
\begin{split}
	& \epsilon_0(s) = \frac{1}{2}- \sqrt{1 - 4 \left(1- \frac{1}{N}\right) a (1-a)}\\
	& \epsilon_1(s) = \frac{1}{2}+\sqrt{1 - 4 \left(1- \frac{1}{N}\right) a (1-a)}.
\end{split}
\end{equation}
Therefore we have, for the instantaneous gap:
\begin{equation}
	g(s) = \epsilon_1(s) - \epsilon_0(s) = \sqrt{1 - 4 \left(1- \frac{1}{N}\right) a (1-a)}.
\end{equation}
In conclusion, the minimum gap is obtained at $a=1/2$ (see inset of Fig.~\ref{fig::grover_schedule}).
If we take $a=s$ (linear schedule), we obtain that the minimum gap is
\begin{equation}
	\Delta = g(1/2) = \frac{1}{\sqrt{N}}
\end{equation}
and therefore, by using Eq.~\eqref{eq::adiab_theo} and $\abs{\bra{\epsilon_0(s)} \partial_s H(s) \ket{\epsilon_1(s)}}\leq 1$, we have
\begin{equation}
	\frac{N}{t_f} \ll 1
\end{equation}
and so we need $t_f \gg N$ (notice that this means $t_f = \eta N$ with $\eta$ some small, N-independent parameter). This disappointing result can be improved by a more careful choice of the schedule $a(s)$. Indeed, let us consider again Eq.~\eqref{eq::adiab_theo}. In this case, we have
\begin{equation}
	\frac{1}{t_f} \max_{s\in[0,1]} \abs{\frac{d a}{d s}} \frac{\abs{\bra{\epsilon_1(s)} \partial_a H(a) \ket{\epsilon_0(s)}}}{g(s)^2}\ll 1.
\end{equation}
Therefore, we can require that, for each $s\in[0,1]$,
\begin{equation}
	\frac{1}{t_f}\abs{\frac{d a}{d s}} \frac{1}{g(s)^2} = \eta,
\end{equation}
where $\eta \ll 1$ is a small parameter. We obtain the differential equation for $a(s)$,
\begin{equation}\label{eq::grover_diff_eq}
	\frac{d a}{d s} = t_f \eta \left( 1 - 4 \left(1- \frac{1}{N}\right) a (1-a) \right),
\end{equation}
with the initial condition $a(0)=0$ and $t_f$ has to be chosen such that $a(1)=1$. 
The result of Eq.~\eqref{eq::grover_diff_eq} for $a(s)$ is plotted in Fig.~\ref{fig::grover_schedule}, while for $t_f$ we obtain
\begin{equation}
	t_f=\frac{\pi}{2 \epsilon} \sqrt{N}.
\end{equation}

\begin{figure}[!htb]
	\centering
	\begin{tikzpicture}
	\node[anchor=south west,inner sep=0] (image) at (0,0) {\includegraphics[width=\textwidth]{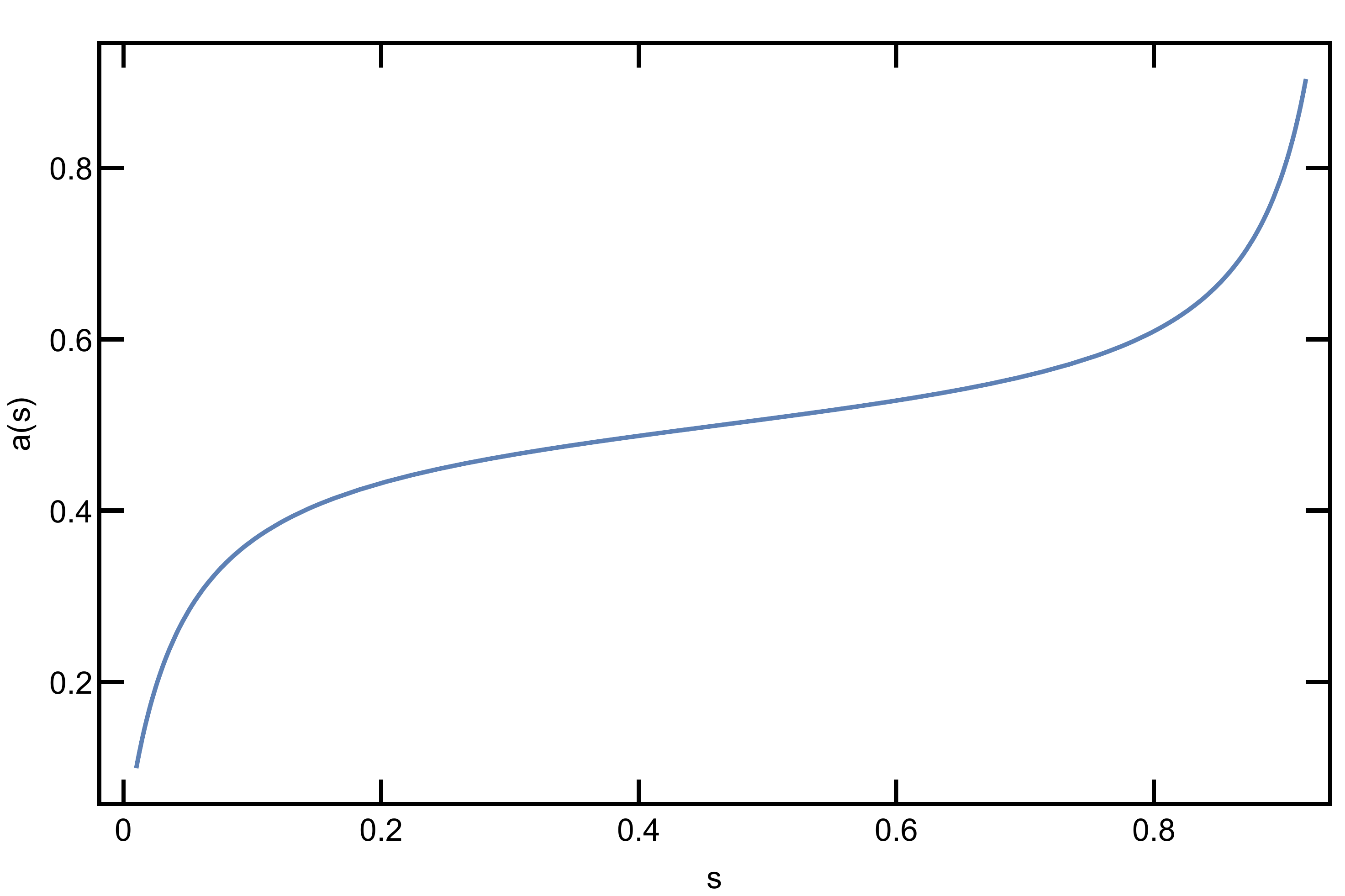}};
	\begin{scope}[x={(image.south east)},y={(image.north west)}]
	\node[anchor=south west,inner sep=0] (image) at (0.1,0.55) {\includegraphics[width=0.4\textwidth]{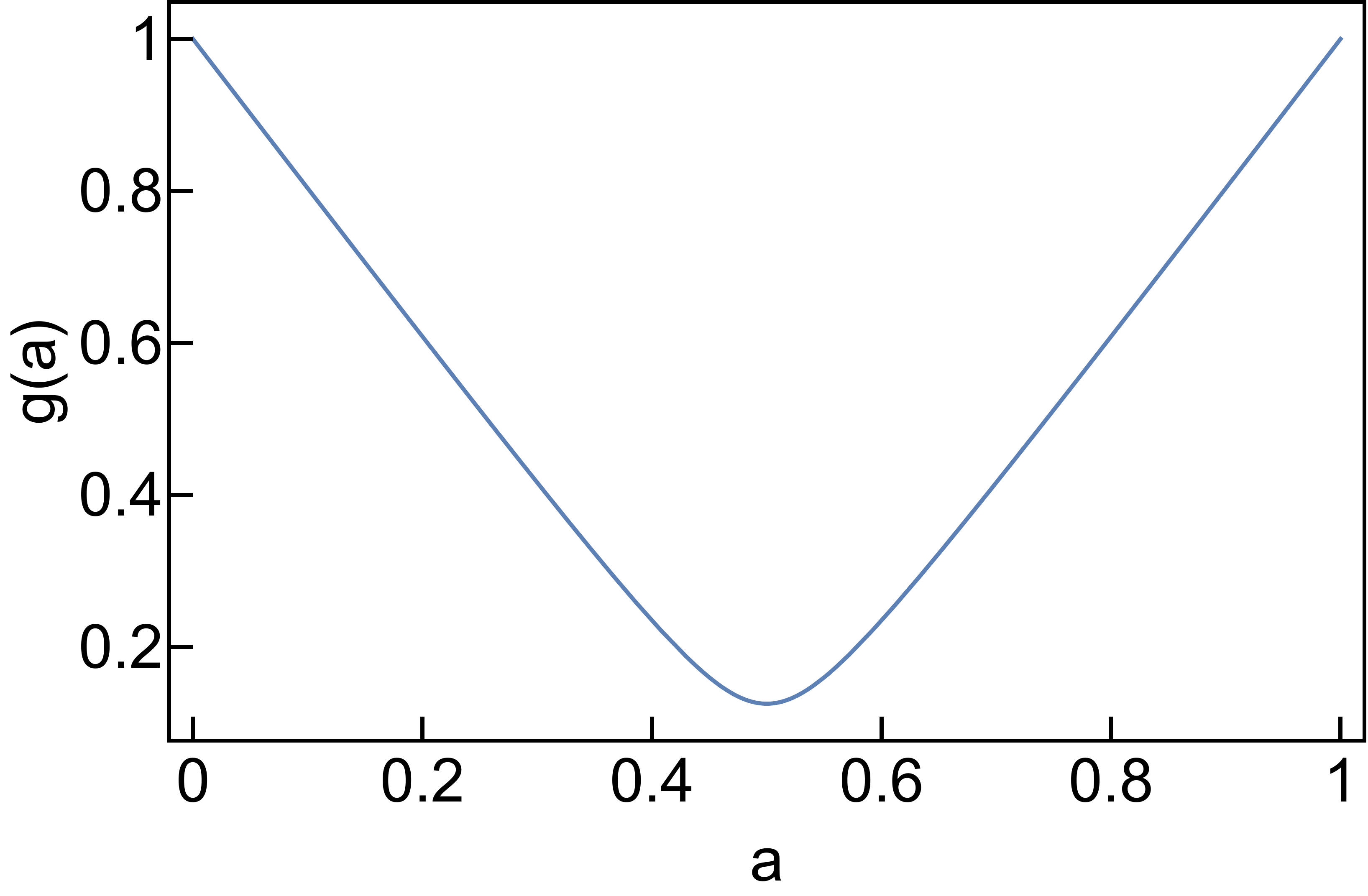}};
	\end{scope}
	\end{tikzpicture}
	\caption{Plot of the solution of Eq.~\eqref{eq::grover_diff_eq}, with $N=64$ and $t_f \epsilon = \sqrt{N} \pi/2.$ As we can see, $a(s)$ changes faster when $s$ is close to 0 or 1, that is (see inset) when the gap is large, and it changes more slowly when $s$ is closer to 0.5, that is to the minimum gap.}\label{fig::grover_schedule}
\end{figure}

This example is instructive and allows us to see an important point: the Hamiltonian can be changed quickly when the gap is large, but the annealing schedule has to slow down where the gap is small. Unfortunately, the computation of the gap is extremely difficult in most cases of interest, and therefore some more heuristic treatment is usually applied.

\subsection{The parameter setting problem}\label{sec::qaa_parset}

Since its introduction, the QAA has been thoroughly studied to understand whether it can be useful to tackle computationally hard problems faster than classical algorithms~\cite{Farhi2000, Kadowaki1998}. The usual comparison is with the classical Simulated Annealing (SA) algorithm and its variants, such as the Parallel Tempering (PT). The general idea at the heart of the hoped success of QAA is that quantum fluctuations could be more effective then thermal ones in exploring rough energy landscapes (even though there are also other possible kind of advantages of quantum algorithms over classical ones \cite{Baldassi2018}). This intuition has been built mostly by using very simple toy-models, such as the highly symmetric Hamming weight problem~\cite{Farhi2002,Muthukrishnan2016} or oracular problems (as the Grover problem analyzed in Sec.~\ref{sec::qaa_grover}), but conceptual arguments proving any kind of quantum speedup for real-world problems lack to date, despite the significant efforts made~\cite{Denchev2016,Mandra2016}. \\
The recent appearance of quantum annealers of relevant size, such as the D-Wave 2000Q, which allows to control about 2000 physical qubits, provided a more pragmatic road: we are now in the exciting position of doing some actual experiments using these annealers to solve certain COPs, and then compare the performances with those of classical solvers.\\
However, many practical issues appear in this case, most of which are related to the fundamental question ``how can we do a fair comparison?''~\cite{Ronnow2014, Mandra2018}. 
It has been soon understood that one needs to carefully choose the problems to be solved.
The first step is to consider COPs that admit a rewriting as Quadratic Unconstrained Binary Optimization (QUBO) problems, that is, in the same spirit of Sec.~\ref{sec::cop_spin}, as
\begin{equation}\label{eq::para_qubo_general}
H_0 = \sum_{i, j} J_{i,j} x_i x_j + \sum_i h_i x_i,
\end{equation}
where $x_i\in\{0,1\}$ and the values of the couplers $J_{i,j}$ and those of the local fields $h_i$ are used to specify the problem and the instance.\\ 
However, this is not enough: to exploit in the best way possible the effect of quantum fluctuations, one has to consider problems with a sufficiently complex energy landscape. Often this is achieved by studying problems whose thermodynamics presents a spin glass phase at low temperature. Unfortunately, the present architecture of qubit interactions in the D-Wave system does not allow to have this kind of difficult problems~\cite{Katzgraber2014} without an extra step, that is the embedding of a different interaction graph into the D-Wave qubit interaction graph (which is called Chimera graph for the D-Wave machines up to 2000Q).
To do that, we need to introduce an extra term in the QUBO Hamiltonian, which embodies some constraints needed to embed the graph~\cite{Choi2008,Choi2011}. \\
Moreover, there are many other COPs whose QUBO formulation itself requires a hard-constraint term, such as the traveling salesman problem, the matching problem, the knapsack problem, the 1-in-3 satisfiability problem and many others~\cite{Lucas2014}. In all these cases, one has an Hamiltonian of the form
\begin{equation}\label{eq::para_h_tot}
H(\lambda) = H_P + \lambda H_C,
\end{equation}
where $H_P$ is the problem Hamiltonian, which is written in QUBO form and $H_C$ is the Hamiltonian, written again in QUBO form, which ensures the constraints by giving a penalty in energy to the configurations which break one or more of them.

Here we address the problem of choosing the value of the parameter $\lambda$. An easy recipe for this choice does not exists: indeed $\lambda$ has to be large enough so that the ground state of our problem (which is the state we are after) has no broken constraint, but it has been argued theoretically \cite{Choi2008} and observed experimentally \cite{Venturelli2015} that a small value provides better performances.

\subsubsection{Optimal choice of parameters: framework}
Consider a COP defined by a cost function $E: \Omega \to \mathbb{R}$, where $\Omega$ is a discrete set. We will refer to this problem as the ``logical'' problem. Consider now that this problem admits a QUBO version. This means that we also have another, ``embedded'', Hamiltonian $H_P : \{0,1\}^N = \mathcal{B} \to \mathbb{R} $ ($N$ is the number of binary variables that we need to encode the problem) and an invertible function $\phi: \Omega \to \mathcal{S} \subseteq \mathcal{B} $ such that $H_P(\phi(\sigma)) = E(\sigma)$ for each $\sigma \in \Omega$. Now consider the case $\mathcal{S} \subset \mathcal{B}$: $H_P$ will give an energy also to elements of the boolean hypercube in $\mathcal{B} \setminus \mathcal{S} = \mathcal{S}^c$, which do not correspond to acceptable configurations of the logical problem.

As an example, let us consider again the matching problem introduced in Sec.~\ref{sec::matching}: given a graph $G = (V, L)$ and a weight $w_\ell \geq 0$ associated to each edge $\ell \in L$, let us call $\mathcal{A}$ the set of all matchings. 
To obtain the QUBO form of this problem, we assign to each edge $\ell$ a binary variable $x_\ell$ which is 1 or 0 if the edge is used or not in the configuration $x$. As we have seen in Sec.~\ref{sec::matching}, if we want a Hamiltonian in QUBO form, we need to introduce a soft constraint and we obtain
%
%
%
\begin{equation}\label{eq::para_matching_qubo}
H_\lambda(x) = H_P + \lambda H_C = \sum_{\ell \in L} w_\ell x_\ell + \lambda \sum_{\nu \in V} \left( 1 - \sum_{\ell \in \partial\nu} x_\ell \right)^2,
\end{equation}
where the quadratic term, provided that $\lambda$ is large enough, enforces the fact that (at least in the ground state) each point has to be connected to exactly one another point.

Let us define 
\begin{equation}
E_{gs} = \min_{\sigma \in \Omega} E(\sigma), \qquad \mathcal{E}_{gs}(\lambda) = \min_{x\in \mathcal{B}}  H_\lambda (x).
\end{equation}
The ``minimum'' value of the parameter, $\lambda^\star$, is the smallest $\lambda \in \mathbb{R}^+$ such that
\begin{equation}
E_{gs} = \mathcal{E}_{gs}(\lambda).
\end{equation}

We define the ``optimal'' value for the parameter $\lambda$, for a fixed heuristic algorithm, as the one such that the time-to-solution (TTS) (see Appendix \ref{app::parameter_tts}, for a definition of TTS) of this algorithm is minimized. 
Therefore the optimal parameter depends in general on the algorithm we are going to use. However, if we focus on annealing algorithms with local moves, it is possible to build some intuition that the optimal parameter is (at least close to) the minimum parameter. 
Indeed, this kind of heuristic algorithms are used to explore complex energy landscapes  
and the idea behind classical/quantum annealing is roughly to exploit thermal/quantum fluctuations to overcome the energy barriers which separate low-energy configurations, so that we can explore these configurations and pick the optimal one. \\
Now consider the case in which the barrier to overcome is given by the $H_C$ term in Eq.~\eqref{eq::para_matching_qubo}, that is because of a penalty term: 
if the coupling term is lowered, the height of the barrier is lowered so the annealing can proceed faster.
This happens, for example, when the Hamming distance between couples of allowed configurations is always larger than 1 (if the algorithm only performs single spin flips): in this case the algorithm has to overcome a barrier given by the penalty term each time it changes the system configuration from one in $\mathcal{S}$ to another in $\mathcal{S}$, passing through $\mathcal{S}^c$. 
An explicit example of this is the matching problem: indeed if the system is in an allowed configuration, the closest allowed configuration is at distance 4 and it corresponds to the swap of two matched points.
Moreover, it is easy to check that this is again the case for many other combinatorial optimization problems relevant for both practical and theoretical analyses.

In Appendix \ref{app::parameter_toy}, we investigate the effect of changing $\lambda$ with a toy model example, where all the computations can be done analytically. In the following, on the other hand, we will firstly provide and discuss an algorithm to find the minimum value of $\lambda$ (in some cases), and we will apply it to study the effect of the choice of $\lambda$ for a specific combinatorial optimization problem.

\subsubsection{Optimal choice of parameters: an algorithm}
The usual strategy to obtain a good constraint term $H_C$ is to find some set of constraints that the binary variables have to respect to be mapped in a logical configuration by $\phi^{-1}$. Then $H_C$ is implemented such that it increases if the number of broken constraints increases, and is zero if no constraint is broken. 
For example, for the matching problem we have that given a vertex $\nu$, only one among the edges in $\partial\nu$ has to be used. So we have one constraint for each point, and the term that we inserted in Eq.~\eqref{eq::para_matching_qubo} is positively correlated to the number of broken constraints.
%

We denote with $E_0^{(k)}$ the minimum over the set of configurations $x$ with $k$ broken constraints of $H_P(x)$. So, for example, $E_{gs} = E_0^{(0)}$.
Therefore, the minimum parameter $\lambda^\star$ is the smallest possible such that
\begin{equation}\label{eq::para_general_minb}
E_0^{(0)} < k \lambda + E_0^{(k)},
\end{equation}
for $k=1,\dots,M$, where $M$ is the maximum number of constraints that can be broken in a single configuration. Therefore we have
\begin{equation}\label{eq::para_general_eq_for_par}
\lambda^\star > \max_{k\in\{ 1,2,\dots, N \}} \frac{E_{0}^{(0)} - E_0^{(k)}}{2 k} = \max_{k\in\{ 1,2,\dots, N \}} \lambda_k,
\end{equation}
where $\lambda_k = (E_{0}^{(0)} - E_0^{(k)})/(2 k)$.\\
This inequality cannot be used efficiently to obtain $\lambda^\star$ as it is: the computation of each $\lambda_k$ could be even more difficult than solving the original problem. On the other hand, one can obtain an approximation of each $\lambda_k$: to do this, one needs to approximate $E_{0}^{(0)}$ from above and $E_0^{(k)}$ from below. But also in this case, one still needs to compute all the $M$ different $\lambda_k$'s, and for most of the interesting problems $M$ scales with the system size $N$. 
This happens, for example, for our working example, the matching problem, where the number of constraints is the number of vertices to be matched.
To worsen the situation, the computation of $E_0^{(k)}$ requires the minimization of the energy over all the possible ways of breaking $k$ constraints, and this number can grow exponentially in $N$ (as it happens, for example, for the matching problem).
However, if we can prove that 
\begin{equation}\label{eq::para_general_order}
\lambda_1 \geq \lambda_2 \geq \dots\geq \lambda_N
\end{equation}
then $\lambda^\star$ can easily be found by estimating $\lambda_1$ and taking the smallest value such that
\begin{equation}\label{eq::para_min_param_1}
\lambda^\star > \lambda_1.
\end{equation}
Let us give some qualitative arguments to understand why Eq.~\eqref{eq::para_general_order} is a reasonable expectation. We have that $\lambda_k \geq \lambda_{k+1}$ if and only if
\begin{equation}\label{eq::para_general_dim}
\begin{split}
E_0^{(0)} - E_0^{(1)} + E_0^{(1)} +\dots -E_0^{(k-1)} & + E_0^{(k-1)}- E_0^{(k)} \\
& \geq k \, (E_0^{(k)}-E_0^{(k+1)}).
\end{split}
\end{equation}
If we prove that
\begin{equation}\label{eq::para_general_gainorder}
E_0^{(n-1)}-E_0^{(n)} \geq E_0^{(n)}-E_0^{(n+1)},
\end{equation}
for each $n=0,1,\dots,N$, then inequality~\eqref{eq::para_general_dim} immediately follows (this is a sufficient but not necessary condition).
This condition is nothing but the fact that the maximum gain in energy that we can obtain by breaking the $n$-th constraint is lower than the one that we obtain by breaking the $(n+1)$-th, for each value of $n$.\\
Actually the inequality given in~\eqref{eq::para_general_order} is satisfied for some problems, but not for all of them. In particular, it depends on both $H_P$ and $H_C$ and in Appendix \ref{app::parameter_failure} we show a specific problem and a specific choice of $H_C$ for which this condition does not hold.
We will see that for the matching problem defined as in Eq.~\eqref{eq::para_matching_qubo} this conditions is satisfied. Finally, notice that there are other algorithms that can be used to find the minimum parameter $\lambda^\star$: when the algorithm we discuss here is not applicable, these methods can be an alternative strategy. However, as we discuss in Appendix \ref{app::parameter_other}  using the example of the matching problem, the differences in performances among these methods can be quite relevant.

\subsubsection{An explicit example: the matching problem}
As we already discussed, the matching problem is in the P complexity class. 
However it is empirically known that for many problems in the P class, heuristic algorithms such as SA still need an exponential time to find the exact solution. 
When this problem is written in QUBO form it is one of the simplest possible constrained problems: quadratic terms are in the penalty term only, and the problem is trivial without it. On the top of that, we have seen that the structure of the physical energy landscape, that is logical configurations separated by non-logical ones, is common to many other problems.
Therefore the matching problem is an ideal starting point to study the effects of the choice of the penalty term coupling parameter.\\
Another, more practical, reason to choose this problem is that, since it is polynomial, we can compute $\lambda_1$ (as defined in Eq.~\eqref{eq::para_general_eq_for_par}) in polynomial time, and we will see that for the QUBO form that we will use for this problem, condition~\eqref{eq::para_general_order} holds. Therefore we can in polynomial time find the minimum parameter, and test in a realistic problem if that is the optimal value.
Notice that we will actually use the exact solution of the problem to obtain the minimum parameter, since our objective is to understand the effect of the choice of the parameter rather than providing an algorithm to find the minimum parameter itself.
Nonetheless, for more interesting (NP-hard) problems one cannot use the solution of the problem, but, as we have discussed previously, approximate solutions together with our technique could be used to obtain good values for the parameter.

Let $\mathcal{A}$ be the set of all the possible matchings for our problem graph $G = (V, L)$. We define
\begin{equation}
E_N = \min_{\sigma\in\mathcal{A}} (E_{\sigma}),
\end{equation}
where $2 N = |V|$ is a measure of the problem size. 

To discuss the inequality~\eqref{eq::para_general_gainorder} in this case, we need to analyze how $E_0^{(k)}$ is obtained. 
Firstly, notice that constraints are always broken in pairs. Now consider a configuration with $2 k$ broken constraints, with $k>0$. Suppose that we can find a point $x$ which is endpoint of $m > 1$ edges. Now consider the configuration obtained by removing $m-1$ links which have $x$ as endpoint: as it is clear from Eq.~\eqref{eq::para_matching_qubo}, we will have a lower cost given by the penalty term of $H_\lambda$, and a lower cost given by the weight term. Therefore, the way to break $2 k$ constraints which minimizes $H_\lambda$ is obtained by configurations which have $2k$ points which are not matched with any other point, that is by \emph{ignoring} $2k$ points of the initial set of $2N$ points in the matching.
Therefore we introduce the symbol
\begin{equation}
\mathcal{E}_{N-k} = E_0^{(k)},
\end{equation}
where we dropped the subscript $0$ to shorten the notation and we stress the fact that we can interpret $E_0^{(k)}$ as the optimal matching when we can ignore $2k$ points. Notice that $\mathcal{E}_{N} = E_N$.
The inequality that we want to prove is then
\begin{equation}\label{eq::para_gainorder}
\mathcal{E}_{n+1}-\mathcal{E}_{n} \geq \mathcal{E}_{n}-\mathcal{E}_{n-1},
\end{equation}
for each $n=0,1,\dots,N$.
The proof is rather technical and is given in full details in Appendix \ref{app::parameter_proof}. Here we only sketch its structure and the main ideas behind it:
\begin{itemize}
	\item we prove that if a point is ignored when $2 k$ constraint can be broken, it will also be ignored when $2(k+1)$ can be broken (``stability'' property);
	\item using the previous fact, we can prove Eq.~\eqref{eq::para_gainorder} (``order'' property).
\end{itemize}
Both the points are proven by building sub-optimal matchings by using pieces of the solutions with cost $\mathcal{E}_{n+1}$ and $\mathcal{E}_{n-1}$ and using the fact that the solutions with costs $\mathcal{E}_{n+1}$, $\mathcal{E}_{n-1}$ and $\mathcal{E}_{n}$ are optimal.

\subsubsection{Numerical results for the matching problem}
The aim of this section is to numerically study the relevance of the choice of the parameter value in terms of performances for the matching problem, where the minimum parameter can be found in polynomial time. This will also allow us to discuss our qualitative picture introduced previously, at least for this specific example. 
We used an exact, polynomial-time, solver to compute the energy of the optimal solution, both when no constraints are broken (to obtain $\mathcal{E}_N$) and when one constraint is broken (we broke it in every possible way, and the minimum of the energies obtained is $\mathcal{E}_{N-1}$). Here, with one broken constraint we mean that we are ignoring 2 points to be matched, as discussed in the previous section. Therefore we can break a constraint in $N (N-1)$ ways, where $N$ is the number of points, and so the procedure to find the minimum parameter is still polynomial. Once we obtained the minimum parameter, we run the classical and quantum algorithms using that value and values at a fixed distance from it. We then computed the time to solution, which is used throughout this section as measure of performance.\\
Let us give now some details about our numerical analysis: we considered the matching problem on a specific graph, which is a 2-dimentional regular lattice of vertices, where each vertex has 4 edges which connect the vertex with its nearest neighbors (the vertices on the boundaries have less edges because we used open boundary conditions). We used this specific graph because in this case the problem of the minor embedding in the D-Wave 2000Q hardware graph (the Chimera graph) is moderate, so we can focus on the effects caused by the change of the penalty term of the matching problem, neglecting the fact that also the penalty term of the minor embedding problem plays a role in determining the performances. The weights of the edges are randomly extracted among the set of integers $\{0, 8, 16 \}$. We need to use multiples of 8 to have integers number after the re-writing of the problem in terms of QUBO form and then again in Ising variables (which is the form of input Hamiltonian used in both out parallel tempering algorithm and the D-Wave 2000Q). We decided to use only integers and such a small set to avoid precision problems, which are particularly severe in noisy devices such as the D-Wave 2000Q. 
Then, for each system size analyzed, the parameter $\lambda$ is chosen as $\lambda = \lambda^\star + 2 \delta$, where $\delta\in\{-2, -1, 0, 1, 2, 3 \}$. We did not consider a finer grid of values for $\delta$ because our computations are done with finite precision, and too close values of $\delta$ would be indistinguishable. In Fig.~\ref{fig_histo} we present histograms of the values of $\lambda^\star$ obtained for system at various sizes. A first important consequence of this plot is that the value of $\lambda^\star$ is not a self-averaging quantity (at least not for the system sizes explored here). On the opposite, the variance of $\lambda^\star$ we obtained is increasing with the system size. 

\begin{figure}[tb]
	\centering
	\includegraphics[width=\columnwidth,keepaspectratio]{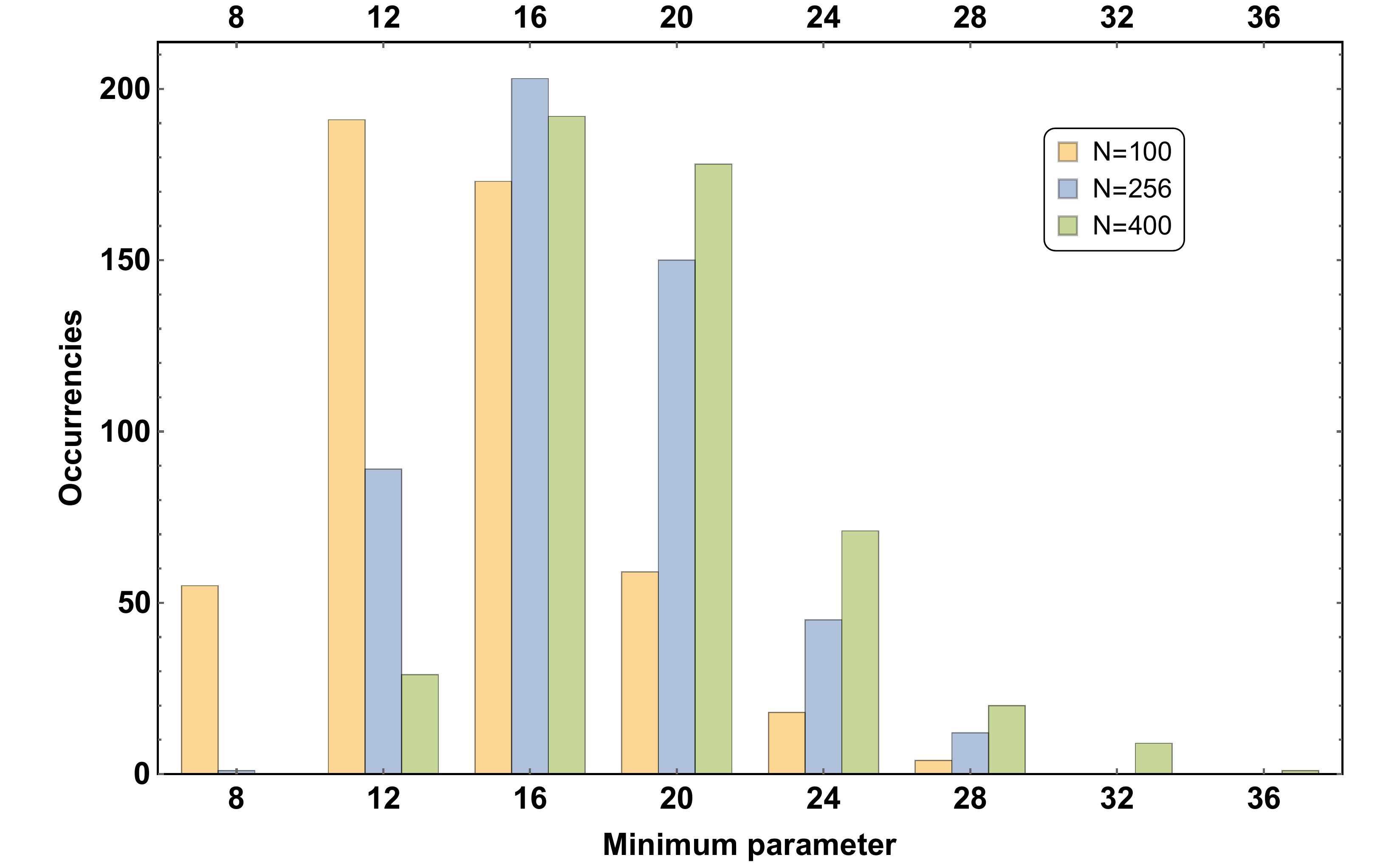}
	\caption{Histograms of the minimum coupling parameter obtained for 500 different instances of the matching problem, for system sizes of 100, 256 and 400. The optimal parameters are always multiple of 4 because of our initial choice of link weights.}\label{fig_histo}
\end{figure}

\paragraph{Classical heuristic algorithm}
We used the Parallel Tempering (PT) algorithm included in the NASA's Unified Framework for Optimization (UFO). We analyzed systems with sizes up to 484 points (that is, a lattice with 22 points of side length) that means, because of our QUBO embedding, about 900 binary variables.  
To choose the temperatures for the PT we considered the energy scale given by the penalty term parameter, and we multiplied it for two constants (one for the lowest and one for the highest temperatures) which are found by maximizing the number of times that the PT algorithm finds either the GS or a forbidden configuration with lower energy than the GS, with $\delta=0$. This is done to have the cleanest possible values of TTS close to $\delta=0$, and we have checked that the qualitative picture (regarding the TTS scaling with $\delta$) does not change when varying the temperatures. 
To obtain the TTS we proceeded as follows: we randomly generated 500 instances, and run the PT 500 times for each instance. For $N=400$ and $N=484$, the number of different instances is reduced to 250.
When the PT algorithm succeeded in finding a good solution, we recorded the time used; when it failed in the time given, or it found a solution with energy lower than the GS (because of broken constraints), we recorded a failure and so ``infinite'' time to find the solution. We do that because once the system is trapped in a local minimum of the energy landscape, to escape from that it will require (typically) much more time than that allowed to each run of the algorithm.
Using the data collected in this way, we can compute the TTS, and the results are shown in Figs. \ref{fig_tts_cl_1} and \ref{fig_tts_cl_2}, as functions respectively of $\delta$ at fixed $N$ and of $N$ at fixed $\delta$.

\begin{figure}[tb]
	\centering
	\includegraphics[width=\columnwidth,keepaspectratio]{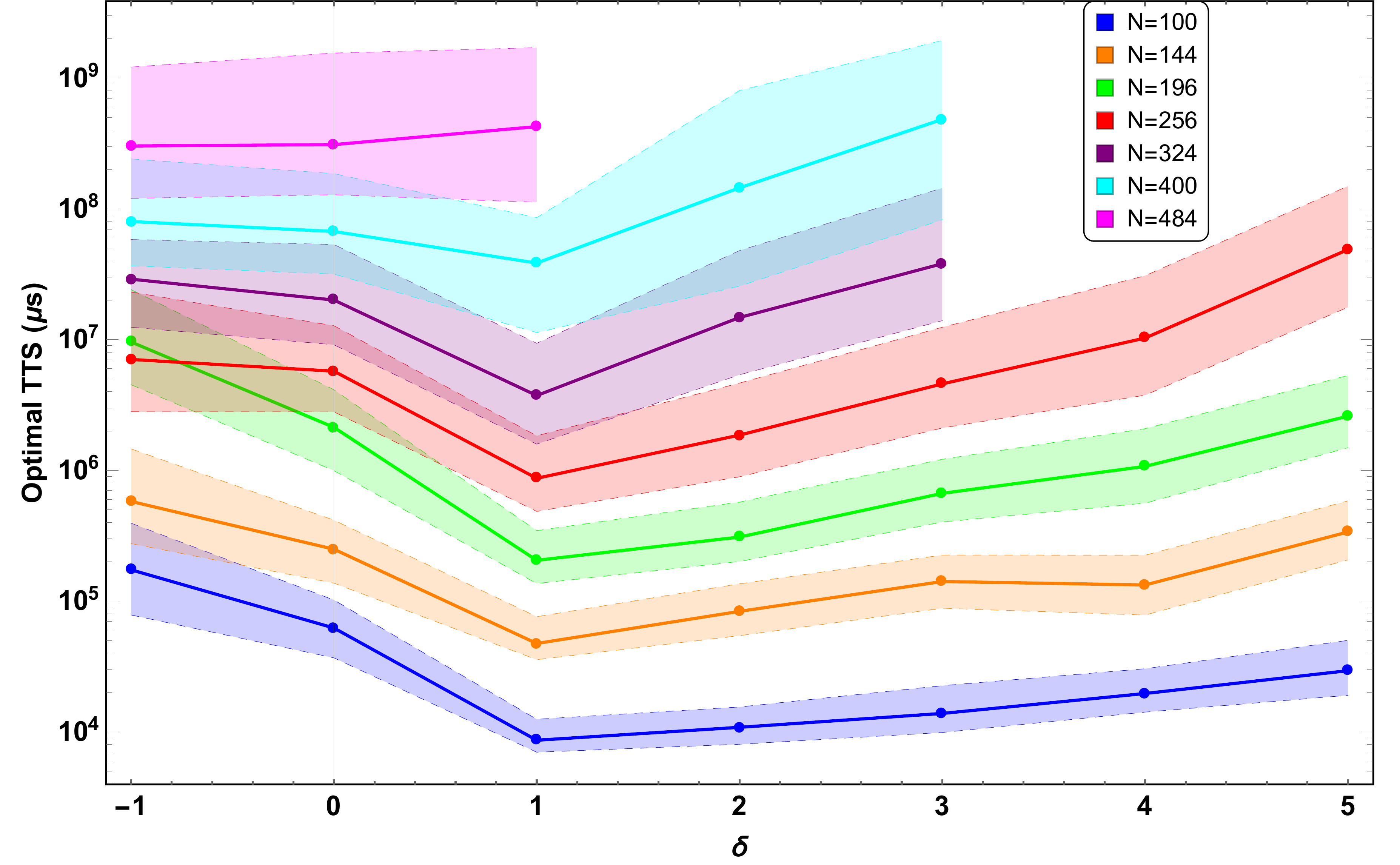}
	\caption{Time-to-solutions for the matching problem, shown at fixed N as function of the distance from the optimal parameter. The solid lines connect points computed using the 50-percentile of instances, the dashed lines corresponds to the 35-percentile (below solid lines) and 65 percentile (above solid lines).}\label{fig_tts_cl_1}
\end{figure}

\begin{figure}[tb]
	\centering
	\includegraphics[width=\columnwidth,keepaspectratio]{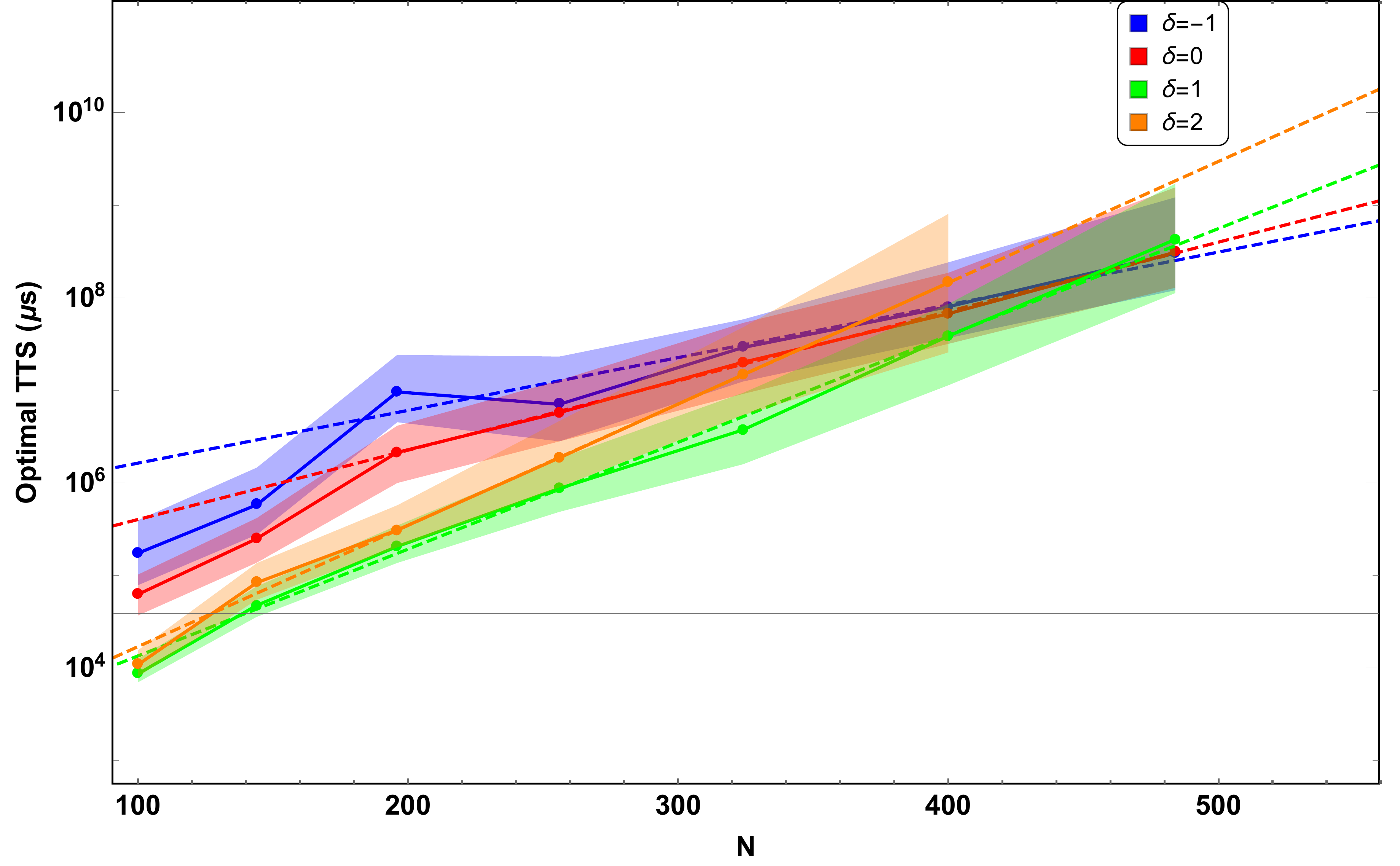}
	\caption{Time-to-solutions for the matching problem, shown at fixed N as function of the distance from the optimal parameter. The solid lines connect points computed using the 50-percentile of instances, the shaded areas corresponds to the 35-percentile (below) and 65-percentile (above), the dashed lines are the best fit of the form $A e^{B N}$
		.}\label{fig_tts_cl_2}
\end{figure}

Let us now comment the results obtained: from Fig. \ref{fig_tts_cl_1} we can see how, as intuitively predicted, the use of parameters close to the minimum results in faster annealing. This effect is more important as $N$ increase.
The 50-percentile shows that, at $N=484$, the maximum system size analyzed here, the optimal choice is $\delta=-1$. Notice that we do not plot $\delta=-2$ in Fig. \ref{fig_tts_cl_1} because at all system sizes considered here, at least one percentile of the TTS exceeded the maximum allowed (which was $10^{20}$ $\mu s$). The reason is that in that case many instances are never solved by the algorithm. On the contrary, for the largest system sizes $\delta=-1$ also maximized the number of instances where the algorithm found the solution at least once.
Therefore it seems that, for large $N$, the use of values of the parameter slightly lower that the minimum is preferred, at least for this problem.
This is investigated in more details in Fig. \ref{fig_tts_cl_1}, where the 50-percentile of the TTS is fitted with a function of the form $A e^{B N}$. The obtained value of $B$ are in Table \ref{table_cl} and show that a moderate exponential speedup in TTS can be obtained using $\delta=-1$ and that $\delta=0$ gives a small exponential speedup against $\delta=1$, which in turn gives a small exponential speedup against $\delta=2$ and so on. 
Another very important consideration is that in Fig. \ref{fig_tts_cl_1} we do not have plotted points corresponding to at least one percentile lines exceeding the maximum limit of $10^{20}$ $\mu s$. This is the reason why the curves become shorter as $N$ increases. This means that outside an ``acceptable interval'' of values around $\lambda^\star$ the performance of the PT algorithm rapidly spoils, and most of the instances are never solved. Moreover, this interval becomes 
smaller and smaller as the system size increases. This means that to use the PT algorithm we need to be more and more precise in finding $\lambda^\star$ and that this parameter has to be found with a pre-processing applied to each instance since, as discussed previously, even at large system size it depends on the specific instance.

\begin{table}[ht]
	\centering
	\begin{tabular}{lc} 
		$\delta$ & $B$ \\
		\hline
		-1 & $(1.3 \pm 0.2) \cdot 10^{-2}$     \\ 
		0  & $(1.72 \pm 0.02) \cdot 10^{-2}$   \\ 
		1  & $(2.6 \pm 0.1) \cdot 10^{-2}$     \\ 
		2  & $(3.01 \pm 0.01) \cdot 10^{-2}$
	\end{tabular}
	\caption{Fitting parameter $B$ for fit equation $A e^{Bx}$ for values value of $\delta$. The fitted data and the fitting curves are those used in Fig. \ref{fig_tts_cl_2}.}\label{table_cl}
\end{table}

\paragraph{Quantum heuristic algorithm}
The quantum computations are performed using the D-Wave 2000Q quantum annealer. In particular, we have embedded the problem in the Chimera graph, so that the logic Ising variables are mapped to ferromagnetic chains of length 4 after the minor embedding, except the Ising variables on the boundaries of the square lattice which correspond to shorter chains. Then we run the QAA for system of side sizes $N=16,36,64$. Larger systems up to $N=256$ (that is, a lattice with 16 points of side length) are in principle possible for the D-Wave 2000Q chip, but they resulted in too few solved instances. Notice that an instance at $N=64$ is a matching of 64 logical points, which corresponds, after the QUBO and the minor embedding, to a problem of $\sim 500$ qubits. Indeed, the starting graph that we have chosen for the matching problem is such that each vertex can be mapped in a 8-qubit unit cell of the Chimera graph, and to do that each binary variable is mapped in a 4-qubit ferromagnetic chain. Notice that the chain lengths are independent on $N$. We used a majority voting technique to correct chains broken at the end of the annealing (when there is no majority, the chain is randomly corrected). 
We set the annealing parameters (annealing time and ferromagnetic coupling for the embedded chains) such that the average number of successful annealings is maximized (we noticed that these settings do not depend in a relevant way on the value of $\delta$ that we use to build the instances).

To obtain the TTS, we generated 100 random instances as discussed in the previous section, and each instance is submitted for $10^4$ runs of the D-Wave 2000Q. The relevant parameter is then the probability of finding the GS for a fixed instance, which is averaged over the instances and plotted in Fig. \ref{fig_tts_q1}.
From these probabilities one can obtain easily the TTS, which we show in Fig. \ref{fig_tts_q2} to ease the comparison with the classical case. Notice that in this case the curves correspond to percentiles, while in Fig. \ref{fig_tts_q1} we have plotted averages and standard deviations as errors.

Unfortunately, due to the small system size that we can analyze in this case, we cannot obtain firm conclusions. However, it seems reasonable to expect that the same problems observed in the classical case can be repeated here: in particular it is still true that a choice of $\delta=-1$ is a particularly bad for small $N$, but this choice improves (i.~e.~it is less worse) as $N$ increases. On the other hand, up to the size $N=8$ the choice of a $\delta>0$ is the optimal, even so a too big value start spoiling the performances. Moreover, as can be seen more explicitly in Fig. \ref{fig_tts_q1}, it is still true that for larger system size the choice of the parameter $\lambda$ becomes more relevant in terms of performances.

An interesting question is why the quantum annealer is not able to solve problem of size $N=100$ or larger. We think that the precision problems have a role, but another reason could be also the structure of the embedded (QUBO) energy landscape itself: in particular, we think that the fact that logical states are always separated by not-acceptable states might be a severe obstacle for quantum annealers. Notice that this is true also for each problem which is embedded in the hardware graph in such a way that each QUBO binary variable is now a chain of qubits. However in this case one can use majority voting or other methods to correct configurations with this constraint broken. In our case (as in other many interesting problems) a simple correction as the majority voting does not exist, so if this is the reason for the failure of the quantum annealer on this problem, other ways to enforce constraints have to be designed to solve this kind of problems.

\begin{figure}[tbp]
	\centering
	\includegraphics[width=\columnwidth,keepaspectratio]{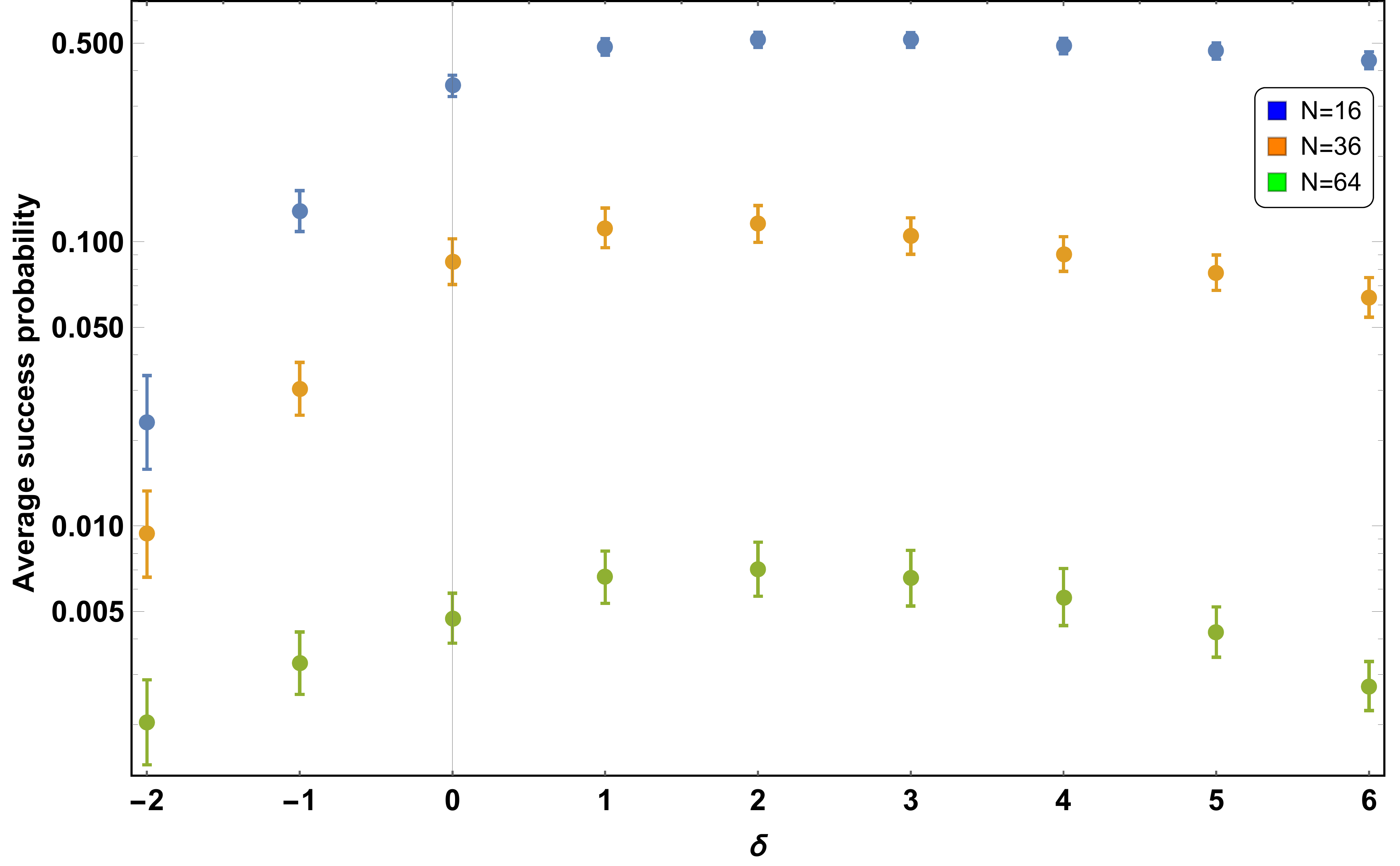}
	\caption{Average probability of finding the solution for the matching problem, shown at fixed N as function of the distance from the optimal parameter. Each point is obtained averaging on 100 different instances, and the probability is computed running the annealing $10^4$ times. These results are obtained using the D-Wave 2000Q quantum annealer hosted at NASA Ames Research Center.}\label{fig_tts_q1}
\end{figure}
\begin{figure}[tbp]
	\centering
	\includegraphics[width=\columnwidth,keepaspectratio]{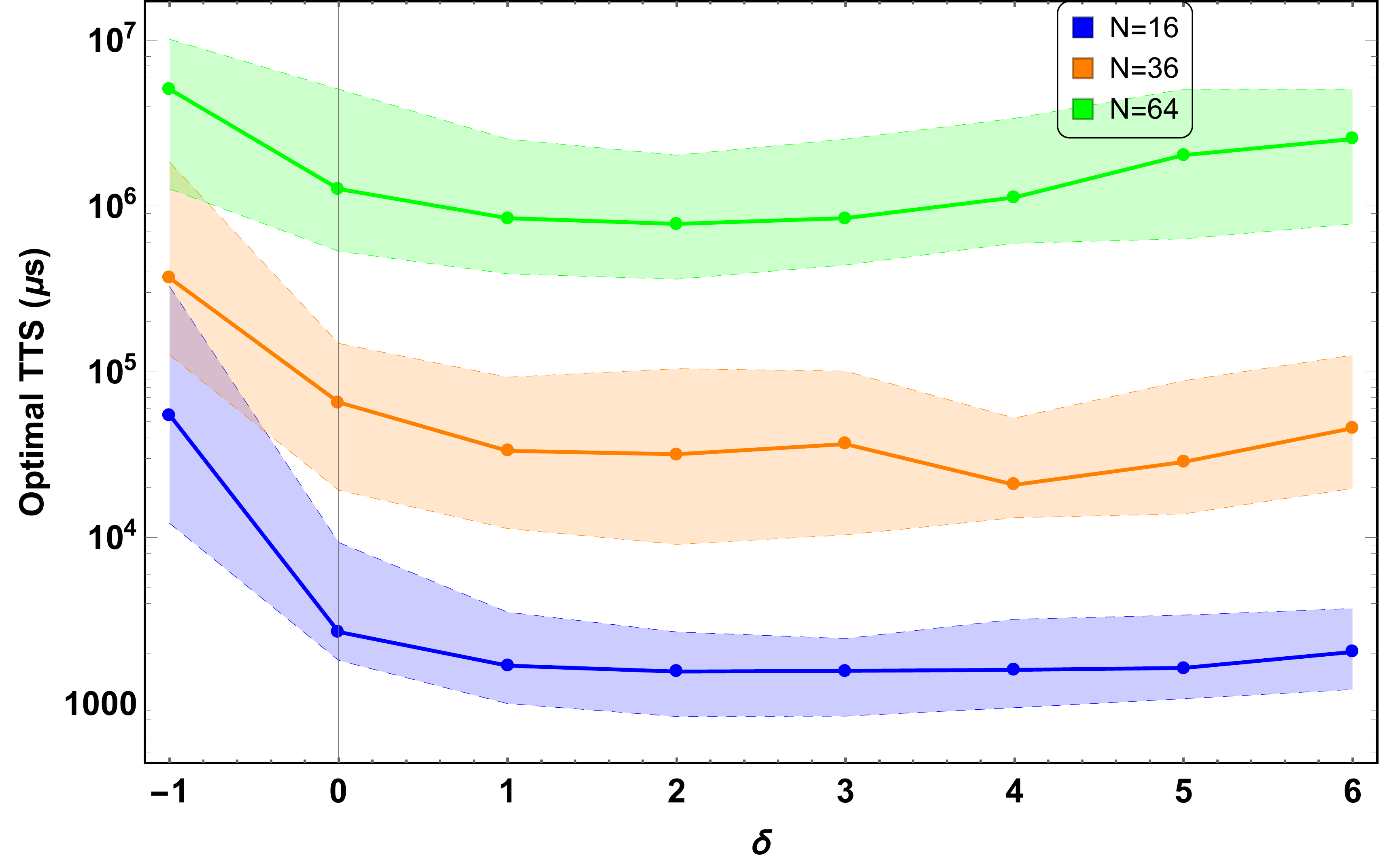}
	\caption{Time-to-solutions for the matching problem, shown at fixed N as function of the distance from the optimal parameter. These results are obtained starting from the same data set used for Fig. \ref{fig_tts_q1}. The solid lines connect points computed using the 50-percentile of instances, the dashed lines corresponds to the 35-percentile (below solid lines) and 65 percentile (above solid lines). These results are obtained using the D-Wave 2000Q quantum annealer hosted at NASA Ames Research Center.}\label{fig_tts_q2}
\end{figure}

\section{Quantum Approximate Optimization Algorithm}\label{sec::qaoa}
As a final section, we mention here briefly a relatively new algorithm to tackle COPs with gate-based quantum computers.

We start from the following question: how can we ``practically'' implement QAA on a gate-based quantum computer? Suppose we are able to apply the following gates
\begin{equation}\label{eq::qaoa_mixing}
	D(\beta) = e^{i \beta H_0}
\end{equation}
and 
\begin{equation}\label{eq::qaoa_phase}
	U(\gamma) = e^{i \gamma H_1}
\end{equation}
for generic values of $\beta$ and $\gamma$, where $H_0$ and $H_1$ are those used in Eq.~\eqref{eq::qaa_ham}. 
Therefore we can apply the Suzuki-Trotter formula,
\begin{equation}
	e^{X+Y} = \lim_{m\to\infty} \left( e^{X/m} e^{X/m} \right)^m,
\end{equation}
so that we write our evolution as\footnote{notice that the Suzuki-Trotter formula can be immediately generalized to time-ordered integrals.}
\begin{equation}
\begin{split}
	\text{Texp} & \left[- \frac{i}{\hbar} \int_0^1 dt \left( A(t) H_0 + B(t) H_1 \right) \right] \\
	&= \lim_{m\to \infty} \left( \text{Texp}\left[- \frac{i}{\hbar} H_0 \int_0^1 dt \frac{A(t)}{m} \right] \text{Texp}\left[- \frac{i}{\hbar} H_1 \int_0^1 dt  \frac{B(t)}{m} \right] \right)^m.
\end{split}	
\end{equation}
Therefore, to implement our QAA, we simply need to use alternatively the gates $D(\gamma)$ and $U(\beta)$, with very small parameters $\beta$ and $\gamma$, many times. 
After applying $p$ times each gate, we obtain the state
\begin{equation}\label{eq::qaa_trotterized}
	\ket{\bm{\beta}, \bm{\gamma}} = D(\beta_p) U(\gamma_p) \cdots D(\beta_1) U(\gamma_1) \ket{\psi}
\end{equation}
where $\ket{\psi}$ is our initial state, which, according to the QAA, has to be the ground state of $H_0$.
Now, the intuition behind QAOA is the following \cite{Farhi2014}: give up the idea of sending $p\to\infty$ and choose freely the sets $\{\beta_1, \dots, \beta_p\}$ and $\{\gamma_1, \dots, \gamma_p\}$. Notice that at this point we are far from the adiabatic situation (where $p\to\infty$ and $\beta_i$, $\gamma_i$ are small). Therefore we can choose the parameters such that
\begin{equation}
	E(\bm{\beta}, \bm{\gamma}) = \abs{\bra{\bm{\beta}, \bm{\gamma}} H_1 \ket{\bm{\beta}, \bm{\gamma}}}^2
\end{equation}
is maximized. Once the parameters are chosen, we can prepare the state by applying our sequence of gates and then measure the state in the computational basis: we are not guaranteed that the final state will be the ground state of our system (unless we have taken $p=\infty$), but because of our choice of the parameters we will end up in a low energy state with high probability. This algorithm, in a slightly generalized form that we will introduce in the remaining part of this section, is called \emph{quantum approximate optimization algorithm} (QAOA).\\
Since all the difficulty is in choosing the $2p$ parameters $\bm{\beta}$ and $\bm{\gamma}$, and we accept to be also very far from the adiabatic limit, usually the initial state is chosen to be the uniform superposition of all the states in the computational basis, independently from the Hamiltonian $H_0$.
In particular, a typical choice is
\begin{equation}
\ket{\psi}= \bigotimes_{i=1}^{N} \ket{+} 
\end{equation}
as initial state and
\begin{equation}
	H_0 = \sum_i X_i
\end{equation}
where $X_i$ is the Pauli matrix $X$ acting on the $i$-th spin. As $H_1$, the Hamiltonian of the COP one wants to solve is used. Notice again that this algorithm is very general, but for fixed $p$ no guarantees that the system ends up in the solution of the problem can be given. For this reason, this algorithm is mainly (but not only, see \cite{Farhi2016, Wecker2016, Jiang2017}) used for approximation purposes.

The operator $D$ defined in Eq.~\eqref{eq::qaoa_mixing} is often called \emph{mixing} operator, while $U$ given in Eq.~\eqref{eq::qaoa_phase} is called \emph{phase} operator. There is a lively line of research about the performances of QAOA, encouraged several findings:
\begin{itemize}
	\item in \cite{Farhi2014_2}, QAOA with $p=1$ is proved to be the best algorithm known to find approximate solutions for a particular COP, the so-called E2Lin2. Notice that shortly after that work,  a new classical algorithm which currently holds the record for this specific problem has been found \cite{Barak2015}. However, this proved the power of QAOA as an approximation algorithm;
	\item under complexity theory assumptions which are generally believed to be true (in the same sense that P is believed to be different from NP), it can be proved \cite{Farhi2016} that the output of a QAOA with $p\geq 1$ cannot be efficiently sampled with a classical algorithm;
	\item the Grover problem is simple enough to allow for an analytical treatment up to $1\ll p \ll N$ and it has been found \cite{Jiang2017} that a periodic choice of the parameters $\bm{\beta}$ and $\bm{\gamma}$ gives a quasi-optimal algorithm (that is, the algorithm requires $\alpha \log(N)$ queries to the oracle for large $N$, but $\alpha$ is higher than that of the Grover algorithm).
\end{itemize}

The interest in low-$p$ QAOA algorithms is also motivated by the recent availability of general-purpose gate-based quantum computers. Indeed shallow circuits as those needed to implement this kind of algorithms should be supported by devices available in the near future.


\chapter{Conclusions}\label{chap::final}
\section{A summary of our results}
In this thesis we reviewed many recent results in the realm of combinatorial optimization problems, mainly from the point of view of statistical mechanics. This perspective allowed us to compute the average value of the solution cost in random versions of many COPs, even in the presence of Euclidean correlations. 
Several original results stemmed during the research we have performed to complete this work:
\begin{itemize}
	\item we computed the average optimal length of the solution of the TSP in one dimension in complete bipartite \cite{Caracciolo2018_1} and complete \cite{Caracciolo2019_2} graphs;
	\item we analyzed a problem with an exponential number of possible solution even in one dimension, the 2-factor problem, and we have been able to provide also in this case bounds on its average solution cost \cite{Caracciolo2018_3};
	\item we discovered a novel application of the famous Selberg integrals to RCOPs, which allowed us to give (in some cases) exact prediction for finite $N$ \cite{Caracciolo2019_1};
	\item we extended our techniques and, thanks to a smart scaling argument, we have been able to compute the average optimal tour length for the NP-hard bipartite TSP and for the bipartite 2-factor problem in two dimensions \cite{Caracciolo2018_2}.
\end{itemize}

Other relevant results originated during this work have been obtained by leaving behind the study of averaged quantities, and focusing on the statistics of large (and very large) fluctuation, in the context of the p-spin spherical model, where we have proved that diverse and interesting regimes of large deviations are tuned by turning on and off an external magnetic field \cite{Pastore2019}.


Finally, we have discussed the possibility of using quantum computing to solve combinatorial optimization problems. 
The analysis of the various possible strategy to do that, and more specifically of the quantum annealing algorithm, has led to interesting findings about how to set the parameters of heuristic algorithms, which could be useful also for classical algorithms as simulated annealing and parallel tempering \cite{DiGioacchino2019}.

\section{Future directions}
As usually in Science, solving old problems has, as a side effect, the challenging result of opening new questions. 

Here we comment upon some possible paths that could start from this work and would (possibly) lead to further understanding of the statistical physics of combinatorial optimization problems with correlations:
\begin{itemize}
	\item Euclidean correlations have been successfully incorporated in the realm of random matrices, with the introduction of Euclidean Random Matrices \cite{Mezard1999,Amir2010,Goetschy2013}; several COPs can be rewritten as the computation of some property (for example, the determinant or the permanent) of a matrix, and therefore the random version of the problem is connected to the random matrix theory. Then we could devise a way to exploit the powerful and well-developed formalism of random matrix theory to tackle the computation of average properties of the solutions of combinatorial optimization problems;
	\item we have discussed here mainly Euclidean correlations. However, these are only one of the many kinds of correlations which in many cases occur in COPs. For example, the structure of real-world data used as input for many tasks, such as classification and feature extraction, is far from trivial and some efforts have been spent to explore the nature and the effect of these correlations \cite{Goldt2019, Erba2019, Rotondo2019}. In the same spirit, we could try to use the techniques developed in this work to deal with other kinds of correlations than the Euclidean ones, and to understand more deeply how correlations affect the ``difficulty'' of combinatorial optimization problems.
\end{itemize}

Our work on large deviations introduces new questions as well. One in particular, which deserves further investigations, regards the application of the magnetic field to obtain (analytically or numerically) the power of $N$ (some measure of the system size) in the exponential suppression of fluctuation. This is relevant for each problem where we find non-standard large deviation scalings with $N$, and would head towards a quite ``general'' method to compute this scaling.


Finally, much work has been done and much more remains to be done to really understand the potential of quantum computers to solve combinatorial optimization problems. While the investigations done here on the parameter setting for the quantum annealing algorithm can be considered a (small) step in that direction, it would be extremely interesting to be able to address the question of the effects of choosing different parameters in QAOA-type algorithms to solve, or approximate, combinatorial optimization problems.

\appendix

\chapter{Volume and surface of a sphere}\label{app::sphere}
Consider a $N$ dimensional sphere of radius $R$. Its volume is given by
\begin{equation}\label{eq::vol_Nsph-1}
	V_N = \int_{-\infty}^\infty d x_1 \cdots \int_{-\infty}^\infty d x_N \, \theta\left(R^2 - \sum_i x^2_i \right)	=  \Omega_N \int_{0}^R dr \, r^{N-1} = \frac{\Omega_N R^N}{N} ,
\end{equation}
where $\Omega_N$ is the surface area of the unit-radius sphere in $N$ dimensions (that is, the integral of all the $N-1$ angular variables in Eq.~\eqref{eq::vol_Nsph-1}).
The value of $\Omega_N$ can be computed with the following trick: consider the integral
\begin{equation}
	\int_{-\infty}^{\infty} dx_1 \cdots	\int_{-\infty}^{\infty} dx_N \, e^{-\sum_i x_i^2} = 	\left(\int_{-\infty}^{\infty} dx \, e^{-x^2}\right)^N = \pi^{N/2}.
\end{equation}
But we can also use spherical coordinates and write
\begin{equation}
	\int_{-\infty}^{\infty} dx_1 \cdots	\int_{-\infty}^{\infty} dx_N \, e^{-\sum_i x_i^2} = \Omega_N \int_{0}^{\infty} dr \, e^{-r^2} r^{N-1} = \frac{1}{2} \Omega_N \, \Gamma\left(\frac{N}{2}\right).
\end{equation}
Therefore
\begin{equation}
	\Omega_N = \frac{2 \pi^{N/2}}{\Gamma\left(\frac{N}{2}\right)}
\end{equation}
and
\begin{equation}
	V_N = \frac{2 \pi^{N/2}}{N} \frac{R^N}{\Gamma\left(\frac{N}{2}\right)}.
\end{equation}
Finally, as one can check from Eq.~\eqref{eq::vol_Nsph-1} (remembering that, in distributinal sense, $\frac{\partial}{\partial x}\theta(x-x_0)= \delta(x-x_0)$), we obtain the surface of a $N$ dimensional sphere of radius $R$ as derivative of its volume:
\begin{equation}
	S_N = \frac{\partial}{\partial R} V_{N} = 2 \pi^{N/2} \frac{R^{N-1}}{\Gamma\left(\frac{N}{2}\right)}.
\end{equation}

\chapter{Calculations for the p-spin spherical model}\label{app::pspin}
\chaptermark{p-spin model calculations}
\section{The replicated partition function}
The averaged replicated partition function of the p-spin spherical model is
\begin{equation}
	\overline{Z^n} = \int DJ  \int D\sigma \exp\left[ \beta \sum_{a=1}^n \sum_{i_1 < \cdots < i_p} J_{i_1\cdots i_p} \sigma_{i_1}^a \cdots \sigma_{i_p}^a \right],
\end{equation}
where
\begin{equation}
	\int DJ = \prod_{i_1<\dots<i_p} \int_{-\infty}^\infty d J_{i_1\cdots i_p} \, p(J_{i_1\cdots i_p})
\end{equation}
and
\begin{equation}
	\int D\sigma = \int_{-\infty}^\infty \prod_{a=1}^{n} \left( \prod_{i=1}^N  d\sigma_i^a \, \delta\left( N - \sum_i \sigma_{i}^a \right) \right).
\end{equation}
Notice that, as discussed in the main text, we are considering an \emph{integer} number of replicas $n\geq 1$. We now integrate over the disorder and, exploiting again Eq.~\eqref{eq::pspin_sum_approx} we get
\begin{equation}\label{eq::pspin_zn_1}
\begin{split}
	\overline{Z^n} & = \int D\sigma \prod_{i_1<\dots<i_p} \exp\left[ \frac{\beta^2}{4} N^{1-p} p!  \sum_{a,b=1}^n \sigma_{i_1}^a \sigma_{i_1}^b \cdots \sigma_{i_p}^a \sigma_{i_p}^b \right] \\
	& \sim \int D\sigma \exp\left[ \frac{\beta^2}{4} N^{1-p} \sum_{a,b=1}^n \left( \sum_{i=1}^N \sigma_{i}^a \sigma_{i}^b \right)^p \right].
\end{split}
\end{equation}
Now we introduce the ``overlap''\footnote{this name comes from the analogous step in problems with binary spin, where this quantity actually is a simple function of the number of spins pointing in the same directions in the replicas $a$ and $b$} between the replicas $a$ and $b$
\begin{equation}
	Q_{ab} = \frac{1}{N} \sum_i \sigma_{i}^a \sigma_{i}^b,
\end{equation}
which, simply by looking at the final form of Eq.~\eqref{eq::pspin_zn_1}, emerges as relevant variable. Clearly $Q_{ab}$ is a symmetric $n \times n$ matrix, and we have $Q_{aa} = 1$ for each $a$ because of the spherical constraint.
In order to rewrite our partition function in term of this new variable, we use the identity 
\begin{equation}
	1 = N^{n(n-1)/2} \int \prod_{a<b} d Q_{ab} \, \delta \left( N Q_{ab} - \sum_i \sigma_{i}^a \sigma_{i}^b \right)
\end{equation}
which we actually rewrite, using that (in distributional sense)
\begin{equation}
	\delta(x-x_0) = \int_{-\infty}^{\infty} \frac{dk}{2 \pi} \, e^{i k (x-x_0)},
\end{equation}
as
\begin{equation}
	1 = \int \left(\prod_{a<b} N d Q_{ab} \right) \int \left( \prod_{a<b} \frac{d \lambda_{ab}}{2 \pi} \right) \exp[i N \sum_{a<b} \lambda_{ab} Q_{ab} - \sum_i \sum_{a<b} \sigma_i^a \lambda_{ab} \sigma_i^b].
\end{equation}
As a last step, we write the spherical constraint as
\begin{equation}
\begin{split}
	\prod_a \delta\left( N - \sum_i \sigma_{i}^a \right) & = \int \left( \prod_a \frac{d\lambda_{aa}}{2 \pi} \right) \exp[i N \sum_a \lambda_{aa} - \sum_i \sum_{a} \sigma_i^a \lambda_{aa} \sigma_i^a] \\
	& = \int \left( \prod_a \frac{d\lambda_{aa}}{4 \pi} \right) \exp[i \frac{N}{2} \sum_a \lambda_{aa} - \frac{1}{2}\sum_i \sum_{a} \sigma_i^a \lambda_{aa} \sigma_i^a]
\end{split}
\end{equation}
and obtain
\begin{equation}\label{eq::pspin_to_compute_zn}
\begin{split}
	\overline{Z^n} & = \left( \prod_{a<b} \int_{-\infty}^{\infty} d Q_{ab} \right)  \left( \prod_{a\leq b} \int_{- i \infty}^{i\infty} d\lambda_{ab} \right) \left(\prod_{i=1}^N \int_{-\infty}^{\infty} d\sigma_{i} \right)  \cdot \\
	& \cdot \exp\left[\frac{N \beta^2}{4} \sum_{a,b} Q_{ab}^p + \frac{N}{2} \sum_{a, b} \lambda_{ab} Q_{ab} - \frac{1}{2} \sum_i \sum_{a, b} \sigma_{i}^a \lambda_{ab} \sigma_i^b \right],
\end{split}
\end{equation}
where $Q_{aa}=1$ (it is not an integration variable), we have written explicitly the integrations ranges to stress that we made change of variables to remove the factors $i$ from the exponents, and we dropped all the pre-factors, since they will not play a role in the final free entropy because they do not scale exponentially in $N$.
The integration over the spin variables can now be performed by using that
\begin{equation}
	\int \left( \prod_{i=1}^N d x_i \right) e^{-\frac{1}{2} \sum_{i,j} x_i A_{ij} x_j} = \sqrt \frac{(2 \pi)^N}{\det A},
\end{equation}
and we finally obtain Eqs.~\eqref{eq::pspin_zn_2} and \eqref{eq::pspin_zn_3} of the main text.
\section{1RSB free energy}
Given the 1RSB ansatz for $Q$, Eq.~\eqref{eq::pspin_1rsb_q}, to obtain the free energy in terms of the variational parameters $q_0,q_1,m$ we need to compute $\frac{1}{n}\sum_{a,b} Q_{ab}^p$ and $ \frac{1}{n} \log\det Q$, and take the limit of small $n$.\\
About the first part, we have that $Q$ has $n$ entries equal to 1 on the diagonal, $m(m-1)$ entries equal to $q_1$ for each block for a total of $n(m-1)$ entries equal to $q_1$, and the remaining $n^2 - n m$ entries equal to $q_0$, therefore\footnote{ this is another weird thing of the replica trick: we are sending $n$ to zero, but we are also supposing $0<m<n$ and we do not want to send $m$ to zero. Actually, the correct thing to do is to suppose that when $n\to 0$ we obtain for $m$ the relations $0\leq m \leq 1$.}:
\begin{equation}\label{eq::pspin_app_1rsb_1}
	\frac{1}{n} \sum_{a,b} Q_{ab}^p = 1 + (m-1) q_1^p + (n - m ) q_0^p \sim 1 + (m-1) q_1^p + m q_0^p.
\end{equation}
The second piece is slightly more tricky: one has to first notice that $[ \mathbb{E}, \mathbb{C} ]=0$, so a single orthogonal matrix such that both $\mathbb{E}$ and $\mathbb{C}$ are diagonalized exists. Now, notice that $(1/m)\mathbb{E}$ and $(1/n)\mathbb{C}$ are both projector. $(1/m)\mathbb{E}$ projects on the subspace of $\mathbb{R}^n$ generated by vectors with the first, second and so on groups of $m$ components equal, $(1/n)\mathbb{C}$ on the subspace generated by the constant vector. 
This observation makes clear that the eigenvalues of $\mathbb{C}$ are $n$ with degeneracy 1 and 0 with degeneracy $n-1$, while $\mathbb{E}$ has 1 eigenvalue equal to $m$ and $m-1$ eigenvalues equal to 0 for each block. Finally, since the constant vector is eigenvector of both $\mathbb{E}$ and $\mathbb{C}$ with eigenvalues $m$ and $n$ respectively, we have that the matrix Q has, as eigenvalues:
\begin{itemize}
	\item $1-q_1 + (q_1-q_0) m + q_0 n$ with degeneracy 1, corresponding to the non-null eigenvalue of $\mathbb{C}$;
	\item $1-q_1 + (q_1-q_0) m$ with degeneracy $n/m-1$, corresponding to the other $n/m-1$ non-null eigenvalues of $\mathbb{E}$;
	\item $1-q_1$ with degeneracy $n-n/m$, corresponding to the null eigenvalues of both $\mathbb{C}$ and $\mathbb{E}$.
\end{itemize}
Therefore we have
\begin{equation}\label{eq::pspin_app_1rsb_2}
\begin{split}
	\frac{1}{n} \log \det Q = & \frac{m - 1}{m} \log(1-q_1) + \frac{1}{m} \log(m (q_1-q_0) + 1- q_1) +\\
	& + \frac{1}{n} \left[- \log(m (q_1-q_0) + 1- q_1) + \log(1-q_1 + (q_1-q_0) m + q_0 n) \right] \\
	= & \frac{m - 1}{m} \log(1-q_1) + \frac{1}{m} \log(m (q_1-q_0) + 1- q_1) +\\
	& + \frac{1}{n} \left[\log(1 + \frac{n q_0}{1-q_1 + (q_1-q_0) m})\right] \\
	\sim & \frac{m - 1}{m} \log(1-q_1) + \frac{1}{m} \log(m (q_1-q_0) + 1- q_1) + \\
	& + \frac{q_0}{m(q_1-q_0) + 1 - q_1}.
\end{split}
\end{equation}
Using Eqs.~\eqref{eq::pspin_app_1rsb_1} and \eqref{eq::pspin_app_1rsb_2} in Eq.~\eqref{eq::pspin_f}, we obtain Eq.~\eqref{eq::pspin_f_1rsb}.

\section{Rammal construction}\label{app::paperld_rammal}
In this appendix we report the details of the geometrical construction reproducing the solution for the SCGF obtained with a 1RSB ansatz with $q_0 = 0$. The following observations are traced back to Rammal's work \cite{Rammal1981} and can be found in \cite{Kondor1983} (similar considerations in \cite{Ogure2004,Nakajima2008,Nakajima2009}). We reproduce here the reasoning not only as an historical curiosity: first of all, we see it as an enlightening approach to the problem of the continuation of the replicated partition function to real number of replicas, particularly suitable for a finite $k$ analysis. Moreover, we note that this interpretation, whenever it works, gives a flavor of ``uniqueness'' (though not in a strict mathematical sense) to the resulting solution, being based only on the properties of convexity and extremality that the SCGF $\psi(k)$ must have. In this respect, a generalization of this result would be of great interest in order to better understand the necessity of Parisi hierarchical RSB procedure, which has been dubbed as ``magic'' even in relatively recent works, like \cite{Dotsenko2011}; however, a true geometrical interpretation of the full machinery of RSB, beyond the simple case considered here, still lacks. Finally, in the context of this paper we are able to show a case where the construction gives the correct answer (the $p$-spin spherical model at zero external magnetic field) and a case where it fails (when the field is switched on).


The explicit evaluation of the SCGF $\psi(k)$ is performed within replica theory: an ansatz is imposed on the form of the replica overlap matrix, the number of replicas $k$ is then continued from integer to real values, the corresponding $G(k)$ is evaluated with the saddle-point method for large $N$ and finally a check is performed \textit{a posteriori} to verify its validity. In the SK model, the system originally considered by Rammal, at low temperatures the replica symmetric ansatz, which still gives the correct values of the positive integer momenta of the partition function, fails to produce a sensible solution for the SCGF at $k<1$, in at least three way:
\begin{itemize}
	\item it becomes unstable under variations around the saddle point (de Almeida-Thouless instability~\cite{deAlmeida1978}) below $k=k_{\text{dAT}}$;
	\item it produces a $G(k)$ that is non-concave (and so a non-convex $\psi(k)$) around $k=k_{\text{conv}}$, meaning that $G''(k)$ changes sign at $k_{\text{conv}}$;
	\item it produces a $G(k)/k$ that loses monotonicity a $k=k_m$.
\end{itemize}
In the SK model $k_{\text{dAT}}$ is the largest ($k_{\text{dAT}}>k_m>k_{\text{conv}}$), and so it is the first problem one encounters in extrapolating the RS solution from integer values of $k$. However, from the point of view of convexity and monotonicity alone, Rammal proposed to build a marginally monotone $G(k)/k$ in a minimal way, starting from the RS and simply keeping it constant below $k_m$ at the value $G(k_m)/k_m$. While the resulting function is not the correct one for the SK model, which needs a full RSB analysis to be solved, surprisingly enough for the spherical $p$-spin in zero magnetic field this approach reproduces the solution obtained with a 1RSB ansatz with $q_0 = 0$ (see Fig.~\ref{fig::paperld_G(k)/k}). Notice that in the present model the RS solution suffers from the same inconsistencies as in the SK model, but now $k_m$ is the largest of the three problematic points.

To convince the reader that the two approaches are actually equivalent we prove, as final part of this appendix, that without an external magnetic field the 1RSB solution of the spherical $p$-spin and the 
Rammal construction coincide. 
In order to obtain this result, we have to prove that:
\begin{itemize}
	\item the 1RSB solution for $G(k)/k$ becomes a constant below $k=k_c$, which is defined as the point where the RS and 1RSB ans\"atze branch out, as we did in the main text;
	\item this constant is the same as the one in the Rammal construction, that is $G(k_m)/k_m$;
	\item the points $k_c$ and $k_m$ are the same.
\end{itemize}
As $k_c$ is the point where the RS solution is not optimal anymore, for $k<k_c$ we have $\bar{q}_0=0$, as discussed in \cite{Crisanti1992}. Let us now consider Eq.~\eqref{eq::paperld_g_h_0} with $q_0=0$: differentiating with respect to $q_1$ and $m$ and setting the results equal to 0 we get the equations for $\bar{q}_1$ and $\bar{m}$, which read

\begin{equation}\label{eq::appr_sp1}
\left\{
\begin{aligned}
& \mu \, \bar{q}_1^{p-2} -  \frac{1}{(1-\bar{q}_1)(1-(1-\bar{m} )\bar{q}_1)} = 0\\
& \frac{\mu}{2} \bar{m}^2 \bar{q}_1^p - \frac{1}{2} \log \left( 1+\frac{\bar{m}\,  \bar{q}_1}{1 - \bar{q}_1} \right)  + \frac{\bar{m}}{2} \frac{\bar{q}_1}{1-(1-\bar{m})\bar{q}_1} = 0
\end{aligned}
\right.
\end{equation}
where $\mu = p (\beta J)^2/2$. These equations can be solved numerically (as we did to obtain the plots in the main text), but to show our point here we do not really need the explicit solution. Indeed it is enough to notice that $\bar{m}$ and $\bar{q}_1$ do not depend on $k$ and therefore $g(k; 0, \bar{q}_1, \bar{m})/k$ is a constant.
Then, we need to check that it is the same constant as the one obtained by Rammal. Again starting from Eq.~\eqref{eq::paperld_g_h_0}, by putting $q_1=q_0=q$ we obtain the RS solution, which is
\begin{equation}
\begin{split}
g_{0} (k;q) = &- \frac{(\beta J )^2}{4} \left[k + k (k-1) q^p\right] - \frac{k-1}{2} \log(1-q)\\
&- \frac{1}{2} \log\left[ 1 - (1-k) q \right] - k s(\infty).
\end{split}
\end{equation}
In this case, extremizing with respect to $q$, we have an equation which gives the RS solution on the saddle point, $\bar{q}$. To find $k_m$, we then require $\frac{\partial}{\partial k} g_{0} /k  = 0$. The two resulting equations are:
\begin{equation}\label{eq::appr_sp2}
\left\{
\begin{aligned}
&     \mu \, \bar{q}^{p-2} -  \frac{1}{(1-\bar{q})\left[1-(1-k_m )\bar{q}\right]} = 0\\
& \frac{\mu}{2} k_m^2 \bar{q}^p - \frac{1}{2} \log \left( 1+\frac{k_m\,  \bar{q}}{1 - \bar{q}} \right)  
+ \frac{k_m}{2} \frac{\bar{q}}{1-(1-k_m)\bar{q}}  = 0
\end{aligned}
\right.
\end{equation}
that are exactly Eqs.~\eqref{eq::appr_sp1} with $k_m$ instead of $\bar{m}$ and $\bar{q}$ instead of $\bar{q}_1$. Therefore $k_m = \bar{m}$ and $\bar{q} = \bar{q}_1$ and one can check that
\begin{equation}
\frac{g(k; 0, \bar{q}, k_m)}{k} = \frac{g_{0} (k_m, q)}{k_m} \,.
\end{equation}
It only remains to prove that $k_c$ and $k_m$, which in general can be different points, are actually the same. As the 1RSB ansatz gives the correct solution for the present model, the corresponding SCGF must be convex and thus, in particular, continuous. The only way to obtain a continuous function which is equal to the RS one above $k_c$ and to the Rammal's constant below, is to take $k_c = k_m$, and so the two functions coincide everywhere.

\section{1RSB with magnetic field}\label{app::pspin_ld}

To obtain \eqref{eq::paperld_g_h_non_0}, the starting point is the p-spin Hamiltonian with magnetic field
\begin{equation}
	H = - \sum_{i_1<i_2<\cdots<i_p} J_{i_1\cdots i_p} \sigma_{i_1} \cdots \sigma_{i_p} - h \sum_i  \sigma_{i}.
\end{equation}
In the presence of a magnetic field, one can perform the average over the disorder and proceed exactly as discussed in  equation exactly as as discussed in this Appendix for the case $h=0$.
The only difference is that in Eq.~\eqref{eq::pspin_to_compute_zn}, the Gaussian integral on the spin degrees of freedom has a linear term. Using the Gaussian integral
\begin{equation}
\int \left( \prod_{i=1}^N d x_i \right) e^{-\frac{1}{2} \sum_{i,j} x_i A_{ij} x_j + \sum_i b_i x_i} = \sqrt \frac{(2 \pi)^N}{\det A} e^{\frac{1}{2} \sum_{i,j} b_i (A^{-1})_{ij} b_j},
\end{equation}
and taking also the extra term into account, we obtain
\begin{equation}
\overline{Z^k} = e^{k N \log(2\pi)/2} \int DQ \, D\lambda \, e^{-N S(Q,\lambda)},
\end{equation}
as Eq.~\eqref{eq::pspin_zn_2}, with
\begin{multline}\label{eq:G_1RSB_full}
S(\mathbf{q},\boldsymbol{\lambda}) = -\frac{\beta^2}{4}\sum_{a, b=1}^k Q_{a b}^p -\frac{1}{2}\sum_{a, b = 1}^k \lambda_{a b} Q_{a b} +\frac{1}{2}\log \det \left(\lambda\right) - \frac{(\beta h)^2}{2}\sum_{a, b=1}^k \left(\lambda^{-1}\right)_{a b}\,.
\end{multline}
Derivation with respect to $\lambda_{a\beta}$ leads to the following saddle-point equations:
\begin{equation}
Q_{a b} - \left(\lambda^{-1}\right)_{a b} -(\beta h)^2 \sum_{\gamma,\delta=1}^k \left(\lambda^{-1}\right)_{\gamma a} \left(\lambda^{-1}\right)_{b \delta} = 0\,,
\label{eq:splambda}
\end{equation}
where we have used the identity
\begin{equation}
\frac{\partial \left(\bm{\lambda}^{-1}\right)_{\gamma \delta}}{\partial\lambda_{a b}} =- \left(\lambda^{-1}\right)_{\gamma a} \left(\lambda^{-1}\right)_{b \delta}\,.
\end{equation}
Equations \eqref{eq:splambda} are solved via successive contractions of the replica indices: a double summation over $a$, $b$ leads to an equation for the scalar $\sum_{a b}\left(\lambda^{-1}\right)_{a b}$ with solutions:
\begin{equation}
\sum_{a, b=1}^k\left(\lambda^{-1}\right)_{a b} = \frac{- 1\pm\sqrt{1+4(\beta h)^2 Q_s}}{2(\beta h)^2} \equiv l_{\pm}(Q_s)\,, \quad  Q_s=\sum_{a b}Q_{a b}\,.
\end{equation}
Similarly a single contraction gives (remember that $\lambda$, and so also its inverse, is symmetric):
\begin{equation}
\sum_{a}\left(\lambda^{-1}\right)_{a b} = \frac{\sum_{a} Q_{a b}}{1 + (\beta h)^2 l_{\pm}(Q_s)} = \frac{Q_r}{1+(\beta h)^2 l_{\pm}(Q_s)},
\end{equation}
where we defined $Q_r=\sum_{a} Q_{a b}$
and finally 
\begin{equation}
\left(\lambda^{-1}\right)_{a b}  = Q_{a b} - \frac{(\beta h)^2 \sum_{\gamma} Q_{\gamma a}\sum_{\delta} Q_{\delta b}}{\left[1 + (\beta h)^2 \, l_{\pm}(Q_s)\right]^2}\,.
\end{equation}
To find which of the signs of $l_\pm$ is the right one, we can consider the limit $k\to 0$ of this result. We obtain that
\begin{equation}
Q_s \to 0\,,\quad  Q_r \to 1 + (m - 1) q_1 - m q_0 \,, \quad l_\pm(Q_s) \to l_\pm(0) = \begin{cases}
- 1/(\beta h)^2 \,,\\
0\,.
\end{cases}
\end{equation}
The only finite limit is for $\hat{q}_-$, for which, in the $k\to0$ limit, we recover the result of \cite{Crisanti1992}.

Given the 1RSB ansatz Eq.~\eqref{eq::pspin_1rsb_q}, $Q_{ab}$ has $k$ elements 1 on the diagonal, $m (m-1) k/m$ elements $q_1$ in the diagonal blocks, the remaining $k^2 - k - k(m-1)$ elements $q_0$, so
\begin{equation}
Q_s = k + k (m - 1) q_1 + k(k - m) q_0
\end{equation}
Every row (column) contains the same elements, so
\begin{equation}
Q_r = 1 + (m - 1) q_1 + (k - m) q_0  \qquad \forall\, a\,.
\end{equation}
In this way, we arrive at the solution of the saddle-point equations for $\bm{\lambda}$:
\begin{equation}
\left({\lambda}^{-1}\right)_{a b}  = Q_{a b} - \hat{q}_-\,.
\end{equation}
where we have defined
\begin{equation}\label{eq::paperld_hatq}
\hat{q}_- = \frac{(\beta h)^2 Q_r^2}{\left[1 + (\beta h)^2 l_{-}(Q_s)\right]^2}\,.
\end{equation}
The structure of the matrix $\lambda^{-1}$ is therefore the same as the one of $Q$, with a constant added to each entry. Thus, the entries of ${\lambda}^{-1}$
can be written as
\begin{equation}
\lambda^{-1} = (1- q_1) \, \mathbb{I} + (q_1 - q_0) \, \mathbb{E} + (q_0 - \hat{q}_{-}) \,  \mathbb{C}.
\end{equation}
Thanks to this equation, we can compute one of the terms containing $\lambda$ in Eq.\eqref{eq:G_1RSB_full}:
\begin{equation}
	\sum_{ab} (\lambda^{-1})_{ab} = k (1- q_1) + k m  (q_1 - q_0) +  (q_0 - \hat{q}_{-}) k^2 = k (\eta_2 - k \hat{q}_{-}).
\end{equation}

Exploiting the fact that $\mathbb{E}/m$ and $\mathbb{C}/k$ are projectors, and using that $\mathbb{E} \cdot \mathbb{C} = \mathbb{C} \cdot \mathbb{E} = m \mathbb{C}$, one can check that the inverse of a matrix with the 1RSB structure is again a matrix with the same structure. In particular, we obtain:
\begin{equation}
\lambda = \frac{1}{\eta_0}  \, \mathbb{I} + \frac{q_0 - q_1}{\eta_0 \eta_1}  \, \mathbb{E} + \frac{\hat{q}_- - q_0}{\eta_1\left(\eta_2 - k \hat{q}_-\right)}   \,  \mathbb{C}
\end{equation}
where $\eta_0 = 1-q_1$, $\eta_1 = 1 - (1-m) q_1- m q_0$ and $\eta_2 = 1 - (1-m) q_1- (m-k) q_0 $ are the three different eigenvalues of $Q$.
Then, as we have done for $Q$, we can compute the eigenvalues of $\lambda$:
\begin{equation}
\begin{aligned}
\kappa_0 &= 1/\eta_0  &  & \text{deg.} = k(m-1)/m \,,\\
\kappa_1 &= 1/\eta_1 & & \text{deg.} = k/m-1 \,,\\
\kappa_2 &= 1/\left(\eta_2 - k \hat{q}_- \right) \qquad & & \text{deg.} = 1 \,.
\end{aligned}
\end{equation}
Using the eigenvalues, we can easily compute, as we did for the $h=0$ case, the term $\log\det(\lambda)$.

The last step is to evaluate the trace appearing in Eq.~\eqref{eq:G_1RSB_full}, by using again the properties of the matrices $\lambda$ and $Q$:
\begin{equation}
\Tr \left({\lambda} \cdot Q \right) = k \left(1 + \frac{\hat{q}_-}{ \eta_2 -  k \hat{q}_{-}} \right)\, .
\end{equation}
Now we have all the ingredients to write Eq.~\eqref{eq::paperld_g_h_non_0}, which has the expected limit for $k\to 0$.

\chapter{Supplemental material to Chapter \ref{chap::second}}
\chaptermark{Supplemental material to Chapter 3}
\section{Order statistics}\label{app::orderstat}
In this section we collect some techniques useful to perform averages on $N$ random points uniformly distributed on a segment of length $L$.
As first step, we notice that, if $f(x_1,\dots, x_N)$ is a symmetric function of its arguments
\begin{equation}
\begin{split}
	\int_0^L dx_1 \cdots & \int_0^L dx_N \, f(x_1,\dots, x_N) \\
	& = N! \int_{0}^{x_2} dx_1 \int_{x_1}^{x_3} dx_2 \cdots \int_{x_{N-2}}^{x_N} dx_{N-1}  \int_{x_{N-1}}^{L} dx_N \, f(x_1,\dots, x_N),
\end{split}
\end{equation}
because of the symmetry of $f$.
Therefore, we can use this result as follows
\begin{equation}
	\int_{0}^{x_2} dx_1 \cdots \int_{x_{k-2}}^{x_k} dx_{k-1} = \frac{1}{(k-1)!} \int_{0}^{x_{k}} dx_1 \cdots \int_{0}^{x_k} dx_{k-1} = \frac{x_k^{k-1}}{(k-1)!}
\end{equation}
and, similarly,
\begin{equation}
	\int_{x_\ell}^{x_{\ell+2}} dx_{\ell+1} \cdots \int_{x_{N-1}}^{L} dx_{N} = \frac{1}{(N-\ell)!} 	\int_{x_\ell}^{L} dx_{\ell+1} \cdots \int_{x_{N-1}}^{L} dx_{N} = \frac{(L - x_\ell)^{N-\ell}}{(N-\ell)!}
\end{equation}
or, more generally
\begin{equation}
	\int_{a}^{x_{m+1}} dx_m \cdots \int_{x_{m+t-1}}^{b} dx_{m+t} = \frac{(b-a)^{t+1}}{(t+1)!}.
\end{equation}
From this we can compute several useful quantities. For example, given $N$ points randomly chosen in the interval [0, 1] and labeled such that they are ordered, $\{x_1 \dots, x_N \}$, the probability that the $k$-th point is in the interval $[x, x+dx]$ is given by the probability that $k-1$ positions are smaller than $x_k$ and $N-k$ positions are larger than $x_k$, then
\begin{equation}\label{eq::ordstat_pk}
\begin{split}
	P_k(x) \, d x & = N! \int_0^{x_2} dx_1 \cdots \int_{x_{k-2}}^{x_k} dx_{k-1} \, d x \int_{x_k}^{x_{k+1}} dx_{k+1} \cdots \int_{x_{N-1}}^1 dx_{N}\\ 
	& = \frac{\Gamma(N+1)}{\Gamma(k) \, \Gamma(N-k+1)} \, x^{k-1} (1-x)^{N-k} d x.
\end{split}
\end{equation}
In a similar way, the probability that the $k$-th point is in $[x,x+dx]$ and the $\ell$-th is in $[y, y+dy]$ is ($k<\ell$ and so $x<y$)
\begin{equation}\label{eq::ordstat_pk_2}
\begin{split}
	P_{k,\ell}(x, y) \, d x \, d y & = N! \int_0^{x_2} dx_1 \cdots \int_{x_{k-2}}^{x_k} dx_{k-1} \, d x \int_{x_k}^{x_{k+1}} dx_{k+1} \cdots\\
	& \quad \cdots \int_{x_{\ell-2}}^{x_\ell} dx_{\ell-1} \, dy \int_{x_\ell}^{x_{\ell+2}} dx_{\ell+1} \cdots \int_{x_{N-1}}^1 dx_{N}\\
	& = \frac{\Gamma(N+1)}{\Gamma(k)\Gamma(\ell-k)\, \Gamma(N-\ell+1)} x^k (y-x)^{\ell - k - 1} (1-y)^{N-\ell}
\end{split}
\end{equation}

\section{Proofs for the traveling salesman problems}\label{app::tsp}

\subsection{Optimal cycle on the complete bipartite graph}\label{app::tsp_bip}
Consider the tour $\tilde{h}$ given by the Eqs.~\eqref{eq::tsp_bip_1d_even} for $N$ even and Eqs.~\eqref{eq::tsp_bip_1d_odd} for $N$ odd.
We will prove now that this cycle is optimal. To do this, we will suggest two moves that lower the energy of a tour and showing that the only Hamiltonian cycle that cannot be modified by these moves is $\tilde{h}$.

We shall make use of the following moves in the ensemble of Hamiltonian cycles.
Given $i, j\in [N]$ with $j>i$ we can partition each cycle as
\begin{equation}
h[(\sigma, \pi)] = (C_1 r_{\sigma(i)} b_{\pi(i)} C_2 b_{\pi(j)} r_{\sigma(j+1)} C_3),
\end{equation}
where the $C_i$ are open paths in the cycle, and we can define the operator $R_{ij}$ that exchanges two blue points $b_{\pi(i)}$ and $b_{\pi(j)}$ and reverses the path between them as
\begin{align}
\begin{split}
h[R_{ij} (\sigma, \pi)] &:=   (C_1 r_{\sigma(i)} [b_{\pi(i)} C_2 b_{\pi(j)}]^{-1} r_{\sigma(j+1)} C_3) \\
& = (C_1 r_{\sigma(i)} b_{\pi(j)} C_2^{-1} b_{\pi(i)} r_{\sigma(j+1)} C_3)  \,.
\end{split}
\end{align}
Analogously by writing
\begin{equation}
h[(\sigma, \pi)] = (C_1 b_{\pi(i-1)} r_{\sigma(i)}  C_2  r_{\sigma(j)} b_{\pi(j)} C_3)
\end{equation}
we can define the corresponding operator $S_{ij}$ that exchanges two red points $r_{\sigma(i)}$ and $r_{\sigma(j)}$ and reverses the path between them
\begin{align}
\begin{split}
h[S_{ij} (\sigma, \pi)] &:=   (C_1 b_{\pi(i-1)} [r_{\sigma(i)} C_2 r_{\sigma(j)}]^{-1} b_{\pi(j)} C_3) \\
& = (C_1  b_{\pi(i-1)} r_{\sigma(j)} C_2^{-1}  r_{\sigma(i)} b_{\pi(j)} C_3) \, .
\end{split}
\end{align}
Two couples of points $(r_{\sigma(k)}, r_{\sigma(l)})$ and $(b_{\pi(j)}, b_{\pi(i)})$ have the same orientation 
if $(r_{\sigma(k)} -r_{\sigma(l)})(b_{\pi(j)}-b_{\pi(i)}) >0$. Remark that as we have ordered both set of points this means also that $(\sigma(k), \sigma(l))$ and $(\pi(j), \pi(i))$ have the same orientation.

Then
\begin{lemma}\label{tsp_bip_lemma1}
	Let $E[(\sigma, \pi)]$ be the cost defined in Eq.~\eqref{eq::tsp_bip_costgen}. Then 
	$E[R_{ij}(\sigma, \pi)] - E[(\sigma, \pi)] >0$ if the couples $(r_{\sigma(j+1)}, r_{\sigma(i)})$ and $(b_{\pi(j)}, b_{\pi(i)})$ have the same orientation and
	$E[S_{ij}(\sigma, \pi)] - E[(\sigma, \pi)] >0$ if the couples $(r_{\sigma(j)}, r_{\sigma(i)})$ and $(b_{\pi(j)}, b_{\pi(i-1)})$ have the same orientation.
\end{lemma}
\begin{proof}
	\begin{multline}
	\begin{split}
	E[R_{ij}(\sigma, \pi)] - E[(\sigma, \pi)] &= w_{(r_{\sigma(i)}, b_{\pi(j)})}  + w_{(b_{\pi(i)}, r_{\sigma(j+1)})} \\
	&- w_{(r_{\sigma(i)}, b_{\pi(i)}) }- w_{(b_{\pi(j)}, r_{\sigma(j+1)})}
	\end{split}
	\end{multline}
	and this is the difference between two matchings which is positive if the couples $(r_{\sigma(j+1)}, r_{\sigma(i)})$ and $(b_{\pi(j)}, b_{\pi(i)})$ have the same orientation (as shown in~\cite{McCannRobert1999, Caracciolo2014} for a weight which is an increasing convex function of the Euclidean distance). 
	The remaining part of the proof is analogous.
\end{proof}

\begin{lemma}
	The only couples of permutations $(\sigma,\pi)$ with $\sigma(1)=1$ such that both $(\sigma(j+1), \sigma(i))$ have the same orientation  as $(\pi(j), \pi(i))$  and $(\pi(j), \pi(i-1))$ and $(\sigma(j), \sigma(i))$, for each $i, j\in [N]$ are $(\tilde{\sigma},\tilde{\pi})$ and its dual $(\tilde{\sigma},\tilde{\pi})^\star$.
\end{lemma}
\begin{proof}
	We have to start  our Hamiltonian cycle from $r_{\sigma(1)} = r_1$.  
	Next we look at $\pi(N)$,  if we assume now that $\pi(N)>1$, there will be  a $j$ such that our cycle would have the form $(r_1 C_1 r_{\sigma(j)} b_1 C_2   b_{\pi(N)})$, if we assume $j>1$ then $(1, \sigma(j))$ and $(\pi(N),1)$ have opposite orientation, so that  necessarily $\pi(N)=1$. In the case $j=1$  our Hamiltonian cycle is of the form $(r_1 b_1 C)$, that is $(b_1 C r_1)$, and this is exactly of the other form if we exchange red and blue points.
	We assume that it is of the form $(r_1 C b_1)$; the other form would give, at the end of the proof, $(\tilde{\sigma},\tilde{\pi})^\star$.
	\\
	Now we shall proceed by induction. Assume that our Hamiltonian cycle is of the form $(r_1 b_2 r_3 \cdots x_k C y_k \cdots b_3 r_2 b_1)$ with $k<N$, where $x_k$ and $y_k$ are, respectively, a red point and a blue point when $k$ is odd and  viceversa when $k$ is even. 
	Then $y_{k+1}$ and $x_{k+1}$ must be in the walk $C$. 
	If $y_{k+1}$ it is not the point on the right of $x_k$ the cycle has the form 
	$(r_1 b_2 r_3 \cdots x_k y_s C_1 y_{k+1} x_l \cdots y_k \cdots b_3 r_2 b_1)$ 
	but then $(x_l , x_k)$ and $(y_{k+1}, y_s)$ have opposite orientation, 
	which is impossible, so that 
	$s=k+1$, that is the point on the right of $x_k$. Where is $x_{k+1}$? If it is not the point on the left of $y_k$ the cycle has the form $(r_1 b_2 r_3 \cdots x_k y_{k+1} \cdots y_l x_{k+1} C_1 x_s \cdots y_k \cdots b_3 r_2 b_1)$, but then $(x_s, x_{k+1})$ and $(y_k, y_l)$ have opposite orientation, which is impossible, so that $s =k+1$, that is the point on the left of $y_k$. We have now shown that the cycle has the form  $(r_1 b_2 r_3 \cdots y_{k+1} C x_{k+1} \cdots b_3 r_2 b_1)$ and can proceed until $C$ is empty.
\end{proof} 

Now that we have understood what is the optimal Hamiltonian cycle, we can look in more details at what are the two matchings which enter in the decomposition we used in Eq.~\eqref{eq::tsp_bip_dec}.
As $\tilde{\pi} = \tilde{\sigma} \circ I$ we have that
\begin{equation}
I = \tilde{\sigma}^{-1} \circ \tilde{\pi} = \tilde{\pi}^{-1} \circ \tilde{\sigma}.
\end{equation}
As a consequence both permutations associated to the matchings appearing inEq.~\eqref{eq::tsp_bip_dec} for the optimal Hamiltonian cycle are involutions:
\begin{subequations}
	\begin{align}
	\begin{split}\label{eq::tsp_bip_m1}
	\tilde{\mu}_1 \equiv \tilde{\pi} \circ \tilde{\sigma}^{-1} & = \tilde{\sigma} \circ I \circ \tilde{\sigma}^{-1} = \tilde{\sigma} \circ  \tilde{\pi}^{-1} \\
	& = \left[ \tilde{\pi} \circ \tilde{\sigma}^{-1}\right]^{-1}
	\end{split}
	\\
	\begin{split}\label{eq::tsp_bip_m2}
	\tilde{\mu}_2 \equiv \tilde{\pi} \circ \tau^{-1} \circ \tilde{\sigma}^{-1} & = \tilde{\sigma} \circ I \circ \tau^{-1} \circ I \circ \tilde{\pi}^{-1} \\ 
	& = \left[\tilde{\pi} \circ \tau^{-1} \circ \tilde{\sigma}^{-1} \right]^{-1}  ,
	\end{split}
	\end{align}
\end{subequations}
where we used Eq.~\eqref{eq::tsp_bip_Itau}.
This implies that those two permutations have at most cycles of period two, a fact which reflects a symmetry by exchange of red and blue points.

When $N$ is odd it happens that
\begin{equation}
I \circ \tilde{\sigma} \circ I = \tilde{\sigma} \circ \tau^{-\frac{N-1}{2}},
\end{equation}
so that
\begin{align}
\begin{split}
I \circ \tilde{\pi} \circ I & = I \circ \tilde{\sigma} \circ I \circ I = \tilde{\sigma} \circ \tau^{-\frac{N-1}{2}} \circ I\\ 
& = \tilde{\pi} \circ I \circ \tau^{-\frac{N-1}{2}} \circ I = \tilde{\pi} \circ \tau^{\frac{N-1}{2}} \,.
\end{split}
\end{align}
It follows that the two permutations in Eq.~\eqref{eq::tsp_bip_m1} and Eq.~\eqref{eq::tsp_bip_m2} are conjugate by $I$
\begin{equation}
I \circ \tilde{\pi} \circ \tau^{-1} \circ \tilde{\sigma}^{-1} \circ I  =  \tilde{\pi}\circ \tau^{\frac{N-1}{2}} \circ \tau \circ \tau^{\frac{N-1}{2}} \circ \tilde{\sigma}^{-1} = \tilde{\pi} \circ \tilde{\sigma}^{-1}
\end{equation}
so that, in this case, they have exactly the same numbers of cycles of order 2.
Indeed we have
\begin{subequations}
	\begin{align}
	\tilde{\mu}_1 = 
	& \, (2,1,4,3,6, \dots , N-1,N-2,N)\\
	\tilde{\mu}_2 = 
	&  \, (1,3,2,5,4, \dots N, N-1)
	\end{align} 
\end{subequations}
and they have $\frac{N-1}{2}$ cycles of order 2 and 1 fixed point. See Fig.~\ref{N5}  for the case $N=5$.

\begin{figure}
	\begin{center}
		\begin{tikzpicture}[scale=0.6]
		\node[draw,circle,inner sep=1.5pt,fill=black,text=white,label=above:{\footnotesize $r_1$}] (r1) at (1,3) {};
		\node[draw,circle,inner sep=1.5pt,fill=black,text=white,label=above:{\footnotesize $r_2$}] (r2) at (3,3) {};
		\node[draw,circle,inner sep=1.5pt,fill=black,text=white,label=above:{\footnotesize $r_3$}] (r3) at (5,3) {};
		\node[draw,circle,inner sep=1.5pt,fill=black,text=white,label=above:{\footnotesize $r_4$}] (r4) at (7,3) {};
		\node[draw,circle,inner sep=1.5pt,fill=black,text=white,label=above:{\footnotesize $r_5$}] (r5) at (9,3) {};
		\node[draw,circle,inner sep=1.5pt,fill=white,text=white,label=below:{\footnotesize $b_1$}] (b1) at (1,0) {};
		\node[draw,circle,inner sep=1.5pt,fill=white,text=white,label=below:{\footnotesize $b_2$}] (b2) at (3,0) {};
		\node[draw,circle,inner sep=1.5pt,fill=white,text=white,label=below:{\footnotesize $b_3$}] (b3) at (5,0) {};
		\node[draw,circle,inner sep=1.5pt,fill=white,text=white,label=below:{\footnotesize $b_4$}] (b4) at (7,0) {};
		\node[draw,circle,inner sep=1.5pt,fill=white,text=white,label=below:{\footnotesize $b_5$}] (b5) at (9,0) {};
		\draw[line width=1pt,gray]  (r2) to (b3);
		\draw[line width=1pt,gray]  (b2) to (r3);
		\draw[line width=1pt,gray]  (r5) to (b4);
		\draw[line width=1pt,gray]  (r4) to (b5);
		\draw[line width=1pt,gray] (r1) to[in=110, out=-110] (b1);
		\end{tikzpicture}
		\qquad\qquad
		\begin{tikzpicture}[scale=0.6]
		\node[draw,circle,inner sep=1.5pt,fill=black,text=white,label=above:{\footnotesize $r_1$}] (r1) at (1,3) {};
		\node[draw,circle,inner sep=1.5pt,fill=black,text=white,label=above:{\footnotesize $r_2$}] (r2) at (3,3) {};
		\node[draw,circle,inner sep=1.5pt,fill=black,text=white,label=above:{\footnotesize $r_3$}] (r3) at (5,3) {};
		\node[draw,circle,inner sep=1.5pt,fill=black,text=white,label=above:{\footnotesize $r_4$}] (r4) at (7,3) {};
		\node[draw,circle,inner sep=1.5pt,fill=black,text=white,label=above:{\footnotesize $r_5$}] (r5) at (9,3) {};
		\node[draw,circle,inner sep=1.5pt,fill=white,text=white,label=below:{\footnotesize $b_1$}] (b1) at (1,0) {};
		\node[draw,circle,inner sep=1.5pt,fill=white,text=white,label=below:{\footnotesize $b_2$}] (b2) at (3,0) {};
		\node[draw,circle,inner sep=1.5pt,fill=white,text=white,label=below:{\footnotesize $b_3$}] (b3) at (5,0) {};
		\node[draw,circle,inner sep=1.5pt,fill=white,text=white,label=below:{\footnotesize $b_4$}] (b4) at (7,0) {};
		\node[draw,circle,inner sep=1.5pt,fill=white,text=white,label=below:{\footnotesize $b_5$}] (b5) at (9,0) {};
		\draw[line width=1pt,gray]  (r1) to (b2);
		\draw[line width=1pt,gray]  (r3) to (b4);
		\draw[line width=1pt,gray]  (b1) to (r2);
		\draw[line width=1pt,gray]  (b3) to (r4);
		\draw[line width=1pt,gray] (r5) to[in=70, out=-70] (b5);
		\end{tikzpicture}
	\end{center}
	\caption{Decomposition of the optimal Hamiltonian cycle $\tilde{h}$ for $N=5$ in two disjoint matchings $\tilde{\mu}_2$ and $\tilde{\mu}_1$.} \label{N5}
\end{figure}
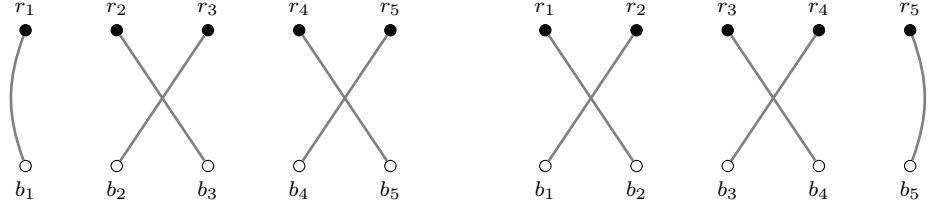

In the case of even $N$ the two permutations have not the same number of cycles of order 2, indeed one has no fixed point and the other has two of them. More explicitly
\begin{subequations}
	\begin{align}
	\tilde{\mu}_1 = 
	& \, (2,1,4,3,6, \dots ,N,N-1)\\
	\tilde{\mu}_2 = 
	&  \, (1,3,2,5,4, \dots N-1,N-2, N)
	\end{align} 
\end{subequations}
See Fig.~\ref{N4} for the case $N=4$.
\begin{figure}
	\begin{center}
		\begin{tikzpicture}[scale=0.6]
		\node[draw,circle,inner sep=1.5pt,fill=black,text=white,label=above:{\footnotesize $r_1$}] (r1) at (1,3) {};
		\node[draw,circle,inner sep=1.5pt,fill=black,text=white,label=above:{\footnotesize $r_2$}] (r2) at (3,3) {};
		\node[draw,circle,inner sep=1.5pt,fill=black,text=white,label=above:{\footnotesize $r_3$}] (r3) at (5,3) {};
		\node[draw,circle,inner sep=1.5pt,fill=black,text=white,label=above:{\footnotesize $r_4$}] (r4) at (7,3) {};
		\node[draw,circle,inner sep=1.5pt,fill=white,text=white,label=below:{\footnotesize $b_1$}] (b1) at (1,0) {};
		\node[draw,circle,inner sep=1.5pt,fill=white,text=white,label=below:{\footnotesize $b_2$}] (b2) at (3,0) {};
		\node[draw,circle,inner sep=1.5pt,fill=white,text=white,label=below:{\footnotesize $b_3$}] (b3) at (5,0) {};
		\node[draw,circle,inner sep=1.5pt,fill=white,text=white,label=below:{\footnotesize $b_4$}] (b4) at (7,0) {};
		\draw[line width=1pt,gray]  (r2) to (b3);
		\draw[line width=1pt,gray]  (b2) to (r3);
		\draw[line width=1pt,gray] (r1) to[in=110, out=-110] (b1);
		\draw[line width=1pt,gray] (r4) to[in=70, out=-70] (b4);
		\end{tikzpicture}
		\qquad\qquad
		\begin{tikzpicture}[scale=0.6]
		\node[draw,circle,inner sep=1.5pt,fill=black,text=white,label=above:{\footnotesize $r_1$}] (r1) at (1,3) {};
		\node[draw,circle,inner sep=1.5pt,fill=black,text=white,label=above:{\footnotesize $r_2$}] (r2) at (3,3) {};
		\node[draw,circle,inner sep=1.5pt,fill=black,text=white,label=above:{\footnotesize $r_3$}] (r3) at (5,3) {};
		\node[draw,circle,inner sep=1.5pt,fill=black,text=white,label=above:{\footnotesize $r_4$}] (r4) at (7,3) {};
		\node[draw,circle,inner sep=1.5pt,fill=white,text=white,label=below:{\footnotesize $b_1$}] (b1) at (1,0) {};
		\node[draw,circle,inner sep=1.5pt,fill=white,text=white,label=below:{\footnotesize $b_2$}] (b2) at (3,0) {};
		\node[draw,circle,inner sep=1.5pt,fill=white,text=white,label=below:{\footnotesize $b_3$}] (b3) at (5,0) {};
		\node[draw,circle,inner sep=1.5pt,fill=white,text=white,label=below:{\footnotesize $b_4$}] (b4) at (7,0) {};
		\draw[line width=1pt,gray]  (r1) to (b2);
		\draw[line width=1pt,gray]  (r3) to (b4);
		\draw[line width=1pt,gray]  (b1) to (r2);
		\draw[line width=1pt,gray]  (b3) to (r4);
		\end{tikzpicture}
	\end{center}
	\caption{Decomposition of the optimal Hamiltonian cycle $\tilde{h}$ for $N=4$ in he two disjoint matchings $\tilde{\mu}_2$ and $\tilde{\mu}_1$.}\label{N4}
\end{figure}

\subsection{Optimal cycle on the complete graph: proofs}\label{app::tsp_mono}
\subsubsection{Proof of the optimal cycle for $p>1$} \label{App::tsp_p>1}
Consider a $\sigma \in \mathcal{S}_N$ with $\sigma(1) =1$. Taking $\sigma(1) =1$ corresponds to the irrelevant choice of the starting point of the cycle.
Let us introduce now a new set of ordered points $\mathcal{B}:=\{ b_j\}_{j=1,\dots,N}\subset [0,1]$
such that
\begin{equation}
b_i = 
\begin{cases}
	r_1 & \hbox{for \, } i =1\\
	r_{i-1} & \hbox{otherwise } 
\end{cases}
\end{equation}
and consider the Hamiltonian cycle on the complete bipartite graph with vertex sets $\mathcal{R}$ and $\mathcal{B}$
\begin{equation}
\begin{split}
	& h[(\sigma, \pi_\sigma)] := (r_1, b_{\pi_\sigma(1)}, r_{\sigma(2)}, b_{\pi_\sigma(2)},\dots, r_{\sigma(N)}, b_{\pi_\sigma(N)}, r_{\sigma(1)})
\end{split}
\end{equation}
so that 
\begin{equation}\label{eq::app_mono_tsp_pi_s}
\pi_\sigma(i) = 
\begin{cases}
	2 & \hbox{for \, } i =1\\
	\sigma(i)+1 & \hbox{for \, } i < k\\
	\sigma(i+1)+1 & \hbox{for \, } i \geq k\\
	1 & \hbox{for \, } i = N\\
\end{cases}
\end{equation}
where $k$ is such that $\sigma(k) = N$.
We have therefore
\begin{equation}\label{eq::app_tsp_mono_b}
\begin{split}
	& (b_{\pi_\sigma(1)}, b_{\pi_\sigma(2)}, \dots, b_{\pi_\sigma(k-1)}, b_{\pi_\sigma(k)}, \dots, b_{\pi_\sigma(N-1)}, b_{\pi_\sigma(N)}) \\
	& = (r_1, r_{\sigma(2)}, \dots, r_{\sigma(k-1)},r_{\sigma(k+1)} ,\dots , r_{\sigma(N)},r_1).
\end{split}
\end{equation}
In other words we are introducing a set of blue points such that we can find a bipartite Hamiltonian tour which only use link available in our ``monopartite'' problem and has the same cost of $\sigma$. Therefore, by construction (using Eq.~\eqref{eq::app_tsp_mono_b}):
\begin{equation}
\begin{split}
	E_N(h[\sigma]) & = E_N(h[(\sigma, \pi_\sigma)] ) \geq E_N(h[(\tilde{\sigma}, \tilde{\pi})]) \\
	& = E_N(h[(\tilde{\sigma}, \pi_{\tilde{\sigma}})]) = E_N(h[\tilde{\sigma}]),
\end{split}
\end{equation}
where the fact that $\tilde{\pi}=\pi_{\tilde{\sigma}}$ can be checked using Eqs.~\eqref{eq::tsp_bip_sigmatilde} and \eqref{eq::tsp_bip_pitilde} and \eqref{eq::app_mono_tsp_pi_s}.

\subsubsection{Proof of the optimal cycle for $0<p<1$} \label{App::0<p<1}
As first step, we enunciate and demonstrate two lemmas that will be useful for the proof. The first one will help us in understand how to remove two crossing arcs without breaking the TSP cycle into multiple ones. The second one, instead will prove that removing a crossing between two arcs will always lower the total number of crossing in the TSP cycle.
\begin{lemma}
	\label{Lemma::tsp_mono_0<p<1}
	Given an Hamiltonian cycle with its edges drawn as arcs in the upper half-plane, let us consider two of the arcs that cannot be drawn without crossing each other. Then, this crossing can be removed only in one way without splitting the original cycle into two disjoint cycles; moreover, this new configuration has a lower cost than the original one.   
\end{lemma}
\begin{proof}
	Let us consider a generic oriented Hamiltonian cycle and let us suppose it contains a matching as in figure:
	\begin{figure}[h!]
		\centering
		\includegraphics[width=0.4\columnwidth]{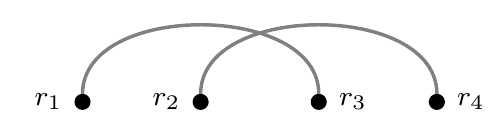}
	\end{figure}

	\noindent There are two possible orientations for the matching that correspond to this two oriented Hamiltonian cycles: 
	
	\begin{enumerate}
		\item $(C_1r_1r_3C_2r_2r_4C_3)\,,$
		\item $(C_1r_1r_3C_2r_4r_2C_3)\,,$
	\end{enumerate}
	where $C_1$, $C_2$ and $C_3$ are paths (possibly visiting other points of our set).
	The other possibilities are the dual of this two, and thus they are equivalent. In both cases, a priori, there are two choices to replace this crossing matching $(r_1, r_3)$, $(r_2,r_4)$ with a non-crossing one: $(r_1, r_2)$, $(r_3, r_4)$ or $(r_1, r_4)$, $(r_2, r_3)$. We now show, for the two possible prototypes of Hamiltonian cycles, which is the right choice for the non-crossing matching, giving a general rule. Let us consider case 1: here, if we replace the crossing matching with $(r_1, r_4)$, $(r_2, r_3)$, the cycle will split; in fact we would have two cycles:  $(C_1r_1r_4C_3)$ and $(r_3C_2r_2)$. Instead, if we use the other non-crossing matching, we would have: $(C_1r_1r_2[C_2]^{-1}r_3r_4C_3)$. This way we have removed the crossing without splitting the cycle. Let us consider now case 2: in this situation, using $(r_1, r_4)$, $(r_2, r_3)$ as the new matching, we would have: $(C_1r_1r_4[C_2]^{-1}r_3r_2C_3)$; the other matching, on the contrary, gives: $(C_1r_1r_2C_3)$ and $(r_3C_2r_4)$.
	
	The general rule is the following: given the oriented matching, consider the four oriented lines going inward and outward the node. Then, the right choice for the non-crossing matching is obtained joining the two couples of lines with opposite orientation.
	
	Since the difference between the cost of the original cycle and the new one simply consists in the difference between a crossing matching and a non-crossing one,  this is positive when $0<p<1$, as shown in~\cite{Boniolo2014}.
\end{proof}

Now we deal with the second point: given an Hamiltonian cycle, in general it is not obvious that replacing non-crossing arcs with a crossing one, the total number of intersections increases. Indeed there could be the chance that one or more crossings are removed in the operation of substituting the matching we are interested in. 
Notice that two arcs forms a matching of 4 points. Therefore, from now on, we will use expressions like ``crossing matching'' (``non-crossing matching'') and ``two crossing arcs'' (``two non-crossing arcs'') indifferently.
We now show that it holds the following

\begin{lemma}
	\label{Lemma::tsp_mono_intersections}
	Given an Hamiltonian cycle with a matching that is non-crossing, if it is replaced by a crossing one, the total number of intersections always increases. Vice versa, if a crossing matching is replaced by a non-crossing one, the total number of crossings always decreases.  
\end{lemma} 
\begin{proof}
	This is a topological property we will prove for cases. To best visualize crossings,
	we change the graphical way we use to represent the complete graph that underlies the problem: now the nodes are organized along a circle, in such a way that they are ordered clockwise (or, equivalently, anti-clockwise) according to the natural ordering given by the positions on the segment $[0,1]$. Links between points here are represented as straight lines. It is easy to see that a crossing as defined in Sec.~\ref{sec::tsp_mono_1d} corresponds to, in this picture, a crossing of lines.
	All the possibilities are displayed in Fig.~\ref{Fig::app_tsp_mono_Crossings}, where we have represented with red lines the edges involved in the matching, while the other lines span all the possible topological configurations.
	
	\begin{figure*}[ht]
		\centering
		\includegraphics[width=1\columnwidth]{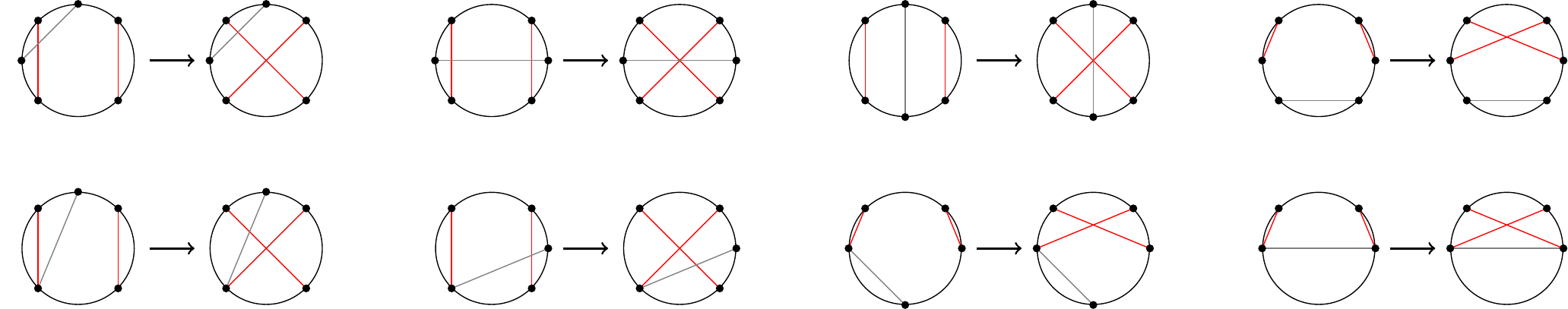}
		\caption{Replacing a non-crossing matching with a crossing one in an Hamiltonian cycle always increase the number of crossings. Here we list all the possible topological configurations one can have.}
		\label{Fig::app_tsp_mono_Crossings}
	\end{figure*}
\end{proof}

Now we can prove that the cycle $h^*$ given in Eq.~\eqref{eq::tsp_mono_0<p<1} is the optimal one:
\begin{proof}
	Consider a generic Hamiltonian cycle and draw the connections between the points in the upper half-plane. Suppose to have an Hamiltonian cycle where there are, let us say, $n$ intersections between edges. Thanks to Lemma~\ref{Lemma::tsp_mono_0<p<1}, we can swap two crossing arcs with a non-crossing one without splitting the Hamiltonian cycle. As shown in Lemma~\ref{Lemma::tsp_mono_intersections}, this operation lowers always the total number of crossings between the edges, and the cost of the new cycle is smaller than the cost of the starting one. Iterating this procedure, it follows that one can find a cycle with no crossings. 
	Now we prove that there are no other cycles out of $h^*$ and its dual with no crossings. This can be easily seen, since $h^*$ is the only cycle that visits all the points, starting from the first, in order. This means that all the other cycles do not visit the points in order and, thus, they have a crossing, due to the fact that the point that is not visited in a first time, must be visited next, creating a crossing.
\end{proof}

\subsubsection{Proof of the optimal cycle for $p<0$, odd $N$} \label{App::p<0 odd N}
To complete the proof given in the main text, we need to discuss two points. Firstly, we address which is the correct move that swap a non-crossing matching with a crossing one; thanks to Lemma~\ref{Lemma::tsp_mono_intersections}, by performing such a move one always increases the total number of crossings. Secondly we prove that there is only one Hamiltonian cycle to which this move cannot be applied (and so it is the optimal solution).

We start with the first point: consider an Hamiltonian cycle with a matching that is non-crossing, then the possible situations are the following two:
\begin{figure}[h!]
	\centering
	\includegraphics[width=0.4\columnwidth]{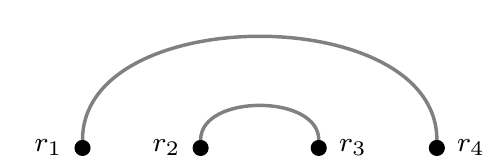} 
\end{figure}
\begin{figure}[h]
	\centering
	\includegraphics[width=0.4\columnwidth]{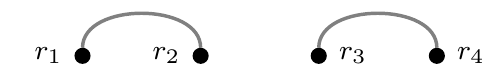}
\end{figure}

\noindent For the first case there are two possible independent orientations:

\begin{enumerate}
	\item $(r_1r_4C_2r_2r_3C_3)\,,$
	\item $(r_1r_4C_2r_3r_2C_3)\,.$
\end{enumerate}

If we try to cross the matchings in the first cycle, we obtain $(r_1r_3C_3)(r_2[C_2]^{-1}r_4)$, and this is not anymore an Hamiltonian cycle. On the other hand, in the second cycle, the non-crossing matching can be replaced by a crossing one without breaking the cycle: $(r_1r_3[C_2]^{-1}r_4r_2C_3)$. For the second case the possible orientations are:
\begin{enumerate}
	\item $(r_1r_2C_2r_4r_3C_3)\,,$
	\item $(r_1r_2C_2r_3r_4C_3)\,.$
\end{enumerate}
By means of the same procedure used in the first case, one finds that the non-crossing matching in the second cycle can be replaced by a crossing one without splitting the cycle, while in the first case the cycle is divided by this operation.\\

The last step is the proof that the Hamiltonian cycle given in Eq.~\eqref{eq::tsp_mono_p<0_odd} has the maximum number of crossings.

Let us consider an Hamiltonian cycle $h[\sigma] = \left( r_{\sigma(1)}, \dots, r_{\sigma(N)} \right)$ on the complete graph $\mathcal{K}_N$. We now want to evaluate what is the maximum number of crossings an edge can have depending on the permutation $\sigma$. Consider the edge connecting two vertices $r_{\sigma(i)}$ and $r_{\sigma(i+1)}$: obviously both the edges $(r_{\sigma(i-1)}, r_{\sigma(i)})$ and $(r_{\sigma(i+1)}, r_{\sigma(i+2)})$ share a common vertex with ($r_{\sigma(i)}, r_{\sigma(i+1)}$), therefore they can never cross it. So, if we have $N$ vertices, each edge has $N-3$ other edges that can cross it. Let us denote with $\mathcal{N}\left[\sigma(i)\right]$ the number of edges that cross the edge  $(r_{\sigma(i)}, r_{\sigma(i+1)})$ and let us define the sets:
\begin{equation}
A_j:=
\begin{cases}
	\{r_k\}_{k=\sigma(i)+1\ (\mathrm{mod}\ N),\dots,\sigma(i+1)-1(\mathrm{mod}\ N)}  & \hbox{for \, } j=1\\
	\{r_k\}_{k=\sigma(i+1)+1(\mathrm{mod}\ N),\dots,\sigma(i)-1(\mathrm{mod}\ N)}  & \hbox{for \, } j=2\\	 
\end{cases}
\label{Sets}
\end{equation}
These two sets contain the points between $r_{\sigma(i)}$ and $r_{\sigma(i+1)}$. In particular, the maximum number of crossings an edge can have is given by:
\begin{equation}	
\max(\mathcal{N}\left[\sigma(i)\right])=
\begin{cases}
	2\min_{j}|A_j|  & \hbox{for \, } |A_1|\not = |A_2|\\
	2|A_1| - 1  & \hbox{for \, } |A_1| = |A_2|\\	
\end{cases}\label{max}
\end{equation}
This is easily seen, since the maximum number of crossings an edge can have is obtained when all the points belonging to the smaller between $A_1$ and $A_2$ contributes with two crossings. This cannot happen when the cardinality of $A_1$ and $A_2$ is the same because at least one of the edges departing from the nodes in $A_1$ for example, must be connected to one of the ends of the edge $(r_{\sigma(i)}, r_{\sigma(i+1)})$, in order to have an Hamiltonian cycle. Note that this case, i.e. $|A_1| = |A_2|$ can happen only if $N$ is even.

Consider the particular case such that $\sigma(i)=a$ and $\sigma(i+1)=a+\frac{N-1}{2}\pmod{N}$ or $\sigma(i+1)=a+\frac{N+1}{2}\pmod{N}$. Then \eqref{max} in this cases is exactly equal to $N-3$, which means that the edges $(r_a, r_{a+\frac{N-1}{2}\pmod{N}})$ and $(r_a, r_{a+\frac{N+1}{2}\pmod{N}})$ can have the maximum number of crossings if the right configuration is chosen.\\
Moreover, if there is a cycle such that every edge has $N-3$ crossings, such a cycle is unique, because the only way of obtaining it is connecting the vertex $r_a$ with $r_{a+\frac{N-1}{2}\pmod{N}}$ and $r_{a+\frac{N+1}{2}\pmod{N}}, \forall a$.\\

\subsection{Optimal TSP and 2-factor for $p<0$ and $N$ even}\label{app::tsp_2f_mono}
We start considering here the 2-factor problem (see Sec.~\ref{sec::2factor} for a definition) for $p<0$ in the even-$N$ case. We will use the shape of its solution to prove that one among the cycles given in Eq.~\eqref{eq::tsp_mono_p<0_even} is the solution of the TSP. 

In the following we will say that, given a permutation $\sigma \in \mathcal{S}_{N}$, the edge $(r_{\sigma(i)}, r_{\sigma(i+1)})$ has length $L \in \mathbb{N}$ if: 
\begin{equation}
L=\mathcal{L}(i):=\min_j|A_j(i)| 
\end{equation}
where $A_j(i)$ was defined in Eq.~\eqref{Sets}.

\subsubsection{$N$ is a multiple of 4}
Let us consider the sequence of points $\mathcal{R} = \{r_i\}_{i=1,\dots,N}$ of $N$ points, with $N$ a multiple of 4, in the interval $[0,1]$, with $r_1 \le \dots \le r_N$, consider the permutations $\sigma_j$, $j=1, 2$ defined by the following cyclic decomposition:
\begin{subequations}\label{N|4}
	\begin{equation}
	\begin{split}
	\sigma_1 & = (r_1, r_{\frac{N}{2}+1}, r_2, r_{\frac{N}{2}+2}) \dots (r_a, r_{a+\frac{N}{2}}, r_{a+1}, r_{a+\frac{N}{2}+1}) \dots(r_{\frac{N}{2}-1}, r_{N-1}, r_{\frac{N}{2}}, r_{N}) \\
	\end{split}
	\end{equation}
	\begin{equation}
	\begin{split}
	\sigma_2 & = (r_1, r_{\frac{N}{2}+1}, r_N, r_{\frac{N}{2}})\dots(r_a, r_{a+\frac{N}{2}}, r_{a-1}, r_{a+\frac{N}{2}-1}) \dots(r_{\frac{N}{2}-1}, r_{N-1}, r_{\frac{N}{2}-2}, r_{N-2})
	\end{split}
	\end{equation}
\end{subequations}
for integer $a = 1, \dots,  \frac{N}{2}-1$. Defined $h^*_1:=h[\sigma_1]$ and $h^*_2:=h[\sigma_2]$, it holds the following:

\begin{pros}\label{b1}
	$h^*_1$ and $h^*_2$ are the 2-factors that contain the maximum number of crossings between the arcs.
\end{pros}

\begin{proof}
	
	An edge can be involved, at most, in $N-3$ crossing matchings. In the even N case, this number is achieved by the edges of the form $(r_a, r_{a+\frac{N}{2}\pmod{N}})$, i.e. by the edges of length $\frac{N}{2}-1$. There can be at most $\frac{N}{2}$ edges of this form in a 2-factor. Thus, in order to maximize the number of crossings, the other $\frac{N}{2}$ edges must be of the form $(r_a, r_{a+\frac{N}{2}+1\pmod{N}})$ or $(r_a, r_{a+\frac{N}{2}-1\pmod{N}})$, i.e. of length $\frac{N}{2}-2$. It is immediate to verify that both $h^*_1$ and $h^*_2$ have this property; we have to prove they are the only ones with this property.\\
	Consider, then, to have already fixed the $\frac{N}{2}$ edges $(r_a, r_{a+\frac{N}{2}\pmod{N}}),   \forall a\in [N]$. Suppose to have fixed also the edge $(r_1, r_{\frac{N}{2}})$ (the other chance is to fix the edge $(r_1, r_{\frac{N}{2}+2})$: this brings to the other 2-factor). Consider now the point $r_{\frac{N}{2}+1}$: suppose it is not connected to the point $r_N$, but to the point $r_2$, i.e., it has a different edge from the cycle $h^*_2$. We now show that this implies it is not possible to construct all the remaining edges of length $\frac{N}{2}-2$. Consider, indeed, of having fixed the edges $(r_1, r_{\frac{N}{2}})$ and $(r_2, r_{\frac{N}{2}+1})$ and focus on the vertex $r_{\frac{N}{2}+2}$: in order to have an edge of length $\frac{N}{2}-2$, this vertex must be connected either with $r_1$ or with $r_3$, but $r_1$ already has two edges, thus, necessarily, there must be the edge $(r_{\frac{N}{2}+2}, r_3)$. By the same reasoning, there must be the edges $(r_{\frac{N}{2}+3}, r_4)$, $(r_{\frac{N}{2}+4}, r_5), \dots, (r_{N-2}, r_{\frac{N}{2}-1})$. Proceeding this way, we have constructed $N-1$ edges; the remaining one is uniquely determined, and it is $(r_{N-1}, r_{N})$, which has null length.\\
	Therefore the edge $(r_2, r_{\frac{N}{2}+1})$ cannot be present in the optimal 2-factor and so, necessarily, there is the edge $(r_{\frac{N}{2}+1}, r_N)$; this creates the cycle $(r_1, r_{\frac{N}{2}}, r_N, r_{\frac{N}{2}+1})$. Proceeding the same way on the set of the remaining vertices $\{r_2, r_3,\dots,r_{\frac{N}{2}-1}, r_{\frac{N}{2}+2},\dots, r_{N-1}\}$, one finds that the only way of obtaining $\frac{N}{2}$ edges of length $\frac{N}{2}-1$ and $\frac{N}{2}$ edges of length $\frac{N}{2}-2$ is generating the loop coverings of the graph $h^*_1$ or $h^*_2$.
\end{proof}

Proposition \ref{b1}, together with the fact that the optimal 2-factor has the maximum number of crossing matchings, guarantees that the optimal 2-factor is either $h^*_1$ or $h^*_2$.

\begin{figure*}[ht]
	\begin{subfigure}[t]{0.49\linewidth}
		\centering
		\includegraphics[width=1\columnwidth]{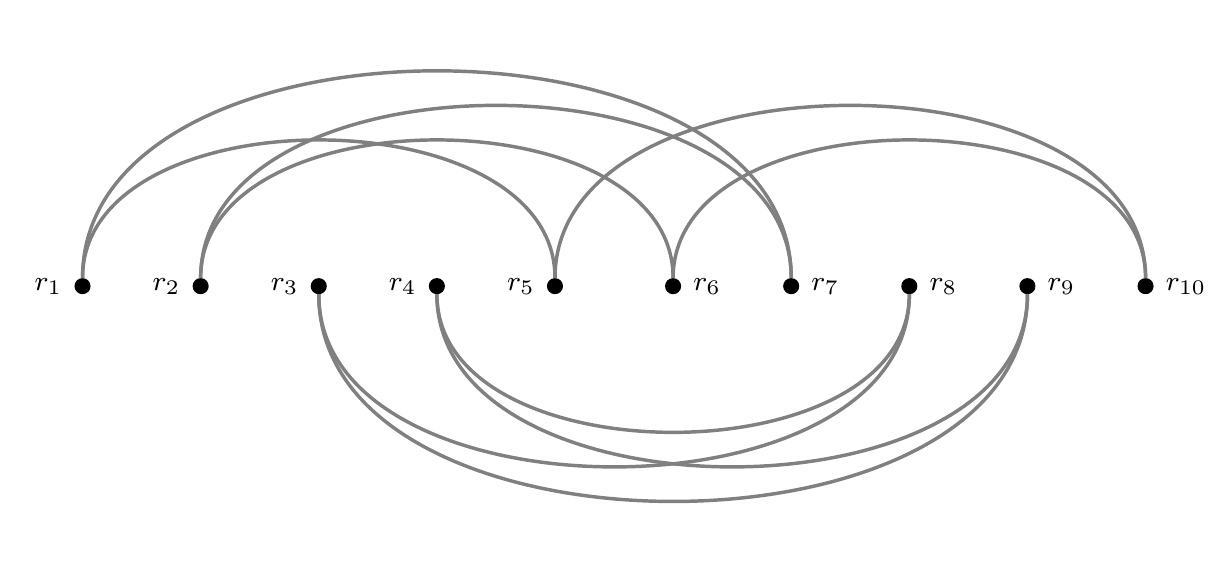}
		\caption{One of the optimal 2-factor solutions for $N=10$ and $p<0$; the others are obtainable cyclically permuting this configuration}
		\label{Fig::2factor_N10_line}
	\end{subfigure} \hfill
	\begin{subfigure}[t]{0.49\linewidth}
		\centering
		\includegraphics[width=0.5\columnwidth]{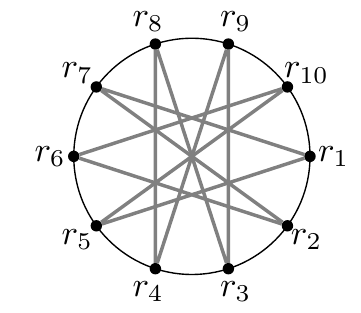}
		\caption{The same optimal 2-factor solution, but represented on a circle, where the symmetries of the solutions are more easily seen} \label{Fig::2factor_N10_circle}
	\end{subfigure}
	\caption{}
\end{figure*}

\subsubsection{$N$ is not a multiple of 4}
Let us consider the usual sequence $\mathcal{R} = \{r_i\}_{i=1,\dots,N}$ of $N$ points, with even $N$ but not a multiple of 4, in the interval $[0,1]$, with $r_1 \le \dots \le r_N$, consider the permutation $\pi$ defined by the following cyclic decomposition:
\begin{equation}
\begin{split}\label{Nnot|4}
\pi & =(r_1, r_{\frac{N}{2}}, r_N, r_{\frac{N}{2}+1}, r_2, r_{\frac{N}{2}+2})(r_3, r_{\frac{N}{2}+3}, r_4, r_{\frac{N}{2}+4}) \dots(r_{\frac{N}{2}-2}, r_{N-1}, r_{\frac{N}{2}-1}, r_{N-2})
\end{split}
\end{equation}
Defined 
\begin{equation}
\pi_k(i):=\pi(i)+k \ (\mathrm{mod}\ N), \; k\in [0, N-1] \label{permut}
\end{equation}
and
\begin{equation}
h^*_k:=h[\pi_k]
\end{equation}
the following proposition holds:
\begin{pros}\label{b2}
	$h^*_k$ are the 2-factors that contain the maximum number of crossings between the arcs.
\end{pros}
\begin{proof}
	Also in this case the observations done in the proof of Proposition~\ref{b1} holds. Thus, in order to maximize the number of crossing matchings, one considers, as in the previous case, the $\frac{N}{2}$ edges of length $\frac{N}{2}-1$, i.e. of the form $(r_a, r_{a+\frac{N}{2}\pmod{N}})$, and then tries to construct the remaining $\frac{N}{2}$ edges of length $\frac{N}{2}-2$, likewise the previous case. Again, if one fixes the edge $(r_1, r_{\frac{N}{2}})$, the edge $(r_2, r_{\frac{N}{2}+1})$ cannot be present, by the same reasoning done in the proof of Proposition B.1. The fact that, in this case, $N$ is not a multiple of 4 makes it impossible to have a 2-factor formed by 4-vertices loops, as in the previous case. The first consequence is that, given $\frac{N}{2}$ edges of length $\frac{N}{2}-1$, it is not possible to have $\frac{N}{2}$ edges of length $\frac{N}{2}-2$. In order to find the maximum-crossing solution, one has the following options:
	\begin{itemize}
		\item to take a 2-factor with $\frac{N}{2}$ edges of length $\frac{N}{2}-1$, $\frac{N}{2}-1$ edges of length $\frac{N}{2}-2$ and one edge of length $\frac{N}{2}-2$: in this case the theoretical maximum number of crossing matchings is $\frac{N(N-3)}{2}+(\frac{N}{2}-1)(N-4)+N-6=N^2-\frac{7N}{2}-2$;
		\item to take a 2-factor with $\frac{N}{2}-1$ edges of length $\frac{N}{2}-1$, $\frac{N}{2}+1$ edges of length $\frac{N}{2}-2$: in this case the theoretical maximum number of crossing matchings is $(\frac{N}{2}-1)(N-3)+(\frac{N}{2}+1)(N-4)=N^2-\frac{7N}{2}-1$.
	\end{itemize}
	Clearly the second option is better, at least in principle, than the first one. The cycles $h^*_k$ belong to the second case and saturate the number of crossing matchings. Suppose, then, to be in this case. Let us fix the $\frac{N}{2}-1$ edges of length $\frac{N}{2}-1$; this operation leaves two vertices without any edge, and this vertices are of the form $r_a$, $r_{a+\frac{N}{2}\pmod{N}}, a\in[1,N]$ (this is the motivation for the degeneracy of solutions). By the reasoning done above, the edges that link this vertices must be of length $\frac{N}{2}-2$, and so they are uniquely determined. They form the 6-points loop $(r_a$, $r_{a-1+\frac{N}{2}\pmod{N}}$, $r_{N-1+a\pmod{N}}$, $r_{a+\frac{N}{2}\pmod{N}}$, $r_{a+1\pmod{N}}$, $r_{a+1+\frac{N}{2}\pmod{N}})$. The remaining $N-6$ points, since $4|(N-6)$, by the same reasoning done in the proof of Proposition~\ref{b1}, necessarily form the $\frac{N-6}{4}$ 4-points loops given by the permutations~\eqref{permut}.
\end{proof}

Proposition \ref{b2}, together with the fact that the optimal 2-factor has the maximum number of crossing matchings, guarantees that the optimal 2-factor is such that $h^*\in \{h^*_k\}_{k=1}^N$. 

\subsubsection{Proof of the optimal cycles for $p<0$, odd $N$}
\begin{proof}
	Let us begin from the permutations that define the optimal solutions for the 2-factor, that is those given in Eqs.~\ref{N|4} if is $N$ a multiple of 4 and in Eq.~\ref{Nnot|4} otherwise.
	In both cases, the optimal solution is formed only by edges of length $\frac{N}{2}-1$ and of length $\frac{N}{2}-2$. Since the optimal 2-factor is not a TSP, in order to obtain an Hamiltonian cycle from the 2-factor solution, couples of crossing edges need to became non-crossing, where one of the two edges belongs to one loop of the covering and the other to another loop. Now we show that the optimal way of joining the loops is replacing two edges of length $\frac{N}{2}-1$ with other two of length $\frac{N}{2}-2$.
	Let us consider two adjacent 4-vertices loops, i.e. two loops of the form:
	\begin{equation}
	(r_a, r_{a+\frac{N}{2}}, r_{a+1}, r_{a+\frac{N}{2}+1}), (r_{a+2}, r_{a+2+\frac{N}{2}}, r_{a+3}, r_{a+\frac{N}{2}+3})
	\end{equation}
	and let us analyze the possible cases:
	\begin{enumerate}
		\item to remove two edges of length $\frac{N}{2}-2$, that can be replaced in two ways:
		\begin{itemize}
			\item either with an edge of length $\frac{N}{2}-2$ and one of length $\frac{N}{2}-4$; in this case the maximum number of crossings decreases by 4;
			\item or with two edges of length $\frac{N}{2}-3$; also in this situation the maximum number of crossings decreases by 4.
		\end{itemize}
		\item to remove one edge of length $\frac{N}{2}-2$ and one of length $\frac{N}{2}-1$, and also this operation can be done in two ways:
		\begin{itemize}
			\item either with an edge of length $\frac{N}{2}-2$ and one of length $\frac{N}{2}-3$; in this case the maximum number of crossings decreases by 3;
			\item or with an edge of length $\frac{N}{2}-3$ and one of length $\frac{N}{2}-4$; in this situation the maximum number of crossings decreases by 7.
		\end{itemize}
		\item the last chance is to remove two edges of length $\frac{N}{2}-1$, and also this can be done in two ways:
		\begin{itemize}
			\item either with two edges of length $\frac{N}{2}-3$; here the maximum number of crossings decreases by 6;
			\item or with two edges of length $\frac{N}{2}-2$; in this situation the maximum number of crossings decreases by 2. This happens when we substitute two adjacent edges of length $\frac{N}{2}-1$, that is, edges of the form $(r_a,r_{\frac{N}{2}+a\pmod{N}})$ and $(r_{a+1},r_{\frac{N}{2}+a+1\pmod{N}})$, with the non-crossing edges $(r_a,r_{\frac{N}{2}+a+1\pmod{N}})$ and $(r_{a+1},r_{\frac{N}{2}+a\pmod{N}})$
		\end{itemize}
	\end{enumerate}
	The last possibility is the optimal one, since our purpose is to find the TSP with the maximum number of crossings, in order to conclude it has the lower cost. Notice that the cases discussed above holds also for the 6-vertices loop and an adjacent 4-vertices loop when N is not a multiple of 4. We have considered here adjacent loops because, if they were not adjacent, then the difference in maximum crossings would have been even bigger.\\
	Now we have a constructive pattern for building the optimal TSP. Let us call $\mathcal{O}$ the operation described in the second point of (3). Then, starting from the optimal 2-factor solution, if it is formed by $n$ points, $\mathcal{O}$ has to be applied $\frac{N}{4}-1$ times if N is a multiple of 4 and $\frac{N-6}{4}$ times otherwise. In both cases it is easily seen that $\mathcal{O}$ always leaves two adjacent edges of length $\frac{N}{2}-1$ invariant, while all the others have length $\frac{N}{2}-2$. The multiplicity of solutions is given by the $\frac{N}{2}$ ways one can choose the two adjacent edges of length $\frac{N}{2}-1$. In particular, the Hamiltonian cycles $h^*_k$ saturates the maximum number of crossings that can be done, i.e., every time that $\mathcal{O}$ is applied, exactly 2 crossings are lost.\\
	We have proved, then, that $h^*_k$ are the Hamiltonian cycles with the maximum number of crossings. Now we prove that any other Hamiltonian cycle has a lower number of crossings. Indeed any other Hamiltonian cycle must have
	\begin{itemize}
		\item either every edge of length $\frac{N}{2}-2$;
		\item or at least one edge of length less than or equal to $\frac{N}{2}-3$.
	\end{itemize} 
	This is easily seen, since it is not possible to build an Hamiltonian cycle with more than two edges or only one edge of length $\frac{N}{2}-1$ and all the others of length $\frac{N}{2}-2$. It is also impossible to build an Hamiltonian cycle with two non-adjacent edges of length $\frac{N}{2}-1$ and all the others of length $\frac{N}{2}-2$: the proof is immediate.
	Consider then the two cases presented above: in the first case the cycle (let us call it $H$) is clearly not optimal, since it differs from $h^*_k, \forall k$ by a matching that is crossing in $h^*_k$ and non-crossing in $H$. Let us consider, then, the second case and suppose the shortest edge, let us call it $b$, has length $\frac{N}{2}-3$: the following reasoning equally holds if the considered edge is shorter. The shortest edge creates two subsets of vertices: in fact, called $x$ and $y$ the vertices of the edge considered and supposing $x<y$, there are the subsets defined by:
	\begin{equation}
	A=\{r\in \mathcal{V}: x<r<y\}
	\end{equation}
	\begin{equation}
	B=\{r\in \mathcal{V}: r<x \vee r>y\}
	\end{equation}
	Suppose, for simplicity, that $|A|<|B|$: then, necessarily $|A|=\frac{N}{2}-3$ and $|B|=\frac{N}{2}+1$. As an immediate consequence, there is a vertex in $B$ whose edges have both vertices in $|B|$. As a consequence, fixed an orientation on the cycle, one of this two edges and $b$ are obviously non-crossing and, moreover, have the right relative orientation so that they can be replaced by two crossing edges without splitting the Hamiltonian cycle. Therefore also in this case the Hamiltonian cycle considered is not optimal.
\end{proof}

\subsection{Second moment of the optimal cost distribution on the complete graph}\label{app::tsp_mono_2Moment}
Here we compute the second moment of the optimal cost distribution. We will restrict for simplicity to the $p>1$ case, where
\begin{equation}\label{eq::app_tsp_mono_OptimalCost}
E_N[h^*]=|r_{2}-r_1|^p+|r_{N}-r_{N-1}|^p+\sum_{i=1}^{N-2}|r_{i+2}-r_{i}|^p \,.
\end{equation}
We begin by writing the probability distribution for $N$ ordered points
\begin{equation}
\rho_N(r_1,\dots,r_N)=N! \prod_{i=0}^{N}\theta(r_{i+1}-r_i)
\end{equation}
where we have defined $r_0\equiv0$ and $r_{N+1}\equiv 1$. The joint probability distribution of their spacings
\begin{equation}
\varphi_i \equiv r_{i+1} - r_i \,,
\end{equation}
is, therefore
\begin{equation}
\rho_N(\varphi_0, \dots , \varphi_{N}) = N! \, \delta\left[ \sum_{i=0}^{N} \varphi_i = 1\right]\, \prod_{i=0}^{N} \theta(\varphi_i)\, .
\end{equation}
If $\{i_1, i_2, \dots, i_k\}$ is a generic subset of $k$ different indices in $\{0,1,\dots,N\}$, we soon get the marginal distributions
\begin{equation}\label{eq::app_tsp_mono_DistributionSpacings}
\rho_N^{(k)} (\varphi_{i_1}, \dots , \varphi_{i_k}) = \frac{N!}{(N-k)!} \left( 1 - \sum_{n=1}^{k}\varphi_{i_n} \right)^{N-k} \theta \left( 1 - \sum_{n=1}^{k}\varphi_{i_n} \right) \, \prod_{n=1}^{k}\theta(\varphi_{i_n}) \,.
\end{equation}
Developing the square of Eq.~\eqref{eq::app_tsp_mono_OptimalCost} one obtains $N^2$ terms, each one describing a particular configuration of two arcs connecting some points on the line. We will denote by $\chi_1$ and $\chi_2$ the length of these arcs; they can only be expressed as a sum of 2 spacings or simply as one spacing. Because the distribution~\eqref{eq::app_tsp_mono_DistributionSpacings} is independent of $i_1$, $\dots$, $i_k$, these terms can be grouped together on the base of their topology on the line with a given multiplicity. All these terms have a weight that can be written as
\begin{equation}\label{eq::app_tsp_mono_GeneralTerm}
\int_{0}^{1} d \chi_1 \, d \chi_2 \; \chi_1^p \, \chi_2^p \, \rho(\chi_1,\chi_2)
\end{equation}
where $\rho$ is a joint distribution of $\chi_1$ and $\chi_2$. Depending on the term in the square of Eq.~\eqref{eq::app_tsp_mono_OptimalCost} one is taking into account, the distribution $\rho$ takes different forms, but it can always be expressed as in function of the distribution Eq.~\eqref{eq::app_tsp_mono_DistributionSpacings}. As an example, we show how to calculate $\overline{|r_{3}-r_1|^{p}|r_{4}-r_2|^{p}}$. In this case $\rho(\chi_1, \chi_2)$ takes the form
\begin{multline}
\rho(\chi_1, \chi_2) = \int d\varphi_1 \, d\varphi_2 \, d\varphi_3 \, \rho^{(3)}_N(\varphi_1, \varphi_2, \varphi_3) \delta\left( \chi_1 - \varphi_1 - \varphi_2 \right) \delta\left( \chi_2 - \varphi_2 - \varphi_3 \right) \\
= N(N-1) \left[ (1-\chi_1)^{N-2}\theta(\chi_1) \theta(\chi_2- \chi_1) \theta(1-\chi_2) \right. \\
+ (1-\chi_2)^{N-2}\theta(\chi_2) \theta(\chi_1- \chi_2) \theta(1-\chi_1) \\
\left. - (1-\chi_1-\chi_2)^{N-2} \theta(\chi_1) \theta(\chi_2) \theta(1-\chi_1-\chi_2)  \right]\,,
\end{multline}
that, plugged into Eq.~\eqref{eq::app_tsp_mono_GeneralTerm} gives
\begin{equation}
\overline{|r_{3}-r_1|^{p}|r_{4}-r_2|^{p}} = \frac{\Gamma(N+1) \left[\Gamma (2p+3)- \Gamma (p+2)^2\right]}{(p+1)^2 \Gamma (N+2 p+1)} \,.
\end{equation}
All the other terms contained can be calculated the same way; in particular there are 7 different topological configurations that contribute. After having counted how many times each configuration appears in $(E_N[h^*])^2$, the final expression that one gets is
\begin{multline}\label{eq::app_tsp_mono_SecondMoment}
\overline{(E_N[h^*])^2}=\frac{\Gamma(N+1)}{\Gamma (N+2 p+1)} \Biggl[ 4(N-3) \Gamma (p+2) \Gamma (p+1) \Biggr. \\
+ \left((N-4) (N-3) (p+1)^2-2 N+8\right) \Gamma(p+1)^2 +\\
\left.+\frac{[N (2 p+1)(p+5)-4 p (p+5)-8]\, \Gamma (2 p+1)}{(p+1)} \right] \,.
\end{multline}

\section{2-factor problem and the plastic constant}\label{app::2factor}
\subsection{The Padovan numbers}
According to the discussion in Sec.~\ref{sec::2factor}, in the optimal 2-factor configuration of the complete bipartite graph there are only loops of length 2 and 3. Here we will count the number of possible optimal solutions for each value of $N$.
Let $f_N$ be the number of ways in which the integer $N$ can be written as a sum in which the addenda are only 2 and 3. For example, $f_4 = 1$ because $N=4$ can be written only as $2+2$, but $f_5=2$ because $N=5$ can be written as $2+3$ and $3+2$. We simply get the recursion relation
\begin{equation}
f_N = f_{N-2} + f_{N-3} \label{app::2f_rec}
\end{equation}
with the initial conditions $f_2 = f_3 = f_4 = 1$. The $N$-th {\em Padovan number} $\Pad(N)$ is defined as $f_{N+2}$. Therefore it satisfies the same recursion relation Eq.~\eqref{app::2f_rec} but with the initial conditions $\Pad(0) = \Pad(1) =\Pad(2) =1$.

A generic solution of Eq.~\eqref{app::2f_rec} can be written in terms of the roots of the equation
\begin{equation}
x^3 = x+1 \,.
\end{equation}
There is one real root 
\begin{equation} \label{app::2f_plastic}
\plas = \frac{(9 + \sqrt{69})^\frac{1}{3} + (9 - \sqrt{69})^\frac{1}{3}}{ 2^\frac{1}{3} 3^\frac{2}{3}} \approx 1.324717957244746\dots 
\end{equation}
known as the {\em plastic} constant and two complex conjugates roots
\begin{equation}
\begin{aligned}
	z_\pm & = \frac{ ( -1 \pm  i\, \sqrt{3}) (9 + \sqrt{69})^\frac{1}{3} +  ( -1 \mp  i\, \sqrt{3})(9 - \sqrt{69})^\frac{1}{3}}{ 2^\frac{4}{3} 3^\frac{2}{3}} \\
	& \approx -0.662359\,\text{\dots} \pm i \, 0.56228\dots 
\end{aligned}
\end{equation}
of modulus less than unity. Therefore
\begin{equation}
\Pad(N) = a\, \plas^N + b \, z_+^N + b^* \, z_-^N
\end{equation}
and by imposing the initial conditions we get
\begin{equation}
\Pad(N) = \frac{(z_+-1)(z_--1)}{(\plas-z_+)(\plas-z_-)}\, \plas^N + \frac{(\plas-1)(z_--1)}{(z_+-\plas)(z_+-z_-)}\, z_+^N +\frac{(\plas-1)(z_+-1)}{(z_--\plas)(z_--z_+)}\, z_-^N\,.
\end{equation}
For large $N$ we get
\begin{equation}
\Pad(N)  \sim \lambda \, \plas^N \label{asym}
\end{equation}
with $\lambda\approx 0.722124\dots$ the real solution of the cubic equation
\begin{equation}
23 \,t^3 - 23\, t^2 + 6\, t -1 = 0 \, .
\end{equation}
In Fig.~\ref{app::2f_fig1} we plot the Padovan sequence for a range of values of $N$ and its asymptotic expression.
\begin{figure}[h!]
	\centering
	\includegraphics[width=0.7\columnwidth]{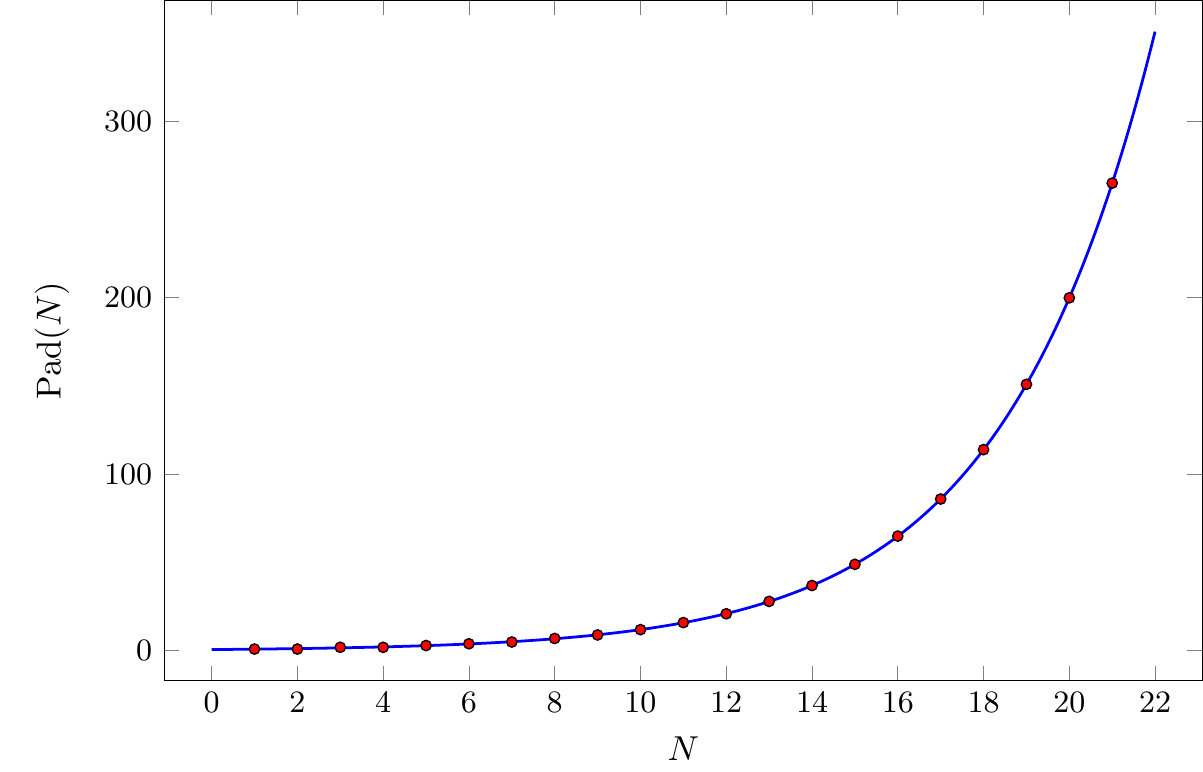} 
	\caption{Padovan numbers and their asymptotic expansion.} \label{app::2f_fig1}
\end{figure}

There is a relation between the Padovan numbers and the Binomial coefficients.
If we consider $k$ addenda equal to 3 and $s$ addenda equal to 2, there are $\binom{k+s}{k}=\binom{k+s}{s}$ possible different orderings. If we fix $N= 3 \, k + 2\, s$ we easily get that
\begin{equation}
\Pad(N-2) = \sum_{k\ge 0} \sum_{s\ge 0} \delta_{N, 3 \, k + 2\, s } \, \binom{k+s}{k} = \sum_{m\ge 0} \sum_{k\ge 0} \delta_{N,  k + 2\, m } \, \binom{m}{k} \,.
\end{equation}

\subsection{The recursion on the complete graph} 
A recursion relation analogous to Eq.~\eqref{app::2f_rec} can be derived for the number of possible solution of the 2-factor problem on the complete graph $\mathcal{K}_N$. Let $g_N$ be the number of ways in which the integer $N$ can be expressed as a sum of 3, 4 and 5. Then $g_N$ satisfies the recursion relation given by
\begin{equation}
g_N = g_{N-3} + g_{N-4} + g_{N-5} \,,
\end{equation}
with the initial conditions $g_3=g_4=g_5=g_6=1$ and $g_7 = 2$. The solution of this recursion relation can be written in function of the roots of the 5-th order polynomial
\begin{equation}
x^5-x^2-x-1=0 \,.
\end{equation}
This polynomial can be written as $(x^2+1)(x^3-x-1)=0$. Therefore the roots will be the same of the complete bipartite case ($\plas$, and $z_{\pm}$) and in addition 
\begin{equation}
y_{\pm} = \pm i \,.
\end{equation}
$g_N$ can be written as
\begin{equation}
g_N = \alpha_1 \plas^N + \alpha_2 z_+^N + \alpha_3 z_-^N + \alpha_4 y_+^N + \alpha_5 y_-^N \,,
\end{equation}
where the constants $\alpha_1$, $\alpha_2$, $\alpha_3$, $\alpha_4$, and $\alpha_5$ are fixed by the initial conditions $g_3=g_4=g_5=g_6=1$ and $g_7 = 2$. When $N$ is large the dominant contribution comes from the plastic constant
\begin{equation}
g_N \simeq \alpha_1 \plas^N \,.
\end{equation}
with $\alpha_1 \approx 0.262126..$.

\subsection{The plastic constant}  \label{app::plastic_constant}
In 1928,  shortly after abandoning his architectural studies and becoming a novice monk of the Benedictine Order,
Hans~van~der~Laan discovered a new, unique system of architectural proportions. Its construction is completely based on a single irrational value which he called the plastic number (also known as the plastic constant)~\cite{Maroni2012}.
This number was originally studied in 1924 by a French engineer, G. Cordonnier, when he was just 17 years old, calling it "radiant number". However, Hans van der Laan was the first who explained how it relates to the human perception of differences in size between three-dimensional objects and demonstrated his discovery in (architectural) design. His main premise was that the plastic number ratio is truly aesthetic in the original Greek sense, i.e. that its concern is not beauty but clarity of perception~\cite{Padovan2002}. 
The word plastic was not intended, therefore, to refer to a specific substance, but rather in its adjectival sense, meaning something that can be given a three-dimensional shape~\cite{Padovan2002}. 
The golden ratio or divine proportion
\begin{equation}
\phi = \frac{1+\sqrt{5}}{2} \approx 1.6180339887 \,,
\end{equation}
which is a solution of the equation
\begin{equation}
x^2 = x+1 \label{app::2f_qe} \,,
\end{equation}
has been studied by Euclid, for example for its appearance in the regular pentagon, and has been used to analyze the most aestetich proportions in the arts.
For example, the golden rectangle, of size $(a+b)\times a$ which may be cut into a square of size $a \times a$ and a smaller rectangle of size $b \times a$ with the same aspect ratio
\begin{equation}
\frac{a+b}{a} = \frac{a}{b} = \phi \, .
\end{equation}
This amounts to the subdivision of the interval $AB$ of length $a+b$ into $AC$ of length $a$ and $BC$ of length $b$. By fixing $a+b=1$ we get
\begin{equation}
\frac{1}{a} = \frac{a}{1-a} = \phi \,,
\end{equation}
which implies that $\phi$ is the solution of Eq.~\eqref{app::2f_qe}. The segments $AC$ and $BC$, of length, respectively $\frac{1}{\phi^2}(\phi, 1)$ are sides of a golden rectangle.

But the golden ratio fails to generate harmonious relations within and between three-dimensional objects. Van~der~Laan therefore elevates definition of the golden rectangle in terms of space dimension.
Van~der~Laan breaks segment $AB$  in a similar manner, but in three parts. If C and D are points of subdivision, plastic number $\plas$  is defined with
\begin{equation}
\frac{ AB }{ AD } = \frac{AD}{BC} = \frac{BC}{ AC} = \frac {AC}{ CD} = \frac{CD}{BD} = \plas
\end{equation} 
and by fixing $AB=1$, from $AC = 1 -BC$, $BD = 1 - AD$ we get
\begin{equation}
\plas^3 = \plas +1 \,.
\end{equation}
The segments $AC$, $CD$ and $BD$, of length, respectively, $\frac{1}{(\plas+1)\plas^2} (\plas^2, \plas, 1)$ can be interpreted as sides of a cuboid analogous to the golden rectangle.

\chapter{Supplemental material to Chapter \ref{chap::third}}
\chaptermark{Supplemental material to Chapter 4}
\section{Time-to-solution}\label{app::parameter_tts}
The time to solution (TTS) is a widely accepted empiric measure of algorithmic performances. It is defined as the time needed to solve an instance of a problem with high probability (here we take the 99\%).
In particular, given the probability $p(t)$ of solving the instance in time $t$, the TTS is given by
\begin{equation}\label{tts_definition}
\mathrm{TTS}(t) = t \frac{\log(0.01)}{\log(1-p(t))}.
\end{equation}
One is usually interested in the minimum TTS, given by
\begin{equation}
\mathrm{TTS} = \min_t \mathrm{TTS}(t).
\end{equation}

When we want to test our algorithm on a set of instances $\mathcal{I}$ and a probability distribution $p$ is defined on such a set, the measure of performance can be the average
\begin{equation}\label{mean_tts_1}
\langle \mathrm{TTS} \rangle = \sum_{I \in \mathcal{I}} p(I) \, \mathrm{TTS}_I,
\end{equation}
where $\mathrm{TTS}_I$ is the $TTS$ for the instance $I$. This average is typically computed as an empirical average on a large number of instances generated with probability $p$.
For some problems and some algorithms, there are instances that are never solved for reasonable running time. What TTS should one use in Eq.~\eqref{mean_tts_1} for them? To avoid this problem, it is often used, instead of the average TTS, the 50-percentile of the TTSs computed for a large set of instance. Typically (and this is the approach used throughout Sec.~\ref{sec::qaa_parset}) this 50-percentile is shown together with the 35-percentile and 65-percentile.

\section{Hamming weight example}\label{app::parameter_toy}
Here we analyze in detail the annealing (both classical and quantum) of a toy problem, to give a concrete example of the effect of choosing the penalty-term parameter.
Consider a cost function defined on $x \in \{1,0\}^N$ with the symmetry $E(x)=E(\sigma(x))$ for each $\sigma\in S_N$ permutation of $N$ objects. This kind of cost functions characterize the so-called Hamming weight problems, since the only thing they can depend on is the Hamming weight (i.e. number of 1) of the configurations.\\
These problems have been extensively used to explore the properties of thermal and quantum annealing, mainly because of their simplicity: their high level of symmetry often allows for exact computations, and the specific form of the cost function can be chosen such that the required annealing time is either polynomial or exponential (or even exponential for classical thermal annealing and polynomial for the quantum version).\\
Here we introduce a constrained version of the problem, with cost function
\begin{equation}\label{eq::para_hamming_ham}
E(x) = \frac{1}{N} (W(x) - N/3)^2,
\end{equation}
where $W(x)$ is the Hamming weight of the configuration $x$. The normalization is chosen to make the cost function an extensive quantity. Indeed, if we define the intensive Hamming weight as $w(x) = W(x)/N$, we have the density-of-cost function
\begin{equation}
e(x) = E(x) / N = (w(x) - 1/3)^2.
\end{equation}
Let us suppose that we have the following constraint: only configurations with density of cost in $\left[0, 1/4 \right] \cup \left[1/2, 1\right]$ are acceptable. 
To implement this constraint, we consider the penalty term (that we write directly as function of the intensive Hamming weight $w$)
\begin{equation}
p(w) = \left\{ 
\begin{array}{ll}
\left(1-4\left( w - \frac{1}{4}\right) \right) \left( w - \frac{1}{4}\right)-4\left( w - \frac{1}{4}\right)\left( w - \frac{1}{2}\right) & \textrm{if $w\in\left[\frac{1}{4},\frac{1}{2}\right]$,}\\
0 & \textrm{otherwise,}\\
\end{array} 
\right.
\end{equation}
where the non-zero term is simply a linear interpolation between $x-1/4$ (linear cost increasing as we break the constraint of having $x>1/4$) and $-x+1/2$ (linear cost increasing as we break the constraint of having $x<1/2$). This is one of the many possible choices of a suitable penalty term for this problem.
Therefore the total cost function to minimize is
\begin{equation}\label{eq::para_hamming_tot}
e_{tot}(w;\lambda) = e(w) + \lambda \, p(w), 
\end{equation}
and our goal is to find the minimum and the optimum value of $\lambda$. 
Notice that this cost function is not given as a local Hamiltonian, and in particular it is not in QUBO form; this is not relevant for our discussion here, since we are only interested in understanding in a simple example the role of the coupling parameter for the penalty term. \\
One can consider both SA and QAA to solve this problem: in both case, as shown in Appendix \ref{app::parameter_toy}, if a too high penalty term $\lambda$ is chosen the system remains trapped an exponentially long time in a local minimum.

\subsection{Classical annealing}
In the simulated annealing algorithm, the probability of a configuration is its free energy $F(W; \lambda)$, defined at temperature $\beta$ by
\begin{equation}
\exp \left(-\beta F(W; \lambda) \right) = \binom{N}{W} \exp \left( - \beta E_{tot}(\lambda) \right).
\end{equation}
Expanding the binomial for large $N$, and keeping only the dominant term, we obtain the following density of free energy
\begin{equation}\label{f_toy}
f(w, \lambda) = \frac{1}{\beta}\left[ (1-w)\log(1-w) + w \log w \right] + e_{tot}(w;\lambda).  
\end{equation}
\begin{figure}[b]
	\centering
	\includegraphics[width=\columnwidth, height=\textheight, 
	keepaspectratio]{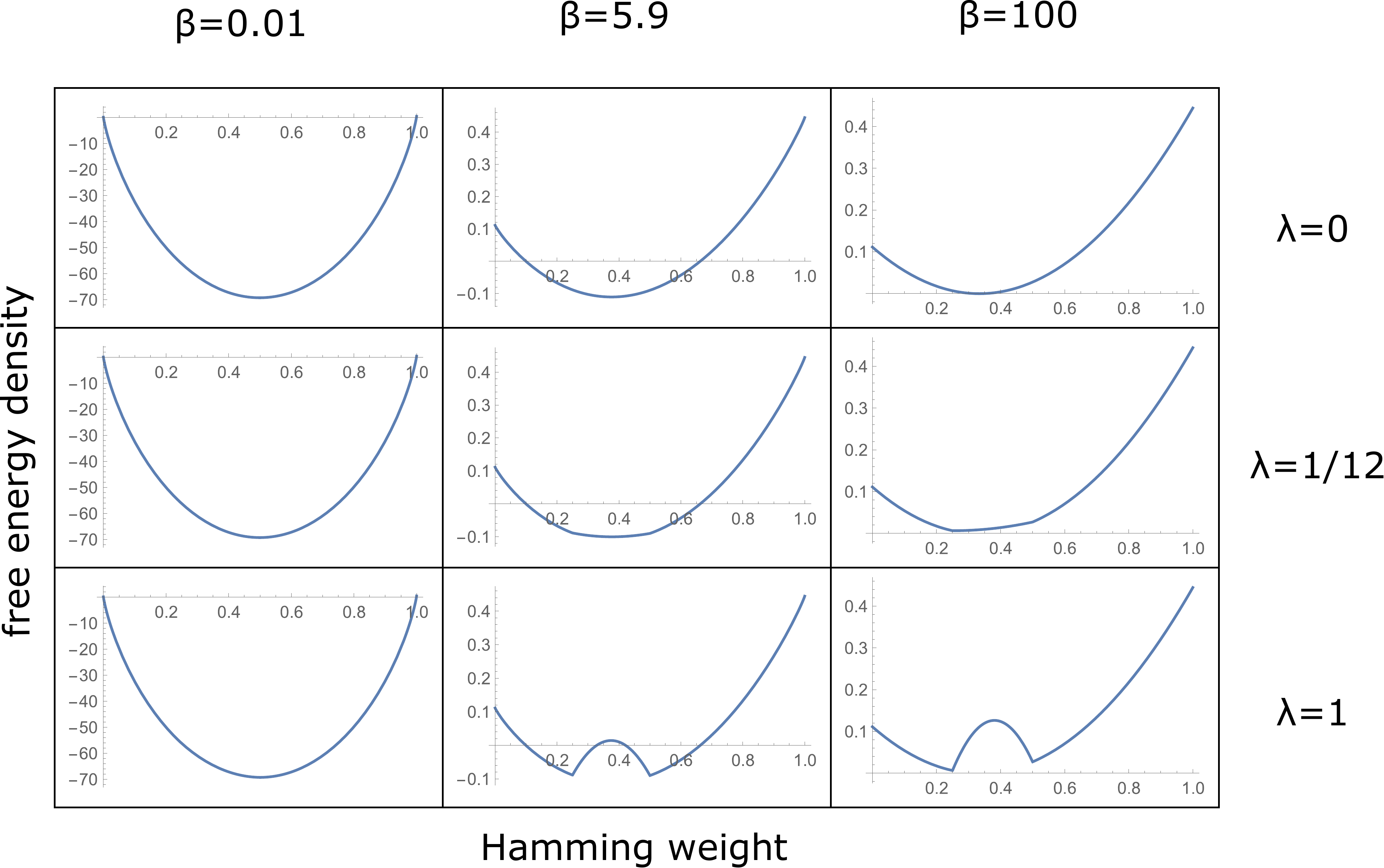} 
	\caption{Free energy landscape for the Hamming weight problem defined in Sec.~\ref{sec::qaa_parset}, for various values of inverse temperature $\beta$ and of the penalty term parameter $\lambda$. In particular, the second column shows to the temperature in which the role of the local and global minimum is exchanged, and the second row shows what happens if we use the minimum value for $\lambda$. Notice that there is no local minimum in this last case.} \label{fig_toy}
\end{figure}

Fig.~\ref{fig_toy} shows the different shapes of the free energy density, varying $\lambda$ and $\beta$. Notice that the probability of configurations which are not corresponding to the minimum of Eq.~\eqref{f_toy} is $\sim \exp(- N \Delta)$, where $\Delta$ is the free energy density of such configurations minus the minimum free energy density.
Therefore the (local) SA algorithm takes exponential time to leave the minimum of $f$, and for $\lambda$ too large (such as $\lambda=1$ in Fig.~\ref{fig_toy}) at $\beta \simeq 5.9$ the SA needs exponential time to pass from the minimum at larger $w$ to the one at smaller $w$, which will become the global minimum. If a lower $\lambda$ is chosen, such as $\lambda=1/12$ (which is the minimum), the previous situation never happens, consisting in an annealing that can proceed in polynomial time. Notice that if $\lambda$ is further decreased the final minimum will be in the forbidden interval, and again we will need to wait an exponential time to reach the minimum acceptable configuration (because it is not a local minimum anymore).

\subsection{Quantum annealing}
For the quantum case, we will consider the following annealing procedure:
\begin{itemize}
	\item the quantum problem Hamiltonian is defined by its action on the computational basis, thus starting from Eq.~\eqref{eq::para_hamming_tot}; 
	\item the quantum driving term of the Hamiltonian (the one which provides quantum fluctuations) is $\sum_i \sigma_i^x$, so the system is initialized in the ground state of this term and the problem Hamiltonian is slowly turned on, while this driving term is slowly turned off.
\end{itemize}

Because of the symmetry of the problem, we can define a semi-classical potential and suppose that its minimum describes the instantaneous ground state of the system during the quantum annealing schedule. The idea is that, because we initialize the system in the factorized superposition state 
\begin{equation}
\ket{+}_N = \otimes_i \frac{\ket{0}_i + \ket{1}_i}{\sqrt{2}}
\end{equation}
and the Hamiltonian is symmetric with respect to qubit permutations, we suppose that all the evolution takes place in the subspace of the Hilbert space spanned by symmetric factorized states of the form\footnote{notice that there are also other possible states, which are entangled; however, this form of quasi-classical potential has been profitably used in many examples (see \cite{Muthukrishnan2016}), so we use it also here.}
\begin{equation}
\ket{\theta} = \bigotimes_{i=1,\dots,N} \left(\cos\theta \ket{0} + \sin\theta\ket{1} \right).
\end{equation}
The semi-classical potential we use is 	$\bra{\theta} H_{tot} \ket{\theta}$,
with
\begin{equation}
H_{tot} = s (H_0 + \lambda H_P) - (1-s) \sum_i \sigma_i^x,
\end{equation}
where $s=s(t)$ is the parameter which defines the annealing schedule, with $s(0)=0$, $s(T)=1$ and $T$ is the total annealing time; $H_0$ and $H_P$ are defined by their action on the computational basis. Therefore, one has
\begin{equation}
\begin{split}
\bra{\theta} H_0 \ket{\theta} & = \sum_{a,b} \braket{\theta}{a} \braket{b}{\theta} \bra{a}H_0\ket{b} \\
& = \sum_{W=0}^N \binom{N}{W} (\sin^2\theta)^W (\cos^2\theta)^{N-W} \left(\frac{1}{N} \left(W-\frac{N}{3}\right)^2\right)\\
& = N \left( \sin^2\theta - \frac{1}{3} \right).
\end{split}
\end{equation}
Notice that one can naturally introduce the ``Hamming-weight operator'' as $\sum_i \frac{1+\sigma_i^z}{2}$ and
\begin{equation}
\bra{\theta} \sum_i \frac{1+\sigma_i^z}{2} \ket{\theta} = N \sin^2 \theta,
\end{equation}
therefore the semi-classical potential is identical to the classical one, where the Hamming weight becomes the expectation value of the Hamming-weight operator.
It is slightly more tricky to deal with the penalty term (we need to take into account also sub-leading terms of the Stirling approximation):
\begin{equation}\label{pen_ham_quantum}
\begin{split}
\bra{\theta} H_P \ket{\theta} & = \sum_{W = N/4}^{N/2} \binom{N}{W} (\sin^2\theta)^W (\cos^2\theta)^{N-W} N p(w)\\
& \sim N^{3/2} \int_{\frac{1}{4}}^{\frac{1}{2}} \mathrm{d}w \, e^{-N g(w, \theta)} \frac{p(w)}{\sqrt{w (1-w)}},
\end{split}
\end{equation}
where
\begin{equation}
g(w, \theta) = w \log \left(w/\sin^2\theta\right) + (1-w)\log\left((1-w)/\cos^2\theta\right).
\end{equation}
Now the integral in Eq.~\eqref{pen_ham_quantum} is done by saddle-point method, where the solution of the saddle point equations is (using the fact that $1/4 \leq w \leq 1/2$)
\begin{equation}
w=\left\{ 
\begin{split}
& \frac{1}{1+\cot(\theta)^2} & \quad \frac{\pi}{6} \leq \theta \leq \frac{\pi}{4} \textrm{ or }  \frac{3\pi}{4} \leq \theta \leq \frac{\pi}{6},\\
& 1/4 & 0 <\theta<\frac{\pi}{6} \textrm{ or } \frac{\pi}{6}<\theta<\pi,\\
& 1/2 & \frac{\pi}{4} <\theta <\frac{3\pi}{4}.\\
\end{split} 
\right.
\end{equation}
Notice that when $w=1/4$ or $w=1/2$ the value of Eq. \eqref{pen_ham_quantum} is exponentially suppressed. However, $g(1/(1+\cot(\theta),\theta)=0$, therefore we have
\begin{equation}
\bra{\theta} H_P \ket{\theta} \sim \left\{
\begin{split}
& N \sqrt{2 \pi} \, \cos(2\theta) \, (1-2\cos(2\theta)) & \quad \frac{\pi}{6} \leq \theta \leq \frac{\pi}{4} \textrm{ or }  \frac{3\pi}{4} \leq \theta \leq \frac{\pi}{6},\\
& 0 & \textrm{otherwise.}
\end{split}
\right.
\end{equation}
The last term, that is the one which provides quantum fluctuation, is
\begin{equation}
\bra{\theta} \sum_i\sigma_i^x \ket{\theta} = N \sin(2\theta).
\end{equation}
Putting all the terms together one can easily see that considerations analogue to those done for the classical case hold also here.

\section{Failure of our method for an instance of the minor embedding problem}\label{app::parameter_failure}

The results obtained for the matching problem raise the question if a similar analysis can be extended to all constrained problems. However, this is not the case. Remember that Eq.~\eqref{eq::para_general_eq_for_par} can be used to obtain efficiently an estimate for the minimum parameter if the chain of inequalities~\eqref{eq::para_general_order} is true (this is a necessary but not sufficient condition: we also need to be able to approximate or solve efficiently the problem under analysis). But there are problems, such as the minor embedding problem, where these inequalities are false.\\
To show that,	 we briefly introduce the minor embedding problem in the formulation that is relevant for us. Then we will choose a specific instance of the problem where one can explicitly see that the inequalities~\eqref{eq::para_general_order} are false.

When a problem is written in QUBO form, an underlying weighted graph can be defined looking at the couplings $J_{i,j}$ in the Hamiltonian: each qubit is associated to a vertex of the graph, and an edge of the graph is present between qubits $i$ and $j$ if $J_{i,j} \neq 0 $. This graph is of great importance, because in real quantum annealing devices there is an effective hardware graph with qubits as vertices, and qubits can interact only if they correspond to two connected vertices in this hardware graph. If the former graph (``problem graph'') and the latter (``hardware graph'') are different, an extra-step is need: the minor embedding.
In the minor embedding problem, we have a QUBO problem defined on a graph $\mathcal{G}$ and we want to embed $\mathcal{G}$ in another graph (which is typically a fixed hardware graph) $\mathcal{U}$, such that we have a QUBO problem on the graph $\mathcal{U}$ whose ground state corresponds through a known map to the ground state of our original problem.
To do so we define a function $\phi: u \to g$, where $g$ and $u$ are the vertices of respectively $\mathcal{G}$ and $\mathcal{U}$, such that if we contract all those vertices in $u$ which are sent by $\phi$ to the same vertex in $g$, we obtain from the graph $\mathcal{U}$ the graph $\mathcal{G}$. In other words, the function $\phi$ defines subsets of $u$ which correspond to the same vertex in $g$. These subsets of spins are often called ``chains''. Then the hardware graph $\mathcal{U}$ can be used for the QAA, with problem Hamiltonian
\begin{equation}\label{minor}
H' = H + J \sum_{\substack{i, j \in u \\ \phi(i) = \phi(j)}} \sigma_i \sigma_j,
\end{equation}
where $H$ is the Hamiltonian of the original problem, where the interaction among two vertices $a$ and $b$ connected in the graph $\mathcal{G}$ is now between two qubits of $u$, $k$ and $\ell$, which are connected in $u$ and such that $\phi(k) = a$ and $\phi(\ell) =b$. \\
The minor embedding problem in general consists in finding a suitable function $\phi$. Here we take another point of view: given a suitable $\phi$, we are interested in finding the minimum value for the parameter $J$ in Eq.~\eqref{minor}. The term with coupling $J$ is a kind of penalty term, whose contribution is minimum when all the spins inside the same chain have equal sign. Notice that only in this situation is possible to go back from the solution of the problem on the hardware graph to the original graph $\mathcal{G}$. Therefore the search for the ground state of $H'$ is a constrained problem, where the only acceptable configurations are those with chains composed of spins with the same sign. This is the problem we are interested in, and that we would like to address with the technique developed in Sec.~\ref{sec::qaa_parset} and used for the matching problem. 

Let us consider a specific (and trivial) example in which the order condition given in~\eqref{eq::para_general_order} is not respected. The starting problem graph and the hardware graph are those given in Fig. \ref{fig_min_embed_prob}, where the couplings in the starting problems are $\pm 1$: the black continuous edges corresponds to $-1$ (ferromagnetic) interactions and the dashed edge correspond to $+1$ (antiferromagnetic) interaction. 
\begin{figure}[htb]
	\begin{subfigure}[t]{0.34\linewidth}
		\centering
		\includegraphics[width=0.95\columnwidth,keepaspectratio, height=0.15\textheight]{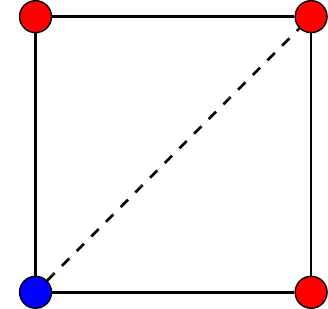}
	\end{subfigure} \hfill
	\begin{subfigure}[t]{0.64\linewidth}
		\centering
		\includegraphics[width=0.95\columnwidth,keepaspectratio,height=0.15\textheight]{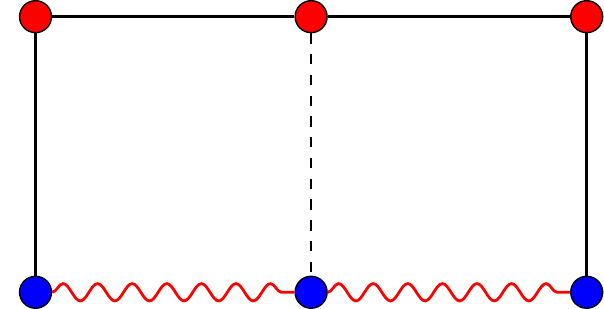}
	\end{subfigure}
	\caption{Example of minor embedding: on the left there is the problem graph, on the right the hardware graph. Each vertex is a spin, the black continuous edges are ferromagnetic couplings, the black dashed edge is a antiferromagnetic couplings and the red wavy edges are the couplings used for the minor embedding. Therefore the blue spin in the problem graph corresponds to the three blue spins in the hardware graph.}
	\label{fig_min_embed_prob}
\end{figure}
Here a constraint breaking is a ``kink'' (two consecutive spins with different sign) in the chain with wavy links, and the analogous of Eq.~\eqref{eq::para_general_eq_for_par} is 
\begin{equation}
\lambda > \max_{k\in\{ 1,2 \}} \frac{\mathcal{E}_0 - \mathcal{E}_k}{k} = \max_{k\in\{ 1,2 \}} \lambda_k,
\end{equation}
where now $\mathcal{E}_k$ is the energy of the problem on the graph without wavy lines, when $k$ kinks are permitted on the wavy lines. Therefore it is easy to check that $\lambda_1=0$ and $\lambda_2=1/2$, therefore $\lambda_2>\lambda_1$ and the inequality given in~\eqref{eq::para_general_order} is not fulfilled. 

This allows us to observe that there are problems (as the matching problems) in which the order relation~\eqref{eq::para_general_order} holds and our method can be efficiently use to estimate the minimum parameter, while in other problems (such as the minor embedding problem) this is not true. This brings in turn an interesting consequence: whatever method one wants to use to decide the value of the parameter, a value that prevents the breaking of (only) one constraint is in general a too weak condition. Indeed, one has to prevent the breaking of any number of constraints.

\section{Other methods to find the minimum penalty term weight}\label{app::parameter_other}
\sectionmark{Other methods}
To the best of our knowledge, there are two ways to choose parameters for penalty terms: one is the ``trial-and-error'' method, the other is the one described by Choi in \cite{Choi2008}.
The first method consists basically in trying many different values and solving the problem with these values to see if the constraints are broken in the ground state. 
A limitation of this method is that, even if the problem can be solved efficiently (which is not the case for hard computational problems of relevant size), one cannot be sure to have found the real minimum parameter if the number of attempts with different parameters is small. 
On the other hand, the strength of this method is that it can be run using the same heuristic algorithm used to solve the problem (so there is no need in principle for additional problem-dependent algorithms). However, especially for large size instances of hard problems, since the heuristic method fails often, to have a reasonable precision on the parameter the algorithms has to be run many times, so it becomes more and more inefficient as the system size increases. Moreover, as we have discussed in the main text, for many constrained problems it is reasonable to expect that even a small error in the parameter setting cause a slowdown which is more and more relevant as the system size increases.

The second method consists in pre-processing the instance and choosing a penalty-term weight high enough to ensure the constraints. We will not review that method in general, but we will discuss how to apply it to the matching problem in the next section. Here we will only highlight the main differences with our method:
\begin{itemize}
	\item it is built to work with the minor embedding problem, but the same idea can easily be applied to other problems; however, always in \cite{Choi2008} the author refines the method in a way that only applies to the minor embedding problem, so we will not discuss that refinement here;
	\item it has a different parameter for each constraint of the problem, and these parameters are individually tuned;
	\item it uses no information about the solution (even approximate) of the instance. 
\end{itemize}
To conclude, this method could be applied to problems where the algorithm we described in Sec.~\ref{sec::qaa_parset} fails (that is, when condition~\eqref{eq::para_general_order} is not true or the problem cannot be approximated in an acceptable way), but the results are often quite far from the real minimum value of the parameters. 

\subsection*{Choi's method applied to the matching problem}\label{sec_examples_choi}

Another interesting feature of the matching problem is that here we can quantify the how good the Choi's upper bound for the minimum parameters is.
Choi's method can be applied to Hamiltonians of the form
\begin{equation}
H = H_0 + \sum_i \mu_i H_P^{(i)},
\end{equation}
where $H_P^{(i)}$ enforces the local $i$-th constraint. The method consists in choosing the values of the $\mu_i$ singularly, in such a way that the $i$-th constraint is never broken, irrespectively of the solution of the specific instance.
More concretely, for the matching problem one would have the following term to ensure that one and only one edge connects the vertex $\nu$ to another vertex:
\begin{equation}
H_P^{(\nu)} = \left( 1 - \sum_{e \in \partial\nu} x_e \right)^2
\end{equation}
and to be sure that independently on the solution of the instance this constraint is not broken, one has to choose a value for $\mu_\nu$ such that
\begin{equation}
\mu_\nu > \max_{e \in \partial\nu} w_e/2,
\end{equation}
where the factor 1/2 is because we have a contribution from two of penalty terms each time we break a constraint.\\
Consider now a specific example: the Euclidean matching in one dimension. In this case, for each instance $2N$ points are randomly thrown on a segment of length 1. The graph of the problem is a complete graph where each vertex corresponds to a point on the segment, and the link weights are the distances on the segment between the two points which correspond to the link endpoints. 
Since each vertex corresponds to a point in one dimension, we can order the points and it can be seen that the distance between the first and last point on the typical instance is going to 1 with $N$. 
Therefore
\begin{equation}
\langle \mu_i \rangle \sim \max \left( \frac{i}{2N+1}, 1 - \frac{i}{2N+1} \right),
\end{equation}
where $\lambda_i$ is the coupling for the $i-$th point once points have been ordered and the angled brackets denotes average on the disorder. Summing all the $\lambda_i$s we obtain
\begin{equation}\label{scaling_choi}
\langle \sum_i \mu_i \rangle \sim N \int_{0}^{1/2} \mathrm{d}x \, (1-x) + \int_{1/2}^{1} \mathrm{d}x \, x = \frac{3 N}{4}.
\end{equation}

On the opposite, the minimum parameter given by Eq.~\eqref{eq::para_min_param_1} is going to zero with $N$. Indeed for this very simple problem $\langle E_0^{(0)} \rangle = O(1)$ (that is, of the same order of the length of the segment) and by removing a couple of points we cannot change this limiting behavior, therefore $\lim_{N\to\infty} \lambda_1 = 0$ and so the sum of $N$ of those parameters is scaling differently from Eq.~\eqref{scaling_choi}, and it is definitely lower than that. In particular, from numerical simulations we see that the minimum parameter $\lambda_1 = O(1/N)$, therefore the sum of $N$ of these gives a constant, rather than going to infinity.

\section{Proof of the inequality \eqref{eq::para_gainorder} for the matching problem}\label{app::parameter_proof}
\sectionmark{Proof for the matching problem}
%
%
In this appendix we give the full proof of~\eqref{eq::para_general_gainorder} for the matching problem. We will use the notation introduced in Sec.~\ref{sec::qaa_parset}. As first step, we introduce the concept of signed path, which will be of used many times in the proof. Then we will proceed with the actual proof.
\paragraph{Definition: signed paths.}
Consider an instance of the problem, that is a given weighted graph with $2N$ vertexes.
Take $m \neq \ell \leq N$ and consider $\mathcal{E}_{\ell}$ and $\mathcal{E}_{m}$. In general, since $\ell\neq m$, the matchings of which $\mathcal{E}_{\ell}$ and $\mathcal{E}_{m}$ are the costs
(with a slight abuse of notation, from now on we will simply say ``the matchings $\mathcal{E}_{\ell}$ and $\mathcal{E}_{m}$'') can be done over two completely different sets of points. Indeed, if for example we have $m=\ell-1$, in $\mathcal{E}_{m}$ we are using 2 points less than those used in $\mathcal{E}_{\ell}$. But this does not necessarily mean that some of the points which are used in $\mathcal{E}_{\ell}$ are used also in $\mathcal{E}_{m}$, so the matching $\mathcal{E}_{m}$ can be completely different from $\mathcal{E}_{\ell}$.
Take a vertex used in $\mathcal{E}_{\ell}$ but not in $\mathcal{E}_{m}$, $x$. Consider the path on the instance graph which starts in $x$ and which is built by using links alternatively of $\mathcal{E}_{\ell}$ and $\mathcal{E}_{m}$. Let us call $y=y(x)$ the ending vertex of this path. We define the ``signed path'' $P_{\ell, m}(x, y)$ as the weight of that path
, which is obtained by summing the weight of each edge of the path used in $\mathcal{E}_{m}$ and by subtracting the weight of each edge of the path used in $\mathcal{E}_{\ell}$.

\paragraph{Proof, part I: stability.}
Let us denote with $\{x\}$ the set of vertexes used in $\mathcal{E}_{n}$, with $\{y\}$ those used in $\mathcal{E}_{n+1}$ but not in $\mathcal{E}_{n}$ and with $\{z\}$ those used in $\mathcal{E}_{n}$ but not in $\mathcal{E}_{n+1}$. 
We prove that $\{z\} = \varnothing$, using a \emph{reductio ad absurdum}.
To do so, we take one among the $z$'s, $z_\star$, and build the signed path $P_{n+1,n}(z_\star, w)$, where $w$ is the point from which we cannot proceed further with the path. Notice that by construction $w$ can only be another point of $\{z\}$ or one of $\{y\}$. These two cases require separate discussions.\\
\emph{Case 1:} consider that $P_{n+1, n}(z_\star, w) = P_{n+1, n}(z_\star, y)$, with $y = y(z_\star)\in\{y\}$. Since $P_{n+1, n}(z_\star, y)$ starts with a link of $\mathcal{E}_{n}$ and the last link is of $\mathcal{E}_{n+1}$, it has the same number of links of both the matchings. Moreover,
\begin{equation}
\mathcal{E}_{n} - P_{n+1, n}(z_\star, y)
\end{equation}
is again an acceptable matching of $2n$ points (although not the same points used in $\mathcal{E}_{n}$), so it has to be greater than or equal to $\mathcal{E}_{n}$ because $\mathcal{E}_{n}$ is the optimal matching of $2n$ points. So we have $P_{n+1, n}(z_\star, y) \geq 0$.
On the other side, also
\begin{equation}
\mathcal{E}_{n+1} + P_{n+1, n}(z_\star, y)
\end{equation}
is an acceptable matching of $2 (n + 1)$ points, which similarly leads to \\\mbox{$P_{n+1, n}(z_\star, y) \leq 0$}.
Therefore we have $P_{n+1, n}(z_\star, y) = 0$, which is the \emph{absurdum}. Notice that actually we can have paths equal to zero and so $\{z\}\neq \varnothing$ if there are ``compatible sub-matchings'' with the same cost. However this kind of degeneracy can be easily taken into account with a slight modification of our arguments (for simplicity we will consider here the non-degenerate case only).

\emph{Case 2:} now we consider $P_{n+1, n}(z_\star, w) = P_{n+1, n}(z_\star, z')$, with $z' = z'(z_\star)\in\{z\}$.
Take $y_1 \in \{y\}$ such that the signed path $\tilde{P}_{n, n+1}(y_1, y_2)$ ends in $y_2\in\{y\}$. A such point $y_1$ have to exist: indeed, a path starting from $y\in\{y\}$ can only end in another point of $\{y\}$ or a point of $\{z\}$. However, since  $\mathcal{E}_{n+1}$ has two points more than $\mathcal{E}_{n}$, the set $\{y\}$ has two more point that the set $\{z\}$, so at least one of the paths starting from points in $\{y\}$ has to finish in $\{y\}$.
As in the case 1, $P_{n+1, n}(z_\star, z') - \tilde{P}_{n, n+1}(y_1, y_2)$ has the same number of links of both the matchings and
\begin{equation}
\mathcal{E}_{n} - P_{n+1, n}(z_\star, z') + \tilde{P}_{n, n+1}(y_1, y_2)
\end{equation}
and 
\begin{equation}
\mathcal{E}_{n+1} + P_{n+1, n}(z_\star, z') - \tilde{P}_{n, n+1}(y_1, y_2)
\end{equation}
are acceptable matchings of, respectively, $n$ and $n+1$ points. So the proof proceeds as the previous case.

\paragraph{Proof, part II: order.}
We want to prove Equation~\eqref{eq::para_gainorder}. Let us denote with $\{x_1, x_2, x_3, x_4\}$ the four points of $\mathcal{E}_{n+1}$ which are not used in $\mathcal{E}_{n-1}$. Let $x_i$ be such that $P_{n-1, n+1}(x_i, x_j)$, that is the ending point of the signed path starting in $x_i$ is $x_j$. Two such points $x_i$ and $x_j$ have to exist for definition of path and because of the stability property (the path starting in $x_i$ cannot end somewhere else than in another of the $x$'s). Then
\begin{equation}
\mathcal{E}_{n+1} - P_{n-1, n+1}(x_i, x_j) \geq \mathcal{E}_{n},
\end{equation}
because this is an acceptable matching of $n$ points and $\mathcal{E}_{n}$ is the optimal among these matchings.
But also 
\begin{equation}
\mathcal{E}_{n-1} + P_{n-1, n+1}(x_i, x_j) \geq \mathcal{E}_{n},
\end{equation}
because this is an acceptable matching of $n$ points and $\mathcal{E}_{n}$ is the optimal among these matchings.
Therefore equation~\eqref{eq::para_gainorder} follows.

\bibliography{mybib}
\bibliographystyle{alpha}

\end{document}